%% file: main.tex
\documentclass[a4paper, twocolumn]{article}
\usepackage[fontsize=9pt]{scrextend}
\usepackage[twoside, columnsep=2pc,
     top=0.77in, bottom=1in, inner=0.77in, outer=0.77in,
     heightrounded]{geometry}
\usepackage[numbers, compress]{natbib}

\usepackage[utf8]{inputenc}
\usepackage[T1]{fontenc}    
\usepackage{amsmath,amsthm}
\usepackage{amssymb}
\usepackage{xspace}
\usepackage{mathtools}
\usepackage{xcolor}
\usepackage{hyperref}
\usepackage[nameinlink]{cleveref}
\usepackage{bbold}
\usepackage{xfrac}
\usepackage{nicefrac}       
\usepackage{microtype}      
\usepackage{caption}
\usepackage{subcaption}
\usepackage[ruled,vlined,linesnumbered]{algorithm2e}
\usepackage{wrapfig}
\usepackage{booktabs}
\usepackage{paralist}
\usepackage{enumitem}
\usepackage{cleveref}
\usepackage{authblk}
\usepackage{float}
\usepackage{graphicx}
\usepackage[
    lambda,
    advantage,
    sets,
    ff,
    adversary,
    operators,
    landau,
    probability,
    notions,    
    logic,
    mm,
    primitives,
    events,
    complexity,
    oracles,
    asymptotics,
    keys]{cryptocode}

\usepackage{tabularx}
\usepackage{pgfplots}
\usepackage{tikzscale}
\pgfplotsset{compat=1.17}
\usepackage{pifont}

\usepgfplotslibrary{external}
\tikzset{external/system call= {lualatex -save-size=80000 
                          -pool-size=10000000 
                          -extra-mem-top=500000000 
                          -extra-mem-bot=10000000 
                          -main-memory=90000000 
                          \tikzexternalcheckshellescape 
                          -interaction=batchmode
                          -jobname "\image" "\texsource"}}

\newcommand{\dpgbdt}{\ensuremath{\text{S-BDT}}\xspace}
\newcommand{\reftoarxiv}{For the complete proof we refer to the appendix.}

\date{}

    \newcommand*\samethanks[1][\value{footnote}]{\footnotemark[#1]}
    
    \author[1]{Thorsten Peinemann \thanks{{\normalfont The first two authors equally contributed to this work.}}}
    \author[1]{Moritz Kirschte \samethanks}
    \author[2]{Joshua Stock}
    \author[3]{Carlos Cotrini}
    \author[1]{Esfandiar Mohammadi}
    
    \affil[1]{Universität zu L\"ubeck, L\"ubeck, Germany}
    
    \affil[2]{Universität Hamburg, Hamburg, Germany}
    \affil[3]{ETH Zurich, Zurich, Switzerland} 
    \affil[1]{\textit{\{t.peinemann,\,m.kirschte,\,esfandiar.mohammadi\}@uni-luebeck.de}}
    \affil[2]{\textit{joshua.stock@uni-hamburg.de}}
    \affil[3]{\textit{ccarlos@inf.ethz.ch}}

\usepackage[]{footmisc}

\newcommand{\eps}{\varepsilon}
\renewcommand{\set}[1]{\bgroup\left\{\,#1\,\right\}\egroup\xspace}

\DeclareMathOperator*{\Eop}{\mathbb{E}}
\newcommand{\E}[1]{{\operatorname{\mathbb{E}}}\bgroup\left[#1\right]\egroup\xspace}
\newcommand{\Econd}[2]{{\operatorname{\mathbb{E}}_{#1}}\bgroup\left[#2\right]\egroup\xspace}

\DeclareMathOperator*{\diag}{diag}

\makeatletter
\newtheorem*{rep@theorem}{\rep@title}
\newcommand{\newreptheorem}[2]{%
\newenvironment{rep#1}[1]{%
 \def\rep@title{#2 \ref{##1}}%
 \begin{rep@theorem}}%
 {\end{rep@theorem}}}
\makeatother

\newtheorem{theorem}{Theorem}
\newreptheorem{theorem}{Theorem}
\newtheorem{lemma}[theorem]{Lemma}
\newreptheorem{lemma}{Lemma}
\newtheorem{corollary}[theorem]{Corollary}
\newreptheorem{corollary}{Corollary}
\theoremstyle{definition}
\newtheorem{definition}[theorem]{Definition}
\theoremstyle{remark}
\newtheorem{remark}[theorem]{Remark}
\crefname{algocf}{alg.}{algs.}
\Crefname{algocf}{Alg.}{Algs.}

\crefname{theorem}{Thm.}{Thms.}
\crefname{lemma}{Lem.}{Lems.}
\crefname{corollary}{Cor.}{Cors.}
\crefname{definition}{Def.}{Defs.}
\crefname{figure}{Fig.}{Figs.}
\crefname{appendix}{Appx.}{Appxs.}
\crefname{section}{Sec.}{Sec.}
\crefname{table}{Tbl.}{Tbls.}
\crefname{equation}{Eq.}{Eqs.}

\newcommand{\boldparagraph}[1]{\textbf{#1}\mbox{}\quad\xspace}
\newcommand{\blue}[1]{{\color{black}#1}}

\title{\dpgbdt: Distributed Differentially Private Boosted Decision Trees}

\begin{document}

\maketitle
\begin{abstract}
    \blue{We introduce \dpgbdt: a novel $(\varepsilon,\delta)$-differentially private distributed gradient boosted decision tree (GBDT)} learner that improves the protection of single training data points (privacy) while achieving meaningful learning goals, such as accuracy or regression error (utility). S-BDT uses less noise by relying on non-spherical multivariate Gaussian noise, for which we show tight subsampling bounds for privacy amplification and incorporate that into a Rényi filter for individual privacy accounting.
    We experimentally reach the same utility while saving $50\%$ in terms of epsilon for $\varepsilon \le 0.5$ on the Abalone regression dataset (dataset size $\approx 4K$), saving $30\%$ in terms of epsilon for $\varepsilon \le 0.08$ for the Adult classification dataset (dataset size $\approx 50K$), and \blue{saving $30\%$ in terms of epsilon for $\varepsilon\leq0.03$ for the Spambase classification dataset (dataset size $\approx 5K$)}. Moreover, we show that for situations where a GBDT is learning a stream of data that originates from different subpopulations (non-IID), S-BDT improves the saving of epsilon even further.
\end{abstract}

\section{Introduction}
    \input{intro}

\section{Overview}
    \input{overview}

\section{Preliminaries}\label{sec:preliminaries}
    \input{preliminaries}

\section{Problem statement \& related work}\label{sec:problem_statement}
    \input{problem_statement}

\section{Tight Individual RDP for subsampled non-spherical multivariate Noise}
    \label{sec:building_blocks}
    \input{building_blocks}
\subsection{Exact RDP bounds for Poisson subsampling}
    \label{sec:exactdp_subsampling}
    \input{subsampling.tex}

\subsection{Individual RDP bound for Leaf-balanced noise}
    \label{sec:tight_leaf_balanced_noise}
    \input{leaf_balanced_noise}

\section{\dpgbdt: Tighter DP GBDT}
    \input{tight_dpgbdt}

\section{\dpgbdt is differentially private}
    \input{proofs}

\section{Empirical evaluation}
    \input{experiments}

\blue{
\section{Other Related work}
In the literature, there are two lines of secure (distributed) GBDT training: Cryptography-based \cite{lu2023squirrel,jiang2024sigbdt} and DP-based methods \cite{bojarski2014differentially,li2020privacy,Maddock_2022,nori2021accuracy,wang2022feverless,Fletcher_2017}. Cryptography-based methods are parallel to our approach as they focus the attack surface on the computation of GBDTs and not on the release. They frequently utilize secure multiparty computation (MPC) or homomorphic encryption (HE). For DP-based GBDTs, \citet{Maddock_2022} is the closest to us as they outperform prior works DP-EBM \cite{nori2021accuracy}, FEVERLESS \cite{wang2022feverless}, and DP-RF \cite{Fletcher_2017}. Other work \cite{bojarski2014differentially} focuses on random forests. Another line of work \cite{li2020privacy} provides suboptimal leaf noising in $\mathcal{O}(g^*)$ and not in $\mathcal{O}(\sfrac{g^*}{n})$ with $g^*$ as the gradient clipping bound and $n$ the number of data points in a leaf together with data-dependent gain-based DP splits which is already covered in \citet{Maddock_2022}.
}

\section{Conclusion}
We introduced \dpgbdt for learning GBDTs, which displays strong utility-privacy performance for the Abalone, Adult, and Spambase datasets. 
Compared to prior work, \dpgbdt incorporates three techniques. To reduce the degree of noise used during training and increase the number of trees, we incorporate subsampling and leaf-balanced noise and prove tight privacy bounds. We adopt so-called individual Rényi filters to ensure that data points that were used in prior trees but were underutilized can be used for training the next trees. As a result, \dpgbdt is well-suited for training models on streams of non-IID data where the intermediary models are released. 

We present a combination of tight individual Rényi DP bound for non-spherical multivariate Gaussian mechanism with leaf-balanced noise and a novel generalized condition for tight privacy bounds for subsampling. This combination might be of independent interest, as it enables a better calibration of the noise for functions with an arbitrary number of output dimensions. The novel condition solely requires the $(\alpha, \rho(\alpha))$-Rényi DP bound of the subsampled mechanism to be linear in $\alpha$.

\bibliographystyle{ACM-Reference-Format}
\bibliography{dp_boosting}

\newpage
\appendix

\section*{Appendix}

\section{Further experiments}

\boldparagraph{Individual Rényi filter for regular training.}
\label{sec:rf_ablation_regular_training}

In \Cref{sec:renyi_filter_streams} we evaluate the individual Rényi filter for learning a stream of non-IID data. Here, we investigate the impact of an individual Rényi filter for regular training.

Prior works \cite{feldman2020individual}\cite{koskela2023individualgdp} have empirically shown that an individual Rényi filter can improve performance when choosing hyperparameters such as a clipping bound suboptimally. Our findings are displayed in \Cref{fig:rf_ablations}. We can also observe the effect of improvement (cf. \Cref{fig:adult_rf_ablation}) when we increase the gradient clipping bound from its optimal value to twice the value. We investigate the number of extra training rounds $T_\text{extra} \in \{T/4, T/2, T \, 3/4, T\}$.

\input{experiments_rf_ablation}

\boldparagraph{Time cost.}
\label{sec:time_cost}

We investigate the time cost of \dpgbdt{}. The individual Rényi filter is quite impactful here. An individual Rényi filter demands updating the individual privacy budget for every data point in every iteration and extra training rounds after the regular training rounds, both increase the time costs. \blue{We update the individual privacy budget by slightly rounding up the individual sensitivities $g_i, h_i$ and then computing the individual privacy loss as an upper bound on the exact individual privacy loss. This strategy saves time cost when the computed individual privacy losses are stored for later use and is proposed in prior work \cite{yu2023individual}.}  We investigate training for 1000 rounds and a tree depth 6 and compare no individual privacy accounting and training with an individual Rényi filter for $T_\text{extra}=1000$ extra rounds. We observe the time cost displayed in \Cref{tbl:time_cost}.

\begin{table}[]
  \footnotesize
  \centering
  \begin{subfigure}[h]{\columnwidth}
  \begin{tabularx}{\columnwidth}{ll}
       \toprule
       Setting  & Time cost (seconds) \\
       \midrule
       \dpgbdt{} no Rényi filter \rule[-2.5ex]{0pt}{0pt} & 3.2 \\
       \dpgbdt{} with Rényi filter \rule[-2.5ex]{0pt}{0pt} & 5.1 \\
       \bottomrule
  \end{tabularx}
  \caption{\textbf{Abalone}}
  \end{subfigure}
  \begin{subfigure}[h]{\columnwidth}
  \begin{tabularx}{\columnwidth}{ll}
       \toprule
       Setting  & Time cost (seconds) \\
       \midrule
       \dpgbdt{} no Rényi filter \rule[-2.5ex]{0pt}{0pt} & 8.2 \\
       \dpgbdt{} with Rényi filter \rule[-2.5ex]{0pt}{0pt} & 25.7 \\
       \bottomrule
  \end{tabularx}
  \caption{\textbf{Adult}}
  \end{subfigure}
  \caption{\textbf{Time cost} of \dpgbdt{} on Abalone and Adult. We train $T_\text{regular}=1000$ rounds with tree depth 6. For the individual Rényi filter we add $T_\text{extra}=1000$ extra rounds.}
  \label{tbl:time_cost}
\end{table}

\section{Postponed proofs}\label{apx:proofs}
    \input{appendix_proofs}

\section{Scalable distributed learning}
    \input{appendix_distributed_learning}

\end{document}

%% file: intro.tex
We present a differentially private distributed learning algorithm for a class of fast learners, the so-called gradient boosted decision tree ensembles (GBDT): these
models combine many weak decision tree learners into an ensemble to prevent overfitting while still capturing complex non-linear patterns.
GBDTs utilize a robust, data-efficient, incremental learning method, traits that are well-suited for strong privacy-preserving approximations.
 
Classical decision trees are vulnerable to privacy attacks on training data \cite{fredrikson_decision_trees_attacks}. 
Each tree consists of splits and leaves which are data-dependent and typically significantly overfitted and thus potentially leak information. 
The state-of-the-art notion for provably protecting against such information leakage is $(\varepsilon,\delta)$-differential privacy (DP), which requires that the impact of single data points be small and deniable.
To protect the training itself, user data should be stored locally and not collected by a trusted 3rd party. This can be realized using distributed training.
Yet, training GBDTs in a privacy-preserving manner needs to achieve strong utility-privacy tradeoffs: protect single training data points (privacy) while keeping the strong machine learning performance (e.g., high classification accuracy or low regression error) of GBDTs (utility).

The best-performing prior work \cite{Maddock_2022} has shown that training GBDTs has the potential to achieve this strong utility-privacy tradeoff.
Yet, we are able to achieve better privacy guarantees (lower $\eps$) for the same utility for regular training as well as in a setting where the data distribution changes during training, i.e. a stream of data that originate from different subpopulations (non-IID):
For such a stream, it is crucial to keep the information from data that arrived at previous rounds, but prior work cannot use that information to avoid additional privacy costs.

Lower $\eps$ values for the same degree of utility are desirable, as for many DP algorithms $\varepsilon$ is in $\mathcal{O}(\sfrac{1}{n})$, for $n$ many data points. Hence, achieving such a better privacy-utility tradeoff means that such a DP algorithm can be applied to smaller datasets. Thus, improving $\eps$ without impeding utility is of practical importance.

\boldparagraph{Contributions} 
We introduce the novel securely distributed DP gradient boosted decision tree algorithm \dpgbdt{}\footnote{\blue{code available at \href{https://github.com/kirschte/sbdt}{https://github.com/kirschte/sbdt}}} that achieves stronger utility-privacy tradeoffs than the best-performing prior work~\cite{Maddock_2022} for strict privacy requirements ($\varepsilon \le  0.5$)  and regression tasks, and provides significant improvements for a stream of non-IID data. 
In contrast to prior work, \dpgbdt{} incorporates three techniques that improve privacy: (1) subsampling to increase the number of trees, (2) leaf-balanced noise for better noise calibration, and (3) Rényi filters for individual privacy accounting.

\input{schematic_overview_figure}

\begin{enumerate}[label=(\arabic*), leftmargin=*] 
\item \emph{Novel condition for exact Rényi DP Bounds for Poisson Subsampling.} Subsampling is known to improve DP bounds. For simple DP algorithms, prior work has shown tight bounds for privacy loss, which leads to less noise and hence better utility. Prior work only showed untight bounds for more general DP algorithms such as used in this work (details in \Cref{sec:single_tree}).
We introduce a novel condition that enables us to prove tight privacy bounds for subsampling for our algorithm. This novel condition generalizes prior results and might be of independent interest.
\item \emph{Individual Rényi DP Bounds for Leaf-balanced Noise.} We show individual RDP for a multivariate Gaussian that balances the scale of the noise added to the first derivative and second derivative of the loss that Newton Boosting requires to construct a single leaf.
These individual RDP bounds neatly fit with subsampling via our contribution above.
\item \emph{Stream of non-IID data.} We adopt a so-called Rényi filter for individual privacy accounting to ensure that underutilized data points that arrived in the past can still be used for training the model, upon the arrival of novel data. We show for a stream of non-IID data that \dpgbdt{} can significantly improve on a naïve baseline which adds extra training rounds but discards formerly known data.
\item \emph{Strong empirical performance.} We experimentally show that our DP GBDT algorithm \dpgbdt{} provides a saving of $>50\%$ in terms of epsilon for $\varepsilon \le 0.5$ (\Cref{fig:all_main_results}) on Abalone regression (dataset size $\approx 4K$) of $>30\%$
for $\varepsilon \le 0.08$ on Adult classification (dataset size $\approx 50K$), and \blue{of $>30\%$ for $\varepsilon \leq 0.03$ on Spambase classification (dataset size $\approx 5K$).} 
Our ablation study evaluates the effectiveness of our improvements (1) to (3).
\end{enumerate}

%% file: schematic_overview_figure.tex
\begin{figure*}
    \centering
   \includegraphics[width=\textwidth]{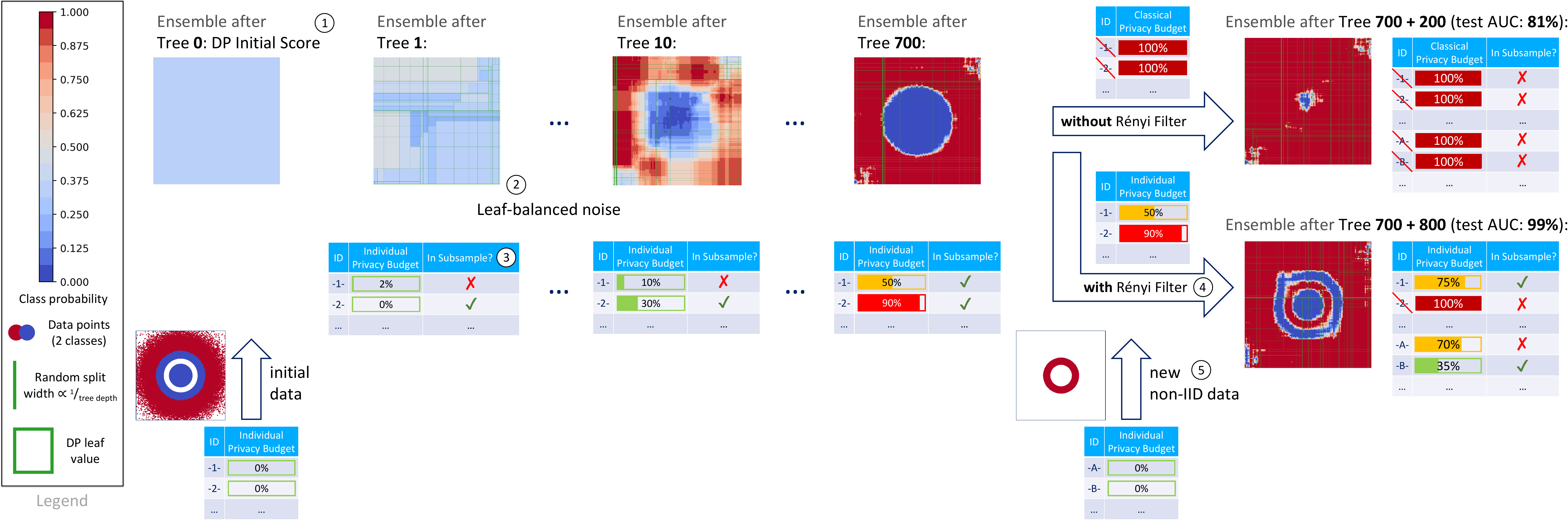}
  \caption{\blue{\textbf{Schematic overview of \dpgbdt ($\eps=0.5$)} classifying two-dimensional concentric circles where the inner yellow circle arrives after the (here: 700) regular training rounds.
  The numbers \ding{192} to \ding{196} refer to \dpgbdt{}'s key technical features (cf. \Cref{sec:overview}).}}
  \label{fig:schematic_overview}
\end{figure*}

%% file: overview.tex
\label{sec:overview}

We briefly sketch gradient boosted decision tree ensembles (GBDT) and then highlight our key technical contributions. As depicted in \Cref{fig:schematic_overview}, decision trees store at each inner node a so-called \emph{split} that partitions the input space along their features. Each path from the root to a leaf thus describes one partition of the input space. For each of these partitions, a decision tree learns a constant leaf value. 
Upon prediction, features of the input are used to choose the corresponding leaf and predict its leaf value. GBDTs iteratively learn a sequence of trees, which leads to overlapping partitions of the input space. Upon prediction, a GBDT sums the leaf values of all trees. During boosted training, each new tree assigns each data point a new label which corrects the current ensemble prediction.

The main sources of privacy leakage for GBDTs are choosing the best splits and leaves. \dpgbdt randomly chooses splits, as prior work~\cite{bojarski2014differentially,Maddock_2022} has shown that this strategy in combination with the iterative improvement of boosting leads to acceptable results.

Our key technical contribution is that we first combine \emph{subsampling} -- a privacy amplification -- with an improved noise calibration (\emph{leaf-balanced noising}) and individual privacy accounting (via \emph{Rényi filters}) and second, derive tight privacy bounds for this combination which might be of independent interest. Our evaluation shows that our modifications lead to DP GBDTs with significantly better utility. 

\boldparagraph{\ding{192} DP Initial score.} We reduce the range of the leaf-value-corrections by first DP-approximating the label-means.

\blue{
\boldparagraph{\ding{193} Tight RDP bounds for leaf-balanced noise.}
DP Newton Boosting for GBDTs \cite{Maddock_2022} sets the leaf value as the fraction of the noisy first derivative of the training loss divided by the noisy second derivative. We propose \emph{leaf-balanced noise} that boosts the utility by balancing the magnitude of both noises. 
On the theoretical side, this noisy derivative pair represents a non-spherical multivariate Gaussian mechanism. We prove tight and individual Rényi DP bounds for this mechanism in \Cref{thm:individual_rdp_nonspherical_gauss}, something prior work \cite{chanyaswad2018mvg,Maddock_2022} did not provide. For instance, MVG \cite{chanyaswad2018mvg} uses untight bounds for a more general noise distribution: matrix-variate Gaussian.
}

\blue{
\boldparagraph{\ding{194} Tight RDP bounds for subsampled multi-variate noising.}
Choosing a random subset of data points (subsampling) is a known privacy amplification.
Yet, prior subsampling results do not apply to DP GBDTs as they either offer untight generic bounds or require univariate noise. The untight bounds do not lead to any privacy amplification (e.g. a factor 15 worse than the tight bound for one of our runs). For a tight bound, \citet{pmlr-v97-zhu19c} prove their tight subsampling bound only for univariate Gaussian noise, yet a DP GBDT adds 2-dimensional noise per leaf. Hence, we prove in \Cref{cor:pearson_vajda_condition_general} a general version for multivariate Gaussian noise including our leaf-balanced noise.
}

\blue{
\boldparagraph{\ding{195} \dpgbdt boosted by Rényi filters.}
}
\citet{feldman2020individual} have shown that in some cases data points whose information has not been used during learning can be re-used in later trees without additional privacy leakage. To this end, a so-called individual privacy accounting via Rényi filters is used. 
\Cref{fig:avgleak} illustrates how this technique utilizes the information of more data points. We prove that the tight Rényi bounds from above constitute a sound individual Privacy accountant  (cf.~\Cref{sec:accounting}) and incorporate a Privacy accountant into \dpgbdt (cf.~\Cref{sec:alg_overview}).

\input{experiments_avg_leakage}

\blue{
\boldparagraph{\ding{196} \dpgbdt is ready for a stream of non-IID data via a Rényi filter.}
We pose the important challenge of consecutively updating a model with later arriving data which are useful to reflect recent data distribution shifts (non-IID). Regularly, a privacy budget once spent cannot be redeemed, thus classical techniques fail to reuse old data to update such a model. Yet, old data is needed to learn new decision boundaries between the old data and the new data points and to ensure that the old data is not unlearned. As a solution, we propose our \dpgbdt tailored with Rényi filters that keeps using those old data that did not consume their accounted budget together with later arriving data points at no additional privacy cost (cf.~\Cref{sec:renyi_filter_streams}).
}

\boldparagraph{\blue{\ding{197}} Scalable distributed learning.} Our differentially private training algorithm for GBDT works seamlessly with distributed learning (cf.~\Cref{sec:distributed_learning}), a setting where multiple parties jointly train a model, without divulging their respective private training datasets to any other user or the public.
This is especially relevant in the medical field with sensitive medical metadata that hospitals might want to keep on-premise but use to train a machine learning model collaboratively with other hospitals and their share of sensitive data.
Our extension to distributed learning builds on prior work \cite{Maddock_2022} that established distributed training of DP GBDT.

%% file: experiments_avg_leakage.tex
\begin{figure}[t]
    \centering
    \resizebox{\columnwidth}{!}{\input{images/avgleak.pgf}\unskip}
  \caption{\blue{
  \textbf{Individual Rényi filter (IRF) boosts the average privacy leakage} (dotted line) closer to the worst-case accounted one. The noise scale is calibrated on the regular tree training rounds, but data points that did not consume their accounted budget are used for free in extra rounds (here: 100 rounds).
  On the abalone dataset, we 1) measure for every data point (x-axis) the privacy leakage by how much the gradient sum of \dpgbdt's leaf changes after removing one data point and 2) aggregate it across the ensemble (y-axis).
  }}
  \label{fig:avgleak}
\end{figure}
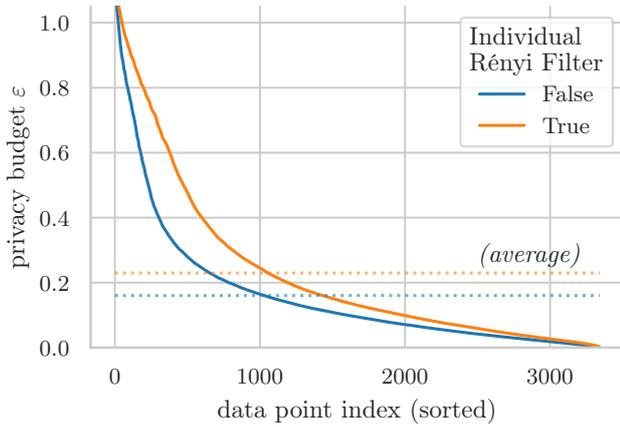

%% file: images/avgleak.pgf
\begingroup%
\makeatletter%
\begin{pgfpicture}%
\pgfpathrectangle{\pgfpointorigin}{\pgfqpoint{3.368250in}{2.296375in}}%
\pgfusepath{use as bounding box, clip}%
\begin{pgfscope}%
\pgfsetbuttcap%
\pgfsetmiterjoin%
\definecolor{currentfill}{rgb}{1.000000,1.000000,1.000000}%
\pgfsetfillcolor{currentfill}%
\pgfsetlinewidth{0.000000pt}%
\definecolor{currentstroke}{rgb}{1.000000,1.000000,1.000000}%
\pgfsetstrokecolor{currentstroke}%
\pgfsetdash{}{0pt}%
\pgfpathmoveto{\pgfqpoint{0.000000in}{0.000000in}}%
\pgfpathlineto{\pgfqpoint{3.368250in}{0.000000in}}%
\pgfpathlineto{\pgfqpoint{3.368250in}{2.296375in}}%
\pgfpathlineto{\pgfqpoint{0.000000in}{2.296375in}}%
\pgfpathlineto{\pgfqpoint{0.000000in}{0.000000in}}%
\pgfpathclose%
\pgfusepath{fill}%
\end{pgfscope}%
\begin{pgfscope}%
\pgfsetbuttcap%
\pgfsetmiterjoin%
\definecolor{currentfill}{rgb}{1.000000,1.000000,1.000000}%
\pgfsetfillcolor{currentfill}%
\pgfsetlinewidth{0.000000pt}%
\definecolor{currentstroke}{rgb}{0.000000,0.000000,0.000000}%
\pgfsetstrokecolor{currentstroke}%
\pgfsetstrokeopacity{0.000000}%
\pgfsetdash{}{0pt}%
\pgfpathmoveto{\pgfqpoint{0.458448in}{0.420833in}}%
\pgfpathlineto{\pgfqpoint{3.368250in}{0.420833in}}%
\pgfpathlineto{\pgfqpoint{3.368250in}{2.296375in}}%
\pgfpathlineto{\pgfqpoint{0.458448in}{2.296375in}}%
\pgfpathlineto{\pgfqpoint{0.458448in}{0.420833in}}%
\pgfpathclose%
\pgfusepath{fill}%
\end{pgfscope}%
\begin{pgfscope}%
\pgfpathrectangle{\pgfqpoint{0.458448in}{0.420833in}}{\pgfqpoint{2.909802in}{1.875542in}}%
\pgfusepath{clip}%
\pgfsetroundcap%
\pgfsetroundjoin%
\pgfsetlinewidth{0.803000pt}%
\definecolor{currentstroke}{rgb}{0.800000,0.800000,0.800000}%
\pgfsetstrokecolor{currentstroke}%
\pgfsetdash{}{0pt}%
\pgfpathmoveto{\pgfqpoint{0.590712in}{0.420833in}}%
\pgfpathlineto{\pgfqpoint{0.590712in}{2.296375in}}%
\pgfusepath{stroke}%
\end{pgfscope}%
\begin{pgfscope}%
\definecolor{textcolor}{rgb}{0.150000,0.150000,0.150000}%
\pgfsetstrokecolor{textcolor}%
\pgfsetfillcolor{textcolor}%
\pgftext[x=0.590712in,y=0.305556in,,top]{\color{textcolor}\rmfamily\fontsize{8.800000}{10.560000}\selectfont \(\displaystyle {0}\)}%
\end{pgfscope}%
\begin{pgfscope}%
\pgfpathrectangle{\pgfqpoint{0.458448in}{0.420833in}}{\pgfqpoint{2.909802in}{1.875542in}}%
\pgfusepath{clip}%
\pgfsetroundcap%
\pgfsetroundjoin%
\pgfsetlinewidth{0.803000pt}%
\definecolor{currentstroke}{rgb}{0.800000,0.800000,0.800000}%
\pgfsetstrokecolor{currentstroke}%
\pgfsetdash{}{0pt}%
\pgfpathmoveto{\pgfqpoint{1.382236in}{0.420833in}}%
\pgfpathlineto{\pgfqpoint{1.382236in}{2.296375in}}%
\pgfusepath{stroke}%
\end{pgfscope}%
\begin{pgfscope}%
\definecolor{textcolor}{rgb}{0.150000,0.150000,0.150000}%
\pgfsetstrokecolor{textcolor}%
\pgfsetfillcolor{textcolor}%
\pgftext[x=1.382236in,y=0.305556in,,top]{\color{textcolor}\rmfamily\fontsize{8.800000}{10.560000}\selectfont \(\displaystyle {1000}\)}%
\end{pgfscope}%
\begin{pgfscope}%
\pgfpathrectangle{\pgfqpoint{0.458448in}{0.420833in}}{\pgfqpoint{2.909802in}{1.875542in}}%
\pgfusepath{clip}%
\pgfsetroundcap%
\pgfsetroundjoin%
\pgfsetlinewidth{0.803000pt}%
\definecolor{currentstroke}{rgb}{0.800000,0.800000,0.800000}%
\pgfsetstrokecolor{currentstroke}%
\pgfsetdash{}{0pt}%
\pgfpathmoveto{\pgfqpoint{2.173760in}{0.420833in}}%
\pgfpathlineto{\pgfqpoint{2.173760in}{2.296375in}}%
\pgfusepath{stroke}%
\end{pgfscope}%
\begin{pgfscope}%
\definecolor{textcolor}{rgb}{0.150000,0.150000,0.150000}%
\pgfsetstrokecolor{textcolor}%
\pgfsetfillcolor{textcolor}%
\pgftext[x=2.173760in,y=0.305556in,,top]{\color{textcolor}\rmfamily\fontsize{8.800000}{10.560000}\selectfont \(\displaystyle {2000}\)}%
\end{pgfscope}%
\begin{pgfscope}%
\pgfpathrectangle{\pgfqpoint{0.458448in}{0.420833in}}{\pgfqpoint{2.909802in}{1.875542in}}%
\pgfusepath{clip}%
\pgfsetroundcap%
\pgfsetroundjoin%
\pgfsetlinewidth{0.803000pt}%
\definecolor{currentstroke}{rgb}{0.800000,0.800000,0.800000}%
\pgfsetstrokecolor{currentstroke}%
\pgfsetdash{}{0pt}%
\pgfpathmoveto{\pgfqpoint{2.965285in}{0.420833in}}%
\pgfpathlineto{\pgfqpoint{2.965285in}{2.296375in}}%
\pgfusepath{stroke}%
\end{pgfscope}%
\begin{pgfscope}%
\definecolor{textcolor}{rgb}{0.150000,0.150000,0.150000}%
\pgfsetstrokecolor{textcolor}%
\pgfsetfillcolor{textcolor}%
\pgftext[x=2.965285in,y=0.305556in,,top]{\color{textcolor}\rmfamily\fontsize{8.800000}{10.560000}\selectfont \(\displaystyle {3000}\)}%
\end{pgfscope}%
\begin{pgfscope}%
\definecolor{textcolor}{rgb}{0.150000,0.150000,0.150000}%
\pgfsetstrokecolor{textcolor}%
\pgfsetfillcolor{textcolor}%
\pgftext[x=1.913349in,y=0.138889in,,top]{\color{textcolor}\rmfamily\fontsize{9.600000}{11.520000}\selectfont data point index (sorted)}%
\end{pgfscope}%
\begin{pgfscope}%
\pgfpathrectangle{\pgfqpoint{0.458448in}{0.420833in}}{\pgfqpoint{2.909802in}{1.875542in}}%
\pgfusepath{clip}%
\pgfsetroundcap%
\pgfsetroundjoin%
\pgfsetlinewidth{0.803000pt}%
\definecolor{currentstroke}{rgb}{0.800000,0.800000,0.800000}%
\pgfsetstrokecolor{currentstroke}%
\pgfsetdash{}{0pt}%
\pgfpathmoveto{\pgfqpoint{0.458448in}{0.420833in}}%
\pgfpathlineto{\pgfqpoint{3.368250in}{0.420833in}}%
\pgfusepath{stroke}%
\end{pgfscope}%
\begin{pgfscope}%
\definecolor{textcolor}{rgb}{0.150000,0.150000,0.150000}%
\pgfsetstrokecolor{textcolor}%
\pgfsetfillcolor{textcolor}%
\pgftext[x=0.179012in, y=0.377431in, left, base]{\color{textcolor}\rmfamily\fontsize{8.800000}{10.560000}\selectfont \(\displaystyle {0.0}\)}%
\end{pgfscope}%
\begin{pgfscope}%
\pgfpathrectangle{\pgfqpoint{0.458448in}{0.420833in}}{\pgfqpoint{2.909802in}{1.875542in}}%
\pgfusepath{clip}%
\pgfsetroundcap%
\pgfsetroundjoin%
\pgfsetlinewidth{0.803000pt}%
\definecolor{currentstroke}{rgb}{0.800000,0.800000,0.800000}%
\pgfsetstrokecolor{currentstroke}%
\pgfsetdash{}{0pt}%
\pgfpathmoveto{\pgfqpoint{0.458448in}{0.778079in}}%
\pgfpathlineto{\pgfqpoint{3.368250in}{0.778079in}}%
\pgfusepath{stroke}%
\end{pgfscope}%
\begin{pgfscope}%
\definecolor{textcolor}{rgb}{0.150000,0.150000,0.150000}%
\pgfsetstrokecolor{textcolor}%
\pgfsetfillcolor{textcolor}%
\pgftext[x=0.179012in, y=0.734677in, left, base]{\color{textcolor}\rmfamily\fontsize{8.800000}{10.560000}\selectfont \(\displaystyle {0.2}\)}%
\end{pgfscope}%
\begin{pgfscope}%
\pgfpathrectangle{\pgfqpoint{0.458448in}{0.420833in}}{\pgfqpoint{2.909802in}{1.875542in}}%
\pgfusepath{clip}%
\pgfsetroundcap%
\pgfsetroundjoin%
\pgfsetlinewidth{0.803000pt}%
\definecolor{currentstroke}{rgb}{0.800000,0.800000,0.800000}%
\pgfsetstrokecolor{currentstroke}%
\pgfsetdash{}{0pt}%
\pgfpathmoveto{\pgfqpoint{0.458448in}{1.135325in}}%
\pgfpathlineto{\pgfqpoint{3.368250in}{1.135325in}}%
\pgfusepath{stroke}%
\end{pgfscope}%
\begin{pgfscope}%
\definecolor{textcolor}{rgb}{0.150000,0.150000,0.150000}%
\pgfsetstrokecolor{textcolor}%
\pgfsetfillcolor{textcolor}%
\pgftext[x=0.179012in, y=1.091923in, left, base]{\color{textcolor}\rmfamily\fontsize{8.800000}{10.560000}\selectfont \(\displaystyle {0.4}\)}%
\end{pgfscope}%
\begin{pgfscope}%
\pgfpathrectangle{\pgfqpoint{0.458448in}{0.420833in}}{\pgfqpoint{2.909802in}{1.875542in}}%
\pgfusepath{clip}%
\pgfsetroundcap%
\pgfsetroundjoin%
\pgfsetlinewidth{0.803000pt}%
\definecolor{currentstroke}{rgb}{0.800000,0.800000,0.800000}%
\pgfsetstrokecolor{currentstroke}%
\pgfsetdash{}{0pt}%
\pgfpathmoveto{\pgfqpoint{0.458448in}{1.492571in}}%
\pgfpathlineto{\pgfqpoint{3.368250in}{1.492571in}}%
\pgfusepath{stroke}%
\end{pgfscope}%
\begin{pgfscope}%
\definecolor{textcolor}{rgb}{0.150000,0.150000,0.150000}%
\pgfsetstrokecolor{textcolor}%
\pgfsetfillcolor{textcolor}%
\pgftext[x=0.179012in, y=1.449169in, left, base]{\color{textcolor}\rmfamily\fontsize{8.800000}{10.560000}\selectfont \(\displaystyle {0.6}\)}%
\end{pgfscope}%
\begin{pgfscope}%
\pgfpathrectangle{\pgfqpoint{0.458448in}{0.420833in}}{\pgfqpoint{2.909802in}{1.875542in}}%
\pgfusepath{clip}%
\pgfsetroundcap%
\pgfsetroundjoin%
\pgfsetlinewidth{0.803000pt}%
\definecolor{currentstroke}{rgb}{0.800000,0.800000,0.800000}%
\pgfsetstrokecolor{currentstroke}%
\pgfsetdash{}{0pt}%
\pgfpathmoveto{\pgfqpoint{0.458448in}{1.849817in}}%
\pgfpathlineto{\pgfqpoint{3.368250in}{1.849817in}}%
\pgfusepath{stroke}%
\end{pgfscope}%
\begin{pgfscope}%
\definecolor{textcolor}{rgb}{0.150000,0.150000,0.150000}%
\pgfsetstrokecolor{textcolor}%
\pgfsetfillcolor{textcolor}%
\pgftext[x=0.179012in, y=1.806415in, left, base]{\color{textcolor}\rmfamily\fontsize{8.800000}{10.560000}\selectfont \(\displaystyle {0.8}\)}%
\end{pgfscope}%
\begin{pgfscope}%
\pgfpathrectangle{\pgfqpoint{0.458448in}{0.420833in}}{\pgfqpoint{2.909802in}{1.875542in}}%
\pgfusepath{clip}%
\pgfsetroundcap%
\pgfsetroundjoin%
\pgfsetlinewidth{0.803000pt}%
\definecolor{currentstroke}{rgb}{0.800000,0.800000,0.800000}%
\pgfsetstrokecolor{currentstroke}%
\pgfsetdash{}{0pt}%
\pgfpathmoveto{\pgfqpoint{0.458448in}{2.207063in}}%
\pgfpathlineto{\pgfqpoint{3.368250in}{2.207063in}}%
\pgfusepath{stroke}%
\end{pgfscope}%
\begin{pgfscope}%
\definecolor{textcolor}{rgb}{0.150000,0.150000,0.150000}%
\pgfsetstrokecolor{textcolor}%
\pgfsetfillcolor{textcolor}%
\pgftext[x=0.179012in, y=2.163661in, left, base]{\color{textcolor}\rmfamily\fontsize{8.800000}{10.560000}\selectfont \(\displaystyle {1.0}\)}%
\end{pgfscope}%
\begin{pgfscope}%
\definecolor{textcolor}{rgb}{0.150000,0.150000,0.150000}%
\pgfsetstrokecolor{textcolor}%
\pgfsetfillcolor{textcolor}%
\pgftext[x=0.123457in,y=1.358604in,,bottom,rotate=90.000000]{\color{textcolor}\rmfamily\fontsize{9.600000}{11.520000}\selectfont privacy budget \(\displaystyle \varepsilon\)}%
\end{pgfscope}%
\begin{pgfscope}%
\pgfpathrectangle{\pgfqpoint{0.458448in}{0.420833in}}{\pgfqpoint{2.909802in}{1.875542in}}%
\pgfusepath{clip}%
\pgfsetroundcap%
\pgfsetroundjoin%
\pgfsetlinewidth{1.204500pt}%
\definecolor{currentstroke}{rgb}{0.121569,0.466667,0.705882}%
\pgfsetstrokecolor{currentstroke}%
\pgfsetdash{}{0pt}%
\pgfpathmoveto{\pgfqpoint{0.601163in}{2.306375in}}%
\pgfpathlineto{\pgfqpoint{0.602585in}{2.276038in}}%
\pgfpathlineto{\pgfqpoint{0.603376in}{2.274427in}}%
\pgfpathlineto{\pgfqpoint{0.621581in}{2.076309in}}%
\pgfpathlineto{\pgfqpoint{0.622373in}{2.074634in}}%
\pgfpathlineto{\pgfqpoint{0.623956in}{2.064086in}}%
\pgfpathlineto{\pgfqpoint{0.626330in}{2.037996in}}%
\pgfpathlineto{\pgfqpoint{0.627122in}{2.035736in}}%
\pgfpathlineto{\pgfqpoint{0.631079in}{2.008846in}}%
\pgfpathlineto{\pgfqpoint{0.633454in}{2.001289in}}%
\pgfpathlineto{\pgfqpoint{0.635829in}{1.986547in}}%
\pgfpathlineto{\pgfqpoint{0.640578in}{1.949226in}}%
\pgfpathlineto{\pgfqpoint{0.642952in}{1.943106in}}%
\pgfpathlineto{\pgfqpoint{0.643744in}{1.941913in}}%
\pgfpathlineto{\pgfqpoint{0.646118in}{1.929192in}}%
\pgfpathlineto{\pgfqpoint{0.647701in}{1.919909in}}%
\pgfpathlineto{\pgfqpoint{0.654034in}{1.874714in}}%
\pgfpathlineto{\pgfqpoint{0.662740in}{1.842322in}}%
\pgfpathlineto{\pgfqpoint{0.665906in}{1.819364in}}%
\pgfpathlineto{\pgfqpoint{0.670656in}{1.807990in}}%
\pgfpathlineto{\pgfqpoint{0.673030in}{1.786461in}}%
\pgfpathlineto{\pgfqpoint{0.673822in}{1.785024in}}%
\pgfpathlineto{\pgfqpoint{0.675405in}{1.777732in}}%
\pgfpathlineto{\pgfqpoint{0.676196in}{1.777075in}}%
\pgfpathlineto{\pgfqpoint{0.678571in}{1.767888in}}%
\pgfpathlineto{\pgfqpoint{0.680154in}{1.759221in}}%
\pgfpathlineto{\pgfqpoint{0.688069in}{1.712415in}}%
\pgfpathlineto{\pgfqpoint{0.689652in}{1.707744in}}%
\pgfpathlineto{\pgfqpoint{0.691235in}{1.698007in}}%
\pgfpathlineto{\pgfqpoint{0.692027in}{1.697715in}}%
\pgfpathlineto{\pgfqpoint{0.695984in}{1.680092in}}%
\pgfpathlineto{\pgfqpoint{0.705483in}{1.628379in}}%
\pgfpathlineto{\pgfqpoint{0.709440in}{1.600241in}}%
\pgfpathlineto{\pgfqpoint{0.710232in}{1.596836in}}%
\pgfpathlineto{\pgfqpoint{0.712606in}{1.579871in}}%
\pgfpathlineto{\pgfqpoint{0.718939in}{1.556153in}}%
\pgfpathlineto{\pgfqpoint{0.721313in}{1.535101in}}%
\pgfpathlineto{\pgfqpoint{0.723688in}{1.521952in}}%
\pgfpathlineto{\pgfqpoint{0.726062in}{1.515757in}}%
\pgfpathlineto{\pgfqpoint{0.736352in}{1.470161in}}%
\pgfpathlineto{\pgfqpoint{0.737144in}{1.468343in}}%
\pgfpathlineto{\pgfqpoint{0.739518in}{1.450827in}}%
\pgfpathlineto{\pgfqpoint{0.745850in}{1.426658in}}%
\pgfpathlineto{\pgfqpoint{0.747433in}{1.420970in}}%
\pgfpathlineto{\pgfqpoint{0.749808in}{1.411772in}}%
\pgfpathlineto{\pgfqpoint{0.752974in}{1.398594in}}%
\pgfpathlineto{\pgfqpoint{0.755349in}{1.382921in}}%
\pgfpathlineto{\pgfqpoint{0.785427in}{1.267300in}}%
\pgfpathlineto{\pgfqpoint{0.788593in}{1.249416in}}%
\pgfpathlineto{\pgfqpoint{0.790176in}{1.243733in}}%
\pgfpathlineto{\pgfqpoint{0.806006in}{1.192921in}}%
\pgfpathlineto{\pgfqpoint{0.809964in}{1.185132in}}%
\pgfpathlineto{\pgfqpoint{0.826586in}{1.143992in}}%
\pgfpathlineto{\pgfqpoint{0.829752in}{1.137420in}}%
\pgfpathlineto{\pgfqpoint{0.834501in}{1.125957in}}%
\pgfpathlineto{\pgfqpoint{0.836084in}{1.124347in}}%
\pgfpathlineto{\pgfqpoint{0.842416in}{1.110967in}}%
\pgfpathlineto{\pgfqpoint{0.844791in}{1.102840in}}%
\pgfpathlineto{\pgfqpoint{0.847166in}{1.097489in}}%
\pgfpathlineto{\pgfqpoint{0.850332in}{1.091087in}}%
\pgfpathlineto{\pgfqpoint{0.859830in}{1.074956in}}%
\pgfpathlineto{\pgfqpoint{0.862996in}{1.071278in}}%
\pgfpathlineto{\pgfqpoint{0.868537in}{1.062785in}}%
\pgfpathlineto{\pgfqpoint{0.875660in}{1.052271in}}%
\pgfpathlineto{\pgfqpoint{0.879618in}{1.043436in}}%
\pgfpathlineto{\pgfqpoint{0.916820in}{0.991948in}}%
\pgfpathlineto{\pgfqpoint{0.919986in}{0.988460in}}%
\pgfpathlineto{\pgfqpoint{0.923943in}{0.983423in}}%
\pgfpathlineto{\pgfqpoint{0.926318in}{0.980887in}}%
\pgfpathlineto{\pgfqpoint{0.928693in}{0.976860in}}%
\pgfpathlineto{\pgfqpoint{0.931067in}{0.975056in}}%
\pgfpathlineto{\pgfqpoint{0.935025in}{0.969612in}}%
\pgfpathlineto{\pgfqpoint{0.940565in}{0.965160in}}%
\pgfpathlineto{\pgfqpoint{0.943732in}{0.960952in}}%
\pgfpathlineto{\pgfqpoint{0.946898in}{0.958123in}}%
\pgfpathlineto{\pgfqpoint{0.955604in}{0.951628in}}%
\pgfpathlineto{\pgfqpoint{0.959562in}{0.946978in}}%
\pgfpathlineto{\pgfqpoint{0.973809in}{0.933725in}}%
\pgfpathlineto{\pgfqpoint{0.976976in}{0.929596in}}%
\pgfpathlineto{\pgfqpoint{0.979350in}{0.927640in}}%
\pgfpathlineto{\pgfqpoint{0.984099in}{0.922416in}}%
\pgfpathlineto{\pgfqpoint{0.988057in}{0.919621in}}%
\pgfpathlineto{\pgfqpoint{0.988848in}{0.919557in}}%
\pgfpathlineto{\pgfqpoint{0.992015in}{0.915292in}}%
\pgfpathlineto{\pgfqpoint{1.014969in}{0.894920in}}%
\pgfpathlineto{\pgfqpoint{1.022884in}{0.888557in}}%
\pgfpathlineto{\pgfqpoint{1.026050in}{0.885998in}}%
\pgfpathlineto{\pgfqpoint{1.068792in}{0.856697in}}%
\pgfpathlineto{\pgfqpoint{1.072750in}{0.854249in}}%
\pgfpathlineto{\pgfqpoint{1.080665in}{0.848192in}}%
\pgfpathlineto{\pgfqpoint{1.085414in}{0.845373in}}%
\pgfpathlineto{\pgfqpoint{1.091747in}{0.842361in}}%
\pgfpathlineto{\pgfqpoint{1.095704in}{0.838896in}}%
\pgfpathlineto{\pgfqpoint{1.107577in}{0.831366in}}%
\pgfpathlineto{\pgfqpoint{1.111535in}{0.829678in}}%
\pgfpathlineto{\pgfqpoint{1.136072in}{0.815928in}}%
\pgfpathlineto{\pgfqpoint{1.140821in}{0.812920in}}%
\pgfpathlineto{\pgfqpoint{1.146362in}{0.809488in}}%
\pgfpathlineto{\pgfqpoint{1.148736in}{0.808003in}}%
\pgfpathlineto{\pgfqpoint{1.156652in}{0.803831in}}%
\pgfpathlineto{\pgfqpoint{1.168524in}{0.799154in}}%
\pgfpathlineto{\pgfqpoint{1.173274in}{0.797045in}}%
\pgfpathlineto{\pgfqpoint{1.211267in}{0.778884in}}%
\pgfpathlineto{\pgfqpoint{1.216016in}{0.777037in}}%
\pgfpathlineto{\pgfqpoint{1.219974in}{0.775282in}}%
\pgfpathlineto{\pgfqpoint{1.227097in}{0.773063in}}%
\pgfpathlineto{\pgfqpoint{1.233429in}{0.770425in}}%
\pgfpathlineto{\pgfqpoint{1.235013in}{0.769664in}}%
\pgfpathlineto{\pgfqpoint{1.238179in}{0.768009in}}%
\pgfpathlineto{\pgfqpoint{1.270631in}{0.754303in}}%
\pgfpathlineto{\pgfqpoint{1.273797in}{0.752276in}}%
\pgfpathlineto{\pgfqpoint{1.305458in}{0.740700in}}%
\pgfpathlineto{\pgfqpoint{1.324455in}{0.733637in}}%
\pgfpathlineto{\pgfqpoint{1.329995in}{0.731756in}}%
\pgfpathlineto{\pgfqpoint{1.443183in}{0.694918in}}%
\pgfpathlineto{\pgfqpoint{1.461389in}{0.690285in}}%
\pgfpathlineto{\pgfqpoint{1.467721in}{0.688051in}}%
\pgfpathlineto{\pgfqpoint{1.474844in}{0.686052in}}%
\pgfpathlineto{\pgfqpoint{1.478802in}{0.684580in}}%
\pgfpathlineto{\pgfqpoint{1.486717in}{0.682661in}}%
\pgfpathlineto{\pgfqpoint{1.535000in}{0.669781in}}%
\pgfpathlineto{\pgfqpoint{1.548456in}{0.666736in}}%
\pgfpathlineto{\pgfqpoint{1.553205in}{0.665135in}}%
\pgfpathlineto{\pgfqpoint{1.567453in}{0.660796in}}%
\pgfpathlineto{\pgfqpoint{1.575368in}{0.658954in}}%
\pgfpathlineto{\pgfqpoint{1.606237in}{0.651150in}}%
\pgfpathlineto{\pgfqpoint{1.614153in}{0.649390in}}%
\pgfpathlineto{\pgfqpoint{1.656895in}{0.639411in}}%
\pgfpathlineto{\pgfqpoint{1.664019in}{0.637770in}}%
\pgfpathlineto{\pgfqpoint{1.832613in}{0.603938in}}%
\pgfpathlineto{\pgfqpoint{1.850027in}{0.600716in}}%
\pgfpathlineto{\pgfqpoint{2.091442in}{0.559640in}}%
\pgfpathlineto{\pgfqpoint{2.106481in}{0.557524in}}%
\pgfpathlineto{\pgfqpoint{2.130227in}{0.554063in}}%
\pgfpathlineto{\pgfqpoint{2.144474in}{0.551702in}}%
\pgfpathlineto{\pgfqpoint{2.212545in}{0.542285in}}%
\pgfpathlineto{\pgfqpoint{2.224418in}{0.540824in}}%
\pgfpathlineto{\pgfqpoint{2.260828in}{0.535221in}}%
\pgfpathlineto{\pgfqpoint{2.271910in}{0.533542in}}%
\pgfpathlineto{\pgfqpoint{2.294864in}{0.530462in}}%
\pgfpathlineto{\pgfqpoint{2.318609in}{0.527053in}}%
\pgfpathlineto{\pgfqpoint{2.445253in}{0.510603in}}%
\pgfpathlineto{\pgfqpoint{2.484038in}{0.505897in}}%
\pgfpathlineto{\pgfqpoint{2.506201in}{0.503047in}}%
\pgfpathlineto{\pgfqpoint{2.527572in}{0.500373in}}%
\pgfpathlineto{\pgfqpoint{2.556858in}{0.497188in}}%
\pgfpathlineto{\pgfqpoint{2.584562in}{0.494062in}}%
\pgfpathlineto{\pgfqpoint{2.672421in}{0.484352in}}%
\pgfpathlineto{\pgfqpoint{2.679545in}{0.483314in}}%
\pgfpathlineto{\pgfqpoint{2.951829in}{0.455018in}}%
\pgfpathlineto{\pgfqpoint{2.974783in}{0.452957in}}%
\pgfpathlineto{\pgfqpoint{3.035731in}{0.446469in}}%
\pgfpathlineto{\pgfqpoint{3.064225in}{0.443489in}}%
\pgfpathlineto{\pgfqpoint{3.079264in}{0.441970in}}%
\pgfpathlineto{\pgfqpoint{3.208283in}{0.427367in}}%
\pgfpathlineto{\pgfqpoint{3.230446in}{0.423446in}}%
\pgfpathlineto{\pgfqpoint{3.235195in}{0.421828in}}%
\pgfpathlineto{\pgfqpoint{3.235195in}{0.421828in}}%
\pgfusepath{stroke}%
\end{pgfscope}%
\begin{pgfscope}%
\pgfpathrectangle{\pgfqpoint{0.458448in}{0.420833in}}{\pgfqpoint{2.909802in}{1.875542in}}%
\pgfusepath{clip}%
\pgfsetroundcap%
\pgfsetroundjoin%
\pgfsetlinewidth{1.204500pt}%
\definecolor{currentstroke}{rgb}{1.000000,0.498039,0.054902}%
\pgfsetstrokecolor{currentstroke}%
\pgfsetdash{}{0pt}%
\pgfpathmoveto{\pgfqpoint{0.611118in}{2.306375in}}%
\pgfpathlineto{\pgfqpoint{0.612874in}{2.292798in}}%
\pgfpathlineto{\pgfqpoint{0.635037in}{2.180229in}}%
\pgfpathlineto{\pgfqpoint{0.635829in}{2.178823in}}%
\pgfpathlineto{\pgfqpoint{0.637412in}{2.170892in}}%
\pgfpathlineto{\pgfqpoint{0.638995in}{2.165133in}}%
\pgfpathlineto{\pgfqpoint{0.640578in}{2.155453in}}%
\pgfpathlineto{\pgfqpoint{0.642161in}{2.150940in}}%
\pgfpathlineto{\pgfqpoint{0.645327in}{2.139591in}}%
\pgfpathlineto{\pgfqpoint{0.646118in}{2.138721in}}%
\pgfpathlineto{\pgfqpoint{0.650868in}{2.122983in}}%
\pgfpathlineto{\pgfqpoint{0.661157in}{2.087732in}}%
\pgfpathlineto{\pgfqpoint{0.663532in}{2.078590in}}%
\pgfpathlineto{\pgfqpoint{0.671447in}{2.055043in}}%
\pgfpathlineto{\pgfqpoint{0.678571in}{2.028253in}}%
\pgfpathlineto{\pgfqpoint{0.681737in}{2.021246in}}%
\pgfpathlineto{\pgfqpoint{0.683320in}{2.018336in}}%
\pgfpathlineto{\pgfqpoint{0.685695in}{2.013202in}}%
\pgfpathlineto{\pgfqpoint{0.688861in}{2.004333in}}%
\pgfpathlineto{\pgfqpoint{0.691235in}{1.999359in}}%
\pgfpathlineto{\pgfqpoint{0.692027in}{1.998778in}}%
\pgfpathlineto{\pgfqpoint{0.695984in}{1.986427in}}%
\pgfpathlineto{\pgfqpoint{0.699150in}{1.974342in}}%
\pgfpathlineto{\pgfqpoint{0.701525in}{1.964919in}}%
\pgfpathlineto{\pgfqpoint{0.714189in}{1.931664in}}%
\pgfpathlineto{\pgfqpoint{0.716564in}{1.927342in}}%
\pgfpathlineto{\pgfqpoint{0.720522in}{1.916770in}}%
\pgfpathlineto{\pgfqpoint{0.729228in}{1.889970in}}%
\pgfpathlineto{\pgfqpoint{0.730811in}{1.887097in}}%
\pgfpathlineto{\pgfqpoint{0.734769in}{1.868920in}}%
\pgfpathlineto{\pgfqpoint{0.737144in}{1.868234in}}%
\pgfpathlineto{\pgfqpoint{0.741101in}{1.861549in}}%
\pgfpathlineto{\pgfqpoint{0.742684in}{1.858089in}}%
\pgfpathlineto{\pgfqpoint{0.747433in}{1.848887in}}%
\pgfpathlineto{\pgfqpoint{0.748225in}{1.847664in}}%
\pgfpathlineto{\pgfqpoint{0.750600in}{1.839421in}}%
\pgfpathlineto{\pgfqpoint{0.754557in}{1.831792in}}%
\pgfpathlineto{\pgfqpoint{0.760098in}{1.810297in}}%
\pgfpathlineto{\pgfqpoint{0.769596in}{1.790040in}}%
\pgfpathlineto{\pgfqpoint{0.773554in}{1.784645in}}%
\pgfpathlineto{\pgfqpoint{0.774345in}{1.783426in}}%
\pgfpathlineto{\pgfqpoint{0.776720in}{1.776895in}}%
\pgfpathlineto{\pgfqpoint{0.782261in}{1.763693in}}%
\pgfpathlineto{\pgfqpoint{0.787010in}{1.743469in}}%
\pgfpathlineto{\pgfqpoint{0.787801in}{1.742979in}}%
\pgfpathlineto{\pgfqpoint{0.790967in}{1.732602in}}%
\pgfpathlineto{\pgfqpoint{0.794925in}{1.724682in}}%
\pgfpathlineto{\pgfqpoint{0.798091in}{1.719656in}}%
\pgfpathlineto{\pgfqpoint{0.801257in}{1.713444in}}%
\pgfpathlineto{\pgfqpoint{0.802840in}{1.711228in}}%
\pgfpathlineto{\pgfqpoint{0.804423in}{1.710401in}}%
\pgfpathlineto{\pgfqpoint{0.807589in}{1.704207in}}%
\pgfpathlineto{\pgfqpoint{0.810755in}{1.698373in}}%
\pgfpathlineto{\pgfqpoint{0.813130in}{1.691375in}}%
\pgfpathlineto{\pgfqpoint{0.826586in}{1.641963in}}%
\pgfpathlineto{\pgfqpoint{0.829752in}{1.636219in}}%
\pgfpathlineto{\pgfqpoint{0.832918in}{1.627011in}}%
\pgfpathlineto{\pgfqpoint{0.836084in}{1.620602in}}%
\pgfpathlineto{\pgfqpoint{0.839250in}{1.612734in}}%
\pgfpathlineto{\pgfqpoint{0.841625in}{1.603183in}}%
\pgfpathlineto{\pgfqpoint{0.843999in}{1.597436in}}%
\pgfpathlineto{\pgfqpoint{0.844791in}{1.595647in}}%
\pgfpathlineto{\pgfqpoint{0.847957in}{1.583117in}}%
\pgfpathlineto{\pgfqpoint{0.855081in}{1.567196in}}%
\pgfpathlineto{\pgfqpoint{0.861413in}{1.557536in}}%
\pgfpathlineto{\pgfqpoint{0.862996in}{1.552082in}}%
\pgfpathlineto{\pgfqpoint{0.867745in}{1.546296in}}%
\pgfpathlineto{\pgfqpoint{0.870911in}{1.540267in}}%
\pgfpathlineto{\pgfqpoint{0.873286in}{1.535850in}}%
\pgfpathlineto{\pgfqpoint{0.878035in}{1.523946in}}%
\pgfpathlineto{\pgfqpoint{0.882784in}{1.510449in}}%
\pgfpathlineto{\pgfqpoint{0.888325in}{1.497715in}}%
\pgfpathlineto{\pgfqpoint{0.895449in}{1.479438in}}%
\pgfpathlineto{\pgfqpoint{0.900989in}{1.462691in}}%
\pgfpathlineto{\pgfqpoint{0.904155in}{1.457107in}}%
\pgfpathlineto{\pgfqpoint{0.906530in}{1.446728in}}%
\pgfpathlineto{\pgfqpoint{0.913654in}{1.425777in}}%
\pgfpathlineto{\pgfqpoint{0.915237in}{1.424137in}}%
\pgfpathlineto{\pgfqpoint{0.917611in}{1.417320in}}%
\pgfpathlineto{\pgfqpoint{0.922360in}{1.404133in}}%
\pgfpathlineto{\pgfqpoint{0.923152in}{1.403829in}}%
\pgfpathlineto{\pgfqpoint{0.929484in}{1.388576in}}%
\pgfpathlineto{\pgfqpoint{0.935816in}{1.371398in}}%
\pgfpathlineto{\pgfqpoint{0.938982in}{1.364781in}}%
\pgfpathlineto{\pgfqpoint{0.941357in}{1.361781in}}%
\pgfpathlineto{\pgfqpoint{0.946898in}{1.349112in}}%
\pgfpathlineto{\pgfqpoint{0.948481in}{1.345943in}}%
\pgfpathlineto{\pgfqpoint{0.957187in}{1.330786in}}%
\pgfpathlineto{\pgfqpoint{0.959562in}{1.323937in}}%
\pgfpathlineto{\pgfqpoint{0.960354in}{1.323544in}}%
\pgfpathlineto{\pgfqpoint{0.968269in}{1.307197in}}%
\pgfpathlineto{\pgfqpoint{0.971435in}{1.301245in}}%
\pgfpathlineto{\pgfqpoint{0.975393in}{1.295096in}}%
\pgfpathlineto{\pgfqpoint{0.989640in}{1.267409in}}%
\pgfpathlineto{\pgfqpoint{1.002304in}{1.238259in}}%
\pgfpathlineto{\pgfqpoint{1.003887in}{1.235898in}}%
\pgfpathlineto{\pgfqpoint{1.014177in}{1.215292in}}%
\pgfpathlineto{\pgfqpoint{1.016552in}{1.210999in}}%
\pgfpathlineto{\pgfqpoint{1.029216in}{1.188172in}}%
\pgfpathlineto{\pgfqpoint{1.031591in}{1.185288in}}%
\pgfpathlineto{\pgfqpoint{1.037923in}{1.174289in}}%
\pgfpathlineto{\pgfqpoint{1.039506in}{1.171519in}}%
\pgfpathlineto{\pgfqpoint{1.044255in}{1.163996in}}%
\pgfpathlineto{\pgfqpoint{1.051379in}{1.154742in}}%
\pgfpathlineto{\pgfqpoint{1.053753in}{1.150765in}}%
\pgfpathlineto{\pgfqpoint{1.056128in}{1.148360in}}%
\pgfpathlineto{\pgfqpoint{1.059294in}{1.142994in}}%
\pgfpathlineto{\pgfqpoint{1.061669in}{1.139370in}}%
\pgfpathlineto{\pgfqpoint{1.067209in}{1.133945in}}%
\pgfpathlineto{\pgfqpoint{1.084623in}{1.109690in}}%
\pgfpathlineto{\pgfqpoint{1.086997in}{1.107067in}}%
\pgfpathlineto{\pgfqpoint{1.090955in}{1.103461in}}%
\pgfpathlineto{\pgfqpoint{1.098079in}{1.094065in}}%
\pgfpathlineto{\pgfqpoint{1.102036in}{1.087464in}}%
\pgfpathlineto{\pgfqpoint{1.106786in}{1.082596in}}%
\pgfpathlineto{\pgfqpoint{1.111535in}{1.076748in}}%
\pgfpathlineto{\pgfqpoint{1.115492in}{1.073096in}}%
\pgfpathlineto{\pgfqpoint{1.120241in}{1.064698in}}%
\pgfpathlineto{\pgfqpoint{1.124199in}{1.060606in}}%
\pgfpathlineto{\pgfqpoint{1.130531in}{1.053287in}}%
\pgfpathlineto{\pgfqpoint{1.133697in}{1.048974in}}%
\pgfpathlineto{\pgfqpoint{1.135280in}{1.047144in}}%
\pgfpathlineto{\pgfqpoint{1.137655in}{1.043595in}}%
\pgfpathlineto{\pgfqpoint{1.141613in}{1.039433in}}%
\pgfpathlineto{\pgfqpoint{1.146362in}{1.032939in}}%
\pgfpathlineto{\pgfqpoint{1.152694in}{1.026832in}}%
\pgfpathlineto{\pgfqpoint{1.155860in}{1.024137in}}%
\pgfpathlineto{\pgfqpoint{1.159818in}{1.020134in}}%
\pgfpathlineto{\pgfqpoint{1.171691in}{1.009449in}}%
\pgfpathlineto{\pgfqpoint{1.175648in}{1.007219in}}%
\pgfpathlineto{\pgfqpoint{1.188313in}{0.994648in}}%
\pgfpathlineto{\pgfqpoint{1.191479in}{0.992069in}}%
\pgfpathlineto{\pgfqpoint{1.193853in}{0.989985in}}%
\pgfpathlineto{\pgfqpoint{1.196228in}{0.987830in}}%
\pgfpathlineto{\pgfqpoint{1.199394in}{0.983964in}}%
\pgfpathlineto{\pgfqpoint{1.204935in}{0.978098in}}%
\pgfpathlineto{\pgfqpoint{1.212058in}{0.973796in}}%
\pgfpathlineto{\pgfqpoint{1.215224in}{0.970203in}}%
\pgfpathlineto{\pgfqpoint{1.219182in}{0.967811in}}%
\pgfpathlineto{\pgfqpoint{1.228680in}{0.960251in}}%
\pgfpathlineto{\pgfqpoint{1.233429in}{0.956907in}}%
\pgfpathlineto{\pgfqpoint{1.285670in}{0.916792in}}%
\pgfpathlineto{\pgfqpoint{1.292002in}{0.912563in}}%
\pgfpathlineto{\pgfqpoint{1.314165in}{0.899336in}}%
\pgfpathlineto{\pgfqpoint{1.316540in}{0.897261in}}%
\pgfpathlineto{\pgfqpoint{1.320497in}{0.895013in}}%
\pgfpathlineto{\pgfqpoint{1.322080in}{0.893267in}}%
\pgfpathlineto{\pgfqpoint{1.328412in}{0.889628in}}%
\pgfpathlineto{\pgfqpoint{1.336328in}{0.884556in}}%
\pgfpathlineto{\pgfqpoint{1.395692in}{0.850858in}}%
\pgfpathlineto{\pgfqpoint{1.405190in}{0.844854in}}%
\pgfpathlineto{\pgfqpoint{1.416272in}{0.837189in}}%
\pgfpathlineto{\pgfqpoint{1.422604in}{0.832790in}}%
\pgfpathlineto{\pgfqpoint{1.433685in}{0.826694in}}%
\pgfpathlineto{\pgfqpoint{1.437643in}{0.825078in}}%
\pgfpathlineto{\pgfqpoint{1.440017in}{0.823724in}}%
\pgfpathlineto{\pgfqpoint{1.447141in}{0.820066in}}%
\pgfpathlineto{\pgfqpoint{1.455848in}{0.815586in}}%
\pgfpathlineto{\pgfqpoint{1.458222in}{0.813956in}}%
\pgfpathlineto{\pgfqpoint{1.462180in}{0.812141in}}%
\pgfpathlineto{\pgfqpoint{1.489883in}{0.798676in}}%
\pgfpathlineto{\pgfqpoint{1.493841in}{0.796934in}}%
\pgfpathlineto{\pgfqpoint{1.512046in}{0.788006in}}%
\pgfpathlineto{\pgfqpoint{1.526294in}{0.781244in}}%
\pgfpathlineto{\pgfqpoint{1.532626in}{0.778866in}}%
\pgfpathlineto{\pgfqpoint{1.535792in}{0.777603in}}%
\pgfpathlineto{\pgfqpoint{1.552414in}{0.769691in}}%
\pgfpathlineto{\pgfqpoint{1.557163in}{0.767357in}}%
\pgfpathlineto{\pgfqpoint{1.563495in}{0.764516in}}%
\pgfpathlineto{\pgfqpoint{1.571410in}{0.761545in}}%
\pgfpathlineto{\pgfqpoint{1.576160in}{0.759526in}}%
\pgfpathlineto{\pgfqpoint{1.587241in}{0.755005in}}%
\pgfpathlineto{\pgfqpoint{1.595948in}{0.751223in}}%
\pgfpathlineto{\pgfqpoint{1.609404in}{0.746217in}}%
\pgfpathlineto{\pgfqpoint{1.615736in}{0.743339in}}%
\pgfpathlineto{\pgfqpoint{1.636315in}{0.735844in}}%
\pgfpathlineto{\pgfqpoint{1.641856in}{0.733927in}}%
\pgfpathlineto{\pgfqpoint{1.652146in}{0.730640in}}%
\pgfpathlineto{\pgfqpoint{1.664019in}{0.726072in}}%
\pgfpathlineto{\pgfqpoint{1.668768in}{0.724619in}}%
\pgfpathlineto{\pgfqpoint{1.681432in}{0.721074in}}%
\pgfpathlineto{\pgfqpoint{1.688556in}{0.718066in}}%
\pgfpathlineto{\pgfqpoint{1.711510in}{0.711122in}}%
\pgfpathlineto{\pgfqpoint{1.718634in}{0.709414in}}%
\pgfpathlineto{\pgfqpoint{1.749503in}{0.700263in}}%
\pgfpathlineto{\pgfqpoint{1.754253in}{0.698274in}}%
\pgfpathlineto{\pgfqpoint{1.761376in}{0.696165in}}%
\pgfpathlineto{\pgfqpoint{1.766125in}{0.694206in}}%
\pgfpathlineto{\pgfqpoint{1.808076in}{0.682396in}}%
\pgfpathlineto{\pgfqpoint{1.826281in}{0.677072in}}%
\pgfpathlineto{\pgfqpoint{1.832613in}{0.675158in}}%
\pgfpathlineto{\pgfqpoint{1.869815in}{0.665156in}}%
\pgfpathlineto{\pgfqpoint{1.895144in}{0.658875in}}%
\pgfpathlineto{\pgfqpoint{1.899893in}{0.657414in}}%
\pgfpathlineto{\pgfqpoint{1.919681in}{0.652185in}}%
\pgfpathlineto{\pgfqpoint{1.946593in}{0.645990in}}%
\pgfpathlineto{\pgfqpoint{1.964798in}{0.641900in}}%
\pgfpathlineto{\pgfqpoint{1.978254in}{0.638869in}}%
\pgfpathlineto{\pgfqpoint{1.990918in}{0.635935in}}%
\pgfpathlineto{\pgfqpoint{2.029703in}{0.627324in}}%
\pgfpathlineto{\pgfqpoint{2.062947in}{0.620587in}}%
\pgfpathlineto{\pgfqpoint{2.075611in}{0.617430in}}%
\pgfpathlineto{\pgfqpoint{2.087484in}{0.615303in}}%
\pgfpathlineto{\pgfqpoint{2.103315in}{0.611857in}}%
\pgfpathlineto{\pgfqpoint{2.117562in}{0.608645in}}%
\pgfpathlineto{\pgfqpoint{2.131018in}{0.605850in}}%
\pgfpathlineto{\pgfqpoint{2.138933in}{0.604132in}}%
\pgfpathlineto{\pgfqpoint{2.210171in}{0.589491in}}%
\pgfpathlineto{\pgfqpoint{2.222835in}{0.587261in}}%
\pgfpathlineto{\pgfqpoint{2.233125in}{0.584770in}}%
\pgfpathlineto{\pgfqpoint{2.378765in}{0.557124in}}%
\pgfpathlineto{\pgfqpoint{2.383514in}{0.555964in}}%
\pgfpathlineto{\pgfqpoint{2.404886in}{0.552525in}}%
\pgfpathlineto{\pgfqpoint{2.488787in}{0.539361in}}%
\pgfpathlineto{\pgfqpoint{2.525197in}{0.533377in}}%
\pgfpathlineto{\pgfqpoint{2.538653in}{0.531165in}}%
\pgfpathlineto{\pgfqpoint{2.594060in}{0.522592in}}%
\pgfpathlineto{\pgfqpoint{2.601184in}{0.521524in}}%
\pgfpathlineto{\pgfqpoint{2.614640in}{0.518972in}}%
\pgfpathlineto{\pgfqpoint{2.622555in}{0.517621in}}%
\pgfpathlineto{\pgfqpoint{2.632053in}{0.515913in}}%
\pgfpathlineto{\pgfqpoint{2.710414in}{0.503740in}}%
\pgfpathlineto{\pgfqpoint{2.727036in}{0.501355in}}%
\pgfpathlineto{\pgfqpoint{2.734160in}{0.500304in}}%
\pgfpathlineto{\pgfqpoint{2.748407in}{0.498228in}}%
\pgfpathlineto{\pgfqpoint{2.757906in}{0.496764in}}%
\pgfpathlineto{\pgfqpoint{2.776111in}{0.494031in}}%
\pgfpathlineto{\pgfqpoint{2.784817in}{0.492646in}}%
\pgfpathlineto{\pgfqpoint{2.798273in}{0.490748in}}%
\pgfpathlineto{\pgfqpoint{2.815687in}{0.488357in}}%
\pgfpathlineto{\pgfqpoint{2.839433in}{0.485338in}}%
\pgfpathlineto{\pgfqpoint{2.853680in}{0.483327in}}%
\pgfpathlineto{\pgfqpoint{3.141795in}{0.445864in}}%
\pgfpathlineto{\pgfqpoint{3.200368in}{0.435211in}}%
\pgfpathlineto{\pgfqpoint{3.209074in}{0.433376in}}%
\pgfpathlineto{\pgfqpoint{3.229654in}{0.427370in}}%
\pgfpathlineto{\pgfqpoint{3.233612in}{0.425128in}}%
\pgfpathlineto{\pgfqpoint{3.235195in}{0.422765in}}%
\pgfpathlineto{\pgfqpoint{3.235195in}{0.422765in}}%
\pgfusepath{stroke}%
\end{pgfscope}%
\begin{pgfscope}%
\pgfpathrectangle{\pgfqpoint{0.458448in}{0.420833in}}{\pgfqpoint{2.909802in}{1.875542in}}%
\pgfusepath{clip}%
\pgfsetbuttcap%
\pgfsetroundjoin%
\pgfsetlinewidth{1.204500pt}%
\definecolor{currentstroke}{rgb}{0.121569,0.466667,0.705882}%
\pgfsetstrokecolor{currentstroke}%
\pgfsetstrokeopacity{0.600000}%
\pgfsetdash{{1.200000pt}{1.980000pt}}{0.000000pt}%
\pgfpathmoveto{\pgfqpoint{0.590712in}{0.707295in}}%
\pgfpathlineto{\pgfqpoint{3.235986in}{0.707295in}}%
\pgfusepath{stroke}%
\end{pgfscope}%
\begin{pgfscope}%
\pgfpathrectangle{\pgfqpoint{0.458448in}{0.420833in}}{\pgfqpoint{2.909802in}{1.875542in}}%
\pgfusepath{clip}%
\pgfsetbuttcap%
\pgfsetroundjoin%
\pgfsetlinewidth{1.204500pt}%
\definecolor{currentstroke}{rgb}{1.000000,0.498039,0.054902}%
\pgfsetstrokecolor{currentstroke}%
\pgfsetstrokeopacity{0.600000}%
\pgfsetdash{{1.200000pt}{1.980000pt}}{0.000000pt}%
\pgfpathmoveto{\pgfqpoint{0.590712in}{0.830940in}}%
\pgfpathlineto{\pgfqpoint{3.235986in}{0.830940in}}%
\pgfusepath{stroke}%
\end{pgfscope}%
\begin{pgfscope}%
\pgfsetrectcap%
\pgfsetmiterjoin%
\pgfsetlinewidth{1.003750pt}%
\definecolor{currentstroke}{rgb}{0.800000,0.800000,0.800000}%
\pgfsetstrokecolor{currentstroke}%
\pgfsetdash{}{0pt}%
\pgfpathmoveto{\pgfqpoint{0.458448in}{0.420833in}}%
\pgfpathlineto{\pgfqpoint{0.458448in}{2.296375in}}%
\pgfusepath{stroke}%
\end{pgfscope}%
\begin{pgfscope}%
\pgfsetrectcap%
\pgfsetmiterjoin%
\pgfsetlinewidth{1.003750pt}%
\definecolor{currentstroke}{rgb}{0.800000,0.800000,0.800000}%
\pgfsetstrokecolor{currentstroke}%
\pgfsetdash{}{0pt}%
\pgfpathmoveto{\pgfqpoint{0.458448in}{0.420833in}}%
\pgfpathlineto{\pgfqpoint{3.368250in}{0.420833in}}%
\pgfusepath{stroke}%
\end{pgfscope}%
\begin{pgfscope}%
\definecolor{textcolor}{rgb}{0.150000,0.150000,0.150000}%
\pgfsetstrokecolor{textcolor}%
\pgfsetfillcolor{textcolor}%
\pgftext[x=2.569523in,y=0.892455in,left,base]{\color{textcolor}\rmfamily\fontsize{9.600000}{11.520000}\itshape\selectfont (average)}%
\end{pgfscope}%
\begin{pgfscope}%
\pgfsetbuttcap%
\pgfsetmiterjoin%
\definecolor{currentfill}{rgb}{1.000000,1.000000,1.000000}%
\pgfsetfillcolor{currentfill}%
\pgfsetfillopacity{0.800000}%
\pgfsetlinewidth{0.803000pt}%
\definecolor{currentstroke}{rgb}{0.800000,0.800000,0.800000}%
\pgfsetstrokecolor{currentstroke}%
\pgfsetstrokeopacity{0.800000}%
\pgfsetdash{}{0pt}%
\pgfpathmoveto{\pgfqpoint{2.500394in}{1.526838in}}%
\pgfpathlineto{\pgfqpoint{3.282694in}{1.526838in}}%
\pgfpathquadraticcurveto{\pgfqpoint{3.307139in}{1.526838in}}{\pgfqpoint{3.307139in}{1.551283in}}%
\pgfpathlineto{\pgfqpoint{3.307139in}{2.210819in}}%
\pgfpathquadraticcurveto{\pgfqpoint{3.307139in}{2.235264in}}{\pgfqpoint{3.282694in}{2.235264in}}%
\pgfpathlineto{\pgfqpoint{2.500394in}{2.235264in}}%
\pgfpathquadraticcurveto{\pgfqpoint{2.475949in}{2.235264in}}{\pgfqpoint{2.475949in}{2.210819in}}%
\pgfpathlineto{\pgfqpoint{2.475949in}{1.551283in}}%
\pgfpathquadraticcurveto{\pgfqpoint{2.475949in}{1.526838in}}{\pgfqpoint{2.500394in}{1.526838in}}%
\pgfpathlineto{\pgfqpoint{2.500394in}{1.526838in}}%
\pgfpathclose%
\pgfusepath{stroke,fill}%
\end{pgfscope}%
\begin{pgfscope}%
\definecolor{textcolor}{rgb}{0.150000,0.150000,0.150000}%
\pgfsetstrokecolor{textcolor}%
\pgfsetfillcolor{textcolor}%
\pgftext[x=2.524838in, y=2.089924in, left, base]{\color{textcolor}\rmfamily\fontsize{9.600000}{11.520000}\selectfont Individual}%
\end{pgfscope}%
\begin{pgfscope}%
\definecolor{textcolor}{rgb}{0.150000,0.150000,0.150000}%
\pgfsetstrokecolor{textcolor}%
\pgfsetfillcolor{textcolor}%
\pgftext[x=2.524838in, y=1.947178in, left, base]{\color{textcolor}\rmfamily\fontsize{9.600000}{11.520000}\selectfont Rényi Filter}%
\end{pgfscope}%
\begin{pgfscope}%
\pgfsetroundcap%
\pgfsetroundjoin%
\pgfsetlinewidth{1.204500pt}%
\definecolor{currentstroke}{rgb}{0.121569,0.466667,0.705882}%
\pgfsetstrokecolor{currentstroke}%
\pgfsetdash{}{0pt}%
\pgfpathmoveto{\pgfqpoint{2.579987in}{1.815033in}}%
\pgfpathlineto{\pgfqpoint{2.702210in}{1.815033in}}%
\pgfpathlineto{\pgfqpoint{2.824432in}{1.815033in}}%
\pgfusepath{stroke}%
\end{pgfscope}%
\begin{pgfscope}%
\definecolor{textcolor}{rgb}{0.150000,0.150000,0.150000}%
\pgfsetstrokecolor{textcolor}%
\pgfsetfillcolor{textcolor}%
\pgftext[x=2.922210in,y=1.772255in,left,base]{\color{textcolor}\rmfamily\fontsize{8.800000}{10.560000}\selectfont False}%
\end{pgfscope}%
\begin{pgfscope}%
\pgfsetroundcap%
\pgfsetroundjoin%
\pgfsetlinewidth{1.204500pt}%
\definecolor{currentstroke}{rgb}{1.000000,0.498039,0.054902}%
\pgfsetstrokecolor{currentstroke}%
\pgfsetdash{}{0pt}%
\pgfpathmoveto{\pgfqpoint{2.579987in}{1.642810in}}%
\pgfpathlineto{\pgfqpoint{2.702210in}{1.642810in}}%
\pgfpathlineto{\pgfqpoint{2.824432in}{1.642810in}}%
\pgfusepath{stroke}%
\end{pgfscope}%
\begin{pgfscope}%
\definecolor{textcolor}{rgb}{0.150000,0.150000,0.150000}%
\pgfsetstrokecolor{textcolor}%
\pgfsetfillcolor{textcolor}%
\pgftext[x=2.922210in,y=1.600033in,left,base]{\color{textcolor}\rmfamily\fontsize{8.800000}{10.560000}\selectfont True}%
\end{pgfscope}%
\end{pgfpicture}%
\makeatother%
\endgroup%

%% file: preliminaries.tex
\subsection{Differential Privacy} 

Differential privacy (DP) \cite{DBLP:conf/tcc/DworkMNS06} is the de facto standard for provable privacy. Given a mechanism that has a data-dependent output, DP requires that the impact of single data points in the output is limited and thus deniable. In line with close prior work ~\cite{li2020privacy,Maddock_2022}, we consider unbounded DP, where the effect of adding or removing a single instance from a dataset on the output is analyzed.

\begin{definition}[Neighboring datasets]
    Given two datasets $X \subseteq \mathcal{X}$ and $X' \subseteq \mathcal{X}$ where $X' \coloneqq X \cup \{x\}$ for some $x \in \mathcal{X}$, then $X$ and $X'$ are neighboring: $X \sim_{x} X'$, or in short $X \sim X'$. 
\end{definition}

\begin{definition}[DP, ~\cite{DBLP:conf/tcc/DworkMNS06}]\label{def:dp}
A randomized mechanism $M: \mathcal{X} \mapsto \mathcal{R}$ satisfies $(\varepsilon, \delta)$-DP if, for any two neighboring datasets $X \sim X'$ and any observation $o$:
$\Pr[M(X) = o] \leq e^{\varepsilon} \cdot \Pr[M(X') = o] + \delta$.
\end{definition}

For accounting, we use a variant of DP: $(\alpha,\rho(\alpha))$-Rényi differential privacy (RDP) which bounds the $\alpha$-th moment of the privacy loss. 
With RDP we can mathematically describe features like composition (cf. \Cref{thmrdpsequentialcomposition}), filter (cf. \Cref{sec:prelim_renyifilter}), or subsampling (cf. \Cref{sec:prelim_subsampling}). Using \Cref{cor:rdp_to_dp}, we can convert Rényi DP to DP. For a tight conversion, we refer to \cite[Theorem 2]{sommer2019privacy} which requires access to all $\alpha$.

\begin{definition}[Rényi Divergence, ~\cite{renyi_divergence}]
    Given two probability distributions $P,Q$ over $\mathcal{R}$ where $P(o)$ denotes the density of $P$ at $o$, then the Rényi divergence of order $\alpha$ is defined as
    \[%
    \textstyle D_\alpha(P||Q) := \frac{1}{\alpha-1} \log \int_{-\infty}^{\infty} \frac{P(o)^\alpha}{Q(o)^{\alpha-1}} do.
    \]%
\end{definition}

\begin{definition}[Rényi DP, Definition 4 in \cite{DBLP:journals/corr/Mironov17}] \label{def:renyi_dp}
    A randomized mechanism $M: \mathcal{X} \mapsto \mathcal{R}$ satisfies $(\alpha, \rho(\alpha))$-Rényi DP if, for any two neighboring datasets $X \sim X'$:
    $D_\alpha(M(X)||M(X') \leq \rho(\alpha)$.%
\end{definition}

\blue{
\begin{theorem}[Adaptive sequential composition for RDP, ~\cite{feldman2020individual} Theorem 3.1]
    Let $\alpha$ and $\rho(\alpha)$ be fixed, $\mathcal{X}$ be a dataset space, and $M$ be a sequence of adaptively chosen mechanisms
    $M_i: \Pi_{j=1}^{i-1} \mathcal{R}_{j} \times \mathcal{X} \mapsto \mathcal{R}_i$ for $i \in \set{1,\dots,k}$, i.e. $M_i$ has the outputs of all previous mechanisms $R_{1},\dots,R_{i-1}$ as an input. If each $M_i$ satisfies $(\alpha, \rho_i(\alpha))$-Rényi DP and $\sum_{i} \rho_i(\alpha) \leq \rho(\alpha)$ then $M$ satisfies $(\alpha, \rho(\alpha))$-Rényi DP.
    \label{thmrdpsequentialcomposition}
\end{theorem}
\Cref{thmrdpsequentialcomposition} also holds if the mechanisms are not adaptively chosen, i.e. have privacy budgets fixed in advance.
}

\begin{corollary}[RDP to DP, Thm. 21 in \cite{balle2019hypothesis} or Prop. 12 in \cite{canonne2020discrete}]\label{cor:rdp_to_dp}
    For any $\delta \in [0,1]$, if a mechanism is $(\alpha,\rho(\alpha))$-RDP, then it is $(\eps,\delta)$-DP with
    $\eps = \rho(\alpha) + \log(\frac{\alpha - 1}{\alpha}) - \frac{\log(\delta) + \log(\alpha)}{\alpha - 1}$.%
\end{corollary}

If the output of a function $q$, e.g. the leaf value in a GBDT, is $s$-sensitivity bounded, then adding Gaussian noise to $q$ is RDP.

\blue{
\begin{definition}[$L_2$ sensitivity]
    For neighboring inputs $X, X' \in \mathcal{X}$, universe $\mathcal{U}$, and a randomized function $q\colon \mathcal{X} \mapsto \mathcal{U}$, the $L_2$ sensitivity $s$ is defined as $s \coloneqq \max_{X \sim X'}\norm{q(X) - q(X')}_{2}$. We call a function with a finite $L_2$ sensitivity $s$, an \emph{$s$-$L_2$-sensitivity bounded function}.
\end{definition}

}

\begin{theorem}[Gaussian mechanism, \cite{privacybook}\cite{DBLP:journals/corr/Mironov17}]
    Let $q\colon \mathcal{X} \mapsto \mathbb{R}^m$ be an $s$-$L_2$-sensitivity bounded function with respect to inputs $X \sim X'$. The Gaussian mechanism $M: \mathcal{X} \mapsto \mathbb{R}^{m}$ of the form
    $M(x) = q(x) + \mathcal{N}(0, s^{2}\sigma^{2} I_m)$
    satisfies $(\alpha, \sfrac{\alpha}{2\sigma^2})$-Rényi DP.
\end{theorem}

\begin{theorem}[Post-Processing~\cite{DBLP:journals/corr/Mironov17}]
    Let $M\colon \mathcal{X} \mapsto \mathcal{R}$ be  $(\alpha, \rho(\alpha))$-RDP mechanism and $f\colon \mathcal{R} \mapsto \mathcal{R}'$ a randomized function. Then, for any pair of neighboring datasets $X \sim X'$
    \[D_\alpha(M(X)||M(X')) \geq D_\alpha(f(M(X))||f(M(X'))),\]
    i.e. Rényi differential privacy is preserved by post-processing.
    \label{thm:rdp_post_processing}
\end{theorem}

\subsection{Gradient Boosted Decision Trees}
\label{sec:gbdt}

\begin{table}[t]
\footnotesize
\begin{center}
\begin{tabular}{ cl } 
    \toprule
    Symbol & Description \\
    \midrule
    $M$ & Randomized mechanism \\
    $(x,y) \in X$ (or $D$) & labeled data point of dataset \\
    $X\sim X'$, $X\sim_{x} X'$ & Neighboring datasets (differing in $x$) \\
    $D_\alpha(M(X)||M(X'))$& Rényi Divergence of order $\alpha$ \\
    $(\alpha, \rho(\alpha))$ & Rényi DP bound $\rho(\alpha)$ for $\alpha $\\
    $(\alpha, \rho_{t}^{(i)}(\alpha))$ & Individual RDP bound of $x_i$ for \\
    & round $t$ \\
    $\mathcal{F}_{\alpha, \rho(\alpha)}$ & Individual Rényi filter \\
    $\tilde{y_i}$ & Prediction for $x_i$ \\
    $(g_i, h_i)$ & Gradient and Hessian of $x_i$ \\
    \blue{$g^*, h^*, m^*$} & Gradient/Hessian/label clipping \\
    & bound \\
    \blue{$r_1, r_2$} & \blue{Noise weights for leaf value} \\
    \blue{$\gamma$} & \blue{Subsampling ratio} \\
    \blue{$\eta$} & \blue{Learning rate} \\
    \blue{$\lambda, \beta$} & \blue{Regularization parameters} \\
    \blue{$\varepsilon_\text{trees}, \varepsilon_\text{init}, \varepsilon_\text{ds}$} & $\eps$ privacy budget shares\\
    & for trees/initial score/dataset size \\
    \blue{$\sigma_\text{leaf}^2$} & \blue{Unweighted variance of leaf Gaussian} \\
    \blue{$T_\text{regular}, T_\text{extra}$} & Number of regular/extra training\\
    & rounds \\
    \bottomrule
\end{tabular}
\end{center}
\caption{\textbf{Notation Table}}\label{table:notation_table}
\end{table}

GBDT~\cite{friedman2001greedy} learns a sequence of decision trees by iteratively correcting errors of prior trees.
Let $X = \{(x_1, y_1), \ldots, (x_n, y_n)\} \subseteq \mathbb{R}^{m} \times \mathbb{R}$ denote a labeled dataset with $n$ data points and $m$ features. For simplicity, we denote $y$ as a label although a regression target applies equally.
A tree ensemble model $\phi_k \coloneqq \begin{bmatrix} f_1, \dots, f_k \end{bmatrix}$ minimizes
\[
\textstyle \mathcal{L}_k(\phi, X) = \sum_{(x_i,y_i) \in X} l(\phi_k(x_i), y_i) + \sum_{t=1}^k \Omega(f_t)
\]
where $l$ is a twice-differentiable convex loss function, e.g. squared error or binary cross-entropy, that measures the difference between the prediction $\tilde{y}_i \coloneqq \phi_k(x_i)$ and the label $y_i$, and $\Omega(f_t) = \sfrac{1}{2}\lambda \lVert V_t \rVert^2$ is a regularization term on the leaves vector $V_t = \mathrm{Leaves}(f_t)$.

Ensemble $\phi_k$ uses $k$ trees $(f_t)_{t=1}^k$ to predict $\tilde{y}_i$ for $x_i$:
\[
\textstyle \tilde{y}_i = \phi_k(x_i) = \sum_{t=1}^k \eta \cdot \mathrm{clamp}(f_t(x_i), -\beta,\beta)
\]
where $\eta$ is a learning rate and $\beta$ a regularization parameter for the clipping routine $\mathrm{clamp}(z,-\beta,\beta) = \max(-\beta, \min(z, \beta))$.

For training each tree $f_t$, XGBoost \cite{chen2016xgboost} proposes Newton boosting, a second-order approximation of the loss function:
\begin{align*}
    &\textstyle\phantom{\approx~}l(\phi_t(x_i), y_i) = l(\phi_{t-1}(x_i) + f_t(x_i), y_i) \\
    &\textstyle\approx l(\phi_{t-1}(x_i), y_i) + g_t(\phi_t(x_i),y_i) f_t(x_i) \\ 
    &+ \frac12 h_t(\phi_t(x_i),y_i) f_t^2(x_i)
\end{align*}
with gradient $g_t(\Tilde{y}_i,y_i) = \sfrac{\partial}{\partial \Tilde{y}_i} l(\Tilde{y}_i, y_i)$ and Hessian $h_t(\Tilde{y}_i,y_i) = \sfrac{\partial^2}{\partial \Tilde{y}_i^2} l(\Tilde{y}_i, y_i)$.
If we have a squared error loss function like for regression, we have a closed form of the gradient $g_t(\tilde{y}_i, y_i) = y_i - \tilde{y}_i$ and Hessian $h_t(\tilde{y}_i, y_i) = 1$. If we have a binary cross-entropy loss like for classification, then $g_t(\tilde{y}_i, y_i) = y_i - \mathrm{sigmoid}(\tilde{y}_i)$ and $h_t(\tilde{y}_i, y_i) = -\mathrm{sigmoid}(\tilde{y}_i)\cdot(1- \mathrm{sigmoid}(\tilde{y}_i))$ with $\mathrm{sigmoid}(z) = (1 + \exp(-z))^{-1}$.

\boldparagraph{Optimal split.} Each tree $f_t$ is recursively built from a root node to the leaves: each node splits a dataset $X$ in a left $I_{L}$ and right child $I_{R}$ given some split criterion $s$. Each child with its remaining dataset is then the basis for the next subtree. The process stops until a stopping criterion, e.g. the maximal tree depth, is reached. An optima-preserving equivalent of the optimal split criterion for each split $s$ is derived as:
\begin{align*}
&\textstyle G_t(I_{L}, I_{R}) = \\
&\frac{\Big(\sum_{(x_i,y_i) \in I_{L}} g_t(\phi_t(x_i),y_i)\Big)^2}{\sum_{(x_i,y_i) \in I_{L}} h_t(\phi_t(x_i),y_i) + \lambda} \\
&+ \frac{\Big(\sum_{(x_i,y_i) \in I_{R}} g_t(\phi_t(x_i),y_i)\Big)^2}{\sum_{(x_i,y_i) \in I_{R}} h_t(\phi_t(x_i),y_i) + \lambda}
\end{align*}
where $\lambda$ is the regularization parameter of $\Omega$.

\boldparagraph{Optimal leaf value.} The leaves of $f_t$ contain the tree's prediction, which is derived \blue{with the Newton method} as: 
\begin{align}\label{eq:leaf_value}
\textstyle V_t(I_{\mathit{Leaf}}) = -\frac{\sum_{(x_i,y_i) \in I_{\text{Leaf}}} g_t(\phi_t(x_i),y_i)}{\sum_{(x_i,y_i) \in I_{\text{Leaf}}} h_t(\phi_t(x_i),y_i) + \lambda}
\end{align}%
for a subset $I_{\mathit{Leaf}}$ of training data $X$ that ended up in this leaf.

\subsection{Individual Rényi filter} \label{sec:prelim_renyifilter}
The conventional privacy accounting approach involves a worst-case analysis of the privacy loss, assuming the global sensitivity for all individuals. This results in an overly conservative estimation of the privacy loss, as often a data point with many similar data points in the dataset has an individual sensitivity that is smaller than the global sensitivity and only those data points which have few similar data points in the dataset utilize the full sensitivity.

An individual Rényi filter ~\cite{feldman2020individual}, is a way to implement personalized privacy accounting. The Rényi filter measures privacy losses individually via individual Rényi Differential Privacy (cf. \Cref{def:individual_privacy}) and guarantees that the privacy loss of no data point will surpass a predefined upper bound $\rho(\alpha)$. This enables a differentially private mechanism, comprising a composed sequence of mechanisms $M_1, M_2, ..., M_{T_\text{max}}$, to execute as many rounds as desired using data points that have not expended their privacy budget, as long as the accounting for individual Rényi privacy losses is sound.

\begin{definition}[Individual Rényi Differential Privacy]
    Fix $n \in \mathbb{N}$ and a data point $x_i$. A randomized mechanism $M$ satisfies $(\alpha, \rho(\alpha))$-individual Rényi differential privacy if for all neighboring datasets $X,X'$ that differ in $x_{i}$, denoted as $X \sim_{x_{i}} X'$, and satisfy $|X|, |X'| \leq n$, it holds that 
   $D_{\alpha}(M(X)||M(X') \leq \rho(\alpha)$
    \label{def:individual_privacy}
\end{definition}

\Cref{alg:adaptivecompwithprivfilt} shows the individual Rényi filter algorithm. The algorithm obtains the individual Rényi privacy loss of round $t$ for each data point $x_i$ (line \ref{alg:adaptivecompwithprivfiltrho}). Next, the algorithm filters out all the data points that have surpassed the predefined upper bound on the privacy loss $\rho(\alpha)$ (line \ref{alg:adaptivecompwithprivfiltfilter}) using the privacy filter from \Cref{thmprivacyfilter} and then continues execution of mechanism $M_t$ in round $t$ only on the active data points that have not been filtered out. \citet{feldman2020individual} show that \Cref{alg:adaptivecompwithprivfilt} satisfies $(\alpha, \rho(\alpha))$-Rényi DP (cf. \Cref{thm:adaptive_composition_individual_privacy}).

\begin{theorem}[Rényi Privacy Filter, \cite{feldman2020individual} Theorem 4.3]
    Let \[\textstyle\mathcal{F}_{\alpha, \rho(\alpha)}(\rho_1, \rho_2, \dots, \rho_k) = \begin{cases}
        \text{CONT, if $\sum_{i=1}^k \rho_i(\alpha) \leq \rho(\alpha)$} \\
        \text{HALT, if $\sum_{i=1}^k \rho_i(\alpha) > \rho(\alpha)$}
    \end{cases}\]
    
    where $\rho(\alpha)$ is the upper bound on the privacy loss and $\rho_i(\alpha)$ ($i \in \{1,2,\dots,k\}$) the individual privacy loss of a data point for round $i$. Then $\mathcal{F}_{\alpha, \rho}$ is a valid Rényi privacy filter.
    \label{thmprivacyfilter}
\end{theorem}

\begin{algorithm}
    \caption{Adaptive composition with individual privacy filtering (cf. \cite[Algorithm 3]{feldman2020individual})}
    \KwIn{$D : \text{dataset}$} 
    \SetKwInOut{KwIn}{\color{white}\phantom{Input}}
    \KwIn{$(M_1, M_2, \dots, M_{T_\text{max}}) : \text{sequence of mechanisms}$} 
    \KwIn{$\hat{\alpha} : \text{Rényi DP parameter}$}
    \KwIn{$\rho(\hat{\alpha}) : \text{upper bound on  Rényi DP privacy loss}$}
    \For{$t=1$ to $T_\text{max}$}{
        \For{$x_i \in D$}{ 
        $\rho_{t}^{(i)}(\alpha) := \sup{}_{X \sim{}_{x_i} X'}
        D_{\alpha}(M_{t}(a_1, \dots, a_{t-1}, X)$
        $|| M_{t}(a_1, \dots, a_{t-1}, X'))$ \label{alg:adaptivecompwithprivfiltrho} \\
        }
        Determine active set $D_{t} = (x_i : \mathcal{F}_{\hat{\alpha}, \rho(\hat{\alpha})}(\rho_{1}^{(i)}, \rho_{2}^{(i)}, \dots, \rho_{t}^{(i)}) = \text{CONT})$ \label{alg:adaptivecompwithprivfiltfilter} \\
        For all $x_i \in D$, set $\rho_{t}^{(i)} \leftarrow 0$ if $x_i \notin D_t$ \\
        Compute $a_t = M_{t}(a_1, \dots, a_{t-1}, D_t)$
    }
    \Return $(a_1, a_2, \dots, a_{T_\text{max}})$
    \label{alg:adaptivecompwithprivfilt}
\end{algorithm}

\begin{theorem}[Theorem 4.5 in \cite{feldman2020individual}]
    Adaptive composition with individual Rényi filters (~\Cref{alg:adaptivecompwithprivfilt}) using the Rényi filter from ~\Cref{thmprivacyfilter} satisfies $(\alpha, \rho(\alpha))$-Rényi differential privacy.
    \label{thm:adaptive_composition_individual_privacy}
\end{theorem}

The intuition for the proof of \Cref{thm:adaptive_composition_individual_privacy} is that, whether a data point $x$ is active in some round, does not depend on the rest of the input dataset but only on the outputs of prior rounds, which are known to the adversary. Once $x$ is excluded from training, it does not lose any more privacy, because the output of any round from there on looks the same whether $x$ is present in the dataset or not.

\subsection{Subsampled Rényi differential privacy}\label{sec:prelim_subsampling}

Privacy amplification by subsampling allows a stronger privacy guarantee when choosing the data points for a single training round randomly from the training dataset rather than training on a fixed subset of the training dataset or even on the whole training dataset. Privacy amplification by subsampling was first analyzed by \citet{li2011sampling}. In this work, we utilize Poisson Subsampling~\cite{balle2018privacy} where each data point is chosen according to a Bernoulli experiment with probability $\gamma$, to obtain a batch of training data.
We utilize the bound of ~\Cref{thm:subsampling_preliminaries} by \citet{pmlr-v97-zhu19c} for subsampled Rényi DP which demands that the to-be-subsampled mechanism satisfies a lower bound on its Pearson-Vajda $\mathcal{X}^{l}$ pseudo-divergence.

\begin{definition}[Pearson-Vajda pseudo-divergence, Vajda (1973)]
    Let $P, Q$ be two probability distributions defined over $\mathcal{R}$. Let $P(x)$, $Q(x)$ be the density of $P, Q$ at $x$. The Pearson-Vajda $\mathcal{X}^{l}$ pseudo-divergence is defined as
    $D_{\mathcal{X}^{l}}(P||Q) := \int_{-\infty}^{\infty} Q(x) \Big(\frac{P(x)}{Q(x)}-1\Big)^{l}$
\end{definition}

\begin{theorem}[Privacy Amplification by Subsampling, Theorem 8 in \cite{pmlr-v97-zhu19c}]
    Let $\mathcal{M}$ be any randomized mechanism that obeys $(\alpha, \rho'(\alpha))$-Rényi differential privacy. Let $\gamma$ be the subsampling ratio and $\alpha \geq 2$. Let $M^{\mathcal{P}_\gamma} = \mathcal{M} \circ \mathcal{P}_\gamma$ and $\mathcal{P}_\gamma$ generating a Poisson subsample with subsampling ratio $\gamma$. 
    If for all neighboring datasets $X \sim X'$ and all odd $3 \leq l \leq \alpha$,
    $D_{\mathcal{X}^{l}}(M(X)||M(X')) \geq 0$
    then $M^{\mathcal{P}_\gamma}$ is tightly $(\alpha, \rho(\alpha))$-Rényi differentially private with $\rho(\alpha) = \frac{1}{\alpha-1} \log{}\bigg((1-\gamma)^{\alpha-1} (\alpha\gamma - \gamma +1) + \sum_{l=2}^{\alpha} \binom{\alpha}{l} (1-\gamma)^{\alpha - l} \gamma^l e^{(l-1) \cdot \rho'(l)}\bigg)$.
    \label{thm:subsampling_preliminaries}\label{thmsubsampling}
\end{theorem}

%% file: problem_statement.tex
Our task is to build a differentially private gradient boosted decision tree ensemble (GBDT). As indicated in the preliminaries, a GBDT has two primary data-dependent parts: splits and leaves. \blue{In prior work, the splits are selected data-independent, e.g. a random split selection, whereas the leaves are built with an additive noise mechanism. \dpgbdt uses both techniques as a foundation.}

\subsection{Random split selection}\label{sec:prelim_randomsplit}
A decision tree ensemble with randomized splits can yield good utility~\cite{bojarski2014differentially} and limits the privacy leakage to learning leaf-values. Random splits are constructed by randomly selecting a feature of the dataset and a value for that feature as a split.
For a numerical feature $i$, we assume a fixed feature range $(v^{(i)}_{\text{min}}, v^{(i)}_{\text{max}})$ for sampling.

\citet{bojarski2014differentially} propose to use random splits for binary classification using random forests: As the trees in the ensemble are random, most trees do not improve the decision to which class a data point belongs. This yields an almost even distribution between trees predicting class `0' and class `1'.
With high probability, however, a few random trees do improve the class-prediction, which suffices to tilt the prediction of the ensemble toward the right class.

\citet{nori2021accuracy} propose \emph{cyclical feature interaction}: Split in all nodes of each tree $t$ on one feature $i \in \{1,\dots,m\}$ only and the next tree splits on the next feature ($i = t \mod m$). This technique boosts the performance in a DP setting.

\citet{Maddock_2022} also investigate random splits and propose different candidate selections: one variant is based on candidates that are chosen equidistantly from a predefined split candidate set. This prevents too fine-grained splits but is also not flexible. Another variant is called \emph{iterative Hessian} which refines an initial equidistant split candidate set during ensemble training using tree-specific information: the aggregated Hessian as used in the leaf. For each split candidate, the aggregated Hessian is calculated: small values indicate data point absence which is handled by merging split candidates while large values indicate a large data point density which is handled by further subdividing this split candidate. A differentially private aggregated Hessian is obtained the same as in the leaf. This method refines the split candidates for the first $s$ training rounds and increases the privacy budget for our algorithm as follows. If \dpgbdt without split refinement is $(\alpha, \rho_1(\alpha))$-RDP, then \dpgbdt with split refinement is $(\alpha, \rho_1(\alpha) + s\rho_2(\alpha))$-RDP where $\rho_2(\alpha)$ is the RDP bound for only releasing the Hessian. The bound is derived via sequential composition (cf. \Cref{thmrdpsequentialcomposition}) and because the split refinement uses all training data. Our experiments show that the split refinement does not improve the utility of \dpgbdt.

\subsection{Leaf Noising}\label{sec:prelim_leafnoising}
The leaf value $V$ is stated in \Cref{eq:leaf_value}. Prior work~\cite{li2020privacy} proposes to additively noise each leaf value $V$ with Laplace noise proportional to the sensitivity $s$ of the leaf value: $V + \mathrm{Lap}(0,\sfrac{s}{\eps_{\text{leaf}}})$. Since each data point is only in one leaf of a tree, parallel composition applies which means that the privacy budget for each leaf is the same as for all leaves of a tree. While they provide a loose sensitivity to the leaf value, \citet{Maddock_2022} propose to Gaussian noise the numerator, i.e. gradient sum, and denominator, i.e. Hessian sum, individually. Thus, we only need to bound each sensitivity individually which is significantly less loose. For classification and due to its loss function, each gradient $g$ is naturally bounded between $[-1,1]$ and each Hessian $h$ between $[0,0.25]$. For regression, each gradient is unbounded but each Hessian is exactly $1$. To accommodate both settings, we clamp each gradient $g$ by $g^*$ and each Hessian $h$ by $h^*$ which can even for classification boost the utility-privacy tradeoff as the noise scales with the clipping bound. Concretely, the sensitivity of a clamped version of the gradient sum $\sum_{(x_i,y_i) \in I_{\text{Leaf}}} \mathrm{clamp}(g_t(\phi_t(x_i),y_i), -g^*, g^*)$ is bounded by a fixed $g^*$ where $\mathrm{clamp}(z,-g^*,g^*)$ clips input $z$ to $[-g^*, g^*]$ and $I_{\text{Leaf}}$ are the data points within a leaf. Similarly, the sensitivity of a clamped version of the Hessian sum $\sum_{(x_i,y_i) \in I_{\text{Leaf}}} \mathrm{clamp}(h_t(\phi_t(x_i),y_i), -h^*, h^*)$ is bounded by $h^*$.

\blue{
\subsection{Distributed learning: Batched Updates}\label{sec:prelim_batchedupdates}
To reduce the communication cost in distributed GBDT training, \citet{Maddock_2022} propose batched updates: A technique that only syncs the prediction updates among the users and thus updates the gradient and Hessian every $B$ batches of the $T$ total rounds. For syncing, $B$ many prediction updates are averaged. This technique effectively reduces the communication cost from $\mathcal{O}(T)$ to $\mathcal{O}(\sfrac{T}{B})$ as only $\sfrac{T}{B}$ many trees are trained sequentially and the others in parallel like in a random forest. Utility-wise, less information is used, yet \citet{Maddock_2022} suggest in their experiments no significant utility loss and a slight utility advantage in the high-privacy regime ($\eps \le 0.1$) for small batches. Privacy-wise, the accounting remains the same as the same number of trees are released.
}

\subsection{Discussion of closely related work}

\blue{
\boldparagraph{DP GBDT learners.} 
\citet{Maddock_2022} establish GBDTs as a promising direction for differentially private machine learning. We contribute significantly to the utility-privacy tradeoff by privacy amplification by \emph{subsampling} and our \emph{leaf-balanced noise}. We pose a solution to learning on a stream of non-IID data using Rényi filters.

\begin{itemize}[leftmargin=*]
    \item Subsampling does not improve \citet{Maddock_2022} without our tight bounds in \Cref{cor:pearson_vajda_condition_general} as it either leads to an untight generic bounds (e.g. a factor 15 worse in one of our runs) or requires univariate noise (due to a missing condition that we prove). As we add 2-dimensional noise per leaf in Newton Boosting, a subsampled \citet{Maddock_2022} would need to fall back to an inefficient gradient boosting like in \cite{li2020privacy} where we add noise per leaf proportional to $\mathcal{O}(g^*)$ instead of $\mathcal{O}(\sfrac{g^*}{n})$ with the gradient clipping bound $g^*$ and the number of data points in a leaf $n$.
    \item Leaf-balanced noise does not improve \citet{Maddock_2022} without our tight bounds in \Cref{thm:individual_rdp_nonspherical_gauss} as prior mechanisms like MVG \cite{chanyaswad2018mvg} do not provide tight privacy bounds (e.g. RDP). As in MVG no RDP bounds are provided, other improvements like subsampling and individual Rényi filters remain an open challenge.
    \item \citet{Maddock_2022} do not investigate individual Rényi filters, training an initial score, or learning on a stream of non-IID data.
    \item For distributed learning, we use similar techniques as in \cite{Maddock_2022}.
\end{itemize}

}

\boldparagraph{Subsampling for individual RDP.} The work by \citet{yu2023individual} applies individual privacy accounting to DP-SGD and shows that individual RDP works seamlessly with privacy amplification by subsampling. Their work uses the moments accountant \cite{Abadi_2016} which can be applied to spherical multivariate Gaussians and needs numerical evaluation that can suffer from imprecision compared to an analytical bound. In this work, we \emph{prove} a tight RDP bound for subsampling by utilizing the analytical bound by \citet{pmlr-v97-zhu19c}. We also generalize the condition required for using the bound by \citet{pmlr-v97-zhu19c} to a novel condition which only states that the RDP bound of the to-be-subsampled mechanism has to be linear in RDP $\alpha$. Using our novel condition, we can precisely evaluate the bound of \citet{pmlr-v97-zhu19c} even for our leaf-balanced noise which resembles a non-spherical multivariate Gaussian.

The original individual Rényi Filter work \cite{feldman2020individual} does not cover subsampled RDP bounds. A follow-up work by \citet{koskela2023individualgdp} applies individual Gaussian DP to a non-subsampled DP-GD; however, their work does cover a subsampled DP-SGD by a numerical implementation of a privacy loss accountant instead of an analytical bound as we have.

\boldparagraph{Leaf-balanced noise.} Prior work \cite{Maddock_2022} on DP GBDTs does not offer the noise weighting of our leaf-balanced noise. There is work on using non-spherical multivariate Gaussian noise for DP: MVG \cite[Theorem 3]{chanyaswad2018mvg} proposes differential privacy accounting of matrix-variate Gaussian noise. For that, they bound the product of the norms of the singular values of the inversed covariance matrices proportional to $\mathcal{O}(\eps)$ and $\mathcal{O}(ln^2 \delta)$. This offers a more general setting, as \dpgbdt only handles diagonal multi-variate Gaussian noise\footnote{Our leaf-balanced noise could be generalized if the covariance matrix is public, since rotations preserve distances and thus the sensitivity remains unchanged.}, yet 
our work provides tight RDP accounting which is essential for applying the tight subsampling bounds (cf. \Cref{sec:exactdp_subsampling}) as well as individual RDP accounting (cf. \Cref{thm:individual_rdp_nonspherical_gauss}), which has not been covered by MVG.

Another work \cite[Lemma 1]{mcmahan2017learning} proposes a technique that noises the ratio of gradient sum to subsampled batch size (in our case the Hessian sum) once instead of each sum individually. For that, they lower bound the subsampled batch size. However, this only works if the subsampled batches have a similar size. In the case of GBDTs, we have a widely varying Hessian sum, where a fixed lower bound would heavily deteriorate utility. In fact, \citet{li2020privacy} assume such a scenario of a subsampled batch size of $1$ (better: $0$) in each leaf.

%% file: building_blocks.tex
In these sections, we present novel technical components for individual Rényi DP that might be of independent interest.

%% file: subsampling.tex
We employ subsampling to boost the utility-privacy tradeoff. To also facilitate an individual Rényi filter we need subsampling bounds in individual Rényi DP accounting, for which prior work \cite{yu2023individual} has shown that an individual Rényi filter works seamlessly with subsampled individual Rényi DP. 

For subsampled RDP, there exists a generic analytical bound by \citet{pmlr-v97-zhu19c} (cf.  \Cref{thm:subsampling_preliminaries}) which is tight if the Pearson-Vajda $\mathcal{X}^l$ pseudo-divergence $D_{\mathcal{X}^l}(M(D_{0})||M(D_{1}))$ is semi-positive for all neighboring datasets $D_0 \sim D_1$, the to-be-subsampled mechanism $M$ and all odd $3 \leq l \leq \alpha$. When $\gamma \, \alpha \, e^{\rho(\alpha)} \ll 1$ the tight bound and the general upper bound \cite[Theorem 5]{pmlr-v97-zhu19c} match up to a multiplicative factor of $1 + \mathcal{O}(\gamma \, \alpha \, e^{\rho(\alpha)})$ for subsampling ratio $\gamma$ and Rényi-DP bound of the to-be-subsampled mechanism $(\alpha, \rho(\alpha))$. In all other cases, the bounds match up to an additive factor of $\frac{\log{(3)}}{\alpha-1}$. For our best-performing hyperparameter setting of \dpgbdt{} on the Abalone regression dataset for $\varepsilon=0.1$, we observe an improvement in $\rho(\alpha)$ from $0.7443$ to $0.04705$ for $\alpha=241$ when using the tight bound compared to the general upper bound.

\citet[Theorem 17]{pmlr-v97-zhu19c} provide a proof for one-\\dimensional Gaussians that enables usage of the tight subsampling bound. Consequently, this proof does not apply to more general mechanisms, including mechanisms using our leaf-balanced noise, a non-spherical multivariate Gaussian. We provide in \Cref{cor:pearson_vajda_condition_general} a general proof of the Pearson-Vajda $\mathcal{X}^l$ pseudo-divergence condition with one notable requirement: \blue{To apply the tight subsampling bound of \Cref{thm:subsampling_preliminaries}~\cite{pmlr-v97-zhu19c}, the to-be-subsampled mechanism $M$ has to have a RDP bound $\rho(\alpha)$ that is linear in $\alpha$ and semi-positive.
This holds for variants of the Gaussian mechanism like multi-variate or non-spherical Gaussian, cf. \Cref{thm:individual_rdp_nonspherical_gauss}, but not Laplace or randomized response (cf. \cite[Table 2]{DBLP:journals/corr/Mironov17}).}

\begin{lemma}
    Let $M$ be the to-be-subsampled mechanism that satisfies $(\alpha,\rho(\alpha))$-Rényi DP and $X \sim X'$ two neighboring datasets. If $\rho(\alpha) = \alpha \cdot t$ for any $t\ge 0$, then the Pearson-Vajda $\mathcal{X}^{\alpha}$ pseudo-divergence of $M$ is semi-positive, i.e. $D_{\mathcal{X}^{\alpha}}(M(X)||M(X')) \geq 0$, for all odd $\alpha \geq 1$.
    \label{cor:pearson_vajda_condition_general}
\end{lemma}
\begin{proof}
    By induction, we show that for all odd $\alpha \geq 1$, $M$ satisfies the condition for any neighboring datasets $X \sim X'$. Let $o$ be an observation of $M$ and $p_{M(X)}, p_{M(X')}$ be the densities of $M(X), M(X')$.
   
    \begin{align*}
       &\text{\textbf{Base case ($\alpha=1$):}} \\ 
       &\textstyle\phantom{=~} D_{\mathcal{X}^{\alpha}}(M(X)||M(X'))
        = \Econd{M(X')}{\frac{M(X)}{M(X')} - 1}^{\alpha} \\
        &\textstyle= \int_{o} p_{M(X')}(o) \cdot \left(\frac{p_{M(X)}(o)}{p_{M(X')}(o)} - 1\right) \, do \\
        &\textstyle= \int_{o} p_{M(X)}(o) - p_{M(X')}(o) \, do = 0
    \end{align*}
    \boldparagraph{Inductive step ($\alpha-2 \to \alpha$):}
    Note that by rewriting the Rényi divergence we have
    \[\textstyle\Econd{M(X')}{\frac{M(X)}{M(X')}}^{\alpha} = \exp{((\alpha-1) \cdot D_{\alpha}(M(X)||M(X')))}.\]
    Using this equality and the binomial identity we get
    \begin{align*}
        &\textstyle D_{\mathcal{X}^{\alpha}}(M(X)||M(X')\textstyle) = \sum_{l=0}^{\alpha} \binom{\alpha}{l} (-1)^{\alpha-l} \, \Econd{M(X')}{\frac{M(X)}{M(X')}}^{l}\\
        &\textstyle= \sum_{l=0}^{\alpha} \binom{\alpha}{l} (-1)^{\alpha-l} \, \exp{\left((l-1) \cdot \rho(l)\right)} \\
        &\textstyle= \alpha - 1 + \sum_{l=2}^{\alpha} \binom{\alpha}{l} (-1)^{\alpha-l} \, \exp{\left((l-1) \cdot l \cdot t\right)}.
    \end{align*}
    We investigate the partial derivative in $t$ of the Pearson-Vajda $\mathcal{X}^{\alpha}$ pseudo-divergence:
    \begin{align*}
        &\textstyle\phantom{=~} \frac{\partial}{\partial t} D_{\mathcal{X}^{\alpha}}(M(X)||M(X')) \\
        &\textstyle= \sum_{l=2}^{\alpha} \binom{\alpha}{l} (-1)^{\alpha-l} \exp{((l-1) \cdot l \cdot t)} \cdot (l-1) \cdot l \\
        &\textstyle= \sum_{l=2}^{\alpha} \binom{\alpha}{l} (-1)^{\alpha-l} \, \Econd{M(X')}{\frac{M(X)}{M(X')}}^{l} \cdot (l-1) \cdot l
        \intertext{Since $\binom{\alpha}{l} \cdot \binom{l}{2} = \binom{\alpha}{2} \cdot \binom{\alpha-2}{l-2}$, for $\tilde{l} \coloneqq l - 2$ and $\tilde{\alpha} \coloneqq \alpha(\alpha-1)$}
        &\textstyle= \tilde{\alpha} \sum_{l=2}^{\alpha} \binom{\alpha-2}{l-2} (-1)^{\alpha-l} \, \Econd{M(X')}{\frac{M(X)}{M(X')}}^{l} \\
        &\textstyle= \tilde{\alpha} \sum_{\tilde{l}=0}^{\alpha-2} \binom{\alpha-2}{\tilde{l}} (-1)^{\alpha-2-\tilde{l}} \, \Econd{M(X')}{\frac{M(X)}{M(X')}}^{\tilde{l}+2} \\
        &\textstyle= \tilde{\alpha} \Econd{M(X')}{\left(\frac{M(X)}{M(X')}\right)^{2} \sum_{\tilde{l}=0}^{\alpha-2} \binom{\alpha-2}{\tilde{l}} (-1)^{\alpha-2-\tilde{l}}\left(\frac{M(X)}{M(X')}\right)^{\tilde{l}} }
        \intertext{Using the binomial identity we get}
        &\textstyle= \tilde{\alpha} \Econd{M(X')}{ \left(\frac{M(X)}{M(X')}\right)^{2} \left(\frac{M(X)}{M(X')} - 1\right)^{\alpha-2} }
        \intertext{By \citet[Lemma 16]{pmlr-v97-zhu19c} we lower bound}
        &\textstyle\geq \tilde{\alpha} \Econd{M(X')}{ \left(\frac{M(X)}{M(X')} - 1\right)^{\alpha-2}} \\
        &\textstyle= \underbrace{\alpha (\alpha-1)}_{\geq 0} \, \cdot \underbrace{D_{\mathcal{X}^{\alpha-2}}(M(X)||M(X'))}_{\geq 0 \text{ by inductive assumption}} \geq 0
    \end{align*}
    Since the partial derivative in $t$ is semi-positive and for $t=0$ we have $\Econd{M(X')}{\left(\frac{M(X)}{M(X')} - 1\right)^{\alpha}} = 0$, we conclude that \break $\Econd{M(X')}{\left(\frac{M(X)}{M(X')} - 1\right)^{\alpha}} \geq 0 ~ \forall_{t\ge0}$. 
    Thus by induction, we conclude that $\Econd{M(X')}{\left(\frac{M(X)}{M(X')} - 1\right)^{\alpha}} \geq 0 ~ \forall_{t\ge0} \forall_{\alpha \geq 1, \alpha \text{ odd}}$.
\end{proof}

If $\rho(\alpha) = \alpha \cdot t$ then the Pearson-Vajda pseudodivergence precondition of \Cref{thm:subsampling_preliminaries} is fulfilled, thus the conclusion of \Cref{thm:subsampling_preliminaries} applies, i.e. the subsampled mechanism is RDP with
   $\rho(\alpha) = \frac{1}{\alpha-1} \log{}\big((1-\gamma)^{\alpha-1} (\alpha\gamma - \gamma +1) + \sum_{l=2}^{\alpha} \binom{\alpha}{l} (1-\gamma)^{\alpha - l} \gamma^l e^{(l-1) \cdot \rho'(l)}\big)$.

%% file: leaf_balanced_noise.tex
To release the value of a leaf in GBDT in a differentially private manner, we consider the function $f(X) = (w, u)$ for an arbitrary dataset $X$. $f$ clips the individual gradients and Hessians from $X$ to bound the sensitivities and computes the sum of gradients $u$ and the sum of Hessians $w$. To preserve differential privacy, Gaussian noise is calibrated to the privacy budget and the clipping bound and applied to the output of $f$. The leaf value is then constructed by post-processing on the released noisy values, $\tilde{u}/\tilde{w}$.

When noising $f$, we utilize leaf-balanced non-spherical multivariate Gaussian noise with covariance matrix 
\[\textstyle\Sigma = \diag\Big(\frac{(h^{*})^{2}\sigma_\text{leaf}^2}{2\cdot r_1}, \frac{(g^{*})^{2}\sigma_\text{leaf}^2}{2\cdot r_2}\Big)\]
with $r_1, r_2$ satisfying $r_1+r_2=1$. We sample $Y \sim \mathcal{N}(0, \Sigma)$ and then release $f(X) + Y = (w, u) + Y$.  Here, $r_1, r_2$ allow us to better calibrate the amount of noise that is applied to the two values $w$ and $u$. 

We believe that our adaptation of the Gaussian mechanism with leaf-balanced noise is of independent interest, as its idea of finetuning the shares of noise over the dimensions of the output can be helpful for different functions even with an arbitrary number of output dimensions. We analytically derive individual Rényi DP bounds for this mechanism in \Cref{thm:individual_rdp_nonspherical_gauss}. 

\begin{theorem}[Individual RDP of Gaussian Mechanism with leaf-balanced non-spherical noise]
    Let $f \colon \mathcal{X} \rightarrow \mathbb{R}^D$ denote a function from an arbitrary input $X \in \mathcal{X}$ to a set of scalars. Let the function $f$ projected to its $d$-th output have bounded sensitivity $s_d \in \mathbb{R}$ and bound individual sensitivity $s^{(i)}_d \in \mathbb{R}$. Let $Y \sim \mathcal{N}(\mathbf{0},\Sigma)$ be a random variable of non-spherical multivariate Gaussian noise with covariance matrix $\Sigma \coloneqq D^{-1}\cdot\diag(r_1^{-1} s_{1}^{2} \sigma^2, \dots, r_D^{-1} s_{D}^{2} \sigma^2 )$ where a given variance $\sigma^2 \in \mathbb{R}^+$ is weighted in each dimension $d$ individually by $s_{d}^{2}$ and $r_d \in \mathbb{R}^+$ such that $\sum_{d=1}^D r_d = 1$. Then the Gaussian mechanism $M(X) \mapsto \set{f(X)_d + Y_d}_{d=1}^D$ satisfies $(\alpha, \rho(\alpha))$-individual RDP for $x_i \in X$ with $\rho(\alpha) = \alpha\cdot \frac{D}{2\sigma^2} \cdot \sum_{d=1}^D \frac{r_d\cdot (s^{(i)}_d)^2}{s_{d}^{2}}$.
    \label{thm:individual_rdp_nonspherical_gauss}
\end{theorem}
\begin{proof}
     We prove $(\alpha,\rho(\alpha))$-individual RDP as follows: (1) We derive the privacy loss distribution (PLD) of the $d$-th output of $M(X)$. (2) We use that the PLD of a $D$-fold sequential composition of $M$ is the same as a $D$-fold convolution of the $1$-dimensional PLD. (3) We show $(\alpha,\rho(\alpha))$-individual RDP via the convoluted PLD, which constitutes a Gaussian distribution.

    (1) Let $X \sim_{x_i} X'$ be neighboring datasets differing in $x_i$. Let $o$ be an atomic event, and $\mathcal{L}_{M(X)/M(X')}(o) \coloneqq \ln(\frac{\Pr[M(X) = o]}{\Pr[M(X') = o]})$ be the privacy loss of mechanism $M$. Then, the privacy loss distribution $\omega(y)$ on support $y \in Y \coloneqq \bigcup_{o} \set{\mathcal{L}_{M(X)/M(X')}(o)}$ as defined in \citet[Definition 2]{sommer2019privacy} resembles a probability distribution over the privacy losses $\mathcal{L}_{M(X)/M(X')}$. Prior work has shown that if there are worst-case distributions, there is a PLD that fully describes the leakage of any mechanism. Any additive mechanism for a sensitivity-bounded query has worst-case distributions: a pair of Gaussians of which one is shifted by the sensitivity. Since $M$ is the Gaussian mechanism, we get a closed form for $\omega^d$ on the $d$-th output of $M$ for $X\sim_{x_i} X'$~\cite[Lemma 11]{sommer2019privacy}:
    \[
        \textstyle\omega^d \sim \mathcal{N}\Big(\frac{r_d\cdot D \cdot (s^{(i)}_d)^2}{2\sigma^2 s_{d}^{2}}, \frac{r_d\cdot D \cdot (s^{(i)}_d)^2}{\sigma^2 s_{d}^{2}}\Big).
    \]
    Note: $\omega^d$ is the same for both privacy losses $\mathcal{L}_{M(X)/M(X')}$ and $ \mathcal{L}_{M(X')/M(X)}$ since a Gaussian is symmetric.

    (2) One characteristic of a PLD is that a $D$-fold sequential composition corresponds to a $D$-fold convolution of $\omega$~\cite[Theorem 1]{sommer2019privacy}. Hence, the convolution of two Gaussians is Gaussian again: $\mathcal{N}(\mu_x, \sigma_x^2) + \mathcal{N}(\mu_y, \sigma_y^2) = \mathcal{N}(\mu_x + \mu_y, \sigma_x^2 + \sigma_y^2),~\forall_{\mu_x,\mu_y,\sigma_x,\sigma_y}$. Thus, we have the following closed form for $\omega$ on the $D$-dimensional output of $M(X)$:
    \begin{align*}
        \textstyle\omega \sim \mathcal{N}\Big(\sum_{d=1}^D \frac{r_d\cdot D \cdot (s^{(i)}_d)^2}{2\sigma^2 s_{d}^{2}}, \sum_{d=1}^D \frac{r_d\cdot D \cdot (s^{(i)}_d)^2}{\sigma^2 s_{d}^{2}}\Big)
        \iff \\
        \textstyle\omega \sim \mathcal{N}\Big(\underbrace{\textstyle\frac{D}{2\sigma^2} \cdot\sum_{d=1}^D \frac{r_d\cdot (s^{(i)}_d)^2}{s_{d}^{2}}}_{\eqqcolon \mu_\omega}, \underbrace{\textstyle\frac{D}{\sigma^2} \cdot\sum_{d=1}^D \frac{r_d\cdot (s^{(i)}_d)^2}{s_{d}^{2}}}_{\eqqcolon \sigma_\omega^2}\Big).
    \end{align*}
        (3) We convert the convoluted PLD $\omega$ to an $(\alpha,\rho(\alpha))$-RDP bound as follows:
    \begin{align*}
        \textstyle(\alpha, \rho(\alpha)) & = (\alpha, \mathcal{D}_\alpha(M(X), M(X') ))
        \intertext{by~\cite[Lemma 8]{sommer2019privacy} we get}
        &\textstyle= (\alpha, \frac{1}{\alpha - 1} \ln\left(\Eop_{y \sim \omega}e^{(\alpha - 1) y}\right))
        \intertext{by the moment generating function:}
        \intertext{$\Eop_{Y \sim \mathcal{N}(\mu,\sigma^2)} [e^{tY}] = e^{t\mu + 0.5\sigma^2t^2},~\forall_{t,\mu,\sigma}$ we get}
        &\textstyle= (\alpha, \mu_\omega + 0.5\sigma_\omega^2(\alpha-1))
        \intertext{since $\sigma_\omega^2 = 2 \mu_\omega$ we conclude}
        &\textstyle= (\alpha, \alpha\cdot\mu_\omega)\\[-6ex]
    \end{align*}
\end{proof}

\blue{\begin{corollary}
    The Gaussian mechanism with leaf-balanced non-spherical noise satisfies $(\alpha, \alpha\cdot \frac{D}{2\sigma^2})$-RDP.
    \label{thm:rdp_nonspherical_gauss}
\end{corollary}
\begin{proof}
    The corollary follows directly from ~\Cref{thm:individual_rdp_nonspherical_gauss} for a challenge data point with worst-case sensitivity $s_d$.
\end{proof}}

We apply \Cref{thm:individual_rdp_nonspherical_gauss} to \dpgbdt{} in \Cref{cor:dpleaf_individual_rdp_main_body} where we have a bivariate case of non-spherical Gaussian noise. In this way, \dpgbdt{} utilizes a leaf-balanced noise. We experimentally show (cf. \Cref{sec:main_results} and \Cref{sec:impact_of_parameters}) that our leaf-balanced noise with $r_1 \neq r_2$ performs better than the standard spherical
Gaussian mechanism where $r_1 = r_2 = 0.5$.

%% file: tight_dpgbdt.tex
\label{sec:sgbdt}
\blue{
In this section, we detail our \dpgbdt algorithm and our improvements to DP GBDTs:
we start with \Cref{algtrainsgbdt} as the starting point of \dpgbdt training (cf. overview \Cref{sec:alg_overview}) and continue with \Cref{algdpinitscore} as our \emph{differentially private initial score}  (cf. \Cref{sec:initial_score}) and with \Cref{algtrainsingletree} as the training routine of a single tree with \emph{subsampling} and random splits (cf. \Cref{sec:single_tree}). Every single tree utilizes our \emph{leaf-balanced noise} which we describe together with the leaf value calculation in \Cref{algdpleafvalue} (cf. \Cref{sec:leaf}). Then, we finish by presenting the individual Rényi DP accountant in \Cref{alg:initialization} (cf. \Cref{sec:accounting}) and our two generalizations: Learning on a stream of non-IID data (cf. \Cref{sec:generalization_streams}) and scalable distributed learning (cf. \Cref{sec:distributed_learning}).
}

\subsection{Algorithm overview}\label{sec:alg_overview}

\blue{
\dpgbdt trains a GBDT ensemble as follows (cf. ~\Cref{algtrainsgbdt}). After initialization of the individual Rényi DP accountant (line \ref{algtrainsgbdtinitialize}), we compute the initial score from the labels in the dataset using ~\Cref{algdpinitscore} (line \ref{algtrainsgbdtcalldpinitscore}) and add it to the ensemble $E$. \dpgbdt then runs $T_{\max} = T_\text{regular} + T_\text{extra}$ training rounds where in each round we 1) apply the individual Rényi filter by computing individual RDP bounds for all data points using \Cref{thm:train_single_tree_irdp_main_body} (line \ref{alg:train_sgbdt_individual}) such that we filter out those data points that have exceeded their privacy budget using \Cref{thmprivacyfilter} (line \ref{alg:train_sgbdt_filter}) and 2) train a single tree on the filtered dataset using \Cref{algtrainsingletree} (line \ref{alg:train_sgbdt_train_single_tree}) and add it to the ensemble $E$.
In ~\Cref{thm:maintheorem} we show that ensemble training with \dpgbdt satisfies $(\alpha, \rho(\alpha))$-RDP.
}

\begin{algorithm}[t]
    \newcommand{\single}[1]{\texttt{TrainSingleTree}}
    \caption{TrainSBDT : Train a DP GBDT ensemble}
    \SetKwFunction{Initialize}{Initialize}
    \SetKwFunction{TrainSingleTree}{TrainSingleTree}
    \SetKwFunction{PrivacyAmplificationBySubsampling}{PrivacyAmplificationBySubsampling}
    \SetKwFunction{ComputeInitScore}{ComputeInitScore}
    \SetKwFunction{ComputeGradients}{ComputeGradients}
    \SetKwFunction{DPInitScore}{DPInitScore}
    \SetKwFunction{UpdateIndividualPrivacyLoss}{UpdateIndividualPrivacyLoss}
    \SetKwFunction{Filter}{Filter}
    \SetKwFunction{PoissonSubsample}{PoissonSubsample}
    \KwIn{$D$ : training data, $\gamma$ : subsampling ratio}
    \SetKwInOut{KwIn}{\color{white}\phantom{Input}}
    \KwIn{$\varepsilon_\text{init}$ : privacy budget for initial score}
    \KwIn{$(\varepsilon_\text{trees}, \delta_\text{trees})$ : DP parameters for training of trees} 
    \KwIn{$(r_1, r_2)$ : noise weights for leaf value} 
    \KwIn{$(g^*, h^*, m^*)$ : gradient/Hessian/label clipping bound}
    \KwIn{$\lambda, \beta$ : regularization parameters}
    \KwIn{$(T_\text{regular}, T_\text{extra})$ : number of regular and extra rounds}
    \KwIn{$\alpha_{\max}$ : largest $\alpha$ to test in Rényi DP, $d$ : depth of trees}
    $T_\text{max} = T_\text{regular} + T_\text{extra}$ \\
    $\hat{\alpha}, \sigma_\text{leaf}^2, \rho(\hat{\alpha}) =$ \Initialize{$\alpha_{\max}$, $(\varepsilon_\text{trees}, \delta_\text{trees}), \varepsilon_\text{init}, \gamma$} \label{algtrainsgbdtinitialize} \\
    $\text{init}_0= \text{\DPInitScore{$D, m^*, \varepsilon_\text{init}$}}$ \label{algtrainsgbdtcalldpinitscore} \\
    $E = (\text{init}_0)$ \\
    \For{$t = 1$ to $T_\text{max}$}{
    \For{$i=1$ to $|D|$}{ \label{algtrainsgbdtfilterstart}
        $\rho_t^{(i)}(\alpha) := a_\gamma(\alpha, \frac{\alpha}{\sigma_\text{leaf}^{2}} \cdot \Big(\frac{r_1\cdot |h_i|^2}{(h^{*})^{2}} + \frac{r_2\cdot |g_i|^2}{(g^{*})^{2}}\Big))$\tcp*{by~\Cref{thm:train_single_tree_irdp_main_body} \blue{(for RDP: \Cref{cor:train_single_tree_rdp_main_body})}}\label{alg:train_sgbdt_individual}
    }
     $D_{t} = (x_i : \mathcal{F}_{\hat{\alpha}, \rho(\hat{\alpha})}(\rho_{1}^{(i)}, \dots, \rho_{t}^{(i)}) = \text{CONT})$\tcp*{by~\Cref{thmprivacyfilter}}
     \label{alg:train_sgbdt_filter} \label{algtrainsgbdtfilterend}
     $\text{tree}_t =$ \TrainSingleTree{$D_t, d,$
     $\sigma_\text{leaf}^2, g^*, h^*, (r_{1}, r_{2}), \lambda, \beta, E$} \label{alg:train_sgbdt_train_single_tree} \\
     $E = (\text{init}_0, \text{tree}_1, ..., \text{tree}_t)$
     }
     \Return E
    \label{algtrainsgbdt}
\end{algorithm}

\subsection{Initial score}\label{sec:initial_score}
GBDT ensemble training needs an initial score so that the ensemble can then add trees to improve on the error of the initial base classifier that outputs the initial score. Prior works chose the initial score simply as $0.0$, however, it can be beneficial (cf. \blue{Experiment \Cref{sec:main_results,sec:impact_of_parameters}}) that the initial score gives a more meaningful starting point for the ensemble. We output the mean of labels from the dataset as the initial score and release this mean with the Laplace mechanism to preserve $(\alpha, \rho(\alpha))$-Rényi DP, which we show in \Cref{thmdpinitscore}:

\blue{
\begin{theorem}
    \input{theorems/initscore_rdp}
    \label{thmdpinitscore}
\end{theorem}
\begin{proof}
    \reftoarxiv{} 
    It uses sequential composition and a generalization of \citet[Proposition 6]{DBLP:journals/corr/Mironov17} for arbitrary sensitivity.
\end{proof}
}

\begin{algorithm}
    \caption{DPInitialScore : Compute a DP initial score}
    \SetKwFunction{Laplace}{Laplace}
    \SetKwFunction{clamp}{clamp}
    \KwIn{$D$ : training dataset, $m^*$ : clipping bound on labels}
    \SetKwInOut{KwIn}{\color{white}\phantom{Input}}
    \KwIn{$\varepsilon_\text{init}$ : DP privacy budget for initial score}
    \KwIn{\blue{$\eps_{\text{ds}} = 0.005$: DP privacy budget for dataset size}}
    \BlankLine
    \blue{$|D|_{\text{priv}} = |D| + \text{\Laplace{$0, 1/\eps_{\text{ds}}$}}$} \label{dpreleaseds} \\
    \If{regression}{
        $M = \frac{1}{|D|_{\blue{\text{priv}}}}\sum_{y_i \in D} \text{\clamp{$y_i, -m^*, m^*$}}$ \label{algdpinitscoreclipmean} \\
         $M_{\text{priv}} = M + \text{\Laplace{$0, m^*/(|D|_{\blue{\text{priv}}} \cdot \varepsilon_\text{init})$}}$ \label{algdpinitscorelaplace}\\
         \Return $M_{\text{priv}}$
     }\ElseIf{\blue{classification}}{\blue{
     $M = \frac{1}{|D|_{\text{priv}}}\sum_{y_i \in D} \text{\clamp{$y_i, 0, m^*$}}$ \\
     $M_{\text{priv}} = M + \text{\Laplace{$0, m^*/(|D|_{\text{priv}} \cdot \varepsilon_\text{init})$}}$ \\
     \Return $\ln(\sfrac{M_{\text{priv}}}{1 - M_{\text{priv}}})$
     }}\label{algdpinitscore}
\end{algorithm}

\boldparagraph{Description of algorithm.}

\blue{In \Cref{algdpinitscore}, we release the dataset size in a DP manner using the Laplace mechanism with privacy budget $\varepsilon_\text{ds}$ (line \ref{dpreleaseds}). We fix $\varepsilon_\text{ds}=0.005$.
We clamp the labels of all data points to a fixed range (line \ref{algdpinitscoreclipmean}) to upper bound the influence of any data point on the initial score. Then we average the clamped labels (line \ref{algdpinitscoreclipmean}) and add Laplace noise calibrated to the clipping bound $m^*$, the privacy budget $\varepsilon_\text{init}$, and the DP approximated number of data points $|D|_\text{priv}$ (line \ref{algdpinitscorelaplace}) to satisfy differential privacy. For classification, we rescale the noised mean with the logit which is the inverse of the sigmoid function.}

\subsection{Training a single tree with subsampling}\label{sec:single_tree}

To boost the utility-privacy tradeoff, we tailor privacy amplification by subsampling which is well discussed in literature~\cite{Abadi_2016,balle2018privacy,li2011sampling,pmlr-v97-zhu19c} to GBDTs. For subsampling, we train each tree in GBDT only on a random subsample of the training data, where we select each training data point with probability $\gamma$. An $(\eps,\delta)$-DP mechanism on all training data becomes an $(\mathcal{O}(\gamma\varepsilon), \gamma\delta)$-DP mechanism on the subsampled data~\cite[Theorem 8]{balle2018privacy}.

\boldparagraph{Description of algorithm.} ~\Cref{algtrainsingletree} describes the training of a single tree. First, we generate a subsample of the training data via Poisson subsampling (line \ref{alg:train_single_tree_poisson_subsample}). 
Second, we sample a random tree of full depth $d$ (line \ref{alg:train_single_tree_random_tree}). Third, we assign all leaf values using differentially private \Cref{algdpleaf} (lines~\ref{alg:train_single_tree_dp_leaf}-\ref{algtrainsingletreecalldpleaf}).

\begin{algorithm}
    \caption{TrainSingleTree: Train a DP decision tree}
    \SetKwFunction{RandomSplit}{RandomSplit}
    \SetKwFunction{SetLeaf}{SetLeaf}
    \SetKwFunction{DPLeaf}{DPLeaf}
    \SetKwFunction{Clamp}{clamp}
    \SetKwFunction{Predict}{Predict}
    \SetKwFunction{DPLeafSupport}{DPLeafSupport}
    \SetKwFunction{PoissonSubsample}{PoissonSubsample}
    \SetKwFunction{RandomTree}{RandomTree}
    \SetKwFunction{ComputeGradients}{GetGradients}
    \SetKwFunction{ComputeHessians}{GetHessians}
    \KwIn{$D$ : training data, $d$ : depth of trees} 
    \SetKwInOut{KwIn}{\color{white}\phantom{Input}}

    \KwIn{$\sigma_\text{leaf}^2$ : unweighted variance of Gaussian for leaves}
    \KwIn{$(g^*, h^*)$ : clipping bound on gradients and Hessians}
    \KwIn{$(r_{1}, r_{2})$ : noise weights for leaf value}
    \KwIn{$\lambda, \beta$ : regularization parameters}
    \KwIn{$E$ : ensemble of trees up to round $t-1$}

    $D^{\mathcal{P}} = \text{\PoissonSubsample{$D, \gamma$}}$\label{alg:train_single_tree_poisson_subsample}\;
    $\text{tree}_t = \text{\RandomTree(d)}$\label{alg:trainSingleTreeRandomTree}\label{alg:train_single_tree_random_tree}\tcp*{cf. \Cref{sec:prelim_randomsplit}}
    \For {\KwSty{each} leaf $l$ \KwSty{in} $\text{tree}_t$}{
        $v = $\DPLeaf{$l, D^\mathcal{P}, 
        g^*,
        $
        $
        h^*, \sigma_\text{leaf}^2,(r_1, r_2), \lambda, \beta$}\label{alg:train_single_tree_dp_leaf}\;
        \SetLeaf($\text{tree}_t, l, v$)\label{algtrainsingletreecalldpleaf}\;
    }
    \Return $\text{tree}_t$\label{algtrainsingletree}\;
\end{algorithm}

We experimentally show an improvement in the utility-privacy tradeoff through subsampling for \dpgbdt{} in \Cref{sec:main_results} and \Cref{sec:impact_of_parameters}.

\boldparagraph{Rényi DP proof.}

\blue{
\begin{theorem}
    \label{thm:train_single_tree_irdp_main_body}
    \input{theorems/traintree_irdp}
\end{theorem}
\begin{proof}
    \reftoarxiv{} The proof uses the individual RDP bound of a leaf $\rho(\alpha) = \alpha\cdot t = \alpha \cdot \frac{2}{2\sigma_\text{leaf}^2} \cdot \Big( \frac{r_1 \cdot |h_i|^{2}}{(h^{*})^{2}} + \frac{r_2\cdot |g_i|^2}{(g^{*})^{2}}\Big)$
    (cf. \Cref{cor:dpleaf_individual_rdp_main_body})
    as well as our novel condition for tight subsampling bounds (cf. \Cref{cor:pearson_vajda_condition_general}), i.e. \Cref{algtrainsingletree} (\texttt{TrainSingleTree}) when no subsampling is applied has a Pearson-Vajda $\mathcal{X}^{\alpha}$ pseudo-divergence of $D_{\mathcal{X}^{\alpha}}(M(X)||M(X')) \geq 0$ for all odd $\alpha \geq 1$. This allows us to use the tight individual RDP bound for subsampling by \citet{pmlr-v97-zhu19c}.
\end{proof}
}

\blue{
\begin{corollary}
    \label{cor:train_single_tree_rdp_main_body}
    \texttt{TrainSingleTree} (cf. \Cref{algtrainsingletree}) satisfies $(\alpha, a_{\gamma}(\alpha, \\\sfrac{\alpha}{\sigma_\text{leaf}^2}))$-RDP.
\end{corollary}
\begin{proof}
    The corollary follows directly from \Cref{thm:train_single_tree_irdp_main_body} with the worst-case sensitivities $g^*$ and $h^*$.
\end{proof}
}

\subsection{Leaf computation with leaf-balanced noise}\label{sec:leaf}

\boldparagraph{Description of algorithm.}
\dpgbdt computes leaves for each tree in a GBDT ensemble using \Cref{algdpleaf}. Our \emph{leaf-balanced} noise uses the foundation of prior work as detailed in \Cref{sec:prelim_leafnoising}: We calculate and Gaussian noise the sum of clamped gradients and clamped Hessians. With clamping we bound the influence of each gradient $g_i$ such that $\abs{g_i} \le g^*$ and each Hessian $h_i$ such that $\abs{h_i} \le h^*$. We release a differentially private leaf value by dividing the noised sum of gradients by the noised sum of Hessians.

\boldparagraph{Leaf-balanced noise.}
Our leaf-balanced noise (\Cref{sec:tight_leaf_balanced_noise}) allows us to better calibrate the amount of noise that is applied to the sum of gradients and sum of Hessians (lines \ref{algdpleafnoisedsupport} and \ref{algdpleafsupportsum} in \Cref{algdpleaf}). We experimentally show (cf. \Cref{sec:main_results} and \Cref{sec:impact_of_parameters}), that our leaf-balanced noise with $r_1 \neq r_2$ performs better than the standard non-spherical multivariate Gaussian mechanism where $r_1 = r_2 = 0.5$. \blue{In our experiments, choosing the range $0.05 \le r_1 \le 0.2$ is most helpful which assigns more budget to the gradient sum (leaf numerator).}

\boldparagraph{Rényi DP proof.}
\blue{

\begin{corollary}
    \input{theorems/dpleaf_irdp}
    \label{cor:dpleaf_individual_rdp_main_body}
\end{corollary}
\begin{proof}
    \reftoarxiv{} The proof directly follows from our leaf-balanced noise (cf. \Cref{thm:individual_rdp_nonspherical_gauss}) which shows that the overall privacy budget remains the same if we weigh each dimension via $r_1, r_2$ differently as long as $r_1 + r_2 = 1$.
    In \texttt{DPLeaf} we represent the leaf noising as adding a 2-dimensional Gaussian to a vector of gradient sum (leaf numerator) and Hessian sum (leaf denominator). Building this ratio is DP due to the post-processing theorem.
\end{proof}

\begin{corollary}
    \label{cor:dpleaf_rdp_main_body}
    \Cref{algdpleaf} (\texttt{DPLeaf}) satisfies $(\alpha, \sfrac{\alpha}{\sigma_\text{leaf}^2})$-RDP.
\end{corollary}
\begin{proof}
    The corollary follows directly from \Cref{cor:dpleaf_individual_rdp_main_body} with a challenge data point with worst-case sensitivities $g^*$ and $h^*$.
\end{proof}

}

\begin{algorithm}
    \caption{DPLeaf : Compute a DP leaf node}
    \SetKwFunction{Gauss}{Gauss}
    \SetKwFunction{Clamp}{Clamp}
    \SetKwFunction{max}{max}

    \KwIn{$D : \text{training data}$, $l$: leaf node identifier}
    \SetKwInOut{KwIn}{\color{white}\phantom{Input}}

    \KwIn{$(g \in D, h \in D)$ : gradient/Hessian of data points in $D$} 
    \KwIn{$(g^*, h^*)$ : clipping bound on gradients and Hessians}
    \KwIn{$(r_{1}, r_{2})$ : noise weights for leaf value}
    \KwIn{$\lambda, \beta$ : regularization parameters}
    \KwIn{$\sigma_\text{leaf}^2$ : unweighted variance of Gaussian for leaves}

    \label{algdpleaf}
    Let $D_l\subseteq D$ be the set of data points in leaf $l$\\
    $w = \sum_{h \in D_l} \text{\Clamp}(h, \blue{0.0}, h^{*})$ \label{algdpleafsupportsum} \\

    $\tilde{w} = \text{$\lambda + w + \text{\Gauss{$0.0, (h^{*})^{2}\sigma_\text{leaf}^2 / (2 \cdot r_1)$}}$}$ \label{algdpleafnoisedsupport} \\
    $u = \sum_{g \in D_l} \text{\Clamp}(g, -g^*, g^*)$ \label{algdpleafsum}\label{alg:dp_leaf_clamp} \\

    $\tilde{u} = u + \text{\Gauss{$0.0, (g^{*})^{2} \sigma_\text{leaf}^2/ (2 \cdot r_2)$}}$ \label{algdpleafnoisedsum} \\
    \Return \Clamp{$v = \tilde{u} / \tilde{w}, -\beta, \beta$} \label{algdpleafvalue}
\end{algorithm}

\subsection{Rényi DP accounting}\label{sec:accounting}

Our Rényi DP accountant picks an order $\hat{\alpha}$ to measure the privacy loss in Rényi DP, sets noise variance $\sigma_\text{leaf}^2$ that is applied to the leaves and sets an upper bound on the individual privacy loss of our training, so that we can utilize an individual Rény filter.  

\boldparagraph{Description of algorithm.}
Our Rényi DP accountant (cf. \Cref{alg:initialization}) iterates over pairs of order $\alpha$ and noise variance $\sigma_\text{leaf}^2$ and uses ~\Cref{cor:train_single_tree_rdp_main_body} (line \ref{alg:initialization_single_tree}) to compute the Rényi DP privacy leakage of order $\alpha$ for subsampled training of a single tree when applying noise with variance $\sigma_\text{leaf}^2$ in the leaves. The accountant investigates orders of $\alpha\geq 2$ as for other $\alpha$ the privacy amplification by subsampling of ~\Cref{thmsubsampling} does not apply. Next, the accountant applies RDP sequential composition (cf. \Cref{thmrdpsequentialcomposition}) to account for the number of training rounds (line \ref{alg:initialization_sequential_composition}) and converts this RDP bound to a $(\varepsilon, \delta)$-DP bound using \Cref{cor:rdp_to_dp}. Over all orders $\alpha$ and all noise variances $\sigma_\text{leaf}^2$ the accountant picks the pair that satisfies the following condition: the converted $(\varepsilon, \delta)$ bound is close to the user-specified privacy budget and the noise variance is minimal (line \ref{alg:initialization_pick}).

The initialization returns (1) the order $\hat{\alpha}$ for which the Rényi DP accounting is done, (2) the noise variance $\sigma_\text{leaf}^2$ applied to the leaves and (3) an upper bound $\rho(\hat{\alpha})$ for the individual privacy loss. The initialization is $(\varepsilon_\text{init} + \varepsilon_\text{trees}, \delta)$-DP for $\delta=5\cdot10^{-8}$ and $\varepsilon_\text{init}$ and $\varepsilon_\text{trees}$ being the DP budget of the initial score and the tree training.

\begin{algorithm}
\caption{Initialize Rényi DP accountant}
    \label{alg:initialization}
    \SetKwFunction{RenyiDPRhoGaussianMechanismNonspherical}{RenyiDPRhoGaussianMechanismNonspherical}
    \SetKwFunction{RenyiDPSubsampledPrivacyAmplification}{RenyiDPSubsampledPrivacyAmplification}
    \SetKwFunction{Append}{Append}
    \KwIn{$\alpha_{\max}$ : largest $\alpha$ to test in RDP, $\gamma$ : subsampling ratio}
    \SetKwInOut{KwIn}{\color{white}\phantom{Input}}
    \KwIn{$(\varepsilon_\text{trees}, \delta_\text{trees})$ : DP parameters for training of trees}
    \KwIn{$\varepsilon_\text{init}$ : privacy budget for DP initial score}
    $\mathcal{T} = ()$ \\
    \For{$\alpha=2$ to $\alpha_\text{max}$}{
        \For{$\sigma_\text{leaf}^2$ in $(0.0, 1000.0]$}{
            $\rho_\text{subsampled-tree} = a_\gamma(\alpha, \frac{\alpha}{\sigma_\text{leaf}^2})$\label{alg:initialization_single_tree}\tcp*{Use \Cref{cor:train_single_tree_rdp_main_body}}
            $\rho(\alpha) = T_{\text{regular}} \cdot \rho_\text{subsampled-tree}$\label{alg:initialization_sequential_composition}\tcp*{Use \Cref{thmrdpsequentialcomposition}}
            $\varepsilon_\text{trees}' = \rho(\alpha) + \log{\frac{\alpha -1}{\alpha} - \frac{\log{\delta_\text{trees}} + \log{\alpha}}{\alpha-1}}$\tcp*{\Cref{cor:rdp_to_dp}}
            \BlankLine
            \Append{$\mathcal{T}, (\alpha, \sigma_\text{leaf}^2, \rho(\alpha), \varepsilon_\text{trees}')$}
        }
    }
    Pick $(\hat{\alpha}, \sigma_\text{leaf}^2, \rho(\hat{\alpha}), \varepsilon_\text{trees}')$ from $\mathcal{T}$ so that $\sigma_\text{leaf}^2$ is smallest and
    $\varepsilon_\text{trees}'$ is close to $\varepsilon_\text{trees}$ \label{alg:initialization_pick} \\
    Report $\varepsilon_\text{init} + \varepsilon_\text{trees}'$ to user\\
    \Return $\hat{\alpha}, \sigma_\text{leaf}^2, \rho(\hat{\alpha})$
\end{algorithm}

\boldparagraph{Discussion.}
Our accountant can measure subsampled Rényi DP, giving \dpgbdt a privacy amplification by subsampling. This sets our accounting apart from the accounting of prior works that utilized parallel and sequential composition for DP. Our accountant is generic: the RDP bound of any mechanism can be plugged in to obtain an accurate and robust subsampled RDP accounting.

\subsection{Generalizing \dpgbdt: Stream of non-IID data}
\label{sec:generalization_streams}

We consider learning a stream of non-IID data, where new data arrives over time, with DP GBDT. Here, we want to incorporate training data that arrives later and also not forget about formerly seen data points. Under privacy, however, training again with formerly seen data points poses the risk of additional privacy leakage. Privacy would hold if the GBDT model simply forgets about formerly seen data points, but this can hurt model performance.

Tailoring an individual Rényi filter to DP GBDT training enables individual privacy accounting, which in turn enables training for an arbitrary amount of rounds with only those data points that still have privacy budget left. This enables \dpgbdt{} to effectively learn on streams of non-IID data: \dpgbdt{} learns with newly arrived data while also including data points that arrived in the past. 

We experimentally show (cf. ~\Cref{sec:renyi_filter_streams}) that when learning a stream of non-IID data, tailoring an individual Rényi filter to \dpgbdt{} can significantly improve the model performance. We compare to a naïve approach for stream learning where the formerly seen data points are discarded once they have been used for the amount of training rounds that the privacy budget has been accounted for.

\subsection{Scalable distributed learning}
    \input{distributed_learning}

%% file: theorems/initscore_rdp.tex
\Cref{algdpinitscore} (\texttt{DPInitialScore}) with clipping bound $m^*$, DP approximated dataset size $|D|_\text{priv}$ 
and $f(\varepsilon) = \log \Big\{ \frac{\alpha}{2\alpha -1} \exp\Big((\alpha-1)\cdot \varepsilon\Big) + \frac{\alpha -1}{2\alpha -1} \exp\Big(-\alpha \cdot \varepsilon\Big) \Big\}$
satisfies $\Big(\alpha,\frac{1}{\alpha-1} 
(f(\varepsilon_\text{init}) + f(\varepsilon_\text{ds}))\Big)$-RDP.

%% file: theorems/traintree_irdp.tex
Let $r_1, r_2$ be defined as in ~\Cref{thm:individual_rdp_nonspherical_gauss}. Let $\sigma_\text{leaf}^2$ be the unweighted variance of the leaf Gaussian. Let $a_\gamma: \mathbb{N} \times \mathbb{R} \mapsto \mathbb{R}$ denote the privacy amplification of ~\Cref{thmsubsampling} with subsampling ratio $\gamma$. Then, for data point $x_i$ with gradient $g_i$ and Hessian $h_i$, \Cref{algtrainsingletree} (\texttt{TrainSingleTree}) satisfies $(\alpha, a_{\gamma}(\alpha, \alpha \cdot \frac{2}{2\sigma_\text{leaf}^2} \cdot \Big( \frac{r_1 \cdot |h_i|^{2}}{(h^{*})^{2}} + \frac{r_2\cdot |g_i|^2}{(g^{*})^{2}}\Big)))$-individual RDP.

%% file: theorems/dpleaf_irdp.tex
\Cref{algdpleaf} (\texttt{DPLeaf}) satisfies $(\alpha, \alpha \cdot \frac{2}{2\sigma_\text{leaf}^2} \cdot \Big( \frac{r_1 \cdot |h_i|^{2}}{(h^{*})^{2}} + \frac{r_2\cdot |g_i|^2}{(g^{*})^{2}}\Big))$-individual RDP for data point $x_i$ with gradient $g_i$ and Hessian $h_i$ and with $\sigma_\text{leaf}^2$ as the unweighted variance of the leaf Gaussian.

%% file: distributed_learning.tex
\label{sec:distributed_learning}

We extend \dpgbdt to work in a distributed learning setting, where multiple users collaboratively train a model, while the private training datasets are not directly shared.
This applies to the field of medical studies where hospitals want to keep sensitive patient data at their hospital but want to train a machine learning model collaboratively. Our extension of \dpgbdt to distributed learning uses similar techniques as \citet{Maddock_2022}.
\blue{For the detailed algorithm, we refer to the appendix. In summary:}

\blue{Every user receives the training hyperparameters from a secure bulletin board and sets up the accounting. For the initial classifier, the user computes DP releases of the sum of labels and the dataset size. All users synchronize and invoke \texttt{SecureAggregation} \cite{secure_aggregation_bonawitz,secure_aggregation_bell} with fixed precision, for aggregating the label sum first and then dataset size. The initial score is then built by dividing the aggregated label sum by the overall dataset size and added to the ensemble. The user commences $T_{\max}$ rounds of training. Each round, the user updates the individual RDP privacy losses for all its data points and filters out data points that have expended their privacy budget. Next, the user locally initializes a tree for the current round. All user synchronize and utilize public uniform sampling \cite[Protocol 1]{sabater_verifiable_sampling} to collaboratively and verifiably sample uniformly random features and feature values for the tree splits.

The user locally only adjusts the leaves of the tree with its local share of sensitive data. All users synchronize again to utilize \texttt{SecureAggregation} with fixed precision for collaboratively generating leaf values for the tree. \texttt{SecureAggregation} is used for aggregating the gradient sum and the Hessian sum of all leaves, then the user sets the leaf values of its local tree and finally adds this tree to the ensemble.}

%% file: proofs.tex
\begin{theorem}[Main theorem (informal)]\label{thm:maintheorem}
    \dpgbdt is  $(\alpha, \rho(\alpha)$)-Rényi differentially private.
\end{theorem}

\begin{proof}
Let $X$ be the training dataset, $T_\text{regular}$ the number regular trees without extra trees due to a Rényi filter, $m^*$ the clipping bound on the labels, and $\varepsilon_\text{init}, \varepsilon_\text{ds}$ the privacy budgets used by the DP initial score. Let $a_\gamma: \mathbb{N} \times \mathbb{R} \mapsto \mathbb{R}$ denote the privacy amplification of ~\Cref{thmsubsampling} with subsampling ratio $\gamma$ and $f(\varepsilon) = \log \Big\{ \frac{\alpha}{2\alpha -1} \exp\Big((\alpha-1)\cdot \varepsilon\Big) + \frac{\alpha -1}{2\alpha -1} \exp\Big(-\alpha \cdot \varepsilon\Big) \Big\}$. We then have
\begin{align*}
    \textstyle\rho(\alpha) =~&\textstyle\frac{1}{\alpha-1} \Big(
    f(\varepsilon_\text{init}) + f(\varepsilon_\text{ds})
    \Big)  
    + T_\text{regular} \cdot a_{\gamma}(\alpha, \frac{\alpha}{\sigma_\text{leaf}^{2}})
\end{align*}

\dpgbdt consists of the differentially private initial score mechanism, and then multiple training rounds, each outputting a single tree. So the proof decomposes \dpgbdt into the initial score step and the training rounds via sequential composition for Rényi DP. We can then separately bound the Rényi divergences for the initial score and the training rounds. 

For the initial score, we derive an analytical $(\alpha, \rho_\text{init}(\alpha))$-Rényi DP bound in ~\Cref{thmdpinitscore}, using a generalization of the Laplace mechanism to arbitrary sensitivity, as the sensitivity of our unnoised initial score depends on the DP approximated dataset size $|D|_\text{priv}$, the privacy budget for the initial score $\varepsilon_\text{init}$ and the clipping bound $m^*$ on the labels from dataset $D$.

Next, we obtain a Rényi DP bound for the training rounds. 
Here is where the individual Rényi filter comes into play: A single training round consists of the individual Rényi filter application, i.e. computing the individual Rényi DP privacy losses via our analytical bound \Cref{thm:train_single_tree_irdp_main_body} \blue{(for RDP: \Cref{cor:train_single_tree_rdp_main_body})} then filtering out those data points that have fully expended their privacy budget, before finally running the training round. In this setup, the individual Rényi filter (cf. \Cref{thm:adaptive_composition_individual_privacy}) guarantees that an arbitrary amount of training rounds satisfies an a priori $(\alpha, \rho_\text{training}(\alpha))$-Rényi DP bound, as long as the individual Rényi DP accounting is sound.
Combining the two analyses, we get that the output of \dpgbdt satisfies $(\alpha, \rho_\text{init}(\alpha) + \rho_\text{training}(\alpha))$.
For the complete proof we refer to the appendix.
\end{proof}

%% file: experiments.tex
\label{sec:experiments}

We evaluate our differentially private gradient boosted decision trees ensemble learner \dpgbdt (cf. \Cref{sec:sgbdt}) as follows. 
In \Cref{sec:main_results} we compare \dpgbdt against the state of the art (SOTA) by Maddock et al. \cite{Maddock_2022}. In \Cref{sec:impact_of_parameters} we perform ablation studies of our improvements in \dpgbdt. We answer the following research questions: 

\noindent\textbf{(RQ1)} \emph{Is \dpgbdt improving on the current SOTA \cite{Maddock_2022}?}
In \Cref{sec:main_results} we find that \dpgbdt significantly improves on the SOTA, saving $50\%$ in terms of $\eps$ for $\varepsilon\leq0.5$ on the regression dataset and $30\%$ for $\varepsilon\leq0.08$ and $\varepsilon\leq0.03$ respectively on the classification datasets for the same utility (AUC/RMSE).

\noindent\textbf{(RQ2)} \emph{Under privacy, are our contributions leaf-balanced noise (cf. \Cref{sec:leaf}), subsampling (cf. \Cref{sec:single_tree}), and DP initial score (cf. \Cref{sec:initial_score}) improving the performance of \dpgbdt{}?}
In \Cref{sec:impact_of_parameters} we show leaf-balanced noise and subsampling consistently improve the performance of \dpgbdt. DP initial score improves utility on the regression dataset.

\noindent\textbf{(RQ3)} \emph{Do individual Rényi filters \cite{feldman2020individual} improve the performance of \dpgbdt{} when learning a stream of non-IID data?}
In \Cref{sec:renyi_filter_streams} we show that an individual Rényi filter can significantly improve the performance of \dpgbdt when learning a stream of non-IID data.

\boldparagraph{Sensitive Datasets.} We use three datasets for our experiments: Abalone, Adult and Spambase. Abalone \cite{misc_abalone_1} is a regression dataset containing $4{,}177$ data points. 
Given eight numerical attributes, e.g. sex, length, and diameter of an abalone, the task is to predict its age.
Adult \cite{misc_adult_2} is a binary classification dataset with more than $48{,}000$ data points. 
Given 14 attributes, e.g. age, sex, and occupation, the task is to determine whether a person earns over $50{,}000 \, $\$ per year.
\blue{Spambase \cite{misc_spambase_3} is a binary classification dataset containing $4{,}601$ data points. Given 57 numerical attributes, e.g. frequencies of words and symbols in a mail, the task is to determine whether that mail is spam.}

\boldparagraph{Experimental Setup.} We tune our \dpgbdt and the state of the art by Maddock et al. \cite{Maddock_2022} in a \blue{randomized grid-search that randomly selects without replacement hyperparameters of a grid search.} In all experiments, we set $\delta=5\cdot10^{-8}$. We use 5-fold cross-validation to obtain training and test datasets. In the hyperparameter search every setting 
is evaluated for 20 runs and the best-performing hyperparameter setting, again for 200 runs \blue{(for Spambase and $\varepsilon\leq0.1$: 1000 runs)}. 
Tests were run with an Intel Xeon Platinum 8168 2.7 GHz CPU and 32GB of RAM. \blue{The code of \dpgbdt is available at \href{https://github.com/kirschte/sbdt}{https://github.com/kirschte/sbdt}.}

\boldparagraph{Model performance.}
 We measure model performance on the test dataset by the root mean squared error (RMSE) for regression and AUC-ROC for classification. We also report the standard error of the mean, defined by the standard deviation of the runs divided by the square root of the number of runs \cite{Altman903}. The standard error represents the standard deviation of the distribution over the means of the runs, and as such gives an estimate on the precision of the sampled mean. The standard error decreases with an increasing number of runs, as the extent of chance variation is reduced. 

\boldparagraph{Non-private baselines.}
For the non-private baselines we implement the XGBoost \cite{chen2016xgboost} algorithm in our framework \blue{and denote \emph{xgboost} for reporting the utility when selecting the optimal, data-dependent split and \emph{xgboost random splits} when selecting splits uniformly at random from the split value range.}

\boldparagraph{Learning a stream of non-IID data.}
For the baseline, a naïve training approach for streams that adds extra training rounds for newly arrived data, and training with Rényi filter we investigate $T' \in \{T/4, T/2, T \, 3/4, T\}$ extra training rounds for $T$ regular training rounds. \blue{We rerun our RDP accountant for the extra training rounds, i.e. we adapt the noise scale of the extra round to the number of extra rounds.} 

\boldparagraph{Hyperparameters.}
 The evaluated hyperparameters are the number of trees $T_\text{regular} \in \{5, 10, 25, 50, 100, 150, 200, 300, 400, 500, \\600\}$, the depth of the trees $d \in \{2,3,5,6\}$ and the number of training rounds of iterative Hessian (cf. \Cref{sec:prelim_randomsplit}) $s \in \{5, 30, 100\}$. We clip gradients with $g^* \in \{0.1, 0.3, 0.5, 0.7, 0.9\}$ and Hessians with $h^* \in \{0.1, 0.25\}$. To simplify gradient clipping for regression we scale the regression labels to the range $[-1,1]$. We evaluate subsampling ratio $\gamma \in \{0.005, 0.05, 0.1, 0.2\}$, leaf-balanced noise parameter $r_1 \in \{0.04, 0.1, 0.2, 0.3, ..., 0.9\}$, ratio $\varepsilon_\text{init}/\varepsilon \in \{10\%, 30\%\}$, clipping bound for initial score $m^* \in \{0.1, 0.5, 1.0\}$. \blue{We fix the privacy budget for DP releasing the dataset size for the initial score $\varepsilon_\text{ds}=0.005$.} We fix one regularization parameter $\beta=2$ and investigate $\lambda \in \{1, 15\}$ (Abalone \blue{and Spambase}) and $\lambda \in \{1, 10\}$ (Adult). We investigate learning rate $\eta \in \{0.1, 0.2, 0.3\}$.

To reduce the overhead for learning a stream of non-IID data, we restrict some hyperparameters to the range that was useful for IID data. For \dpgbdt{} we evaluate $T_\text{regular} \in \{50, 100, 200, 400\}$ for Abalone and $T_\text{regular} \in \{50, 200, 400\}$ for Adult. We fix $\varepsilon_\text{init}/\varepsilon = 10\%$, $m^* = 1.0$,$\gamma=0.1$, $\eta=0.1$ and $r_1=0.2$ (Abalone) and $r_1=0.1$ (Adult). 
For Maddock et al. we evaluate $T_\text{regular} \in \{5, 10, 25\}$ for Abalone and $T_\text{regular} \in \{25, 50, 200\}$ for Adult. We fix 5 rounds of training with iterative Hessian and $\eta=0.3$ 
For both, 
we evaluate $d \in \{2,4,6\}$ and $g^* \in \{0.1, 0.3, 0.5, 0.7\}$, $h^* \in \{0.1, 0.25\}$.

\subsection{Comparing \dpgbdt{} to the SOTA \cite{Maddock_2022}}
\label{sec:main_results}

\input{experiments_main_results}
\input{experiments_main_results_ablations}

We compare \dpgbdt to the best prior work \cite{Maddock_2022} \blue{both including the following relevant algorithmic variations proposed by \citet{Maddock_2022}. For the split selection, we also use random splits which is the best-performing variant in \citet{Maddock_2022}. For random splits, we choose a feature uniformly at random and a split value uniformly at random from a pre-defined feature range. We evaluate two weight update methods, gradient boosting for regression and Newton boosting for classification. For regression, gradient and Newton boosting are the same, as we divide in the leaf by the number of samples in that leaf which is the same as the sum of Hessians. For generating split candidates we evaluate: No prior split candidates, i.e. sampling a split from the full pre-defined feature range uniformly at random, equidistant split candidates from the pre-defined feature range, and equidistant split candidates refined with iterative Hessian, an approach that refines the initial split candidates based on the aggregated Hessians of data points for each split candidate. The latter two approaches are proposed by \citet{Maddock_2022}. We evaluate two feature interactions, cyclical, where each tree is trained on a single feature selected cyclically from all features, and random, where a random feature is selected in every split.}

\input{experiments_ablations}
\input{experiments_streams}

We vary the privacy budget $\varepsilon$ and perform a hyperparameter search over the other parameters. \Cref{fig:all_main_results} displays our findings and \Cref{tbl:abalone_ablations} and \Cref{tbl:adult_ablations} the improvement of  \dpgbdt{} over \citet{Maddock_2022} for a single privacy budget value. 
We observe that \dpgbdt achieves better performance for regression and classification than the SOTA by Maddock et al. For $\varepsilon=0.25$ on Abalone, the RMSE of \dpgbdt{} is 2.64 which Maddock et al. reach for $\varepsilon=0.5$, saving $50\%$ in terms of epsilon (cf. \Cref{fig:abalone_main_results}). For $\varepsilon=0.053$ on Adult, the AUC of \dpgbdt{} is 0.853 which Maddock et al. reach for $\varepsilon=0.08$, saving $30\%$ in terms of epsilon (\Cref{fig:adult_main_results}). \blue{For $\varepsilon=0.02$ on Spambase, the AUC of \dpgbdt{} is 0.79 which Maddock et al. reach for $\varepsilon=0.03$, saving $30\%$ in terms of epsilon (\Cref{fig:spambse_main_results}).} For the split selection, we observe that split refinement using Hessian information as proposed by Maddock et al. \cite{Maddock_2022} is beneficial for their work but not for ours.

\blue{\dpgbdt{} uses up to 600 training rounds, Maddock et al. up to 300. \dpgbdt{} takes 3.2 seconds to train 1000 trees with tree depth 6 on Abalone, and 8.2 seconds on Adult. With individual Rényi filter, \dpgbdt{} takes 5.1 seconds on Abalone and 25.7 seconds on Adult. We update the individual privacy budget
by slightly rounding up the individual sensitivities $g_i, h_i$,
computing the individual privacy loss as an upper bound on the
exact one and storing it for later use. This saves time cost and is
proposed in prior work \cite{yu2023individual}.} 

\blue{\subsection{Random splits vs. data-dependent splits}

Our evaluation supports the results of prior work \cite{bojarski2014differentially,Maddock_2022} that random splits compared to data-dependent splits cost some utility, yet with a limited effect. For random splits, we select uniformly at random a feature and uniformly at random a value from a pre-defined feature range. For data-dependent splits, we evaluate the best split across all features and possible feature values, i.e. data points for this feature, using a variant of the MSE gain as used in xgboost \cite{chen2016xgboost}. In fact, on the three datasets Abalone, Adult, and Spambase we observe in the non-private setting a gap of $0.07$ in RMSE (Abalone), $0.008$ in AUC (Adult), and $0.004$ in AUC (Spambase) between data-independent random splits and data-dependent gain-based splits (cf. \Cref{fig:all_main_results}).
This demonstrates that the privacy leakage of tree learners can be reduced to the leaves and that random splits can be of interest even in a non-private setting due to their significantly faster computation time.}

\subsection{Impact of algorithmic parameters $r_1, \gamma, \varepsilon_\text{init}/\varepsilon$}
\label{sec:impact_of_parameters}

We explore the impact of parameters for our leaf-balanced noise, subsampling, and our initial score on model performance. We fix the optimal setting for our \dpgbdt{} and then vary a single one of the parameters. Our findings are displayed in \Cref{fig:all_ablations}. For leaf-balanced noise, \blue{choosing the range $0.05 \le r_1 \le 0.2$ is most helpful which assigns more budget to the gradient sum (leaf numerator)} (cf. \Cref{fig:abalone_r1_ablation}, \Cref{fig:adult_r1_ablation}).
Choosing a small subsampling ratio $\gamma$ is generally beneficial ($\gamma < 0.3$) but if the chosen $\gamma$ is too small, the privacy amplification cannot outweigh the small amount of data points used for a single tree which can hurt model performance (cf. \Cref{fig:adult_Q_ablation}, \Cref{fig:abalone_Q_ablation}). \blue{We fix $\varepsilon_\text{ds}=0.005$ and tune $\varepsilon_\text{init}$.} Spending privacy budget $\varepsilon_\text{init}$ on the initial score improves performance on Abalone (cf. \Cref{fig:abalone_init_ablation}), but choosing a fraction $\varepsilon_\text{init}/\varepsilon>0.3$ leaves too little privacy budget for the trees and hurts model performance.

\subsection{Learning a stream of non-IID data}
\label{sec:renyi_filter_streams}

We study the impact of an individual Rényi filter for learning a stream of non-IID data. On the stream, all data points with a regression label above the mean of the data set arrive after the regular training rounds are over. For classification, all data points with label 1 arrive late. As the baseline, we adopt a naïve training approach for streams by adding extra training rounds that only train with the newly arrived data. Our \dpgbdt{} with an individual Rényi filter instead can incorporate not only newly arrived data but also formerly seen data points that still have a privacy budget left. We vary the privacy budget $\varepsilon$ and perform a hyperparameter search over the other parameters. For the SOTA by \citet{Maddock_2022} we adopt the naïve training approach for streams. \Cref{fig:all_streams} displays our findings: \dpgbdt{} with individual Rényi filter performs better than the naïve baseline. This benefit comes without spending any additional privacy budget, an important feature of an individual Rényi filter.

%% file: experiments_main_results.tex
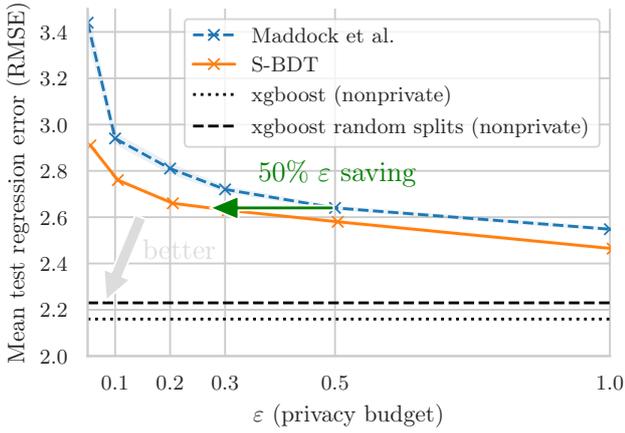
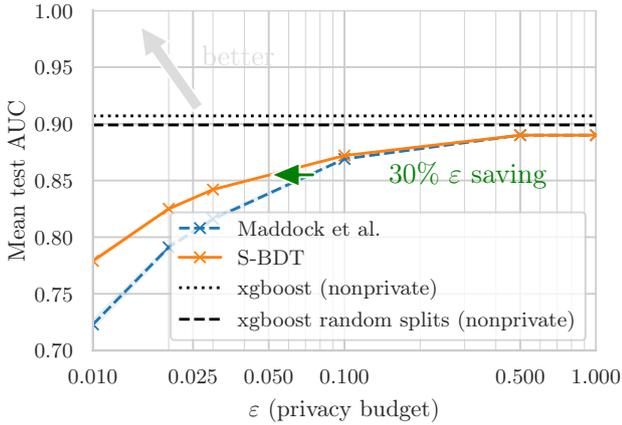
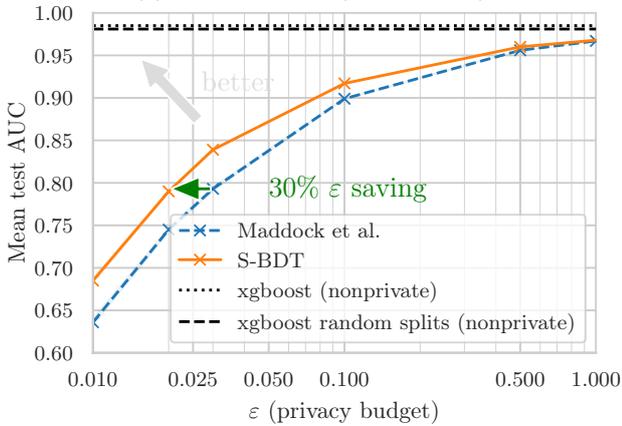
\begin{figure}[!t]
  \centering
      \begin{subfigure}[h]{\columnwidth}
        \resizebox{\textwidth}{!}{\input{images/abalone_main_results.pgf}\unskip}
          \caption{\blue{Dataset: \textbf{Abalone} (Regression)}}
          \label{fig:abalone_main_results}
      \end{subfigure}
      \begin{subfigure}[h]{\columnwidth}
        \resizebox{\textwidth}{!}{\input{images/adult_main_results.pgf}\unskip}
        \caption{\blue{Dataset: \textbf{Adult} (Classification)}}
          \label{fig:adult_main_results}
      \end{subfigure}
      \begin{subfigure}[h]{\columnwidth}
        \resizebox{\textwidth}{!}{\input{images/spambase_main_results.pgf}\unskip}
        \caption{\blue{Dataset: \textbf{Spambase} (Classification)}}
          \label{fig:spambse_main_results}
      \end{subfigure}
  \caption{\textbf{Comparison of utility-privacy tradeoff of our \dpgbdt{} and the SOTA by Maddock et al. \cite{Maddock_2022}}. Regression error (RMSE) (Abalone) and AUC (Adult and Spambase) of 200 runs \blue{(for Spambase and $\varepsilon\leq0.1$: 1000 runs)} vs. privacy budget $\eps$ ($(b)$ and $(c)$ in log-scale). The transparent area is the standard error.}
  \label{fig:all_main_results}
\end{figure}

%% file: images/abalone_main_results.pgf
\begingroup%
\makeatletter%
\begin{pgfpicture}%
\pgfpathrectangle{\pgfpointorigin}{\pgfqpoint{3.612000in}{2.517475in}}%
\pgfusepath{use as bounding box, clip}%
\begin{pgfscope}%
\pgfsetbuttcap%
\pgfsetmiterjoin%
\definecolor{currentfill}{rgb}{1.000000,1.000000,1.000000}%
\pgfsetfillcolor{currentfill}%
\pgfsetlinewidth{0.000000pt}%
\definecolor{currentstroke}{rgb}{1.000000,1.000000,1.000000}%
\pgfsetstrokecolor{currentstroke}%
\pgfsetdash{}{0pt}%
\pgfpathmoveto{\pgfqpoint{0.000000in}{0.000000in}}%
\pgfpathlineto{\pgfqpoint{3.612000in}{0.000000in}}%
\pgfpathlineto{\pgfqpoint{3.612000in}{2.517475in}}%
\pgfpathlineto{\pgfqpoint{0.000000in}{2.517475in}}%
\pgfpathclose%
\pgfusepath{fill}%
\end{pgfscope}%
\begin{pgfscope}%
\pgfsetbuttcap%
\pgfsetmiterjoin%
\definecolor{currentfill}{rgb}{1.000000,1.000000,1.000000}%
\pgfsetfillcolor{currentfill}%
\pgfsetlinewidth{0.000000pt}%
\definecolor{currentstroke}{rgb}{0.000000,0.000000,0.000000}%
\pgfsetstrokecolor{currentstroke}%
\pgfsetstrokeopacity{0.000000}%
\pgfsetdash{}{0pt}%
\pgfpathmoveto{\pgfqpoint{0.473880in}{0.420833in}}%
\pgfpathlineto{\pgfqpoint{3.529921in}{0.420833in}}%
\pgfpathlineto{\pgfqpoint{3.529921in}{2.468667in}}%
\pgfpathlineto{\pgfqpoint{0.473880in}{2.468667in}}%
\pgfpathclose%
\pgfusepath{fill}%
\end{pgfscope}%
\begin{pgfscope}%
\pgfpathrectangle{\pgfqpoint{0.473880in}{0.420833in}}{\pgfqpoint{3.056041in}{2.047833in}}%
\pgfusepath{clip}%
\pgfsetroundcap%
\pgfsetroundjoin%
\pgfsetlinewidth{0.803000pt}%
\definecolor{currentstroke}{rgb}{0.800000,0.800000,0.800000}%
\pgfsetstrokecolor{currentstroke}%
\pgfsetdash{}{0pt}%
\pgfpathmoveto{\pgfqpoint{0.634724in}{0.420833in}}%
\pgfpathlineto{\pgfqpoint{0.634724in}{2.468667in}}%
\pgfusepath{stroke}%
\end{pgfscope}%
\begin{pgfscope}%
\definecolor{textcolor}{rgb}{0.150000,0.150000,0.150000}%
\pgfsetstrokecolor{textcolor}%
\pgfsetfillcolor{textcolor}%
\pgftext[x=0.634724in,y=0.305556in,,top]{\color{textcolor}\rmfamily\fontsize{8.800000}{10.560000}\selectfont \(\displaystyle {0.1}\)}%
\end{pgfscope}%
\begin{pgfscope}%
\pgfpathrectangle{\pgfqpoint{0.473880in}{0.420833in}}{\pgfqpoint{3.056041in}{2.047833in}}%
\pgfusepath{clip}%
\pgfsetroundcap%
\pgfsetroundjoin%
\pgfsetlinewidth{0.803000pt}%
\definecolor{currentstroke}{rgb}{0.800000,0.800000,0.800000}%
\pgfsetstrokecolor{currentstroke}%
\pgfsetdash{}{0pt}%
\pgfpathmoveto{\pgfqpoint{0.956413in}{0.420833in}}%
\pgfpathlineto{\pgfqpoint{0.956413in}{2.468667in}}%
\pgfusepath{stroke}%
\end{pgfscope}%
\begin{pgfscope}%
\definecolor{textcolor}{rgb}{0.150000,0.150000,0.150000}%
\pgfsetstrokecolor{textcolor}%
\pgfsetfillcolor{textcolor}%
\pgftext[x=0.956413in,y=0.305556in,,top]{\color{textcolor}\rmfamily\fontsize{8.800000}{10.560000}\selectfont \(\displaystyle {0.2}\)}%
\end{pgfscope}%
\begin{pgfscope}%
\pgfpathrectangle{\pgfqpoint{0.473880in}{0.420833in}}{\pgfqpoint{3.056041in}{2.047833in}}%
\pgfusepath{clip}%
\pgfsetroundcap%
\pgfsetroundjoin%
\pgfsetlinewidth{0.803000pt}%
\definecolor{currentstroke}{rgb}{0.800000,0.800000,0.800000}%
\pgfsetstrokecolor{currentstroke}%
\pgfsetdash{}{0pt}%
\pgfpathmoveto{\pgfqpoint{1.278101in}{0.420833in}}%
\pgfpathlineto{\pgfqpoint{1.278101in}{2.468667in}}%
\pgfusepath{stroke}%
\end{pgfscope}%
\begin{pgfscope}%
\definecolor{textcolor}{rgb}{0.150000,0.150000,0.150000}%
\pgfsetstrokecolor{textcolor}%
\pgfsetfillcolor{textcolor}%
\pgftext[x=1.278101in,y=0.305556in,,top]{\color{textcolor}\rmfamily\fontsize{8.800000}{10.560000}\selectfont \(\displaystyle {0.3}\)}%
\end{pgfscope}%
\begin{pgfscope}%
\pgfpathrectangle{\pgfqpoint{0.473880in}{0.420833in}}{\pgfqpoint{3.056041in}{2.047833in}}%
\pgfusepath{clip}%
\pgfsetroundcap%
\pgfsetroundjoin%
\pgfsetlinewidth{0.803000pt}%
\definecolor{currentstroke}{rgb}{0.800000,0.800000,0.800000}%
\pgfsetstrokecolor{currentstroke}%
\pgfsetdash{}{0pt}%
\pgfpathmoveto{\pgfqpoint{1.921478in}{0.420833in}}%
\pgfpathlineto{\pgfqpoint{1.921478in}{2.468667in}}%
\pgfusepath{stroke}%
\end{pgfscope}%
\begin{pgfscope}%
\definecolor{textcolor}{rgb}{0.150000,0.150000,0.150000}%
\pgfsetstrokecolor{textcolor}%
\pgfsetfillcolor{textcolor}%
\pgftext[x=1.921478in,y=0.305556in,,top]{\color{textcolor}\rmfamily\fontsize{8.800000}{10.560000}\selectfont \(\displaystyle {0.5}\)}%
\end{pgfscope}%
\begin{pgfscope}%
\pgfpathrectangle{\pgfqpoint{0.473880in}{0.420833in}}{\pgfqpoint{3.056041in}{2.047833in}}%
\pgfusepath{clip}%
\pgfsetroundcap%
\pgfsetroundjoin%
\pgfsetlinewidth{0.803000pt}%
\definecolor{currentstroke}{rgb}{0.800000,0.800000,0.800000}%
\pgfsetstrokecolor{currentstroke}%
\pgfsetdash{}{0pt}%
\pgfpathmoveto{\pgfqpoint{3.529921in}{0.420833in}}%
\pgfpathlineto{\pgfqpoint{3.529921in}{2.468667in}}%
\pgfusepath{stroke}%
\end{pgfscope}%
\begin{pgfscope}%
\definecolor{textcolor}{rgb}{0.150000,0.150000,0.150000}%
\pgfsetstrokecolor{textcolor}%
\pgfsetfillcolor{textcolor}%
\pgftext[x=3.529921in,y=0.305556in,,top]{\color{textcolor}\rmfamily\fontsize{8.800000}{10.560000}\selectfont \(\displaystyle {1.0}\)}%
\end{pgfscope}%
\begin{pgfscope}%
\definecolor{textcolor}{rgb}{0.150000,0.150000,0.150000}%
\pgfsetstrokecolor{textcolor}%
\pgfsetfillcolor{textcolor}%
\pgftext[x=2.001901in,y=0.138889in,,top]{\color{textcolor}\rmfamily\fontsize{9.600000}{11.520000}\selectfont \(\displaystyle \varepsilon\) (privacy budget)}%
\end{pgfscope}%
\begin{pgfscope}%
\pgfpathrectangle{\pgfqpoint{0.473880in}{0.420833in}}{\pgfqpoint{3.056041in}{2.047833in}}%
\pgfusepath{clip}%
\pgfsetroundcap%
\pgfsetroundjoin%
\pgfsetlinewidth{0.803000pt}%
\definecolor{currentstroke}{rgb}{0.800000,0.800000,0.800000}%
\pgfsetstrokecolor{currentstroke}%
\pgfsetdash{}{0pt}%
\pgfpathmoveto{\pgfqpoint{0.473880in}{0.420833in}}%
\pgfpathlineto{\pgfqpoint{3.529921in}{0.420833in}}%
\pgfusepath{stroke}%
\end{pgfscope}%
\begin{pgfscope}%
\definecolor{textcolor}{rgb}{0.150000,0.150000,0.150000}%
\pgfsetstrokecolor{textcolor}%
\pgfsetfillcolor{textcolor}%
\pgftext[x=0.194444in, y=0.377431in, left, base]{\color{textcolor}\rmfamily\fontsize{8.800000}{10.560000}\selectfont \(\displaystyle {2.0}\)}%
\end{pgfscope}%
\begin{pgfscope}%
\pgfpathrectangle{\pgfqpoint{0.473880in}{0.420833in}}{\pgfqpoint{3.056041in}{2.047833in}}%
\pgfusepath{clip}%
\pgfsetroundcap%
\pgfsetroundjoin%
\pgfsetlinewidth{0.803000pt}%
\definecolor{currentstroke}{rgb}{0.800000,0.800000,0.800000}%
\pgfsetstrokecolor{currentstroke}%
\pgfsetdash{}{0pt}%
\pgfpathmoveto{\pgfqpoint{0.473880in}{0.693878in}}%
\pgfpathlineto{\pgfqpoint{3.529921in}{0.693878in}}%
\pgfusepath{stroke}%
\end{pgfscope}%
\begin{pgfscope}%
\definecolor{textcolor}{rgb}{0.150000,0.150000,0.150000}%
\pgfsetstrokecolor{textcolor}%
\pgfsetfillcolor{textcolor}%
\pgftext[x=0.194444in, y=0.650475in, left, base]{\color{textcolor}\rmfamily\fontsize{8.800000}{10.560000}\selectfont \(\displaystyle {2.2}\)}%
\end{pgfscope}%
\begin{pgfscope}%
\pgfpathrectangle{\pgfqpoint{0.473880in}{0.420833in}}{\pgfqpoint{3.056041in}{2.047833in}}%
\pgfusepath{clip}%
\pgfsetroundcap%
\pgfsetroundjoin%
\pgfsetlinewidth{0.803000pt}%
\definecolor{currentstroke}{rgb}{0.800000,0.800000,0.800000}%
\pgfsetstrokecolor{currentstroke}%
\pgfsetdash{}{0pt}%
\pgfpathmoveto{\pgfqpoint{0.473880in}{0.966922in}}%
\pgfpathlineto{\pgfqpoint{3.529921in}{0.966922in}}%
\pgfusepath{stroke}%
\end{pgfscope}%
\begin{pgfscope}%
\definecolor{textcolor}{rgb}{0.150000,0.150000,0.150000}%
\pgfsetstrokecolor{textcolor}%
\pgfsetfillcolor{textcolor}%
\pgftext[x=0.194444in, y=0.923519in, left, base]{\color{textcolor}\rmfamily\fontsize{8.800000}{10.560000}\selectfont \(\displaystyle {2.4}\)}%
\end{pgfscope}%
\begin{pgfscope}%
\pgfpathrectangle{\pgfqpoint{0.473880in}{0.420833in}}{\pgfqpoint{3.056041in}{2.047833in}}%
\pgfusepath{clip}%
\pgfsetroundcap%
\pgfsetroundjoin%
\pgfsetlinewidth{0.803000pt}%
\definecolor{currentstroke}{rgb}{0.800000,0.800000,0.800000}%
\pgfsetstrokecolor{currentstroke}%
\pgfsetdash{}{0pt}%
\pgfpathmoveto{\pgfqpoint{0.473880in}{1.239967in}}%
\pgfpathlineto{\pgfqpoint{3.529921in}{1.239967in}}%
\pgfusepath{stroke}%
\end{pgfscope}%
\begin{pgfscope}%
\definecolor{textcolor}{rgb}{0.150000,0.150000,0.150000}%
\pgfsetstrokecolor{textcolor}%
\pgfsetfillcolor{textcolor}%
\pgftext[x=0.194444in, y=1.196564in, left, base]{\color{textcolor}\rmfamily\fontsize{8.800000}{10.560000}\selectfont \(\displaystyle {2.6}\)}%
\end{pgfscope}%
\begin{pgfscope}%
\pgfpathrectangle{\pgfqpoint{0.473880in}{0.420833in}}{\pgfqpoint{3.056041in}{2.047833in}}%
\pgfusepath{clip}%
\pgfsetroundcap%
\pgfsetroundjoin%
\pgfsetlinewidth{0.803000pt}%
\definecolor{currentstroke}{rgb}{0.800000,0.800000,0.800000}%
\pgfsetstrokecolor{currentstroke}%
\pgfsetdash{}{0pt}%
\pgfpathmoveto{\pgfqpoint{0.473880in}{1.513011in}}%
\pgfpathlineto{\pgfqpoint{3.529921in}{1.513011in}}%
\pgfusepath{stroke}%
\end{pgfscope}%
\begin{pgfscope}%
\definecolor{textcolor}{rgb}{0.150000,0.150000,0.150000}%
\pgfsetstrokecolor{textcolor}%
\pgfsetfillcolor{textcolor}%
\pgftext[x=0.194444in, y=1.469608in, left, base]{\color{textcolor}\rmfamily\fontsize{8.800000}{10.560000}\selectfont \(\displaystyle {2.8}\)}%
\end{pgfscope}%
\begin{pgfscope}%
\pgfpathrectangle{\pgfqpoint{0.473880in}{0.420833in}}{\pgfqpoint{3.056041in}{2.047833in}}%
\pgfusepath{clip}%
\pgfsetroundcap%
\pgfsetroundjoin%
\pgfsetlinewidth{0.803000pt}%
\definecolor{currentstroke}{rgb}{0.800000,0.800000,0.800000}%
\pgfsetstrokecolor{currentstroke}%
\pgfsetdash{}{0pt}%
\pgfpathmoveto{\pgfqpoint{0.473880in}{1.786056in}}%
\pgfpathlineto{\pgfqpoint{3.529921in}{1.786056in}}%
\pgfusepath{stroke}%
\end{pgfscope}%
\begin{pgfscope}%
\definecolor{textcolor}{rgb}{0.150000,0.150000,0.150000}%
\pgfsetstrokecolor{textcolor}%
\pgfsetfillcolor{textcolor}%
\pgftext[x=0.194444in, y=1.742653in, left, base]{\color{textcolor}\rmfamily\fontsize{8.800000}{10.560000}\selectfont \(\displaystyle {3.0}\)}%
\end{pgfscope}%
\begin{pgfscope}%
\pgfpathrectangle{\pgfqpoint{0.473880in}{0.420833in}}{\pgfqpoint{3.056041in}{2.047833in}}%
\pgfusepath{clip}%
\pgfsetroundcap%
\pgfsetroundjoin%
\pgfsetlinewidth{0.803000pt}%
\definecolor{currentstroke}{rgb}{0.800000,0.800000,0.800000}%
\pgfsetstrokecolor{currentstroke}%
\pgfsetdash{}{0pt}%
\pgfpathmoveto{\pgfqpoint{0.473880in}{2.059100in}}%
\pgfpathlineto{\pgfqpoint{3.529921in}{2.059100in}}%
\pgfusepath{stroke}%
\end{pgfscope}%
\begin{pgfscope}%
\definecolor{textcolor}{rgb}{0.150000,0.150000,0.150000}%
\pgfsetstrokecolor{textcolor}%
\pgfsetfillcolor{textcolor}%
\pgftext[x=0.194444in, y=2.015697in, left, base]{\color{textcolor}\rmfamily\fontsize{8.800000}{10.560000}\selectfont \(\displaystyle {3.2}\)}%
\end{pgfscope}%
\begin{pgfscope}%
\pgfpathrectangle{\pgfqpoint{0.473880in}{0.420833in}}{\pgfqpoint{3.056041in}{2.047833in}}%
\pgfusepath{clip}%
\pgfsetroundcap%
\pgfsetroundjoin%
\pgfsetlinewidth{0.803000pt}%
\definecolor{currentstroke}{rgb}{0.800000,0.800000,0.800000}%
\pgfsetstrokecolor{currentstroke}%
\pgfsetdash{}{0pt}%
\pgfpathmoveto{\pgfqpoint{0.473880in}{2.332144in}}%
\pgfpathlineto{\pgfqpoint{3.529921in}{2.332144in}}%
\pgfusepath{stroke}%
\end{pgfscope}%
\begin{pgfscope}%
\definecolor{textcolor}{rgb}{0.150000,0.150000,0.150000}%
\pgfsetstrokecolor{textcolor}%
\pgfsetfillcolor{textcolor}%
\pgftext[x=0.194444in, y=2.288742in, left, base]{\color{textcolor}\rmfamily\fontsize{8.800000}{10.560000}\selectfont \(\displaystyle {3.4}\)}%
\end{pgfscope}%
\begin{pgfscope}%
\definecolor{textcolor}{rgb}{0.150000,0.150000,0.150000}%
\pgfsetstrokecolor{textcolor}%
\pgfsetfillcolor{textcolor}%
\pgftext[x=0.138889in,y=1.444750in,,bottom,rotate=90.000000]{\color{textcolor}\rmfamily\fontsize{9.600000}{11.520000}\selectfont Mean test regression error (RMSE)}%
\end{pgfscope}%
\begin{pgfscope}%
\pgfpathrectangle{\pgfqpoint{0.473880in}{0.420833in}}{\pgfqpoint{3.056041in}{2.047833in}}%
\pgfusepath{clip}%
\pgfsetbuttcap%
\pgfsetroundjoin%
\definecolor{currentfill}{rgb}{0.121569,0.466667,0.705882}%
\pgfsetfillcolor{currentfill}%
\pgfsetfillopacity{0.100000}%
\pgfsetlinewidth{0.803000pt}%
\definecolor{currentstroke}{rgb}{0.121569,0.466667,0.705882}%
\pgfsetstrokecolor{currentstroke}%
\pgfsetstrokeopacity{0.100000}%
\pgfsetdash{}{0pt}%
\pgfsys@defobject{currentmarker}{\pgfqpoint{0.473880in}{1.160784in}}{\pgfqpoint{3.529921in}{2.455014in}}{%
\pgfpathmoveto{\pgfqpoint{0.473880in}{2.455014in}}%
\pgfpathlineto{\pgfqpoint{0.473880in}{2.318492in}}%
\pgfpathlineto{\pgfqpoint{0.634724in}{1.676838in}}%
\pgfpathlineto{\pgfqpoint{0.956413in}{1.499359in}}%
\pgfpathlineto{\pgfqpoint{1.278101in}{1.390141in}}%
\pgfpathlineto{\pgfqpoint{1.921478in}{1.286384in}}%
\pgfpathlineto{\pgfqpoint{3.529921in}{1.160784in}}%
\pgfpathlineto{\pgfqpoint{3.529921in}{1.177166in}}%
\pgfpathlineto{\pgfqpoint{3.529921in}{1.177166in}}%
\pgfpathlineto{\pgfqpoint{1.921478in}{1.302767in}}%
\pgfpathlineto{\pgfqpoint{1.278101in}{1.417446in}}%
\pgfpathlineto{\pgfqpoint{0.956413in}{1.553968in}}%
\pgfpathlineto{\pgfqpoint{0.634724in}{1.731447in}}%
\pgfpathlineto{\pgfqpoint{0.473880in}{2.455014in}}%
\pgfpathclose%
\pgfusepath{stroke,fill}%
}%
\begin{pgfscope}%
\pgfsys@transformshift{0.000000in}{0.000000in}%
\pgfsys@useobject{currentmarker}{}%
\end{pgfscope}%
\end{pgfscope}%
\begin{pgfscope}%
\pgfpathrectangle{\pgfqpoint{0.473880in}{0.420833in}}{\pgfqpoint{3.056041in}{2.047833in}}%
\pgfusepath{clip}%
\pgfsetbuttcap%
\pgfsetroundjoin%
\definecolor{currentfill}{rgb}{1.000000,0.498039,0.054902}%
\pgfsetfillcolor{currentfill}%
\pgfsetfillopacity{0.100000}%
\pgfsetlinewidth{0.803000pt}%
\definecolor{currentstroke}{rgb}{1.000000,0.498039,0.054902}%
\pgfsetstrokecolor{currentstroke}%
\pgfsetstrokeopacity{0.100000}%
\pgfsetdash{}{0pt}%
\pgfsys@defobject{currentmarker}{\pgfqpoint{0.489965in}{1.046105in}}{\pgfqpoint{3.546005in}{1.676838in}}{%
\pgfpathmoveto{\pgfqpoint{0.489965in}{1.676838in}}%
\pgfpathlineto{\pgfqpoint{0.489965in}{1.649533in}}%
\pgfpathlineto{\pgfqpoint{0.650809in}{1.448846in}}%
\pgfpathlineto{\pgfqpoint{0.972497in}{1.310958in}}%
\pgfpathlineto{\pgfqpoint{1.294186in}{1.275462in}}%
\pgfpathlineto{\pgfqpoint{1.937563in}{1.203106in}}%
\pgfpathlineto{\pgfqpoint{3.546005in}{1.046105in}}%
\pgfpathlineto{\pgfqpoint{3.546005in}{1.062488in}}%
\pgfpathlineto{\pgfqpoint{3.546005in}{1.062488in}}%
\pgfpathlineto{\pgfqpoint{1.937563in}{1.222219in}}%
\pgfpathlineto{\pgfqpoint{1.294186in}{1.286384in}}%
\pgfpathlineto{\pgfqpoint{0.972497in}{1.332802in}}%
\pgfpathlineto{\pgfqpoint{0.650809in}{1.467959in}}%
\pgfpathlineto{\pgfqpoint{0.489965in}{1.676838in}}%
\pgfpathclose%
\pgfusepath{stroke,fill}%
}%
\begin{pgfscope}%
\pgfsys@transformshift{0.000000in}{0.000000in}%
\pgfsys@useobject{currentmarker}{}%
\end{pgfscope}%
\end{pgfscope}%
\begin{pgfscope}%
\pgfpathrectangle{\pgfqpoint{0.473880in}{0.420833in}}{\pgfqpoint{3.056041in}{2.047833in}}%
\pgfusepath{clip}%
\pgfsetbuttcap%
\pgfsetroundjoin%
\pgfsetlinewidth{1.204500pt}%
\definecolor{currentstroke}{rgb}{0.121569,0.466667,0.705882}%
\pgfsetstrokecolor{currentstroke}%
\pgfsetdash{{4.440000pt}{1.920000pt}}{0.000000pt}%
\pgfpathmoveto{\pgfqpoint{0.473880in}{2.386753in}}%
\pgfpathlineto{\pgfqpoint{0.634724in}{1.704142in}}%
\pgfpathlineto{\pgfqpoint{0.956413in}{1.526663in}}%
\pgfpathlineto{\pgfqpoint{1.278101in}{1.403793in}}%
\pgfpathlineto{\pgfqpoint{1.921478in}{1.294576in}}%
\pgfpathlineto{\pgfqpoint{3.529921in}{1.168975in}}%
\pgfusepath{stroke}%
\end{pgfscope}%
\begin{pgfscope}%
\pgfpathrectangle{\pgfqpoint{0.473880in}{0.420833in}}{\pgfqpoint{3.056041in}{2.047833in}}%
\pgfusepath{clip}%
\pgfsetbuttcap%
\pgfsetroundjoin%
\definecolor{currentfill}{rgb}{0.121569,0.466667,0.705882}%
\pgfsetfillcolor{currentfill}%
\pgfsetlinewidth{0.752812pt}%
\definecolor{currentstroke}{rgb}{0.121569,0.466667,0.705882}%
\pgfsetstrokecolor{currentstroke}%
\pgfsetdash{}{0pt}%
\pgfsys@defobject{currentmarker}{\pgfqpoint{-0.033333in}{-0.033333in}}{\pgfqpoint{0.033333in}{0.033333in}}{%
\pgfpathmoveto{\pgfqpoint{-0.033333in}{-0.033333in}}%
\pgfpathlineto{\pgfqpoint{0.033333in}{0.033333in}}%
\pgfpathmoveto{\pgfqpoint{-0.033333in}{0.033333in}}%
\pgfpathlineto{\pgfqpoint{0.033333in}{-0.033333in}}%
\pgfusepath{stroke,fill}%
}%
\begin{pgfscope}%
\pgfsys@transformshift{0.473880in}{2.386753in}%
\pgfsys@useobject{currentmarker}{}%
\end{pgfscope}%
\begin{pgfscope}%
\pgfsys@transformshift{0.634724in}{1.704142in}%
\pgfsys@useobject{currentmarker}{}%
\end{pgfscope}%
\begin{pgfscope}%
\pgfsys@transformshift{0.956413in}{1.526663in}%
\pgfsys@useobject{currentmarker}{}%
\end{pgfscope}%
\begin{pgfscope}%
\pgfsys@transformshift{1.278101in}{1.403793in}%
\pgfsys@useobject{currentmarker}{}%
\end{pgfscope}%
\begin{pgfscope}%
\pgfsys@transformshift{1.921478in}{1.294576in}%
\pgfsys@useobject{currentmarker}{}%
\end{pgfscope}%
\begin{pgfscope}%
\pgfsys@transformshift{3.529921in}{1.168975in}%
\pgfsys@useobject{currentmarker}{}%
\end{pgfscope}%
\end{pgfscope}%
\begin{pgfscope}%
\pgfpathrectangle{\pgfqpoint{0.473880in}{0.420833in}}{\pgfqpoint{3.056041in}{2.047833in}}%
\pgfusepath{clip}%
\pgfsetroundcap%
\pgfsetroundjoin%
\pgfsetlinewidth{1.204500pt}%
\definecolor{currentstroke}{rgb}{1.000000,0.498039,0.054902}%
\pgfsetstrokecolor{currentstroke}%
\pgfsetdash{}{0pt}%
\pgfpathmoveto{\pgfqpoint{0.489965in}{1.663186in}}%
\pgfpathlineto{\pgfqpoint{0.650809in}{1.458402in}}%
\pgfpathlineto{\pgfqpoint{0.972497in}{1.321880in}}%
\pgfpathlineto{\pgfqpoint{1.294186in}{1.280923in}}%
\pgfpathlineto{\pgfqpoint{1.937563in}{1.212662in}}%
\pgfpathlineto{\pgfqpoint{3.539921in}{1.054896in}}%
\pgfusepath{stroke}%
\end{pgfscope}%
\begin{pgfscope}%
\pgfpathrectangle{\pgfqpoint{0.473880in}{0.420833in}}{\pgfqpoint{3.056041in}{2.047833in}}%
\pgfusepath{clip}%
\pgfsetbuttcap%
\pgfsetroundjoin%
\definecolor{currentfill}{rgb}{1.000000,0.498039,0.054902}%
\pgfsetfillcolor{currentfill}%
\pgfsetlinewidth{0.752812pt}%
\definecolor{currentstroke}{rgb}{1.000000,0.498039,0.054902}%
\pgfsetstrokecolor{currentstroke}%
\pgfsetdash{}{0pt}%
\pgfsys@defobject{currentmarker}{\pgfqpoint{-0.033333in}{-0.033333in}}{\pgfqpoint{0.033333in}{0.033333in}}{%
\pgfpathmoveto{\pgfqpoint{-0.033333in}{-0.033333in}}%
\pgfpathlineto{\pgfqpoint{0.033333in}{0.033333in}}%
\pgfpathmoveto{\pgfqpoint{-0.033333in}{0.033333in}}%
\pgfpathlineto{\pgfqpoint{0.033333in}{-0.033333in}}%
\pgfusepath{stroke,fill}%
}%
\begin{pgfscope}%
\pgfsys@transformshift{0.489965in}{1.663186in}%
\pgfsys@useobject{currentmarker}{}%
\end{pgfscope}%
\begin{pgfscope}%
\pgfsys@transformshift{0.650809in}{1.458402in}%
\pgfsys@useobject{currentmarker}{}%
\end{pgfscope}%
\begin{pgfscope}%
\pgfsys@transformshift{0.972497in}{1.321880in}%
\pgfsys@useobject{currentmarker}{}%
\end{pgfscope}%
\begin{pgfscope}%
\pgfsys@transformshift{1.294186in}{1.280923in}%
\pgfsys@useobject{currentmarker}{}%
\end{pgfscope}%
\begin{pgfscope}%
\pgfsys@transformshift{1.937563in}{1.212662in}%
\pgfsys@useobject{currentmarker}{}%
\end{pgfscope}%
\begin{pgfscope}%
\pgfsys@transformshift{3.546005in}{1.054296in}%
\pgfsys@useobject{currentmarker}{}%
\end{pgfscope}%
\end{pgfscope}%
\begin{pgfscope}%
\pgfpathrectangle{\pgfqpoint{0.473880in}{0.420833in}}{\pgfqpoint{3.056041in}{2.047833in}}%
\pgfusepath{clip}%
\pgfsetbuttcap%
\pgfsetroundjoin%
\pgfsetlinewidth{1.204500pt}%
\definecolor{currentstroke}{rgb}{0.000000,0.000000,0.000000}%
\pgfsetstrokecolor{currentstroke}%
\pgfsetdash{{1.200000pt}{1.980000pt}}{0.000000pt}%
\pgfpathmoveto{\pgfqpoint{0.473880in}{0.639269in}}%
\pgfpathlineto{\pgfqpoint{3.529921in}{0.639269in}}%
\pgfusepath{stroke}%
\end{pgfscope}%
\begin{pgfscope}%
\pgfpathrectangle{\pgfqpoint{0.473880in}{0.420833in}}{\pgfqpoint{3.056041in}{2.047833in}}%
\pgfusepath{clip}%
\pgfsetbuttcap%
\pgfsetroundjoin%
\pgfsetlinewidth{1.204500pt}%
\definecolor{currentstroke}{rgb}{0.000000,0.000000,0.000000}%
\pgfsetstrokecolor{currentstroke}%
\pgfsetdash{{4.440000pt}{1.920000pt}}{0.000000pt}%
\pgfpathmoveto{\pgfqpoint{0.473880in}{0.734834in}}%
\pgfpathlineto{\pgfqpoint{3.529921in}{0.734834in}}%
\pgfusepath{stroke}%
\end{pgfscope}%
\begin{pgfscope}%
\pgfsetrectcap%
\pgfsetmiterjoin%
\pgfsetlinewidth{1.003750pt}%
\definecolor{currentstroke}{rgb}{0.800000,0.800000,0.800000}%
\pgfsetstrokecolor{currentstroke}%
\pgfsetdash{}{0pt}%
\pgfpathmoveto{\pgfqpoint{0.473880in}{0.420833in}}%
\pgfpathlineto{\pgfqpoint{0.473880in}{2.468667in}}%
\pgfusepath{stroke}%
\end{pgfscope}%
\begin{pgfscope}%
\pgfsetrectcap%
\pgfsetmiterjoin%
\pgfsetlinewidth{1.003750pt}%
\definecolor{currentstroke}{rgb}{0.800000,0.800000,0.800000}%
\pgfsetstrokecolor{currentstroke}%
\pgfsetdash{}{0pt}%
\pgfpathmoveto{\pgfqpoint{0.473880in}{0.420833in}}%
\pgfpathlineto{\pgfqpoint{3.529921in}{0.420833in}}%
\pgfusepath{stroke}%
\end{pgfscope}%
\begin{pgfscope}%
\pgfsetroundcap%
\pgfsetroundjoin%
\definecolor{currentfill}{rgb}{0.862745,0.862745,0.862745}%
\pgfsetfillcolor{currentfill}%
\pgfsetlinewidth{0.803000pt}%
\definecolor{currentstroke}{rgb}{1.000000,1.000000,1.000000}%
\pgfsetstrokecolor{currentstroke}%
\pgfsetdash{}{0pt}%
\pgfpathmoveto{\pgfqpoint{0.816413in}{1.226737in}}%
\pgfpathquadraticcurveto{\pgfqpoint{0.746854in}{1.058051in}}{\pgfqpoint{0.677296in}{0.889365in}}%
\pgfpathlineto{\pgfqpoint{0.722236in}{0.870833in}}%
\pgfpathquadraticcurveto{\pgfqpoint{0.651945in}{0.809671in}}{\pgfqpoint{0.581655in}{0.748509in}}%
\pgfpathquadraticcurveto{\pgfqpoint{0.574905in}{0.841439in}}{\pgfqpoint{0.568155in}{0.934369in}}%
\pgfpathlineto{\pgfqpoint{0.613095in}{0.915838in}}%
\pgfpathquadraticcurveto{\pgfqpoint{0.682654in}{1.084524in}}{\pgfqpoint{0.752212in}{1.253210in}}%
\pgfpathlineto{\pgfqpoint{0.816413in}{1.226737in}}%
\pgfpathclose%
\pgfusepath{stroke,fill}%
\end{pgfscope}%
\begin{pgfscope}%
\definecolor{textcolor}{rgb}{0.862745,0.862745,0.862745}%
\pgfsetstrokecolor{textcolor}%
\pgfsetfillcolor{textcolor}%
\pgftext[x=0.795569in,y=1.048836in,left,]{\color{textcolor}\rmfamily\fontsize{12.000000}{14.400000}\selectfont better}%
\end{pgfscope}%
\begin{pgfscope}%
\pgfsetroundcap%
\pgfsetroundjoin%
\definecolor{currentfill}{rgb}{0.000000,0.501961,0.000000}%
\pgfsetfillcolor{currentfill}%
\pgfsetlinewidth{0.803000pt}%
\definecolor{currentstroke}{rgb}{1.000000,1.000000,1.000000}%
\pgfsetstrokecolor{currentstroke}%
\pgfsetdash{}{0pt}%
\pgfpathmoveto{\pgfqpoint{1.913446in}{1.280687in}}%
\pgfpathquadraticcurveto{\pgfqpoint{1.634895in}{1.280687in}}{\pgfqpoint{1.356345in}{1.280687in}}%
\pgfpathlineto{\pgfqpoint{1.356345in}{1.225131in}}%
\pgfpathquadraticcurveto{\pgfqpoint{1.273007in}{1.259853in}}{\pgfqpoint{1.189670in}{1.294576in}}%
\pgfpathquadraticcurveto{\pgfqpoint{1.273007in}{1.329298in}}{\pgfqpoint{1.356345in}{1.364020in}}%
\pgfpathlineto{\pgfqpoint{1.356345in}{1.308464in}}%
\pgfpathquadraticcurveto{\pgfqpoint{1.634895in}{1.308464in}}{\pgfqpoint{1.913446in}{1.308464in}}%
\pgfpathlineto{\pgfqpoint{1.913446in}{1.280687in}}%
\pgfpathclose%
\pgfusepath{stroke,fill}%
\end{pgfscope}%
\begin{pgfscope}%
\definecolor{textcolor}{rgb}{0.000000,0.501961,0.000000}%
\pgfsetstrokecolor{textcolor}%
\pgfsetfillcolor{textcolor}%
\pgftext[x=1.471115in,y=1.499359in,left,]{\color{textcolor}\rmfamily\fontsize{12.000000}{14.400000}\selectfont 50\% \(\displaystyle \varepsilon\) saving}%
\end{pgfscope}%
\begin{pgfscope}%
\pgfsetbuttcap%
\pgfsetmiterjoin%
\definecolor{currentfill}{rgb}{1.000000,1.000000,1.000000}%
\pgfsetfillcolor{currentfill}%
\pgfsetfillopacity{0.800000}%
\pgfsetlinewidth{0.803000pt}%
\definecolor{currentstroke}{rgb}{0.800000,0.800000,0.800000}%
\pgfsetstrokecolor{currentstroke}%
\pgfsetstrokeopacity{0.800000}%
\pgfsetdash{}{0pt}%
\pgfpathmoveto{\pgfqpoint{1.066317in}{1.654222in}}%
\pgfpathlineto{\pgfqpoint{3.444365in}{1.654222in}}%
\pgfpathquadraticcurveto{\pgfqpoint{3.468810in}{1.654222in}}{\pgfqpoint{3.468810in}{1.678667in}}%
\pgfpathlineto{\pgfqpoint{3.468810in}{2.383111in}}%
\pgfpathquadraticcurveto{\pgfqpoint{3.468810in}{2.407556in}}{\pgfqpoint{3.444365in}{2.407556in}}%
\pgfpathlineto{\pgfqpoint{1.066317in}{2.407556in}}%
\pgfpathquadraticcurveto{\pgfqpoint{1.041873in}{2.407556in}}{\pgfqpoint{1.041873in}{2.383111in}}%
\pgfpathlineto{\pgfqpoint{1.041873in}{1.678667in}}%
\pgfpathquadraticcurveto{\pgfqpoint{1.041873in}{1.654222in}}{\pgfqpoint{1.066317in}{1.654222in}}%
\pgfpathclose%
\pgfusepath{stroke,fill}%
\end{pgfscope}%
\begin{pgfscope}%
\pgfsetbuttcap%
\pgfsetroundjoin%
\pgfsetlinewidth{1.204500pt}%
\definecolor{currentstroke}{rgb}{0.121569,0.466667,0.705882}%
\pgfsetstrokecolor{currentstroke}%
\pgfsetdash{{4.440000pt}{1.920000pt}}{0.000000pt}%
\pgfpathmoveto{\pgfqpoint{1.090762in}{2.314639in}}%
\pgfpathlineto{\pgfqpoint{1.335206in}{2.314639in}}%
\pgfusepath{stroke}%
\end{pgfscope}%
\begin{pgfscope}%
\pgfsetbuttcap%
\pgfsetroundjoin%
\definecolor{currentfill}{rgb}{0.121569,0.466667,0.705882}%
\pgfsetfillcolor{currentfill}%
\pgfsetlinewidth{0.752812pt}%
\definecolor{currentstroke}{rgb}{0.121569,0.466667,0.705882}%
\pgfsetstrokecolor{currentstroke}%
\pgfsetdash{}{0pt}%
\pgfsys@defobject{currentmarker}{\pgfqpoint{-0.033333in}{-0.033333in}}{\pgfqpoint{0.033333in}{0.033333in}}{%
\pgfpathmoveto{\pgfqpoint{-0.033333in}{-0.033333in}}%
\pgfpathlineto{\pgfqpoint{0.033333in}{0.033333in}}%
\pgfpathmoveto{\pgfqpoint{-0.033333in}{0.033333in}}%
\pgfpathlineto{\pgfqpoint{0.033333in}{-0.033333in}}%
\pgfusepath{stroke,fill}%
}%
\begin{pgfscope}%
\pgfsys@transformshift{1.212984in}{2.314639in}%
\pgfsys@useobject{currentmarker}{}%
\end{pgfscope}%
\end{pgfscope}%
\begin{pgfscope}%
\definecolor{textcolor}{rgb}{0.150000,0.150000,0.150000}%
\pgfsetstrokecolor{textcolor}%
\pgfsetfillcolor{textcolor}%
\pgftext[x=1.432984in,y=2.271861in,left,base]{\color{textcolor}\rmfamily\fontsize{8.800000}{10.560000}\selectfont Maddock et al.}%
\end{pgfscope}%
\begin{pgfscope}%
\pgfsetroundcap%
\pgfsetroundjoin%
\pgfsetlinewidth{1.204500pt}%
\definecolor{currentstroke}{rgb}{1.000000,0.498039,0.054902}%
\pgfsetstrokecolor{currentstroke}%
\pgfsetdash{}{0pt}%
\pgfpathmoveto{\pgfqpoint{1.090762in}{2.142417in}}%
\pgfpathlineto{\pgfqpoint{1.335206in}{2.142417in}}%
\pgfusepath{stroke}%
\end{pgfscope}%
\begin{pgfscope}%
\pgfsetbuttcap%
\pgfsetroundjoin%
\definecolor{currentfill}{rgb}{1.000000,0.498039,0.054902}%
\pgfsetfillcolor{currentfill}%
\pgfsetlinewidth{0.752812pt}%
\definecolor{currentstroke}{rgb}{1.000000,0.498039,0.054902}%
\pgfsetstrokecolor{currentstroke}%
\pgfsetdash{}{0pt}%
\pgfsys@defobject{currentmarker}{\pgfqpoint{-0.033333in}{-0.033333in}}{\pgfqpoint{0.033333in}{0.033333in}}{%
\pgfpathmoveto{\pgfqpoint{-0.033333in}{-0.033333in}}%
\pgfpathlineto{\pgfqpoint{0.033333in}{0.033333in}}%
\pgfpathmoveto{\pgfqpoint{-0.033333in}{0.033333in}}%
\pgfpathlineto{\pgfqpoint{0.033333in}{-0.033333in}}%
\pgfusepath{stroke,fill}%
}%
\begin{pgfscope}%
\pgfsys@transformshift{1.212984in}{2.142417in}%
\pgfsys@useobject{currentmarker}{}%
\end{pgfscope}%
\end{pgfscope}%
\begin{pgfscope}%
\definecolor{textcolor}{rgb}{0.150000,0.150000,0.150000}%
\pgfsetstrokecolor{textcolor}%
\pgfsetfillcolor{textcolor}%
\pgftext[x=1.432984in,y=2.099639in,left,base]{\color{textcolor}\rmfamily\fontsize{8.800000}{10.560000}\selectfont S-BDT}%
\end{pgfscope}%
\begin{pgfscope}%
\pgfsetbuttcap%
\pgfsetroundjoin%
\pgfsetlinewidth{1.204500pt}%
\definecolor{currentstroke}{rgb}{0.000000,0.000000,0.000000}%
\pgfsetstrokecolor{currentstroke}%
\pgfsetdash{{1.200000pt}{1.980000pt}}{0.000000pt}%
\pgfpathmoveto{\pgfqpoint{1.090762in}{1.963250in}}%
\pgfpathlineto{\pgfqpoint{1.335206in}{1.963250in}}%
\pgfusepath{stroke}%
\end{pgfscope}%
\begin{pgfscope}%
\definecolor{textcolor}{rgb}{0.150000,0.150000,0.150000}%
\pgfsetstrokecolor{textcolor}%
\pgfsetfillcolor{textcolor}%
\pgftext[x=1.432984in,y=1.920472in,left,base]{\color{textcolor}\rmfamily\fontsize{8.800000}{10.560000}\selectfont xgboost (nonprivate)}%
\end{pgfscope}%
\begin{pgfscope}%
\pgfsetbuttcap%
\pgfsetroundjoin%
\pgfsetlinewidth{1.204500pt}%
\definecolor{currentstroke}{rgb}{0.000000,0.000000,0.000000}%
\pgfsetstrokecolor{currentstroke}%
\pgfsetdash{{4.440000pt}{1.920000pt}}{0.000000pt}%
\pgfpathmoveto{\pgfqpoint{1.090762in}{1.777139in}}%
\pgfpathlineto{\pgfqpoint{1.335206in}{1.777139in}}%
\pgfusepath{stroke}%
\end{pgfscope}%
\begin{pgfscope}%
\definecolor{textcolor}{rgb}{0.150000,0.150000,0.150000}%
\pgfsetstrokecolor{textcolor}%
\pgfsetfillcolor{textcolor}%
\pgftext[x=1.432984in,y=1.734361in,left,base]{\color{textcolor}\rmfamily\fontsize{8.800000}{10.560000}\selectfont xgboost random splits (nonprivate)}%
\end{pgfscope}%
\end{pgfpicture}%
\makeatother%
\endgroup%

%% file: images/adult_main_results.pgf
\begingroup%
\makeatletter%
\begin{pgfpicture}%
\pgfpathrectangle{\pgfpointorigin}{\pgfqpoint{3.612000in}{2.468667in}}%
\pgfusepath{use as bounding box, clip}%
\begin{pgfscope}%
\pgfsetbuttcap%
\pgfsetmiterjoin%
\definecolor{currentfill}{rgb}{1.000000,1.000000,1.000000}%
\pgfsetfillcolor{currentfill}%
\pgfsetlinewidth{0.000000pt}%
\definecolor{currentstroke}{rgb}{1.000000,1.000000,1.000000}%
\pgfsetstrokecolor{currentstroke}%
\pgfsetdash{}{0pt}%
\pgfpathmoveto{\pgfqpoint{0.000000in}{0.000000in}}%
\pgfpathlineto{\pgfqpoint{3.612000in}{0.000000in}}%
\pgfpathlineto{\pgfqpoint{3.612000in}{2.468667in}}%
\pgfpathlineto{\pgfqpoint{0.000000in}{2.468667in}}%
\pgfpathclose%
\pgfusepath{fill}%
\end{pgfscope}%
\begin{pgfscope}%
\pgfsetbuttcap%
\pgfsetmiterjoin%
\definecolor{currentfill}{rgb}{1.000000,1.000000,1.000000}%
\pgfsetfillcolor{currentfill}%
\pgfsetlinewidth{0.000000pt}%
\definecolor{currentstroke}{rgb}{0.000000,0.000000,0.000000}%
\pgfsetstrokecolor{currentstroke}%
\pgfsetstrokeopacity{0.000000}%
\pgfsetdash{}{0pt}%
\pgfpathmoveto{\pgfqpoint{0.522684in}{0.420833in}}%
\pgfpathlineto{\pgfqpoint{3.465685in}{0.420833in}}%
\pgfpathlineto{\pgfqpoint{3.465685in}{2.425264in}}%
\pgfpathlineto{\pgfqpoint{0.522684in}{2.425264in}}%
\pgfpathclose%
\pgfusepath{fill}%
\end{pgfscope}%
\begin{pgfscope}%
\pgfpathrectangle{\pgfqpoint{0.522684in}{0.420833in}}{\pgfqpoint{2.943002in}{2.004431in}}%
\pgfusepath{clip}%
\pgfsetroundcap%
\pgfsetroundjoin%
\pgfsetlinewidth{0.803000pt}%
\definecolor{currentstroke}{rgb}{0.800000,0.800000,0.800000}%
\pgfsetstrokecolor{currentstroke}%
\pgfsetdash{}{0pt}%
\pgfpathmoveto{\pgfqpoint{0.522684in}{0.420833in}}%
\pgfpathlineto{\pgfqpoint{0.522684in}{2.425264in}}%
\pgfusepath{stroke}%
\end{pgfscope}%
\begin{pgfscope}%
\definecolor{textcolor}{rgb}{0.150000,0.150000,0.150000}%
\pgfsetstrokecolor{textcolor}%
\pgfsetfillcolor{textcolor}%
\pgftext[x=0.522684in,y=0.305556in,,top]{\color{textcolor}\rmfamily\fontsize{8.800000}{10.560000}\selectfont \(\displaystyle {0.010}\)}%
\end{pgfscope}%
\begin{pgfscope}%
\pgfpathrectangle{\pgfqpoint{0.522684in}{0.420833in}}{\pgfqpoint{2.943002in}{2.004431in}}%
\pgfusepath{clip}%
\pgfsetroundcap%
\pgfsetroundjoin%
\pgfsetlinewidth{0.803000pt}%
\definecolor{currentstroke}{rgb}{0.800000,0.800000,0.800000}%
\pgfsetstrokecolor{currentstroke}%
\pgfsetdash{}{0pt}%
\pgfpathmoveto{\pgfqpoint{1.108253in}{0.420833in}}%
\pgfpathlineto{\pgfqpoint{1.108253in}{2.425264in}}%
\pgfusepath{stroke}%
\end{pgfscope}%
\begin{pgfscope}%
\definecolor{textcolor}{rgb}{0.150000,0.150000,0.150000}%
\pgfsetstrokecolor{textcolor}%
\pgfsetfillcolor{textcolor}%
\pgftext[x=1.108253in,y=0.305556in,,top]{\color{textcolor}\rmfamily\fontsize{8.800000}{10.560000}\selectfont \(\displaystyle {0.025}\)}%
\end{pgfscope}%
\begin{pgfscope}%
\pgfpathrectangle{\pgfqpoint{0.522684in}{0.420833in}}{\pgfqpoint{2.943002in}{2.004431in}}%
\pgfusepath{clip}%
\pgfsetroundcap%
\pgfsetroundjoin%
\pgfsetlinewidth{0.803000pt}%
\definecolor{currentstroke}{rgb}{0.800000,0.800000,0.800000}%
\pgfsetstrokecolor{currentstroke}%
\pgfsetdash{}{0pt}%
\pgfpathmoveto{\pgfqpoint{1.551219in}{0.420833in}}%
\pgfpathlineto{\pgfqpoint{1.551219in}{2.425264in}}%
\pgfusepath{stroke}%
\end{pgfscope}%
\begin{pgfscope}%
\definecolor{textcolor}{rgb}{0.150000,0.150000,0.150000}%
\pgfsetstrokecolor{textcolor}%
\pgfsetfillcolor{textcolor}%
\pgftext[x=1.551219in,y=0.305556in,,top]{\color{textcolor}\rmfamily\fontsize{8.800000}{10.560000}\selectfont \(\displaystyle {0.050}\)}%
\end{pgfscope}%
\begin{pgfscope}%
\pgfpathrectangle{\pgfqpoint{0.522684in}{0.420833in}}{\pgfqpoint{2.943002in}{2.004431in}}%
\pgfusepath{clip}%
\pgfsetroundcap%
\pgfsetroundjoin%
\pgfsetlinewidth{0.803000pt}%
\definecolor{currentstroke}{rgb}{0.800000,0.800000,0.800000}%
\pgfsetstrokecolor{currentstroke}%
\pgfsetdash{}{0pt}%
\pgfpathmoveto{\pgfqpoint{1.994185in}{0.420833in}}%
\pgfpathlineto{\pgfqpoint{1.994185in}{2.425264in}}%
\pgfusepath{stroke}%
\end{pgfscope}%
\begin{pgfscope}%
\definecolor{textcolor}{rgb}{0.150000,0.150000,0.150000}%
\pgfsetstrokecolor{textcolor}%
\pgfsetfillcolor{textcolor}%
\pgftext[x=1.994185in,y=0.305556in,,top]{\color{textcolor}\rmfamily\fontsize{8.800000}{10.560000}\selectfont \(\displaystyle {0.100}\)}%
\end{pgfscope}%
\begin{pgfscope}%
\pgfpathrectangle{\pgfqpoint{0.522684in}{0.420833in}}{\pgfqpoint{2.943002in}{2.004431in}}%
\pgfusepath{clip}%
\pgfsetroundcap%
\pgfsetroundjoin%
\pgfsetlinewidth{0.803000pt}%
\definecolor{currentstroke}{rgb}{0.800000,0.800000,0.800000}%
\pgfsetstrokecolor{currentstroke}%
\pgfsetdash{}{0pt}%
\pgfpathmoveto{\pgfqpoint{3.022720in}{0.420833in}}%
\pgfpathlineto{\pgfqpoint{3.022720in}{2.425264in}}%
\pgfusepath{stroke}%
\end{pgfscope}%
\begin{pgfscope}%
\definecolor{textcolor}{rgb}{0.150000,0.150000,0.150000}%
\pgfsetstrokecolor{textcolor}%
\pgfsetfillcolor{textcolor}%
\pgftext[x=3.022720in,y=0.305556in,,top]{\color{textcolor}\rmfamily\fontsize{8.800000}{10.560000}\selectfont \(\displaystyle {0.500}\)}%
\end{pgfscope}%
\begin{pgfscope}%
\pgfpathrectangle{\pgfqpoint{0.522684in}{0.420833in}}{\pgfqpoint{2.943002in}{2.004431in}}%
\pgfusepath{clip}%
\pgfsetroundcap%
\pgfsetroundjoin%
\pgfsetlinewidth{0.803000pt}%
\definecolor{currentstroke}{rgb}{0.800000,0.800000,0.800000}%
\pgfsetstrokecolor{currentstroke}%
\pgfsetdash{}{0pt}%
\pgfpathmoveto{\pgfqpoint{3.465685in}{0.420833in}}%
\pgfpathlineto{\pgfqpoint{3.465685in}{2.425264in}}%
\pgfusepath{stroke}%
\end{pgfscope}%
\begin{pgfscope}%
\definecolor{textcolor}{rgb}{0.150000,0.150000,0.150000}%
\pgfsetstrokecolor{textcolor}%
\pgfsetfillcolor{textcolor}%
\pgftext[x=3.465685in,y=0.305556in,,top]{\color{textcolor}\rmfamily\fontsize{8.800000}{10.560000}\selectfont \(\displaystyle {1.000}\)}%
\end{pgfscope}%
\begin{pgfscope}%
\pgfpathrectangle{\pgfqpoint{0.522684in}{0.420833in}}{\pgfqpoint{2.943002in}{2.004431in}}%
\pgfusepath{clip}%
\pgfsetroundcap%
\pgfsetroundjoin%
\pgfsetlinewidth{0.075281pt}%
\definecolor{currentstroke}{rgb}{0.827451,0.827451,0.827451}%
\pgfsetstrokecolor{currentstroke}%
\pgfsetdash{}{0pt}%
\pgfpathmoveto{\pgfqpoint{0.965650in}{0.420833in}}%
\pgfpathlineto{\pgfqpoint{0.965650in}{2.425264in}}%
\pgfusepath{stroke}%
\end{pgfscope}%
\begin{pgfscope}%
\pgfpathrectangle{\pgfqpoint{0.522684in}{0.420833in}}{\pgfqpoint{2.943002in}{2.004431in}}%
\pgfusepath{clip}%
\pgfsetroundcap%
\pgfsetroundjoin%
\pgfsetlinewidth{0.075281pt}%
\definecolor{currentstroke}{rgb}{0.827451,0.827451,0.827451}%
\pgfsetstrokecolor{currentstroke}%
\pgfsetdash{}{0pt}%
\pgfpathmoveto{\pgfqpoint{1.224768in}{0.420833in}}%
\pgfpathlineto{\pgfqpoint{1.224768in}{2.425264in}}%
\pgfusepath{stroke}%
\end{pgfscope}%
\begin{pgfscope}%
\pgfpathrectangle{\pgfqpoint{0.522684in}{0.420833in}}{\pgfqpoint{2.943002in}{2.004431in}}%
\pgfusepath{clip}%
\pgfsetroundcap%
\pgfsetroundjoin%
\pgfsetlinewidth{0.075281pt}%
\definecolor{currentstroke}{rgb}{0.827451,0.827451,0.827451}%
\pgfsetstrokecolor{currentstroke}%
\pgfsetdash{}{0pt}%
\pgfpathmoveto{\pgfqpoint{1.408615in}{0.420833in}}%
\pgfpathlineto{\pgfqpoint{1.408615in}{2.425264in}}%
\pgfusepath{stroke}%
\end{pgfscope}%
\begin{pgfscope}%
\pgfpathrectangle{\pgfqpoint{0.522684in}{0.420833in}}{\pgfqpoint{2.943002in}{2.004431in}}%
\pgfusepath{clip}%
\pgfsetroundcap%
\pgfsetroundjoin%
\pgfsetlinewidth{0.075281pt}%
\definecolor{currentstroke}{rgb}{0.827451,0.827451,0.827451}%
\pgfsetstrokecolor{currentstroke}%
\pgfsetdash{}{0pt}%
\pgfpathmoveto{\pgfqpoint{1.667734in}{0.420833in}}%
\pgfpathlineto{\pgfqpoint{1.667734in}{2.425264in}}%
\pgfusepath{stroke}%
\end{pgfscope}%
\begin{pgfscope}%
\pgfpathrectangle{\pgfqpoint{0.522684in}{0.420833in}}{\pgfqpoint{2.943002in}{2.004431in}}%
\pgfusepath{clip}%
\pgfsetroundcap%
\pgfsetroundjoin%
\pgfsetlinewidth{0.075281pt}%
\definecolor{currentstroke}{rgb}{0.827451,0.827451,0.827451}%
\pgfsetstrokecolor{currentstroke}%
\pgfsetdash{}{0pt}%
\pgfpathmoveto{\pgfqpoint{1.766246in}{0.420833in}}%
\pgfpathlineto{\pgfqpoint{1.766246in}{2.425264in}}%
\pgfusepath{stroke}%
\end{pgfscope}%
\begin{pgfscope}%
\pgfpathrectangle{\pgfqpoint{0.522684in}{0.420833in}}{\pgfqpoint{2.943002in}{2.004431in}}%
\pgfusepath{clip}%
\pgfsetroundcap%
\pgfsetroundjoin%
\pgfsetlinewidth{0.075281pt}%
\definecolor{currentstroke}{rgb}{0.827451,0.827451,0.827451}%
\pgfsetstrokecolor{currentstroke}%
\pgfsetdash{}{0pt}%
\pgfpathmoveto{\pgfqpoint{1.851581in}{0.420833in}}%
\pgfpathlineto{\pgfqpoint{1.851581in}{2.425264in}}%
\pgfusepath{stroke}%
\end{pgfscope}%
\begin{pgfscope}%
\pgfpathrectangle{\pgfqpoint{0.522684in}{0.420833in}}{\pgfqpoint{2.943002in}{2.004431in}}%
\pgfusepath{clip}%
\pgfsetroundcap%
\pgfsetroundjoin%
\pgfsetlinewidth{0.075281pt}%
\definecolor{currentstroke}{rgb}{0.827451,0.827451,0.827451}%
\pgfsetstrokecolor{currentstroke}%
\pgfsetdash{}{0pt}%
\pgfpathmoveto{\pgfqpoint{1.926852in}{0.420833in}}%
\pgfpathlineto{\pgfqpoint{1.926852in}{2.425264in}}%
\pgfusepath{stroke}%
\end{pgfscope}%
\begin{pgfscope}%
\pgfpathrectangle{\pgfqpoint{0.522684in}{0.420833in}}{\pgfqpoint{2.943002in}{2.004431in}}%
\pgfusepath{clip}%
\pgfsetroundcap%
\pgfsetroundjoin%
\pgfsetlinewidth{0.075281pt}%
\definecolor{currentstroke}{rgb}{0.827451,0.827451,0.827451}%
\pgfsetstrokecolor{currentstroke}%
\pgfsetdash{}{0pt}%
\pgfpathmoveto{\pgfqpoint{2.437150in}{0.420833in}}%
\pgfpathlineto{\pgfqpoint{2.437150in}{2.425264in}}%
\pgfusepath{stroke}%
\end{pgfscope}%
\begin{pgfscope}%
\pgfpathrectangle{\pgfqpoint{0.522684in}{0.420833in}}{\pgfqpoint{2.943002in}{2.004431in}}%
\pgfusepath{clip}%
\pgfsetroundcap%
\pgfsetroundjoin%
\pgfsetlinewidth{0.075281pt}%
\definecolor{currentstroke}{rgb}{0.827451,0.827451,0.827451}%
\pgfsetstrokecolor{currentstroke}%
\pgfsetdash{}{0pt}%
\pgfpathmoveto{\pgfqpoint{2.696269in}{0.420833in}}%
\pgfpathlineto{\pgfqpoint{2.696269in}{2.425264in}}%
\pgfusepath{stroke}%
\end{pgfscope}%
\begin{pgfscope}%
\pgfpathrectangle{\pgfqpoint{0.522684in}{0.420833in}}{\pgfqpoint{2.943002in}{2.004431in}}%
\pgfusepath{clip}%
\pgfsetroundcap%
\pgfsetroundjoin%
\pgfsetlinewidth{0.075281pt}%
\definecolor{currentstroke}{rgb}{0.827451,0.827451,0.827451}%
\pgfsetstrokecolor{currentstroke}%
\pgfsetdash{}{0pt}%
\pgfpathmoveto{\pgfqpoint{2.880116in}{0.420833in}}%
\pgfpathlineto{\pgfqpoint{2.880116in}{2.425264in}}%
\pgfusepath{stroke}%
\end{pgfscope}%
\begin{pgfscope}%
\pgfpathrectangle{\pgfqpoint{0.522684in}{0.420833in}}{\pgfqpoint{2.943002in}{2.004431in}}%
\pgfusepath{clip}%
\pgfsetroundcap%
\pgfsetroundjoin%
\pgfsetlinewidth{0.075281pt}%
\definecolor{currentstroke}{rgb}{0.827451,0.827451,0.827451}%
\pgfsetstrokecolor{currentstroke}%
\pgfsetdash{}{0pt}%
\pgfpathmoveto{\pgfqpoint{3.139235in}{0.420833in}}%
\pgfpathlineto{\pgfqpoint{3.139235in}{2.425264in}}%
\pgfusepath{stroke}%
\end{pgfscope}%
\begin{pgfscope}%
\pgfpathrectangle{\pgfqpoint{0.522684in}{0.420833in}}{\pgfqpoint{2.943002in}{2.004431in}}%
\pgfusepath{clip}%
\pgfsetroundcap%
\pgfsetroundjoin%
\pgfsetlinewidth{0.075281pt}%
\definecolor{currentstroke}{rgb}{0.827451,0.827451,0.827451}%
\pgfsetstrokecolor{currentstroke}%
\pgfsetdash{}{0pt}%
\pgfpathmoveto{\pgfqpoint{3.237747in}{0.420833in}}%
\pgfpathlineto{\pgfqpoint{3.237747in}{2.425264in}}%
\pgfusepath{stroke}%
\end{pgfscope}%
\begin{pgfscope}%
\pgfpathrectangle{\pgfqpoint{0.522684in}{0.420833in}}{\pgfqpoint{2.943002in}{2.004431in}}%
\pgfusepath{clip}%
\pgfsetroundcap%
\pgfsetroundjoin%
\pgfsetlinewidth{0.075281pt}%
\definecolor{currentstroke}{rgb}{0.827451,0.827451,0.827451}%
\pgfsetstrokecolor{currentstroke}%
\pgfsetdash{}{0pt}%
\pgfpathmoveto{\pgfqpoint{3.323082in}{0.420833in}}%
\pgfpathlineto{\pgfqpoint{3.323082in}{2.425264in}}%
\pgfusepath{stroke}%
\end{pgfscope}%
\begin{pgfscope}%
\pgfpathrectangle{\pgfqpoint{0.522684in}{0.420833in}}{\pgfqpoint{2.943002in}{2.004431in}}%
\pgfusepath{clip}%
\pgfsetroundcap%
\pgfsetroundjoin%
\pgfsetlinewidth{0.075281pt}%
\definecolor{currentstroke}{rgb}{0.827451,0.827451,0.827451}%
\pgfsetstrokecolor{currentstroke}%
\pgfsetdash{}{0pt}%
\pgfpathmoveto{\pgfqpoint{3.398353in}{0.420833in}}%
\pgfpathlineto{\pgfqpoint{3.398353in}{2.425264in}}%
\pgfusepath{stroke}%
\end{pgfscope}%
\begin{pgfscope}%
\definecolor{textcolor}{rgb}{0.150000,0.150000,0.150000}%
\pgfsetstrokecolor{textcolor}%
\pgfsetfillcolor{textcolor}%
\pgftext[x=1.994185in,y=0.138889in,,top]{\color{textcolor}\rmfamily\fontsize{9.600000}{11.520000}\selectfont \(\displaystyle \varepsilon\) (privacy budget)}%
\end{pgfscope}%
\begin{pgfscope}%
\pgfpathrectangle{\pgfqpoint{0.522684in}{0.420833in}}{\pgfqpoint{2.943002in}{2.004431in}}%
\pgfusepath{clip}%
\pgfsetroundcap%
\pgfsetroundjoin%
\pgfsetlinewidth{0.803000pt}%
\definecolor{currentstroke}{rgb}{0.800000,0.800000,0.800000}%
\pgfsetstrokecolor{currentstroke}%
\pgfsetdash{}{0pt}%
\pgfpathmoveto{\pgfqpoint{0.522684in}{0.420833in}}%
\pgfpathlineto{\pgfqpoint{3.465685in}{0.420833in}}%
\pgfusepath{stroke}%
\end{pgfscope}%
\begin{pgfscope}%
\definecolor{textcolor}{rgb}{0.150000,0.150000,0.150000}%
\pgfsetstrokecolor{textcolor}%
\pgfsetfillcolor{textcolor}%
\pgftext[x=0.179012in, y=0.377431in, left, base]{\color{textcolor}\rmfamily\fontsize{8.800000}{10.560000}\selectfont \(\displaystyle {0.70}\)}%
\end{pgfscope}%
\begin{pgfscope}%
\pgfpathrectangle{\pgfqpoint{0.522684in}{0.420833in}}{\pgfqpoint{2.943002in}{2.004431in}}%
\pgfusepath{clip}%
\pgfsetroundcap%
\pgfsetroundjoin%
\pgfsetlinewidth{0.803000pt}%
\definecolor{currentstroke}{rgb}{0.800000,0.800000,0.800000}%
\pgfsetstrokecolor{currentstroke}%
\pgfsetdash{}{0pt}%
\pgfpathmoveto{\pgfqpoint{0.522684in}{0.754905in}}%
\pgfpathlineto{\pgfqpoint{3.465685in}{0.754905in}}%
\pgfusepath{stroke}%
\end{pgfscope}%
\begin{pgfscope}%
\definecolor{textcolor}{rgb}{0.150000,0.150000,0.150000}%
\pgfsetstrokecolor{textcolor}%
\pgfsetfillcolor{textcolor}%
\pgftext[x=0.179012in, y=0.711502in, left, base]{\color{textcolor}\rmfamily\fontsize{8.800000}{10.560000}\selectfont \(\displaystyle {0.75}\)}%
\end{pgfscope}%
\begin{pgfscope}%
\pgfpathrectangle{\pgfqpoint{0.522684in}{0.420833in}}{\pgfqpoint{2.943002in}{2.004431in}}%
\pgfusepath{clip}%
\pgfsetroundcap%
\pgfsetroundjoin%
\pgfsetlinewidth{0.803000pt}%
\definecolor{currentstroke}{rgb}{0.800000,0.800000,0.800000}%
\pgfsetstrokecolor{currentstroke}%
\pgfsetdash{}{0pt}%
\pgfpathmoveto{\pgfqpoint{0.522684in}{1.088977in}}%
\pgfpathlineto{\pgfqpoint{3.465685in}{1.088977in}}%
\pgfusepath{stroke}%
\end{pgfscope}%
\begin{pgfscope}%
\definecolor{textcolor}{rgb}{0.150000,0.150000,0.150000}%
\pgfsetstrokecolor{textcolor}%
\pgfsetfillcolor{textcolor}%
\pgftext[x=0.179012in, y=1.045574in, left, base]{\color{textcolor}\rmfamily\fontsize{8.800000}{10.560000}\selectfont \(\displaystyle {0.80}\)}%
\end{pgfscope}%
\begin{pgfscope}%
\pgfpathrectangle{\pgfqpoint{0.522684in}{0.420833in}}{\pgfqpoint{2.943002in}{2.004431in}}%
\pgfusepath{clip}%
\pgfsetroundcap%
\pgfsetroundjoin%
\pgfsetlinewidth{0.803000pt}%
\definecolor{currentstroke}{rgb}{0.800000,0.800000,0.800000}%
\pgfsetstrokecolor{currentstroke}%
\pgfsetdash{}{0pt}%
\pgfpathmoveto{\pgfqpoint{0.522684in}{1.423049in}}%
\pgfpathlineto{\pgfqpoint{3.465685in}{1.423049in}}%
\pgfusepath{stroke}%
\end{pgfscope}%
\begin{pgfscope}%
\definecolor{textcolor}{rgb}{0.150000,0.150000,0.150000}%
\pgfsetstrokecolor{textcolor}%
\pgfsetfillcolor{textcolor}%
\pgftext[x=0.179012in, y=1.379646in, left, base]{\color{textcolor}\rmfamily\fontsize{8.800000}{10.560000}\selectfont \(\displaystyle {0.85}\)}%
\end{pgfscope}%
\begin{pgfscope}%
\pgfpathrectangle{\pgfqpoint{0.522684in}{0.420833in}}{\pgfqpoint{2.943002in}{2.004431in}}%
\pgfusepath{clip}%
\pgfsetroundcap%
\pgfsetroundjoin%
\pgfsetlinewidth{0.803000pt}%
\definecolor{currentstroke}{rgb}{0.800000,0.800000,0.800000}%
\pgfsetstrokecolor{currentstroke}%
\pgfsetdash{}{0pt}%
\pgfpathmoveto{\pgfqpoint{0.522684in}{1.757120in}}%
\pgfpathlineto{\pgfqpoint{3.465685in}{1.757120in}}%
\pgfusepath{stroke}%
\end{pgfscope}%
\begin{pgfscope}%
\definecolor{textcolor}{rgb}{0.150000,0.150000,0.150000}%
\pgfsetstrokecolor{textcolor}%
\pgfsetfillcolor{textcolor}%
\pgftext[x=0.179012in, y=1.713718in, left, base]{\color{textcolor}\rmfamily\fontsize{8.800000}{10.560000}\selectfont \(\displaystyle {0.90}\)}%
\end{pgfscope}%
\begin{pgfscope}%
\pgfpathrectangle{\pgfqpoint{0.522684in}{0.420833in}}{\pgfqpoint{2.943002in}{2.004431in}}%
\pgfusepath{clip}%
\pgfsetroundcap%
\pgfsetroundjoin%
\pgfsetlinewidth{0.803000pt}%
\definecolor{currentstroke}{rgb}{0.800000,0.800000,0.800000}%
\pgfsetstrokecolor{currentstroke}%
\pgfsetdash{}{0pt}%
\pgfpathmoveto{\pgfqpoint{0.522684in}{2.091192in}}%
\pgfpathlineto{\pgfqpoint{3.465685in}{2.091192in}}%
\pgfusepath{stroke}%
\end{pgfscope}%
\begin{pgfscope}%
\definecolor{textcolor}{rgb}{0.150000,0.150000,0.150000}%
\pgfsetstrokecolor{textcolor}%
\pgfsetfillcolor{textcolor}%
\pgftext[x=0.179012in, y=2.047789in, left, base]{\color{textcolor}\rmfamily\fontsize{8.800000}{10.560000}\selectfont \(\displaystyle {0.95}\)}%
\end{pgfscope}%
\begin{pgfscope}%
\pgfpathrectangle{\pgfqpoint{0.522684in}{0.420833in}}{\pgfqpoint{2.943002in}{2.004431in}}%
\pgfusepath{clip}%
\pgfsetroundcap%
\pgfsetroundjoin%
\pgfsetlinewidth{0.803000pt}%
\definecolor{currentstroke}{rgb}{0.800000,0.800000,0.800000}%
\pgfsetstrokecolor{currentstroke}%
\pgfsetdash{}{0pt}%
\pgfpathmoveto{\pgfqpoint{0.522684in}{2.425264in}}%
\pgfpathlineto{\pgfqpoint{3.465685in}{2.425264in}}%
\pgfusepath{stroke}%
\end{pgfscope}%
\begin{pgfscope}%
\definecolor{textcolor}{rgb}{0.150000,0.150000,0.150000}%
\pgfsetstrokecolor{textcolor}%
\pgfsetfillcolor{textcolor}%
\pgftext[x=0.179012in, y=2.381861in, left, base]{\color{textcolor}\rmfamily\fontsize{8.800000}{10.560000}\selectfont \(\displaystyle {1.00}\)}%
\end{pgfscope}%
\begin{pgfscope}%
\definecolor{textcolor}{rgb}{0.150000,0.150000,0.150000}%
\pgfsetstrokecolor{textcolor}%
\pgfsetfillcolor{textcolor}%
\pgftext[x=0.123457in,y=1.423049in,,bottom,rotate=90.000000]{\color{textcolor}\rmfamily\fontsize{9.600000}{11.520000}\selectfont Mean test AUC}%
\end{pgfscope}%
\begin{pgfscope}%
\pgfpathrectangle{\pgfqpoint{0.522684in}{0.420833in}}{\pgfqpoint{2.943002in}{2.004431in}}%
\pgfusepath{clip}%
\pgfsetbuttcap%
\pgfsetroundjoin%
\definecolor{currentfill}{rgb}{0.121569,0.466667,0.705882}%
\pgfsetfillcolor{currentfill}%
\pgfsetfillopacity{0.100000}%
\pgfsetlinewidth{0.803000pt}%
\definecolor{currentstroke}{rgb}{0.121569,0.466667,0.705882}%
\pgfsetstrokecolor{currentstroke}%
\pgfsetstrokeopacity{0.100000}%
\pgfsetdash{}{0pt}%
\pgfsys@defobject{currentmarker}{\pgfqpoint{0.522684in}{0.547781in}}{\pgfqpoint{3.465685in}{1.690306in}}{%
\pgfpathmoveto{\pgfqpoint{0.522684in}{0.601232in}}%
\pgfpathlineto{\pgfqpoint{0.522684in}{0.547781in}}%
\pgfpathlineto{\pgfqpoint{0.965650in}{1.015481in}}%
\pgfpathlineto{\pgfqpoint{1.224768in}{1.182517in}}%
\pgfpathlineto{\pgfqpoint{1.994185in}{1.549996in}}%
\pgfpathlineto{\pgfqpoint{3.022720in}{1.690306in}}%
\pgfpathlineto{\pgfqpoint{3.465685in}{1.690306in}}%
\pgfpathlineto{\pgfqpoint{3.465685in}{1.690306in}}%
\pgfpathlineto{\pgfqpoint{3.465685in}{1.690306in}}%
\pgfpathlineto{\pgfqpoint{3.022720in}{1.690306in}}%
\pgfpathlineto{\pgfqpoint{1.994185in}{1.549996in}}%
\pgfpathlineto{\pgfqpoint{1.224768in}{1.209243in}}%
\pgfpathlineto{\pgfqpoint{0.965650in}{1.042207in}}%
\pgfpathlineto{\pgfqpoint{0.522684in}{0.601232in}}%
\pgfpathclose%
\pgfusepath{stroke,fill}%
}%
\begin{pgfscope}%
\pgfsys@transformshift{0.000000in}{0.000000in}%
\pgfsys@useobject{currentmarker}{}%
\end{pgfscope}%
\end{pgfscope}%
\begin{pgfscope}%
\pgfpathrectangle{\pgfqpoint{0.522684in}{0.420833in}}{\pgfqpoint{2.943002in}{2.004431in}}%
\pgfusepath{clip}%
\pgfsetbuttcap%
\pgfsetroundjoin%
\definecolor{currentfill}{rgb}{1.000000,0.498039,0.054902}%
\pgfsetfillcolor{currentfill}%
\pgfsetfillopacity{0.100000}%
\pgfsetlinewidth{0.803000pt}%
\definecolor{currentstroke}{rgb}{1.000000,0.498039,0.054902}%
\pgfsetstrokecolor{currentstroke}%
\pgfsetstrokeopacity{0.100000}%
\pgfsetdash{}{0pt}%
\pgfsys@defobject{currentmarker}{\pgfqpoint{0.522684in}{0.935304in}}{\pgfqpoint{3.465685in}{1.690306in}}{%
\pgfpathmoveto{\pgfqpoint{0.522684in}{0.962030in}}%
\pgfpathlineto{\pgfqpoint{0.522684in}{0.935304in}}%
\pgfpathlineto{\pgfqpoint{0.965650in}{1.249331in}}%
\pgfpathlineto{\pgfqpoint{1.224768in}{1.369597in}}%
\pgfpathlineto{\pgfqpoint{1.994185in}{1.570040in}}%
\pgfpathlineto{\pgfqpoint{3.022720in}{1.690306in}}%
\pgfpathlineto{\pgfqpoint{3.465685in}{1.690306in}}%
\pgfpathlineto{\pgfqpoint{3.465685in}{1.690306in}}%
\pgfpathlineto{\pgfqpoint{3.465685in}{1.690306in}}%
\pgfpathlineto{\pgfqpoint{3.022720in}{1.690306in}}%
\pgfpathlineto{\pgfqpoint{1.994185in}{1.570040in}}%
\pgfpathlineto{\pgfqpoint{1.224768in}{1.369597in}}%
\pgfpathlineto{\pgfqpoint{0.965650in}{1.262694in}}%
\pgfpathlineto{\pgfqpoint{0.522684in}{0.962030in}}%
\pgfpathclose%
\pgfusepath{stroke,fill}%
}%
\begin{pgfscope}%
\pgfsys@transformshift{0.000000in}{0.000000in}%
\pgfsys@useobject{currentmarker}{}%
\end{pgfscope}%
\end{pgfscope}%
\begin{pgfscope}%
\pgfpathrectangle{\pgfqpoint{0.522684in}{0.420833in}}{\pgfqpoint{2.943002in}{2.004431in}}%
\pgfusepath{clip}%
\pgfsetbuttcap%
\pgfsetroundjoin%
\pgfsetlinewidth{1.204500pt}%
\definecolor{currentstroke}{rgb}{0.121569,0.466667,0.705882}%
\pgfsetstrokecolor{currentstroke}%
\pgfsetdash{{4.440000pt}{1.920000pt}}{0.000000pt}%
\pgfpathmoveto{\pgfqpoint{0.522684in}{0.574506in}}%
\pgfpathlineto{\pgfqpoint{0.965650in}{1.028844in}}%
\pgfpathlineto{\pgfqpoint{1.224768in}{1.195880in}}%
\pgfpathlineto{\pgfqpoint{1.994185in}{1.549996in}}%
\pgfpathlineto{\pgfqpoint{3.022720in}{1.690306in}}%
\pgfpathlineto{\pgfqpoint{3.465685in}{1.690306in}}%
\pgfusepath{stroke}%
\end{pgfscope}%
\begin{pgfscope}%
\pgfpathrectangle{\pgfqpoint{0.522684in}{0.420833in}}{\pgfqpoint{2.943002in}{2.004431in}}%
\pgfusepath{clip}%
\pgfsetbuttcap%
\pgfsetroundjoin%
\definecolor{currentfill}{rgb}{0.121569,0.466667,0.705882}%
\pgfsetfillcolor{currentfill}%
\pgfsetlinewidth{0.752812pt}%
\definecolor{currentstroke}{rgb}{0.121569,0.466667,0.705882}%
\pgfsetstrokecolor{currentstroke}%
\pgfsetdash{}{0pt}%
\pgfsys@defobject{currentmarker}{\pgfqpoint{-0.033333in}{-0.033333in}}{\pgfqpoint{0.033333in}{0.033333in}}{%
\pgfpathmoveto{\pgfqpoint{-0.033333in}{-0.033333in}}%
\pgfpathlineto{\pgfqpoint{0.033333in}{0.033333in}}%
\pgfpathmoveto{\pgfqpoint{-0.033333in}{0.033333in}}%
\pgfpathlineto{\pgfqpoint{0.033333in}{-0.033333in}}%
\pgfusepath{stroke,fill}%
}%
\begin{pgfscope}%
\pgfsys@transformshift{0.522684in}{0.574506in}%
\pgfsys@useobject{currentmarker}{}%
\end{pgfscope}%
\begin{pgfscope}%
\pgfsys@transformshift{0.965650in}{1.028844in}%
\pgfsys@useobject{currentmarker}{}%
\end{pgfscope}%
\begin{pgfscope}%
\pgfsys@transformshift{1.224768in}{1.195880in}%
\pgfsys@useobject{currentmarker}{}%
\end{pgfscope}%
\begin{pgfscope}%
\pgfsys@transformshift{1.994185in}{1.549996in}%
\pgfsys@useobject{currentmarker}{}%
\end{pgfscope}%
\begin{pgfscope}%
\pgfsys@transformshift{3.022720in}{1.690306in}%
\pgfsys@useobject{currentmarker}{}%
\end{pgfscope}%
\begin{pgfscope}%
\pgfsys@transformshift{3.465685in}{1.690306in}%
\pgfsys@useobject{currentmarker}{}%
\end{pgfscope}%
\end{pgfscope}%
\begin{pgfscope}%
\pgfpathrectangle{\pgfqpoint{0.522684in}{0.420833in}}{\pgfqpoint{2.943002in}{2.004431in}}%
\pgfusepath{clip}%
\pgfsetroundcap%
\pgfsetroundjoin%
\pgfsetlinewidth{1.204500pt}%
\definecolor{currentstroke}{rgb}{1.000000,0.498039,0.054902}%
\pgfsetstrokecolor{currentstroke}%
\pgfsetdash{}{0pt}%
\pgfpathmoveto{\pgfqpoint{0.522684in}{0.948667in}}%
\pgfpathlineto{\pgfqpoint{0.965650in}{1.256013in}}%
\pgfpathlineto{\pgfqpoint{1.224768in}{1.369597in}}%
\pgfpathlineto{\pgfqpoint{1.994185in}{1.570040in}}%
\pgfpathlineto{\pgfqpoint{3.022720in}{1.690306in}}%
\pgfpathlineto{\pgfqpoint{3.465685in}{1.690306in}}%
\pgfusepath{stroke}%
\end{pgfscope}%
\begin{pgfscope}%
\pgfpathrectangle{\pgfqpoint{0.522684in}{0.420833in}}{\pgfqpoint{2.943002in}{2.004431in}}%
\pgfusepath{clip}%
\pgfsetbuttcap%
\pgfsetroundjoin%
\definecolor{currentfill}{rgb}{1.000000,0.498039,0.054902}%
\pgfsetfillcolor{currentfill}%
\pgfsetlinewidth{0.752812pt}%
\definecolor{currentstroke}{rgb}{1.000000,0.498039,0.054902}%
\pgfsetstrokecolor{currentstroke}%
\pgfsetdash{}{0pt}%
\pgfsys@defobject{currentmarker}{\pgfqpoint{-0.033333in}{-0.033333in}}{\pgfqpoint{0.033333in}{0.033333in}}{%
\pgfpathmoveto{\pgfqpoint{-0.033333in}{-0.033333in}}%
\pgfpathlineto{\pgfqpoint{0.033333in}{0.033333in}}%
\pgfpathmoveto{\pgfqpoint{-0.033333in}{0.033333in}}%
\pgfpathlineto{\pgfqpoint{0.033333in}{-0.033333in}}%
\pgfusepath{stroke,fill}%
}%
\begin{pgfscope}%
\pgfsys@transformshift{0.522684in}{0.948667in}%
\pgfsys@useobject{currentmarker}{}%
\end{pgfscope}%
\begin{pgfscope}%
\pgfsys@transformshift{0.965650in}{1.256013in}%
\pgfsys@useobject{currentmarker}{}%
\end{pgfscope}%
\begin{pgfscope}%
\pgfsys@transformshift{1.224768in}{1.369597in}%
\pgfsys@useobject{currentmarker}{}%
\end{pgfscope}%
\begin{pgfscope}%
\pgfsys@transformshift{1.994185in}{1.570040in}%
\pgfsys@useobject{currentmarker}{}%
\end{pgfscope}%
\begin{pgfscope}%
\pgfsys@transformshift{3.022720in}{1.690306in}%
\pgfsys@useobject{currentmarker}{}%
\end{pgfscope}%
\begin{pgfscope}%
\pgfsys@transformshift{3.465685in}{1.690306in}%
\pgfsys@useobject{currentmarker}{}%
\end{pgfscope}%
\end{pgfscope}%
\begin{pgfscope}%
\pgfpathrectangle{\pgfqpoint{0.522684in}{0.420833in}}{\pgfqpoint{2.943002in}{2.004431in}}%
\pgfusepath{clip}%
\pgfsetbuttcap%
\pgfsetroundjoin%
\pgfsetlinewidth{1.204500pt}%
\definecolor{currentstroke}{rgb}{0.000000,0.000000,0.000000}%
\pgfsetstrokecolor{currentstroke}%
\pgfsetdash{{1.200000pt}{1.980000pt}}{0.000000pt}%
\pgfpathmoveto{\pgfqpoint{0.522684in}{1.803890in}}%
\pgfpathlineto{\pgfqpoint{3.465685in}{1.803890in}}%
\pgfusepath{stroke}%
\end{pgfscope}%
\begin{pgfscope}%
\pgfpathrectangle{\pgfqpoint{0.522684in}{0.420833in}}{\pgfqpoint{2.943002in}{2.004431in}}%
\pgfusepath{clip}%
\pgfsetbuttcap%
\pgfsetroundjoin%
\pgfsetlinewidth{1.204500pt}%
\definecolor{currentstroke}{rgb}{0.000000,0.000000,0.000000}%
\pgfsetstrokecolor{currentstroke}%
\pgfsetdash{{4.440000pt}{1.920000pt}}{0.000000pt}%
\pgfpathmoveto{\pgfqpoint{0.522684in}{1.750439in}}%
\pgfpathlineto{\pgfqpoint{3.465685in}{1.750439in}}%
\pgfusepath{stroke}%
\end{pgfscope}%
\begin{pgfscope}%
\pgfsetrectcap%
\pgfsetmiterjoin%
\pgfsetlinewidth{1.003750pt}%
\definecolor{currentstroke}{rgb}{0.800000,0.800000,0.800000}%
\pgfsetstrokecolor{currentstroke}%
\pgfsetdash{}{0pt}%
\pgfpathmoveto{\pgfqpoint{0.522684in}{0.420833in}}%
\pgfpathlineto{\pgfqpoint{0.522684in}{2.425264in}}%
\pgfusepath{stroke}%
\end{pgfscope}%
\begin{pgfscope}%
\pgfsetrectcap%
\pgfsetmiterjoin%
\pgfsetlinewidth{1.003750pt}%
\definecolor{currentstroke}{rgb}{0.800000,0.800000,0.800000}%
\pgfsetstrokecolor{currentstroke}%
\pgfsetdash{}{0pt}%
\pgfpathmoveto{\pgfqpoint{0.522684in}{0.420833in}}%
\pgfpathlineto{\pgfqpoint{3.465685in}{0.420833in}}%
\pgfusepath{stroke}%
\end{pgfscope}%
\begin{pgfscope}%
\pgfsetroundcap%
\pgfsetroundjoin%
\definecolor{currentfill}{rgb}{0.862745,0.862745,0.862745}%
\pgfsetfillcolor{currentfill}%
\pgfsetlinewidth{0.803000pt}%
\definecolor{currentstroke}{rgb}{1.000000,1.000000,1.000000}%
\pgfsetstrokecolor{currentstroke}%
\pgfsetdash{}{0pt}%
\pgfpathmoveto{\pgfqpoint{1.110250in}{1.830690in}}%
\pgfpathquadraticcurveto{\pgfqpoint{0.989141in}{2.003023in}}{\pgfqpoint{0.868033in}{2.175357in}}%
\pgfpathlineto{\pgfqpoint{0.828261in}{2.147406in}}%
\pgfpathquadraticcurveto{\pgfqpoint{0.814430in}{2.239554in}}{\pgfqpoint{0.800599in}{2.331702in}}%
\pgfpathquadraticcurveto{\pgfqpoint{0.882611in}{2.287469in}}{\pgfqpoint{0.964623in}{2.243236in}}%
\pgfpathlineto{\pgfqpoint{0.924850in}{2.215285in}}%
\pgfpathquadraticcurveto{\pgfqpoint{1.045959in}{2.042952in}}{\pgfqpoint{1.167067in}{1.870618in}}%
\pgfpathlineto{\pgfqpoint{1.110250in}{1.830690in}}%
\pgfpathclose%
\pgfusepath{stroke,fill}%
\end{pgfscope}%
\begin{pgfscope}%
\definecolor{textcolor}{rgb}{0.862745,0.862745,0.862745}%
\pgfsetstrokecolor{textcolor}%
\pgfsetfillcolor{textcolor}%
\pgftext[x=1.157436in,y=2.158006in,left,]{\color{textcolor}\rmfamily\fontsize{12.000000}{14.400000}\selectfont better}%
\end{pgfscope}%
\begin{pgfscope}%
\pgfsetroundcap%
\pgfsetroundjoin%
\definecolor{currentfill}{rgb}{0.000000,0.501961,0.000000}%
\pgfsetfillcolor{currentfill}%
\pgfsetlinewidth{0.803000pt}%
\definecolor{currentstroke}{rgb}{1.000000,1.000000,1.000000}%
\pgfsetstrokecolor{currentstroke}%
\pgfsetdash{}{0pt}%
\pgfpathmoveto{\pgfqpoint{1.821213in}{1.442567in}}%
\pgfpathquadraticcurveto{\pgfqpoint{1.776641in}{1.442567in}}{\pgfqpoint{1.732068in}{1.442567in}}%
\pgfpathlineto{\pgfqpoint{1.732068in}{1.387011in}}%
\pgfpathquadraticcurveto{\pgfqpoint{1.648746in}{1.421734in}}{\pgfqpoint{1.565424in}{1.456456in}}%
\pgfpathquadraticcurveto{\pgfqpoint{1.648746in}{1.491178in}}{\pgfqpoint{1.732068in}{1.525900in}}%
\pgfpathlineto{\pgfqpoint{1.732068in}{1.470345in}}%
\pgfpathquadraticcurveto{\pgfqpoint{1.776641in}{1.470345in}}{\pgfqpoint{1.821213in}{1.470345in}}%
\pgfpathlineto{\pgfqpoint{1.821213in}{1.442567in}}%
\pgfpathclose%
\pgfusepath{stroke,fill}%
\end{pgfscope}%
\begin{pgfscope}%
\definecolor{textcolor}{rgb}{0.000000,0.501961,0.000000}%
\pgfsetstrokecolor{textcolor}%
\pgfsetfillcolor{textcolor}%
\pgftext[x=2.253303in,y=1.456456in,left,]{\color{textcolor}\rmfamily\fontsize{12.000000}{14.400000}\selectfont 30\% \(\displaystyle \varepsilon\) saving}%
\end{pgfscope}%
\begin{pgfscope}%
\pgfsetbuttcap%
\pgfsetmiterjoin%
\definecolor{currentfill}{rgb}{1.000000,1.000000,1.000000}%
\pgfsetfillcolor{currentfill}%
\pgfsetfillopacity{0.800000}%
\pgfsetlinewidth{0.803000pt}%
\definecolor{currentstroke}{rgb}{0.800000,0.800000,0.800000}%
\pgfsetstrokecolor{currentstroke}%
\pgfsetstrokeopacity{0.800000}%
\pgfsetdash{}{0pt}%
\pgfpathmoveto{\pgfqpoint{1.002082in}{0.481944in}}%
\pgfpathlineto{\pgfqpoint{3.380130in}{0.481944in}}%
\pgfpathquadraticcurveto{\pgfqpoint{3.404574in}{0.481944in}}{\pgfqpoint{3.404574in}{0.506389in}}%
\pgfpathlineto{\pgfqpoint{3.404574in}{1.210833in}}%
\pgfpathquadraticcurveto{\pgfqpoint{3.404574in}{1.235278in}}{\pgfqpoint{3.380130in}{1.235278in}}%
\pgfpathlineto{\pgfqpoint{1.002082in}{1.235278in}}%
\pgfpathquadraticcurveto{\pgfqpoint{0.977637in}{1.235278in}}{\pgfqpoint{0.977637in}{1.210833in}}%
\pgfpathlineto{\pgfqpoint{0.977637in}{0.506389in}}%
\pgfpathquadraticcurveto{\pgfqpoint{0.977637in}{0.481944in}}{\pgfqpoint{1.002082in}{0.481944in}}%
\pgfpathclose%
\pgfusepath{stroke,fill}%
\end{pgfscope}%
\begin{pgfscope}%
\pgfsetbuttcap%
\pgfsetroundjoin%
\pgfsetlinewidth{1.204500pt}%
\definecolor{currentstroke}{rgb}{0.121569,0.466667,0.705882}%
\pgfsetstrokecolor{currentstroke}%
\pgfsetdash{{4.440000pt}{1.920000pt}}{0.000000pt}%
\pgfpathmoveto{\pgfqpoint{1.026526in}{1.142361in}}%
\pgfpathlineto{\pgfqpoint{1.270971in}{1.142361in}}%
\pgfusepath{stroke}%
\end{pgfscope}%
\begin{pgfscope}%
\pgfsetbuttcap%
\pgfsetroundjoin%
\definecolor{currentfill}{rgb}{0.121569,0.466667,0.705882}%
\pgfsetfillcolor{currentfill}%
\pgfsetlinewidth{0.752812pt}%
\definecolor{currentstroke}{rgb}{0.121569,0.466667,0.705882}%
\pgfsetstrokecolor{currentstroke}%
\pgfsetdash{}{0pt}%
\pgfsys@defobject{currentmarker}{\pgfqpoint{-0.033333in}{-0.033333in}}{\pgfqpoint{0.033333in}{0.033333in}}{%
\pgfpathmoveto{\pgfqpoint{-0.033333in}{-0.033333in}}%
\pgfpathlineto{\pgfqpoint{0.033333in}{0.033333in}}%
\pgfpathmoveto{\pgfqpoint{-0.033333in}{0.033333in}}%
\pgfpathlineto{\pgfqpoint{0.033333in}{-0.033333in}}%
\pgfusepath{stroke,fill}%
}%
\begin{pgfscope}%
\pgfsys@transformshift{1.148748in}{1.142361in}%
\pgfsys@useobject{currentmarker}{}%
\end{pgfscope}%
\end{pgfscope}%
\begin{pgfscope}%
\definecolor{textcolor}{rgb}{0.150000,0.150000,0.150000}%
\pgfsetstrokecolor{textcolor}%
\pgfsetfillcolor{textcolor}%
\pgftext[x=1.368748in,y=1.099583in,left,base]{\color{textcolor}\rmfamily\fontsize{8.800000}{10.560000}\selectfont Maddock et al.}%
\end{pgfscope}%
\begin{pgfscope}%
\pgfsetroundcap%
\pgfsetroundjoin%
\pgfsetlinewidth{1.204500pt}%
\definecolor{currentstroke}{rgb}{1.000000,0.498039,0.054902}%
\pgfsetstrokecolor{currentstroke}%
\pgfsetdash{}{0pt}%
\pgfpathmoveto{\pgfqpoint{1.026526in}{0.970139in}}%
\pgfpathlineto{\pgfqpoint{1.270971in}{0.970139in}}%
\pgfusepath{stroke}%
\end{pgfscope}%
\begin{pgfscope}%
\pgfsetbuttcap%
\pgfsetroundjoin%
\definecolor{currentfill}{rgb}{1.000000,0.498039,0.054902}%
\pgfsetfillcolor{currentfill}%
\pgfsetlinewidth{0.752812pt}%
\definecolor{currentstroke}{rgb}{1.000000,0.498039,0.054902}%
\pgfsetstrokecolor{currentstroke}%
\pgfsetdash{}{0pt}%
\pgfsys@defobject{currentmarker}{\pgfqpoint{-0.033333in}{-0.033333in}}{\pgfqpoint{0.033333in}{0.033333in}}{%
\pgfpathmoveto{\pgfqpoint{-0.033333in}{-0.033333in}}%
\pgfpathlineto{\pgfqpoint{0.033333in}{0.033333in}}%
\pgfpathmoveto{\pgfqpoint{-0.033333in}{0.033333in}}%
\pgfpathlineto{\pgfqpoint{0.033333in}{-0.033333in}}%
\pgfusepath{stroke,fill}%
}%
\begin{pgfscope}%
\pgfsys@transformshift{1.148748in}{0.970139in}%
\pgfsys@useobject{currentmarker}{}%
\end{pgfscope}%
\end{pgfscope}%
\begin{pgfscope}%
\definecolor{textcolor}{rgb}{0.150000,0.150000,0.150000}%
\pgfsetstrokecolor{textcolor}%
\pgfsetfillcolor{textcolor}%
\pgftext[x=1.368748in,y=0.927361in,left,base]{\color{textcolor}\rmfamily\fontsize{8.800000}{10.560000}\selectfont S-BDT}%
\end{pgfscope}%
\begin{pgfscope}%
\pgfsetbuttcap%
\pgfsetroundjoin%
\pgfsetlinewidth{1.204500pt}%
\definecolor{currentstroke}{rgb}{0.000000,0.000000,0.000000}%
\pgfsetstrokecolor{currentstroke}%
\pgfsetdash{{1.200000pt}{1.980000pt}}{0.000000pt}%
\pgfpathmoveto{\pgfqpoint{1.026526in}{0.790972in}}%
\pgfpathlineto{\pgfqpoint{1.270971in}{0.790972in}}%
\pgfusepath{stroke}%
\end{pgfscope}%
\begin{pgfscope}%
\definecolor{textcolor}{rgb}{0.150000,0.150000,0.150000}%
\pgfsetstrokecolor{textcolor}%
\pgfsetfillcolor{textcolor}%
\pgftext[x=1.368748in,y=0.748194in,left,base]{\color{textcolor}\rmfamily\fontsize{8.800000}{10.560000}\selectfont xgboost (nonprivate)}%
\end{pgfscope}%
\begin{pgfscope}%
\pgfsetbuttcap%
\pgfsetroundjoin%
\pgfsetlinewidth{1.204500pt}%
\definecolor{currentstroke}{rgb}{0.000000,0.000000,0.000000}%
\pgfsetstrokecolor{currentstroke}%
\pgfsetdash{{4.440000pt}{1.920000pt}}{0.000000pt}%
\pgfpathmoveto{\pgfqpoint{1.026526in}{0.604861in}}%
\pgfpathlineto{\pgfqpoint{1.270971in}{0.604861in}}%
\pgfusepath{stroke}%
\end{pgfscope}%
\begin{pgfscope}%
\definecolor{textcolor}{rgb}{0.150000,0.150000,0.150000}%
\pgfsetstrokecolor{textcolor}%
\pgfsetfillcolor{textcolor}%
\pgftext[x=1.368748in,y=0.562083in,left,base]{\color{textcolor}\rmfamily\fontsize{8.800000}{10.560000}\selectfont xgboost random splits (nonprivate)}%
\end{pgfscope}%
\end{pgfpicture}%
\makeatother%
\endgroup%

%% file: images/spambase_main_results.pgf
\begingroup%
\makeatletter%
\begin{pgfpicture}%
\pgfpathrectangle{\pgfpointorigin}{\pgfqpoint{3.612000in}{2.468667in}}%
\pgfusepath{use as bounding box, clip}%
\begin{pgfscope}%
\pgfsetbuttcap%
\pgfsetmiterjoin%
\definecolor{currentfill}{rgb}{1.000000,1.000000,1.000000}%
\pgfsetfillcolor{currentfill}%
\pgfsetlinewidth{0.000000pt}%
\definecolor{currentstroke}{rgb}{1.000000,1.000000,1.000000}%
\pgfsetstrokecolor{currentstroke}%
\pgfsetdash{}{0pt}%
\pgfpathmoveto{\pgfqpoint{0.000000in}{0.000000in}}%
\pgfpathlineto{\pgfqpoint{3.612000in}{0.000000in}}%
\pgfpathlineto{\pgfqpoint{3.612000in}{2.468667in}}%
\pgfpathlineto{\pgfqpoint{0.000000in}{2.468667in}}%
\pgfpathclose%
\pgfusepath{fill}%
\end{pgfscope}%
\begin{pgfscope}%
\pgfsetbuttcap%
\pgfsetmiterjoin%
\definecolor{currentfill}{rgb}{1.000000,1.000000,1.000000}%
\pgfsetfillcolor{currentfill}%
\pgfsetlinewidth{0.000000pt}%
\definecolor{currentstroke}{rgb}{0.000000,0.000000,0.000000}%
\pgfsetstrokecolor{currentstroke}%
\pgfsetstrokeopacity{0.000000}%
\pgfsetdash{}{0pt}%
\pgfpathmoveto{\pgfqpoint{0.522684in}{0.420833in}}%
\pgfpathlineto{\pgfqpoint{3.465685in}{0.420833in}}%
\pgfpathlineto{\pgfqpoint{3.465685in}{2.425264in}}%
\pgfpathlineto{\pgfqpoint{0.522684in}{2.425264in}}%
\pgfpathclose%
\pgfusepath{fill}%
\end{pgfscope}%
\begin{pgfscope}%
\pgfpathrectangle{\pgfqpoint{0.522684in}{0.420833in}}{\pgfqpoint{2.943002in}{2.004431in}}%
\pgfusepath{clip}%
\pgfsetroundcap%
\pgfsetroundjoin%
\pgfsetlinewidth{0.803000pt}%
\definecolor{currentstroke}{rgb}{0.800000,0.800000,0.800000}%
\pgfsetstrokecolor{currentstroke}%
\pgfsetdash{}{0pt}%
\pgfpathmoveto{\pgfqpoint{0.522684in}{0.420833in}}%
\pgfpathlineto{\pgfqpoint{0.522684in}{2.425264in}}%
\pgfusepath{stroke}%
\end{pgfscope}%
\begin{pgfscope}%
\definecolor{textcolor}{rgb}{0.150000,0.150000,0.150000}%
\pgfsetstrokecolor{textcolor}%
\pgfsetfillcolor{textcolor}%
\pgftext[x=0.522684in,y=0.305556in,,top]{\color{textcolor}\rmfamily\fontsize{8.800000}{10.560000}\selectfont \(\displaystyle {0.010}\)}%
\end{pgfscope}%
\begin{pgfscope}%
\pgfpathrectangle{\pgfqpoint{0.522684in}{0.420833in}}{\pgfqpoint{2.943002in}{2.004431in}}%
\pgfusepath{clip}%
\pgfsetroundcap%
\pgfsetroundjoin%
\pgfsetlinewidth{0.803000pt}%
\definecolor{currentstroke}{rgb}{0.800000,0.800000,0.800000}%
\pgfsetstrokecolor{currentstroke}%
\pgfsetdash{}{0pt}%
\pgfpathmoveto{\pgfqpoint{1.108253in}{0.420833in}}%
\pgfpathlineto{\pgfqpoint{1.108253in}{2.425264in}}%
\pgfusepath{stroke}%
\end{pgfscope}%
\begin{pgfscope}%
\definecolor{textcolor}{rgb}{0.150000,0.150000,0.150000}%
\pgfsetstrokecolor{textcolor}%
\pgfsetfillcolor{textcolor}%
\pgftext[x=1.108253in,y=0.305556in,,top]{\color{textcolor}\rmfamily\fontsize{8.800000}{10.560000}\selectfont \(\displaystyle {0.025}\)}%
\end{pgfscope}%
\begin{pgfscope}%
\pgfpathrectangle{\pgfqpoint{0.522684in}{0.420833in}}{\pgfqpoint{2.943002in}{2.004431in}}%
\pgfusepath{clip}%
\pgfsetroundcap%
\pgfsetroundjoin%
\pgfsetlinewidth{0.803000pt}%
\definecolor{currentstroke}{rgb}{0.800000,0.800000,0.800000}%
\pgfsetstrokecolor{currentstroke}%
\pgfsetdash{}{0pt}%
\pgfpathmoveto{\pgfqpoint{1.551219in}{0.420833in}}%
\pgfpathlineto{\pgfqpoint{1.551219in}{2.425264in}}%
\pgfusepath{stroke}%
\end{pgfscope}%
\begin{pgfscope}%
\definecolor{textcolor}{rgb}{0.150000,0.150000,0.150000}%
\pgfsetstrokecolor{textcolor}%
\pgfsetfillcolor{textcolor}%
\pgftext[x=1.551219in,y=0.305556in,,top]{\color{textcolor}\rmfamily\fontsize{8.800000}{10.560000}\selectfont \(\displaystyle {0.050}\)}%
\end{pgfscope}%
\begin{pgfscope}%
\pgfpathrectangle{\pgfqpoint{0.522684in}{0.420833in}}{\pgfqpoint{2.943002in}{2.004431in}}%
\pgfusepath{clip}%
\pgfsetroundcap%
\pgfsetroundjoin%
\pgfsetlinewidth{0.803000pt}%
\definecolor{currentstroke}{rgb}{0.800000,0.800000,0.800000}%
\pgfsetstrokecolor{currentstroke}%
\pgfsetdash{}{0pt}%
\pgfpathmoveto{\pgfqpoint{1.994185in}{0.420833in}}%
\pgfpathlineto{\pgfqpoint{1.994185in}{2.425264in}}%
\pgfusepath{stroke}%
\end{pgfscope}%
\begin{pgfscope}%
\definecolor{textcolor}{rgb}{0.150000,0.150000,0.150000}%
\pgfsetstrokecolor{textcolor}%
\pgfsetfillcolor{textcolor}%
\pgftext[x=1.994185in,y=0.305556in,,top]{\color{textcolor}\rmfamily\fontsize{8.800000}{10.560000}\selectfont \(\displaystyle {0.100}\)}%
\end{pgfscope}%
\begin{pgfscope}%
\pgfpathrectangle{\pgfqpoint{0.522684in}{0.420833in}}{\pgfqpoint{2.943002in}{2.004431in}}%
\pgfusepath{clip}%
\pgfsetroundcap%
\pgfsetroundjoin%
\pgfsetlinewidth{0.803000pt}%
\definecolor{currentstroke}{rgb}{0.800000,0.800000,0.800000}%
\pgfsetstrokecolor{currentstroke}%
\pgfsetdash{}{0pt}%
\pgfpathmoveto{\pgfqpoint{3.022720in}{0.420833in}}%
\pgfpathlineto{\pgfqpoint{3.022720in}{2.425264in}}%
\pgfusepath{stroke}%
\end{pgfscope}%
\begin{pgfscope}%
\definecolor{textcolor}{rgb}{0.150000,0.150000,0.150000}%
\pgfsetstrokecolor{textcolor}%
\pgfsetfillcolor{textcolor}%
\pgftext[x=3.022720in,y=0.305556in,,top]{\color{textcolor}\rmfamily\fontsize{8.800000}{10.560000}\selectfont \(\displaystyle {0.500}\)}%
\end{pgfscope}%
\begin{pgfscope}%
\pgfpathrectangle{\pgfqpoint{0.522684in}{0.420833in}}{\pgfqpoint{2.943002in}{2.004431in}}%
\pgfusepath{clip}%
\pgfsetroundcap%
\pgfsetroundjoin%
\pgfsetlinewidth{0.803000pt}%
\definecolor{currentstroke}{rgb}{0.800000,0.800000,0.800000}%
\pgfsetstrokecolor{currentstroke}%
\pgfsetdash{}{0pt}%
\pgfpathmoveto{\pgfqpoint{3.465685in}{0.420833in}}%
\pgfpathlineto{\pgfqpoint{3.465685in}{2.425264in}}%
\pgfusepath{stroke}%
\end{pgfscope}%
\begin{pgfscope}%
\definecolor{textcolor}{rgb}{0.150000,0.150000,0.150000}%
\pgfsetstrokecolor{textcolor}%
\pgfsetfillcolor{textcolor}%
\pgftext[x=3.465685in,y=0.305556in,,top]{\color{textcolor}\rmfamily\fontsize{8.800000}{10.560000}\selectfont \(\displaystyle {1.000}\)}%
\end{pgfscope}%
\begin{pgfscope}%
\pgfpathrectangle{\pgfqpoint{0.522684in}{0.420833in}}{\pgfqpoint{2.943002in}{2.004431in}}%
\pgfusepath{clip}%
\pgfsetroundcap%
\pgfsetroundjoin%
\pgfsetlinewidth{0.075281pt}%
\definecolor{currentstroke}{rgb}{0.827451,0.827451,0.827451}%
\pgfsetstrokecolor{currentstroke}%
\pgfsetdash{}{0pt}%
\pgfpathmoveto{\pgfqpoint{0.965650in}{0.420833in}}%
\pgfpathlineto{\pgfqpoint{0.965650in}{2.425264in}}%
\pgfusepath{stroke}%
\end{pgfscope}%
\begin{pgfscope}%
\pgfpathrectangle{\pgfqpoint{0.522684in}{0.420833in}}{\pgfqpoint{2.943002in}{2.004431in}}%
\pgfusepath{clip}%
\pgfsetroundcap%
\pgfsetroundjoin%
\pgfsetlinewidth{0.075281pt}%
\definecolor{currentstroke}{rgb}{0.827451,0.827451,0.827451}%
\pgfsetstrokecolor{currentstroke}%
\pgfsetdash{}{0pt}%
\pgfpathmoveto{\pgfqpoint{1.224768in}{0.420833in}}%
\pgfpathlineto{\pgfqpoint{1.224768in}{2.425264in}}%
\pgfusepath{stroke}%
\end{pgfscope}%
\begin{pgfscope}%
\pgfpathrectangle{\pgfqpoint{0.522684in}{0.420833in}}{\pgfqpoint{2.943002in}{2.004431in}}%
\pgfusepath{clip}%
\pgfsetroundcap%
\pgfsetroundjoin%
\pgfsetlinewidth{0.075281pt}%
\definecolor{currentstroke}{rgb}{0.827451,0.827451,0.827451}%
\pgfsetstrokecolor{currentstroke}%
\pgfsetdash{}{0pt}%
\pgfpathmoveto{\pgfqpoint{1.408615in}{0.420833in}}%
\pgfpathlineto{\pgfqpoint{1.408615in}{2.425264in}}%
\pgfusepath{stroke}%
\end{pgfscope}%
\begin{pgfscope}%
\pgfpathrectangle{\pgfqpoint{0.522684in}{0.420833in}}{\pgfqpoint{2.943002in}{2.004431in}}%
\pgfusepath{clip}%
\pgfsetroundcap%
\pgfsetroundjoin%
\pgfsetlinewidth{0.075281pt}%
\definecolor{currentstroke}{rgb}{0.827451,0.827451,0.827451}%
\pgfsetstrokecolor{currentstroke}%
\pgfsetdash{}{0pt}%
\pgfpathmoveto{\pgfqpoint{1.667734in}{0.420833in}}%
\pgfpathlineto{\pgfqpoint{1.667734in}{2.425264in}}%
\pgfusepath{stroke}%
\end{pgfscope}%
\begin{pgfscope}%
\pgfpathrectangle{\pgfqpoint{0.522684in}{0.420833in}}{\pgfqpoint{2.943002in}{2.004431in}}%
\pgfusepath{clip}%
\pgfsetroundcap%
\pgfsetroundjoin%
\pgfsetlinewidth{0.075281pt}%
\definecolor{currentstroke}{rgb}{0.827451,0.827451,0.827451}%
\pgfsetstrokecolor{currentstroke}%
\pgfsetdash{}{0pt}%
\pgfpathmoveto{\pgfqpoint{1.766246in}{0.420833in}}%
\pgfpathlineto{\pgfqpoint{1.766246in}{2.425264in}}%
\pgfusepath{stroke}%
\end{pgfscope}%
\begin{pgfscope}%
\pgfpathrectangle{\pgfqpoint{0.522684in}{0.420833in}}{\pgfqpoint{2.943002in}{2.004431in}}%
\pgfusepath{clip}%
\pgfsetroundcap%
\pgfsetroundjoin%
\pgfsetlinewidth{0.075281pt}%
\definecolor{currentstroke}{rgb}{0.827451,0.827451,0.827451}%
\pgfsetstrokecolor{currentstroke}%
\pgfsetdash{}{0pt}%
\pgfpathmoveto{\pgfqpoint{1.851581in}{0.420833in}}%
\pgfpathlineto{\pgfqpoint{1.851581in}{2.425264in}}%
\pgfusepath{stroke}%
\end{pgfscope}%
\begin{pgfscope}%
\pgfpathrectangle{\pgfqpoint{0.522684in}{0.420833in}}{\pgfqpoint{2.943002in}{2.004431in}}%
\pgfusepath{clip}%
\pgfsetroundcap%
\pgfsetroundjoin%
\pgfsetlinewidth{0.075281pt}%
\definecolor{currentstroke}{rgb}{0.827451,0.827451,0.827451}%
\pgfsetstrokecolor{currentstroke}%
\pgfsetdash{}{0pt}%
\pgfpathmoveto{\pgfqpoint{1.926852in}{0.420833in}}%
\pgfpathlineto{\pgfqpoint{1.926852in}{2.425264in}}%
\pgfusepath{stroke}%
\end{pgfscope}%
\begin{pgfscope}%
\pgfpathrectangle{\pgfqpoint{0.522684in}{0.420833in}}{\pgfqpoint{2.943002in}{2.004431in}}%
\pgfusepath{clip}%
\pgfsetroundcap%
\pgfsetroundjoin%
\pgfsetlinewidth{0.075281pt}%
\definecolor{currentstroke}{rgb}{0.827451,0.827451,0.827451}%
\pgfsetstrokecolor{currentstroke}%
\pgfsetdash{}{0pt}%
\pgfpathmoveto{\pgfqpoint{2.437150in}{0.420833in}}%
\pgfpathlineto{\pgfqpoint{2.437150in}{2.425264in}}%
\pgfusepath{stroke}%
\end{pgfscope}%
\begin{pgfscope}%
\pgfpathrectangle{\pgfqpoint{0.522684in}{0.420833in}}{\pgfqpoint{2.943002in}{2.004431in}}%
\pgfusepath{clip}%
\pgfsetroundcap%
\pgfsetroundjoin%
\pgfsetlinewidth{0.075281pt}%
\definecolor{currentstroke}{rgb}{0.827451,0.827451,0.827451}%
\pgfsetstrokecolor{currentstroke}%
\pgfsetdash{}{0pt}%
\pgfpathmoveto{\pgfqpoint{2.696269in}{0.420833in}}%
\pgfpathlineto{\pgfqpoint{2.696269in}{2.425264in}}%
\pgfusepath{stroke}%
\end{pgfscope}%
\begin{pgfscope}%
\pgfpathrectangle{\pgfqpoint{0.522684in}{0.420833in}}{\pgfqpoint{2.943002in}{2.004431in}}%
\pgfusepath{clip}%
\pgfsetroundcap%
\pgfsetroundjoin%
\pgfsetlinewidth{0.075281pt}%
\definecolor{currentstroke}{rgb}{0.827451,0.827451,0.827451}%
\pgfsetstrokecolor{currentstroke}%
\pgfsetdash{}{0pt}%
\pgfpathmoveto{\pgfqpoint{2.880116in}{0.420833in}}%
\pgfpathlineto{\pgfqpoint{2.880116in}{2.425264in}}%
\pgfusepath{stroke}%
\end{pgfscope}%
\begin{pgfscope}%
\pgfpathrectangle{\pgfqpoint{0.522684in}{0.420833in}}{\pgfqpoint{2.943002in}{2.004431in}}%
\pgfusepath{clip}%
\pgfsetroundcap%
\pgfsetroundjoin%
\pgfsetlinewidth{0.075281pt}%
\definecolor{currentstroke}{rgb}{0.827451,0.827451,0.827451}%
\pgfsetstrokecolor{currentstroke}%
\pgfsetdash{}{0pt}%
\pgfpathmoveto{\pgfqpoint{3.139235in}{0.420833in}}%
\pgfpathlineto{\pgfqpoint{3.139235in}{2.425264in}}%
\pgfusepath{stroke}%
\end{pgfscope}%
\begin{pgfscope}%
\pgfpathrectangle{\pgfqpoint{0.522684in}{0.420833in}}{\pgfqpoint{2.943002in}{2.004431in}}%
\pgfusepath{clip}%
\pgfsetroundcap%
\pgfsetroundjoin%
\pgfsetlinewidth{0.075281pt}%
\definecolor{currentstroke}{rgb}{0.827451,0.827451,0.827451}%
\pgfsetstrokecolor{currentstroke}%
\pgfsetdash{}{0pt}%
\pgfpathmoveto{\pgfqpoint{3.237747in}{0.420833in}}%
\pgfpathlineto{\pgfqpoint{3.237747in}{2.425264in}}%
\pgfusepath{stroke}%
\end{pgfscope}%
\begin{pgfscope}%
\pgfpathrectangle{\pgfqpoint{0.522684in}{0.420833in}}{\pgfqpoint{2.943002in}{2.004431in}}%
\pgfusepath{clip}%
\pgfsetroundcap%
\pgfsetroundjoin%
\pgfsetlinewidth{0.075281pt}%
\definecolor{currentstroke}{rgb}{0.827451,0.827451,0.827451}%
\pgfsetstrokecolor{currentstroke}%
\pgfsetdash{}{0pt}%
\pgfpathmoveto{\pgfqpoint{3.323082in}{0.420833in}}%
\pgfpathlineto{\pgfqpoint{3.323082in}{2.425264in}}%
\pgfusepath{stroke}%
\end{pgfscope}%
\begin{pgfscope}%
\pgfpathrectangle{\pgfqpoint{0.522684in}{0.420833in}}{\pgfqpoint{2.943002in}{2.004431in}}%
\pgfusepath{clip}%
\pgfsetroundcap%
\pgfsetroundjoin%
\pgfsetlinewidth{0.075281pt}%
\definecolor{currentstroke}{rgb}{0.827451,0.827451,0.827451}%
\pgfsetstrokecolor{currentstroke}%
\pgfsetdash{}{0pt}%
\pgfpathmoveto{\pgfqpoint{3.398353in}{0.420833in}}%
\pgfpathlineto{\pgfqpoint{3.398353in}{2.425264in}}%
\pgfusepath{stroke}%
\end{pgfscope}%
\begin{pgfscope}%
\definecolor{textcolor}{rgb}{0.150000,0.150000,0.150000}%
\pgfsetstrokecolor{textcolor}%
\pgfsetfillcolor{textcolor}%
\pgftext[x=1.994185in,y=0.138889in,,top]{\color{textcolor}\rmfamily\fontsize{9.600000}{11.520000}\selectfont \(\displaystyle \varepsilon\) (privacy budget)}%
\end{pgfscope}%
\begin{pgfscope}%
\pgfpathrectangle{\pgfqpoint{0.522684in}{0.420833in}}{\pgfqpoint{2.943002in}{2.004431in}}%
\pgfusepath{clip}%
\pgfsetroundcap%
\pgfsetroundjoin%
\pgfsetlinewidth{0.803000pt}%
\definecolor{currentstroke}{rgb}{0.800000,0.800000,0.800000}%
\pgfsetstrokecolor{currentstroke}%
\pgfsetdash{}{0pt}%
\pgfpathmoveto{\pgfqpoint{0.522684in}{0.420833in}}%
\pgfpathlineto{\pgfqpoint{3.465685in}{0.420833in}}%
\pgfusepath{stroke}%
\end{pgfscope}%
\begin{pgfscope}%
\definecolor{textcolor}{rgb}{0.150000,0.150000,0.150000}%
\pgfsetstrokecolor{textcolor}%
\pgfsetfillcolor{textcolor}%
\pgftext[x=0.179012in, y=0.377431in, left, base]{\color{textcolor}\rmfamily\fontsize{8.800000}{10.560000}\selectfont \(\displaystyle {0.60}\)}%
\end{pgfscope}%
\begin{pgfscope}%
\pgfpathrectangle{\pgfqpoint{0.522684in}{0.420833in}}{\pgfqpoint{2.943002in}{2.004431in}}%
\pgfusepath{clip}%
\pgfsetroundcap%
\pgfsetroundjoin%
\pgfsetlinewidth{0.803000pt}%
\definecolor{currentstroke}{rgb}{0.800000,0.800000,0.800000}%
\pgfsetstrokecolor{currentstroke}%
\pgfsetdash{}{0pt}%
\pgfpathmoveto{\pgfqpoint{0.522684in}{0.671387in}}%
\pgfpathlineto{\pgfqpoint{3.465685in}{0.671387in}}%
\pgfusepath{stroke}%
\end{pgfscope}%
\begin{pgfscope}%
\definecolor{textcolor}{rgb}{0.150000,0.150000,0.150000}%
\pgfsetstrokecolor{textcolor}%
\pgfsetfillcolor{textcolor}%
\pgftext[x=0.179012in, y=0.627984in, left, base]{\color{textcolor}\rmfamily\fontsize{8.800000}{10.560000}\selectfont \(\displaystyle {0.65}\)}%
\end{pgfscope}%
\begin{pgfscope}%
\pgfpathrectangle{\pgfqpoint{0.522684in}{0.420833in}}{\pgfqpoint{2.943002in}{2.004431in}}%
\pgfusepath{clip}%
\pgfsetroundcap%
\pgfsetroundjoin%
\pgfsetlinewidth{0.803000pt}%
\definecolor{currentstroke}{rgb}{0.800000,0.800000,0.800000}%
\pgfsetstrokecolor{currentstroke}%
\pgfsetdash{}{0pt}%
\pgfpathmoveto{\pgfqpoint{0.522684in}{0.921941in}}%
\pgfpathlineto{\pgfqpoint{3.465685in}{0.921941in}}%
\pgfusepath{stroke}%
\end{pgfscope}%
\begin{pgfscope}%
\definecolor{textcolor}{rgb}{0.150000,0.150000,0.150000}%
\pgfsetstrokecolor{textcolor}%
\pgfsetfillcolor{textcolor}%
\pgftext[x=0.179012in, y=0.878538in, left, base]{\color{textcolor}\rmfamily\fontsize{8.800000}{10.560000}\selectfont \(\displaystyle {0.70}\)}%
\end{pgfscope}%
\begin{pgfscope}%
\pgfpathrectangle{\pgfqpoint{0.522684in}{0.420833in}}{\pgfqpoint{2.943002in}{2.004431in}}%
\pgfusepath{clip}%
\pgfsetroundcap%
\pgfsetroundjoin%
\pgfsetlinewidth{0.803000pt}%
\definecolor{currentstroke}{rgb}{0.800000,0.800000,0.800000}%
\pgfsetstrokecolor{currentstroke}%
\pgfsetdash{}{0pt}%
\pgfpathmoveto{\pgfqpoint{0.522684in}{1.172495in}}%
\pgfpathlineto{\pgfqpoint{3.465685in}{1.172495in}}%
\pgfusepath{stroke}%
\end{pgfscope}%
\begin{pgfscope}%
\definecolor{textcolor}{rgb}{0.150000,0.150000,0.150000}%
\pgfsetstrokecolor{textcolor}%
\pgfsetfillcolor{textcolor}%
\pgftext[x=0.179012in, y=1.129092in, left, base]{\color{textcolor}\rmfamily\fontsize{8.800000}{10.560000}\selectfont \(\displaystyle {0.75}\)}%
\end{pgfscope}%
\begin{pgfscope}%
\pgfpathrectangle{\pgfqpoint{0.522684in}{0.420833in}}{\pgfqpoint{2.943002in}{2.004431in}}%
\pgfusepath{clip}%
\pgfsetroundcap%
\pgfsetroundjoin%
\pgfsetlinewidth{0.803000pt}%
\definecolor{currentstroke}{rgb}{0.800000,0.800000,0.800000}%
\pgfsetstrokecolor{currentstroke}%
\pgfsetdash{}{0pt}%
\pgfpathmoveto{\pgfqpoint{0.522684in}{1.423049in}}%
\pgfpathlineto{\pgfqpoint{3.465685in}{1.423049in}}%
\pgfusepath{stroke}%
\end{pgfscope}%
\begin{pgfscope}%
\definecolor{textcolor}{rgb}{0.150000,0.150000,0.150000}%
\pgfsetstrokecolor{textcolor}%
\pgfsetfillcolor{textcolor}%
\pgftext[x=0.179012in, y=1.379646in, left, base]{\color{textcolor}\rmfamily\fontsize{8.800000}{10.560000}\selectfont \(\displaystyle {0.80}\)}%
\end{pgfscope}%
\begin{pgfscope}%
\pgfpathrectangle{\pgfqpoint{0.522684in}{0.420833in}}{\pgfqpoint{2.943002in}{2.004431in}}%
\pgfusepath{clip}%
\pgfsetroundcap%
\pgfsetroundjoin%
\pgfsetlinewidth{0.803000pt}%
\definecolor{currentstroke}{rgb}{0.800000,0.800000,0.800000}%
\pgfsetstrokecolor{currentstroke}%
\pgfsetdash{}{0pt}%
\pgfpathmoveto{\pgfqpoint{0.522684in}{1.673602in}}%
\pgfpathlineto{\pgfqpoint{3.465685in}{1.673602in}}%
\pgfusepath{stroke}%
\end{pgfscope}%
\begin{pgfscope}%
\definecolor{textcolor}{rgb}{0.150000,0.150000,0.150000}%
\pgfsetstrokecolor{textcolor}%
\pgfsetfillcolor{textcolor}%
\pgftext[x=0.179012in, y=1.630200in, left, base]{\color{textcolor}\rmfamily\fontsize{8.800000}{10.560000}\selectfont \(\displaystyle {0.85}\)}%
\end{pgfscope}%
\begin{pgfscope}%
\pgfpathrectangle{\pgfqpoint{0.522684in}{0.420833in}}{\pgfqpoint{2.943002in}{2.004431in}}%
\pgfusepath{clip}%
\pgfsetroundcap%
\pgfsetroundjoin%
\pgfsetlinewidth{0.803000pt}%
\definecolor{currentstroke}{rgb}{0.800000,0.800000,0.800000}%
\pgfsetstrokecolor{currentstroke}%
\pgfsetdash{}{0pt}%
\pgfpathmoveto{\pgfqpoint{0.522684in}{1.924156in}}%
\pgfpathlineto{\pgfqpoint{3.465685in}{1.924156in}}%
\pgfusepath{stroke}%
\end{pgfscope}%
\begin{pgfscope}%
\definecolor{textcolor}{rgb}{0.150000,0.150000,0.150000}%
\pgfsetstrokecolor{textcolor}%
\pgfsetfillcolor{textcolor}%
\pgftext[x=0.179012in, y=1.880753in, left, base]{\color{textcolor}\rmfamily\fontsize{8.800000}{10.560000}\selectfont \(\displaystyle {0.90}\)}%
\end{pgfscope}%
\begin{pgfscope}%
\pgfpathrectangle{\pgfqpoint{0.522684in}{0.420833in}}{\pgfqpoint{2.943002in}{2.004431in}}%
\pgfusepath{clip}%
\pgfsetroundcap%
\pgfsetroundjoin%
\pgfsetlinewidth{0.803000pt}%
\definecolor{currentstroke}{rgb}{0.800000,0.800000,0.800000}%
\pgfsetstrokecolor{currentstroke}%
\pgfsetdash{}{0pt}%
\pgfpathmoveto{\pgfqpoint{0.522684in}{2.174710in}}%
\pgfpathlineto{\pgfqpoint{3.465685in}{2.174710in}}%
\pgfusepath{stroke}%
\end{pgfscope}%
\begin{pgfscope}%
\definecolor{textcolor}{rgb}{0.150000,0.150000,0.150000}%
\pgfsetstrokecolor{textcolor}%
\pgfsetfillcolor{textcolor}%
\pgftext[x=0.179012in, y=2.131307in, left, base]{\color{textcolor}\rmfamily\fontsize{8.800000}{10.560000}\selectfont \(\displaystyle {0.95}\)}%
\end{pgfscope}%
\begin{pgfscope}%
\pgfpathrectangle{\pgfqpoint{0.522684in}{0.420833in}}{\pgfqpoint{2.943002in}{2.004431in}}%
\pgfusepath{clip}%
\pgfsetroundcap%
\pgfsetroundjoin%
\pgfsetlinewidth{0.803000pt}%
\definecolor{currentstroke}{rgb}{0.800000,0.800000,0.800000}%
\pgfsetstrokecolor{currentstroke}%
\pgfsetdash{}{0pt}%
\pgfpathmoveto{\pgfqpoint{0.522684in}{2.425264in}}%
\pgfpathlineto{\pgfqpoint{3.465685in}{2.425264in}}%
\pgfusepath{stroke}%
\end{pgfscope}%
\begin{pgfscope}%
\definecolor{textcolor}{rgb}{0.150000,0.150000,0.150000}%
\pgfsetstrokecolor{textcolor}%
\pgfsetfillcolor{textcolor}%
\pgftext[x=0.179012in, y=2.381861in, left, base]{\color{textcolor}\rmfamily\fontsize{8.800000}{10.560000}\selectfont \(\displaystyle {1.00}\)}%
\end{pgfscope}%
\begin{pgfscope}%
\definecolor{textcolor}{rgb}{0.150000,0.150000,0.150000}%
\pgfsetstrokecolor{textcolor}%
\pgfsetfillcolor{textcolor}%
\pgftext[x=0.123457in,y=1.423049in,,bottom,rotate=90.000000]{\color{textcolor}\rmfamily\fontsize{9.600000}{11.520000}\selectfont Mean test AUC}%
\end{pgfscope}%
\begin{pgfscope}%
\pgfpathrectangle{\pgfqpoint{0.522684in}{0.420833in}}{\pgfqpoint{2.943002in}{2.004431in}}%
\pgfusepath{clip}%
\pgfsetbuttcap%
\pgfsetroundjoin%
\definecolor{currentfill}{rgb}{0.121569,0.466667,0.705882}%
\pgfsetfillcolor{currentfill}%
\pgfsetfillopacity{0.100000}%
\pgfsetlinewidth{0.803000pt}%
\definecolor{currentstroke}{rgb}{0.121569,0.466667,0.705882}%
\pgfsetstrokecolor{currentstroke}%
\pgfsetstrokeopacity{0.100000}%
\pgfsetdash{}{0pt}%
\pgfsys@defobject{currentmarker}{\pgfqpoint{0.522684in}{0.581188in}}{\pgfqpoint{3.465685in}{2.259898in}}{%
\pgfpathmoveto{\pgfqpoint{0.522684in}{0.621276in}}%
\pgfpathlineto{\pgfqpoint{0.522684in}{0.581188in}}%
\pgfpathlineto{\pgfqpoint{0.965650in}{1.137417in}}%
\pgfpathlineto{\pgfqpoint{1.224768in}{1.377949in}}%
\pgfpathlineto{\pgfqpoint{1.994185in}{1.919145in}}%
\pgfpathlineto{\pgfqpoint{3.022720in}{2.204777in}}%
\pgfpathlineto{\pgfqpoint{3.465685in}{2.259898in}}%
\pgfpathlineto{\pgfqpoint{3.465685in}{2.259898in}}%
\pgfpathlineto{\pgfqpoint{3.465685in}{2.259898in}}%
\pgfpathlineto{\pgfqpoint{3.022720in}{2.204777in}}%
\pgfpathlineto{\pgfqpoint{1.994185in}{1.919145in}}%
\pgfpathlineto{\pgfqpoint{1.224768in}{1.397993in}}%
\pgfpathlineto{\pgfqpoint{0.965650in}{1.157462in}}%
\pgfpathlineto{\pgfqpoint{0.522684in}{0.621276in}}%
\pgfpathclose%
\pgfusepath{stroke,fill}%
}%
\begin{pgfscope}%
\pgfsys@transformshift{0.000000in}{0.000000in}%
\pgfsys@useobject{currentmarker}{}%
\end{pgfscope}%
\end{pgfscope}%
\begin{pgfscope}%
\pgfpathrectangle{\pgfqpoint{0.522684in}{0.420833in}}{\pgfqpoint{2.943002in}{2.004431in}}%
\pgfusepath{clip}%
\pgfsetbuttcap%
\pgfsetroundjoin%
\definecolor{currentfill}{rgb}{1.000000,0.498039,0.054902}%
\pgfsetfillcolor{currentfill}%
\pgfsetfillopacity{0.100000}%
\pgfsetlinewidth{0.803000pt}%
\definecolor{currentstroke}{rgb}{1.000000,0.498039,0.054902}%
\pgfsetstrokecolor{currentstroke}%
\pgfsetstrokeopacity{0.100000}%
\pgfsetdash{}{0pt}%
\pgfsys@defobject{currentmarker}{\pgfqpoint{0.522684in}{0.831742in}}{\pgfqpoint{3.526595in}{2.269921in}}{%
\pgfpathmoveto{\pgfqpoint{0.522684in}{0.861808in}}%
\pgfpathlineto{\pgfqpoint{0.522684in}{0.831742in}}%
\pgfpathlineto{\pgfqpoint{0.965650in}{1.367927in}}%
\pgfpathlineto{\pgfqpoint{1.224768in}{1.618481in}}%
\pgfpathlineto{\pgfqpoint{1.994185in}{2.009345in}}%
\pgfpathlineto{\pgfqpoint{3.022720in}{2.224821in}}%
\pgfpathlineto{\pgfqpoint{3.526595in}{2.269921in}}%
\pgfpathlineto{\pgfqpoint{3.526595in}{2.269921in}}%
\pgfpathlineto{\pgfqpoint{3.526595in}{2.269921in}}%
\pgfpathlineto{\pgfqpoint{3.022720in}{2.224821in}}%
\pgfpathlineto{\pgfqpoint{1.994185in}{2.009345in}}%
\pgfpathlineto{\pgfqpoint{1.224768in}{1.618481in}}%
\pgfpathlineto{\pgfqpoint{0.965650in}{1.377949in}}%
\pgfpathlineto{\pgfqpoint{0.522684in}{0.861808in}}%
\pgfpathclose%
\pgfusepath{stroke,fill}%
}%
\begin{pgfscope}%
\pgfsys@transformshift{0.000000in}{0.000000in}%
\pgfsys@useobject{currentmarker}{}%
\end{pgfscope}%
\end{pgfscope}%
\begin{pgfscope}%
\pgfpathrectangle{\pgfqpoint{0.522684in}{0.420833in}}{\pgfqpoint{2.943002in}{2.004431in}}%
\pgfusepath{clip}%
\pgfsetbuttcap%
\pgfsetroundjoin%
\pgfsetlinewidth{1.204500pt}%
\definecolor{currentstroke}{rgb}{0.121569,0.466667,0.705882}%
\pgfsetstrokecolor{currentstroke}%
\pgfsetdash{{4.440000pt}{1.920000pt}}{0.000000pt}%
\pgfpathmoveto{\pgfqpoint{0.522684in}{0.601232in}}%
\pgfpathlineto{\pgfqpoint{0.965650in}{1.147439in}}%
\pgfpathlineto{\pgfqpoint{1.224768in}{1.387971in}}%
\pgfpathlineto{\pgfqpoint{1.994185in}{1.919145in}}%
\pgfpathlineto{\pgfqpoint{3.022720in}{2.204777in}}%
\pgfpathlineto{\pgfqpoint{3.465685in}{2.259898in}}%
\pgfusepath{stroke}%
\end{pgfscope}%
\begin{pgfscope}%
\pgfpathrectangle{\pgfqpoint{0.522684in}{0.420833in}}{\pgfqpoint{2.943002in}{2.004431in}}%
\pgfusepath{clip}%
\pgfsetbuttcap%
\pgfsetroundjoin%
\definecolor{currentfill}{rgb}{0.121569,0.466667,0.705882}%
\pgfsetfillcolor{currentfill}%
\pgfsetlinewidth{0.752812pt}%
\definecolor{currentstroke}{rgb}{0.121569,0.466667,0.705882}%
\pgfsetstrokecolor{currentstroke}%
\pgfsetdash{}{0pt}%
\pgfsys@defobject{currentmarker}{\pgfqpoint{-0.033333in}{-0.033333in}}{\pgfqpoint{0.033333in}{0.033333in}}{%
\pgfpathmoveto{\pgfqpoint{-0.033333in}{-0.033333in}}%
\pgfpathlineto{\pgfqpoint{0.033333in}{0.033333in}}%
\pgfpathmoveto{\pgfqpoint{-0.033333in}{0.033333in}}%
\pgfpathlineto{\pgfqpoint{0.033333in}{-0.033333in}}%
\pgfusepath{stroke,fill}%
}%
\begin{pgfscope}%
\pgfsys@transformshift{0.522684in}{0.601232in}%
\pgfsys@useobject{currentmarker}{}%
\end{pgfscope}%
\begin{pgfscope}%
\pgfsys@transformshift{0.965650in}{1.147439in}%
\pgfsys@useobject{currentmarker}{}%
\end{pgfscope}%
\begin{pgfscope}%
\pgfsys@transformshift{1.224768in}{1.387971in}%
\pgfsys@useobject{currentmarker}{}%
\end{pgfscope}%
\begin{pgfscope}%
\pgfsys@transformshift{1.994185in}{1.919145in}%
\pgfsys@useobject{currentmarker}{}%
\end{pgfscope}%
\begin{pgfscope}%
\pgfsys@transformshift{3.022720in}{2.204777in}%
\pgfsys@useobject{currentmarker}{}%
\end{pgfscope}%
\begin{pgfscope}%
\pgfsys@transformshift{3.465685in}{2.259898in}%
\pgfsys@useobject{currentmarker}{}%
\end{pgfscope}%
\end{pgfscope}%
\begin{pgfscope}%
\pgfpathrectangle{\pgfqpoint{0.522684in}{0.420833in}}{\pgfqpoint{2.943002in}{2.004431in}}%
\pgfusepath{clip}%
\pgfsetroundcap%
\pgfsetroundjoin%
\pgfsetlinewidth{1.204500pt}%
\definecolor{currentstroke}{rgb}{1.000000,0.498039,0.054902}%
\pgfsetstrokecolor{currentstroke}%
\pgfsetdash{}{0pt}%
\pgfpathmoveto{\pgfqpoint{0.522684in}{0.846775in}}%
\pgfpathlineto{\pgfqpoint{0.965650in}{1.372938in}}%
\pgfpathlineto{\pgfqpoint{1.224768in}{1.618481in}}%
\pgfpathlineto{\pgfqpoint{1.994185in}{2.009345in}}%
\pgfpathlineto{\pgfqpoint{3.022720in}{2.224821in}}%
\pgfpathlineto{\pgfqpoint{3.475685in}{2.265364in}}%
\pgfusepath{stroke}%
\end{pgfscope}%
\begin{pgfscope}%
\pgfpathrectangle{\pgfqpoint{0.522684in}{0.420833in}}{\pgfqpoint{2.943002in}{2.004431in}}%
\pgfusepath{clip}%
\pgfsetbuttcap%
\pgfsetroundjoin%
\definecolor{currentfill}{rgb}{1.000000,0.498039,0.054902}%
\pgfsetfillcolor{currentfill}%
\pgfsetlinewidth{0.752812pt}%
\definecolor{currentstroke}{rgb}{1.000000,0.498039,0.054902}%
\pgfsetstrokecolor{currentstroke}%
\pgfsetdash{}{0pt}%
\pgfsys@defobject{currentmarker}{\pgfqpoint{-0.033333in}{-0.033333in}}{\pgfqpoint{0.033333in}{0.033333in}}{%
\pgfpathmoveto{\pgfqpoint{-0.033333in}{-0.033333in}}%
\pgfpathlineto{\pgfqpoint{0.033333in}{0.033333in}}%
\pgfpathmoveto{\pgfqpoint{-0.033333in}{0.033333in}}%
\pgfpathlineto{\pgfqpoint{0.033333in}{-0.033333in}}%
\pgfusepath{stroke,fill}%
}%
\begin{pgfscope}%
\pgfsys@transformshift{0.522684in}{0.846775in}%
\pgfsys@useobject{currentmarker}{}%
\end{pgfscope}%
\begin{pgfscope}%
\pgfsys@transformshift{0.965650in}{1.372938in}%
\pgfsys@useobject{currentmarker}{}%
\end{pgfscope}%
\begin{pgfscope}%
\pgfsys@transformshift{1.224768in}{1.618481in}%
\pgfsys@useobject{currentmarker}{}%
\end{pgfscope}%
\begin{pgfscope}%
\pgfsys@transformshift{1.994185in}{2.009345in}%
\pgfsys@useobject{currentmarker}{}%
\end{pgfscope}%
\begin{pgfscope}%
\pgfsys@transformshift{3.022720in}{2.224821in}%
\pgfsys@useobject{currentmarker}{}%
\end{pgfscope}%
\begin{pgfscope}%
\pgfsys@transformshift{3.526595in}{2.269921in}%
\pgfsys@useobject{currentmarker}{}%
\end{pgfscope}%
\end{pgfscope}%
\begin{pgfscope}%
\pgfpathrectangle{\pgfqpoint{0.522684in}{0.420833in}}{\pgfqpoint{2.943002in}{2.004431in}}%
\pgfusepath{clip}%
\pgfsetbuttcap%
\pgfsetroundjoin%
\pgfsetlinewidth{1.204500pt}%
\definecolor{currentstroke}{rgb}{0.000000,0.000000,0.000000}%
\pgfsetstrokecolor{currentstroke}%
\pgfsetdash{{1.200000pt}{1.980000pt}}{0.000000pt}%
\pgfpathmoveto{\pgfqpoint{0.522684in}{2.350098in}}%
\pgfpathlineto{\pgfqpoint{3.465685in}{2.350098in}}%
\pgfusepath{stroke}%
\end{pgfscope}%
\begin{pgfscope}%
\pgfpathrectangle{\pgfqpoint{0.522684in}{0.420833in}}{\pgfqpoint{2.943002in}{2.004431in}}%
\pgfusepath{clip}%
\pgfsetbuttcap%
\pgfsetroundjoin%
\pgfsetlinewidth{1.204500pt}%
\definecolor{currentstroke}{rgb}{0.000000,0.000000,0.000000}%
\pgfsetstrokecolor{currentstroke}%
\pgfsetdash{{4.440000pt}{1.920000pt}}{0.000000pt}%
\pgfpathmoveto{\pgfqpoint{0.522684in}{2.330053in}}%
\pgfpathlineto{\pgfqpoint{3.465685in}{2.330053in}}%
\pgfusepath{stroke}%
\end{pgfscope}%
\begin{pgfscope}%
\pgfsetrectcap%
\pgfsetmiterjoin%
\pgfsetlinewidth{1.003750pt}%
\definecolor{currentstroke}{rgb}{0.800000,0.800000,0.800000}%
\pgfsetstrokecolor{currentstroke}%
\pgfsetdash{}{0pt}%
\pgfpathmoveto{\pgfqpoint{0.522684in}{0.420833in}}%
\pgfpathlineto{\pgfqpoint{0.522684in}{2.425264in}}%
\pgfusepath{stroke}%
\end{pgfscope}%
\begin{pgfscope}%
\pgfsetrectcap%
\pgfsetmiterjoin%
\pgfsetlinewidth{1.003750pt}%
\definecolor{currentstroke}{rgb}{0.800000,0.800000,0.800000}%
\pgfsetstrokecolor{currentstroke}%
\pgfsetdash{}{0pt}%
\pgfpathmoveto{\pgfqpoint{0.522684in}{0.420833in}}%
\pgfpathlineto{\pgfqpoint{3.465685in}{0.420833in}}%
\pgfusepath{stroke}%
\end{pgfscope}%
\begin{pgfscope}%
\pgfsetroundcap%
\pgfsetroundjoin%
\definecolor{currentfill}{rgb}{0.862745,0.862745,0.862745}%
\pgfsetfillcolor{currentfill}%
\pgfsetlinewidth{0.803000pt}%
\definecolor{currentstroke}{rgb}{1.000000,1.000000,1.000000}%
\pgfsetstrokecolor{currentstroke}%
\pgfsetdash{}{0pt}%
\pgfpathmoveto{\pgfqpoint{1.113321in}{1.770122in}}%
\pgfpathquadraticcurveto{\pgfqpoint{1.001271in}{1.889705in}}{\pgfqpoint{0.889221in}{2.009288in}}%
\pgfpathlineto{\pgfqpoint{0.853748in}{1.976050in}}%
\pgfpathquadraticcurveto{\pgfqpoint{0.827174in}{2.065350in}}{\pgfqpoint{0.800599in}{2.154650in}}%
\pgfpathquadraticcurveto{\pgfqpoint{0.887983in}{2.122329in}}{\pgfqpoint{0.975368in}{2.090009in}}%
\pgfpathlineto{\pgfqpoint{0.939895in}{2.056771in}}%
\pgfpathquadraticcurveto{\pgfqpoint{1.051946in}{1.937188in}}{\pgfqpoint{1.163996in}{1.817605in}}%
\pgfpathlineto{\pgfqpoint{1.113321in}{1.770122in}}%
\pgfpathclose%
\pgfusepath{stroke,fill}%
\end{pgfscope}%
\begin{pgfscope}%
\definecolor{textcolor}{rgb}{0.862745,0.862745,0.862745}%
\pgfsetstrokecolor{textcolor}%
\pgfsetfillcolor{textcolor}%
\pgftext[x=1.157436in,y=2.024378in,left,]{\color{textcolor}\rmfamily\fontsize{12.000000}{14.400000}\selectfont better}%
\end{pgfscope}%
\begin{pgfscope}%
\pgfsetroundcap%
\pgfsetroundjoin%
\definecolor{currentfill}{rgb}{0.000000,0.501961,0.000000}%
\pgfsetfillcolor{currentfill}%
\pgfsetlinewidth{0.803000pt}%
\definecolor{currentstroke}{rgb}{1.000000,1.000000,1.000000}%
\pgfsetstrokecolor{currentstroke}%
\pgfsetdash{}{0pt}%
\pgfpathmoveto{\pgfqpoint{1.211831in}{1.374082in}}%
\pgfpathquadraticcurveto{\pgfqpoint{1.178563in}{1.374082in}}{\pgfqpoint{1.145296in}{1.374082in}}%
\pgfpathlineto{\pgfqpoint{1.145296in}{1.318527in}}%
\pgfpathquadraticcurveto{\pgfqpoint{1.061949in}{1.353249in}}{\pgfqpoint{0.978602in}{1.387971in}}%
\pgfpathquadraticcurveto{\pgfqpoint{1.061949in}{1.422693in}}{\pgfqpoint{1.145296in}{1.457416in}}%
\pgfpathlineto{\pgfqpoint{1.145296in}{1.401860in}}%
\pgfpathquadraticcurveto{\pgfqpoint{1.178563in}{1.401860in}}{\pgfqpoint{1.211831in}{1.401860in}}%
\pgfpathlineto{\pgfqpoint{1.211831in}{1.374082in}}%
\pgfpathclose%
\pgfusepath{stroke,fill}%
\end{pgfscope}%
\begin{pgfscope}%
\definecolor{textcolor}{rgb}{0.000000,0.501961,0.000000}%
\pgfsetstrokecolor{textcolor}%
\pgfsetfillcolor{textcolor}%
\pgftext[x=1.551219in,y=1.387971in,left,]{\color{textcolor}\rmfamily\fontsize{12.000000}{14.400000}\selectfont 30\% \(\displaystyle \varepsilon\) saving}%
\end{pgfscope}%
\begin{pgfscope}%
\pgfsetbuttcap%
\pgfsetmiterjoin%
\definecolor{currentfill}{rgb}{1.000000,1.000000,1.000000}%
\pgfsetfillcolor{currentfill}%
\pgfsetfillopacity{0.800000}%
\pgfsetlinewidth{0.803000pt}%
\definecolor{currentstroke}{rgb}{0.800000,0.800000,0.800000}%
\pgfsetstrokecolor{currentstroke}%
\pgfsetstrokeopacity{0.800000}%
\pgfsetdash{}{0pt}%
\pgfpathmoveto{\pgfqpoint{1.002082in}{0.481944in}}%
\pgfpathlineto{\pgfqpoint{3.380130in}{0.481944in}}%
\pgfpathquadraticcurveto{\pgfqpoint{3.404574in}{0.481944in}}{\pgfqpoint{3.404574in}{0.506389in}}%
\pgfpathlineto{\pgfqpoint{3.404574in}{1.210833in}}%
\pgfpathquadraticcurveto{\pgfqpoint{3.404574in}{1.235278in}}{\pgfqpoint{3.380130in}{1.235278in}}%
\pgfpathlineto{\pgfqpoint{1.002082in}{1.235278in}}%
\pgfpathquadraticcurveto{\pgfqpoint{0.977637in}{1.235278in}}{\pgfqpoint{0.977637in}{1.210833in}}%
\pgfpathlineto{\pgfqpoint{0.977637in}{0.506389in}}%
\pgfpathquadraticcurveto{\pgfqpoint{0.977637in}{0.481944in}}{\pgfqpoint{1.002082in}{0.481944in}}%
\pgfpathclose%
\pgfusepath{stroke,fill}%
\end{pgfscope}%
\begin{pgfscope}%
\pgfsetbuttcap%
\pgfsetroundjoin%
\pgfsetlinewidth{1.204500pt}%
\definecolor{currentstroke}{rgb}{0.121569,0.466667,0.705882}%
\pgfsetstrokecolor{currentstroke}%
\pgfsetdash{{4.440000pt}{1.920000pt}}{0.000000pt}%
\pgfpathmoveto{\pgfqpoint{1.026526in}{1.142361in}}%
\pgfpathlineto{\pgfqpoint{1.270971in}{1.142361in}}%
\pgfusepath{stroke}%
\end{pgfscope}%
\begin{pgfscope}%
\pgfsetbuttcap%
\pgfsetroundjoin%
\definecolor{currentfill}{rgb}{0.121569,0.466667,0.705882}%
\pgfsetfillcolor{currentfill}%
\pgfsetlinewidth{0.752812pt}%
\definecolor{currentstroke}{rgb}{0.121569,0.466667,0.705882}%
\pgfsetstrokecolor{currentstroke}%
\pgfsetdash{}{0pt}%
\pgfsys@defobject{currentmarker}{\pgfqpoint{-0.033333in}{-0.033333in}}{\pgfqpoint{0.033333in}{0.033333in}}{%
\pgfpathmoveto{\pgfqpoint{-0.033333in}{-0.033333in}}%
\pgfpathlineto{\pgfqpoint{0.033333in}{0.033333in}}%
\pgfpathmoveto{\pgfqpoint{-0.033333in}{0.033333in}}%
\pgfpathlineto{\pgfqpoint{0.033333in}{-0.033333in}}%
\pgfusepath{stroke,fill}%
}%
\begin{pgfscope}%
\pgfsys@transformshift{1.148748in}{1.142361in}%
\pgfsys@useobject{currentmarker}{}%
\end{pgfscope}%
\end{pgfscope}%
\begin{pgfscope}%
\definecolor{textcolor}{rgb}{0.150000,0.150000,0.150000}%
\pgfsetstrokecolor{textcolor}%
\pgfsetfillcolor{textcolor}%
\pgftext[x=1.368748in,y=1.099583in,left,base]{\color{textcolor}\rmfamily\fontsize{8.800000}{10.560000}\selectfont Maddock et al.}%
\end{pgfscope}%
\begin{pgfscope}%
\pgfsetroundcap%
\pgfsetroundjoin%
\pgfsetlinewidth{1.204500pt}%
\definecolor{currentstroke}{rgb}{1.000000,0.498039,0.054902}%
\pgfsetstrokecolor{currentstroke}%
\pgfsetdash{}{0pt}%
\pgfpathmoveto{\pgfqpoint{1.026526in}{0.970139in}}%
\pgfpathlineto{\pgfqpoint{1.270971in}{0.970139in}}%
\pgfusepath{stroke}%
\end{pgfscope}%
\begin{pgfscope}%
\pgfsetbuttcap%
\pgfsetroundjoin%
\definecolor{currentfill}{rgb}{1.000000,0.498039,0.054902}%
\pgfsetfillcolor{currentfill}%
\pgfsetlinewidth{0.752812pt}%
\definecolor{currentstroke}{rgb}{1.000000,0.498039,0.054902}%
\pgfsetstrokecolor{currentstroke}%
\pgfsetdash{}{0pt}%
\pgfsys@defobject{currentmarker}{\pgfqpoint{-0.033333in}{-0.033333in}}{\pgfqpoint{0.033333in}{0.033333in}}{%
\pgfpathmoveto{\pgfqpoint{-0.033333in}{-0.033333in}}%
\pgfpathlineto{\pgfqpoint{0.033333in}{0.033333in}}%
\pgfpathmoveto{\pgfqpoint{-0.033333in}{0.033333in}}%
\pgfpathlineto{\pgfqpoint{0.033333in}{-0.033333in}}%
\pgfusepath{stroke,fill}%
}%
\begin{pgfscope}%
\pgfsys@transformshift{1.148748in}{0.970139in}%
\pgfsys@useobject{currentmarker}{}%
\end{pgfscope}%
\end{pgfscope}%
\begin{pgfscope}%
\definecolor{textcolor}{rgb}{0.150000,0.150000,0.150000}%
\pgfsetstrokecolor{textcolor}%
\pgfsetfillcolor{textcolor}%
\pgftext[x=1.368748in,y=0.927361in,left,base]{\color{textcolor}\rmfamily\fontsize{8.800000}{10.560000}\selectfont S-BDT}%
\end{pgfscope}%
\begin{pgfscope}%
\pgfsetbuttcap%
\pgfsetroundjoin%
\pgfsetlinewidth{1.204500pt}%
\definecolor{currentstroke}{rgb}{0.000000,0.000000,0.000000}%
\pgfsetstrokecolor{currentstroke}%
\pgfsetdash{{1.200000pt}{1.980000pt}}{0.000000pt}%
\pgfpathmoveto{\pgfqpoint{1.026526in}{0.790972in}}%
\pgfpathlineto{\pgfqpoint{1.270971in}{0.790972in}}%
\pgfusepath{stroke}%
\end{pgfscope}%
\begin{pgfscope}%
\definecolor{textcolor}{rgb}{0.150000,0.150000,0.150000}%
\pgfsetstrokecolor{textcolor}%
\pgfsetfillcolor{textcolor}%
\pgftext[x=1.368748in,y=0.748194in,left,base]{\color{textcolor}\rmfamily\fontsize{8.800000}{10.560000}\selectfont xgboost (nonprivate)}%
\end{pgfscope}%
\begin{pgfscope}%
\pgfsetbuttcap%
\pgfsetroundjoin%
\pgfsetlinewidth{1.204500pt}%
\definecolor{currentstroke}{rgb}{0.000000,0.000000,0.000000}%
\pgfsetstrokecolor{currentstroke}%
\pgfsetdash{{4.440000pt}{1.920000pt}}{0.000000pt}%
\pgfpathmoveto{\pgfqpoint{1.026526in}{0.604861in}}%
\pgfpathlineto{\pgfqpoint{1.270971in}{0.604861in}}%
\pgfusepath{stroke}%
\end{pgfscope}%
\begin{pgfscope}%
\definecolor{textcolor}{rgb}{0.150000,0.150000,0.150000}%
\pgfsetstrokecolor{textcolor}%
\pgfsetfillcolor{textcolor}%
\pgftext[x=1.368748in,y=0.562083in,left,base]{\color{textcolor}\rmfamily\fontsize{8.800000}{10.560000}\selectfont xgboost random splits (nonprivate)}%
\end{pgfscope}%
\end{pgfpicture}%
\makeatother%
\endgroup%

%% file: experiments_main_results_ablations.tex
\begin{table}[!ht]
  \footnotesize
  \centering
  \begin{subfigure}[h]{0.99\columnwidth}
  \begin{tabularx}{\textwidth}{ll}
       \toprule
       Technique & Mean test regression error (RMSE)  \rule[-0.9ex]{0pt}{0pt} \\
       \midrule
       \citet{Maddock_2022} & 2.939 $\pm 0.019$ \rule{0pt}{2.6ex} \\
       \raisebox{0.25ex}{\footnotesize+} Subsampling & 2.782 $\pm 0.008$ \\
       \raisebox{0.25ex}{\footnotesize+}  DP initial score & 2.760 $\pm 0.009$ \\
       \raisebox{0.25ex}{\footnotesize+}  Leaf-balanced noise & 2.745 $\pm 0.008$ \\
       \bottomrule
  \end{tabularx}
  \caption{Dataset: \textbf{Abalone} (Regression)}
  \label{tbl:abalone_ablations}
  \end{subfigure}
  \begin{subfigure}[h]{0.99\columnwidth}
  \begin{tabularx}{\textwidth}{ll}
       \toprule
       Technique & Mean test AUC  \rule[-0.9ex]{0pt}{0pt} \\
       \midrule
       \citet{Maddock_2022} & 0.791 $\pm 0.002$ \rule{0pt}{2.6ex} \\
       \raisebox{0.25ex}{\footnotesize+} Subsampling & 0.811 $\pm 0.001$ \\
       \raisebox{0.25ex}{\footnotesize+}  Leaf-balanced noise & 0.825 $\pm 0.001$ \\
       \bottomrule
  \end{tabularx}
  \caption{Dataset: \textbf{Adult} (Classification)}
  \label{tbl:adult_ablations}
  \end{subfigure}
  \caption{\textbf{Ablation study}: Improvement of our \dpgbdt{} over the SOTA by Maddock et al. \cite{Maddock_2022}. We report the RMSE (Abalone, $\eps=0.105$) and the AUC (Adult, $\eps=0.02$) each with a standard error of $200$ runs. 
  The parameters are chosen as in \Cref{fig:all_ablations}.
  }
  \label{tbl:abalone_and_adult_ablations}
\end{table}

%% file: experiments_ablations.tex
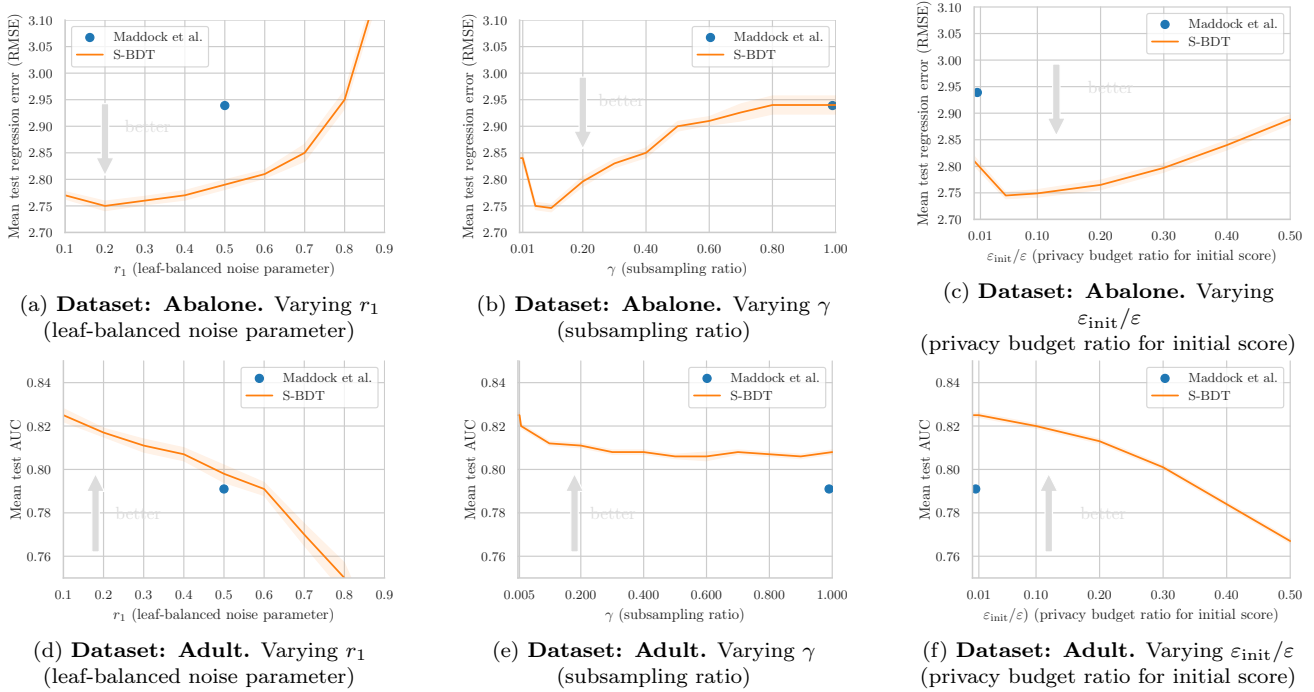
\begin{figure*}[!t]
  \centering
  \captionsetup{justification=centering}
  \begin{subfigure}[h]{\textwidth}
    \begin{subfigure}[h]{0.3\textwidth}
        \centering
        \captionsetup{justification=centering}
        \resizebox{0.99\columnwidth}{!}{\input{images/abalone_lbn_ablation.pgf}\unskip}
        \caption{\textbf{Dataset: Abalone.} Varying $r_1$ \\ (leaf-balanced noise parameter)}
          \label{fig:abalone_r1_ablation}
    \end{subfigure}
    \hfill
    \begin{subfigure}[h]{0.3\textwidth}
        \centering
        \captionsetup{justification=centering}
        \resizebox{0.99\columnwidth}{!}{\input{images/abalone_subsampling_ablation.pgf}\unskip}
        \caption{\textbf{Dataset: Abalone.} Varying $\gamma$ \\ (subsampling ratio)}
          \label{fig:abalone_Q_ablation}
    \end{subfigure}
    \hfill
    \begin{subfigure}[h]{0.3\textwidth}
        \centering
        \captionsetup{justification=centering}
        \resizebox{0.99\columnwidth}{!}{\input{images/abalone_init_ablation.pgf}\unskip}
        \caption{\textbf{Dataset: Abalone.} Varying $\varepsilon_\text{init} / \varepsilon$ \\ (privacy budget ratio for initial score)}
          \label{fig:abalone_init_ablation}
    \end{subfigure}
  \end{subfigure}
  \begin{subfigure}[h]{\textwidth}
    \begin{subfigure}[h]{0.3\textwidth}
        \centering
        \captionsetup{justification=centering}
        \resizebox{0.99\columnwidth}{!}{\input{images/adult_lbn_ablation.pgf}\unskip}
        \caption{\textbf{Dataset: Adult.} Varying $r_1$ \\ (leaf-balanced noise parameter)}
          \label{fig:adult_r1_ablation}
    \end{subfigure}
    \hfill
    \begin{subfigure}[h]{0.3\textwidth}
        \centering
        \captionsetup{justification=centering}
        \resizebox{0.99\columnwidth}{!}{\input{images/adult_subsampling_ablation.pgf}\unskip}
        \caption{\textbf{Dataset: Adult.} Varying $\gamma$ \\ (subsampling ratio)}
          \label{fig:adult_Q_ablation}
    \end{subfigure}
    \hfill
    \begin{subfigure}[h]{0.3\textwidth}
        \centering
        \captionsetup{justification=centering}
        \resizebox{0.99\columnwidth}{!}{\input{images/adult_init_ablation.pgf}\unskip}
        \caption{\textbf{Dataset: Adult.} Varying $\varepsilon_\text{init} / \varepsilon$ \\ (privacy budget ratio for initial score)}
          \label{fig:adult_init_ablation}
    \end{subfigure}
  \end{subfigure}
  \caption{\textbf{Ablation studies of our improvements:} Regression error (RMSE) and AUC of 200 runs. The transparent area is the standard error. 
  For Abalone we set $\varepsilon=0.105$, number of trees $T_\text{regular}=150$, depth $d=2$, subsampling ratio $\gamma=0.1$, leaf-balanced noise parameter $r_1=0.2$, privacy budget ratio for initial score $\varepsilon_\text{init}/\varepsilon=0.1$, \blue{$\varepsilon_\text{ds}=0.005$} and clipping bound $g^*=0.1$.
  For Adult, we set $\varepsilon=0.02$, $T_\text{regular}=200$, $d=5$, $\gamma=0.005$, $r_1=0.1$, $g^*=0.5, h^*=0.1$.
  A smaller $r_1$ value means more privacy budget for Hessian sum compared to gradient sum. $r_1=0.5$ deactivates leaf-balanced noise, $\gamma=1.0$ subsampling and $\varepsilon_\text{init}/\varepsilon=0.0$ 
  the initial score.
  }
  \label{fig:all_ablations}
\end{figure*}

%% file: images/abalone_lbn_ablation.pgf
\begingroup%
\makeatletter%
\begin{pgfpicture}%
\pgfpathrectangle{\pgfpointorigin}{\pgfqpoint{3.612000in}{2.495774in}}%
\pgfusepath{use as bounding box, clip}%
\begin{pgfscope}%
\pgfsetbuttcap%
\pgfsetmiterjoin%
\definecolor{currentfill}{rgb}{1.000000,1.000000,1.000000}%
\pgfsetfillcolor{currentfill}%
\pgfsetlinewidth{0.000000pt}%
\definecolor{currentstroke}{rgb}{1.000000,1.000000,1.000000}%
\pgfsetstrokecolor{currentstroke}%
\pgfsetdash{}{0pt}%
\pgfpathmoveto{\pgfqpoint{0.000000in}{0.000000in}}%
\pgfpathlineto{\pgfqpoint{3.612000in}{0.000000in}}%
\pgfpathlineto{\pgfqpoint{3.612000in}{2.495774in}}%
\pgfpathlineto{\pgfqpoint{0.000000in}{2.495774in}}%
\pgfpathlineto{\pgfqpoint{0.000000in}{0.000000in}}%
\pgfpathclose%
\pgfusepath{fill}%
\end{pgfscope}%
\begin{pgfscope}%
\pgfsetbuttcap%
\pgfsetmiterjoin%
\definecolor{currentfill}{rgb}{1.000000,1.000000,1.000000}%
\pgfsetfillcolor{currentfill}%
\pgfsetlinewidth{0.000000pt}%
\definecolor{currentstroke}{rgb}{0.000000,0.000000,0.000000}%
\pgfsetstrokecolor{currentstroke}%
\pgfsetstrokeopacity{0.000000}%
\pgfsetdash{}{0pt}%
\pgfpathmoveto{\pgfqpoint{0.538116in}{0.420833in}}%
\pgfpathlineto{\pgfqpoint{3.529921in}{0.420833in}}%
\pgfpathlineto{\pgfqpoint{3.529921in}{2.425264in}}%
\pgfpathlineto{\pgfqpoint{0.538116in}{2.425264in}}%
\pgfpathlineto{\pgfqpoint{0.538116in}{0.420833in}}%
\pgfpathclose%
\pgfusepath{fill}%
\end{pgfscope}%
\begin{pgfscope}%
\pgfpathrectangle{\pgfqpoint{0.538116in}{0.420833in}}{\pgfqpoint{2.991805in}{2.004431in}}%
\pgfusepath{clip}%
\pgfsetroundcap%
\pgfsetroundjoin%
\pgfsetlinewidth{0.803000pt}%
\definecolor{currentstroke}{rgb}{0.800000,0.800000,0.800000}%
\pgfsetstrokecolor{currentstroke}%
\pgfsetdash{}{0pt}%
\pgfpathmoveto{\pgfqpoint{0.538116in}{0.420833in}}%
\pgfpathlineto{\pgfqpoint{0.538116in}{2.425264in}}%
\pgfusepath{stroke}%
\end{pgfscope}%
\begin{pgfscope}%
\definecolor{textcolor}{rgb}{0.150000,0.150000,0.150000}%
\pgfsetstrokecolor{textcolor}%
\pgfsetfillcolor{textcolor}%
\pgftext[x=0.538116in,y=0.305556in,,top]{\color{textcolor}\rmfamily\fontsize{8.800000}{10.560000}\selectfont \(\displaystyle {0.1}\)}%
\end{pgfscope}%
\begin{pgfscope}%
\pgfpathrectangle{\pgfqpoint{0.538116in}{0.420833in}}{\pgfqpoint{2.991805in}{2.004431in}}%
\pgfusepath{clip}%
\pgfsetroundcap%
\pgfsetroundjoin%
\pgfsetlinewidth{0.803000pt}%
\definecolor{currentstroke}{rgb}{0.800000,0.800000,0.800000}%
\pgfsetstrokecolor{currentstroke}%
\pgfsetdash{}{0pt}%
\pgfpathmoveto{\pgfqpoint{0.912091in}{0.420833in}}%
\pgfpathlineto{\pgfqpoint{0.912091in}{2.425264in}}%
\pgfusepath{stroke}%
\end{pgfscope}%
\begin{pgfscope}%
\definecolor{textcolor}{rgb}{0.150000,0.150000,0.150000}%
\pgfsetstrokecolor{textcolor}%
\pgfsetfillcolor{textcolor}%
\pgftext[x=0.912091in,y=0.305556in,,top]{\color{textcolor}\rmfamily\fontsize{8.800000}{10.560000}\selectfont \(\displaystyle {0.2}\)}%
\end{pgfscope}%
\begin{pgfscope}%
\pgfpathrectangle{\pgfqpoint{0.538116in}{0.420833in}}{\pgfqpoint{2.991805in}{2.004431in}}%
\pgfusepath{clip}%
\pgfsetroundcap%
\pgfsetroundjoin%
\pgfsetlinewidth{0.803000pt}%
\definecolor{currentstroke}{rgb}{0.800000,0.800000,0.800000}%
\pgfsetstrokecolor{currentstroke}%
\pgfsetdash{}{0pt}%
\pgfpathmoveto{\pgfqpoint{1.286067in}{0.420833in}}%
\pgfpathlineto{\pgfqpoint{1.286067in}{2.425264in}}%
\pgfusepath{stroke}%
\end{pgfscope}%
\begin{pgfscope}%
\definecolor{textcolor}{rgb}{0.150000,0.150000,0.150000}%
\pgfsetstrokecolor{textcolor}%
\pgfsetfillcolor{textcolor}%
\pgftext[x=1.286067in,y=0.305556in,,top]{\color{textcolor}\rmfamily\fontsize{8.800000}{10.560000}\selectfont \(\displaystyle {0.3}\)}%
\end{pgfscope}%
\begin{pgfscope}%
\pgfpathrectangle{\pgfqpoint{0.538116in}{0.420833in}}{\pgfqpoint{2.991805in}{2.004431in}}%
\pgfusepath{clip}%
\pgfsetroundcap%
\pgfsetroundjoin%
\pgfsetlinewidth{0.803000pt}%
\definecolor{currentstroke}{rgb}{0.800000,0.800000,0.800000}%
\pgfsetstrokecolor{currentstroke}%
\pgfsetdash{}{0pt}%
\pgfpathmoveto{\pgfqpoint{1.660043in}{0.420833in}}%
\pgfpathlineto{\pgfqpoint{1.660043in}{2.425264in}}%
\pgfusepath{stroke}%
\end{pgfscope}%
\begin{pgfscope}%
\definecolor{textcolor}{rgb}{0.150000,0.150000,0.150000}%
\pgfsetstrokecolor{textcolor}%
\pgfsetfillcolor{textcolor}%
\pgftext[x=1.660043in,y=0.305556in,,top]{\color{textcolor}\rmfamily\fontsize{8.800000}{10.560000}\selectfont \(\displaystyle {0.4}\)}%
\end{pgfscope}%
\begin{pgfscope}%
\pgfpathrectangle{\pgfqpoint{0.538116in}{0.420833in}}{\pgfqpoint{2.991805in}{2.004431in}}%
\pgfusepath{clip}%
\pgfsetroundcap%
\pgfsetroundjoin%
\pgfsetlinewidth{0.803000pt}%
\definecolor{currentstroke}{rgb}{0.800000,0.800000,0.800000}%
\pgfsetstrokecolor{currentstroke}%
\pgfsetdash{}{0pt}%
\pgfpathmoveto{\pgfqpoint{2.034018in}{0.420833in}}%
\pgfpathlineto{\pgfqpoint{2.034018in}{2.425264in}}%
\pgfusepath{stroke}%
\end{pgfscope}%
\begin{pgfscope}%
\definecolor{textcolor}{rgb}{0.150000,0.150000,0.150000}%
\pgfsetstrokecolor{textcolor}%
\pgfsetfillcolor{textcolor}%
\pgftext[x=2.034018in,y=0.305556in,,top]{\color{textcolor}\rmfamily\fontsize{8.800000}{10.560000}\selectfont \(\displaystyle {0.5}\)}%
\end{pgfscope}%
\begin{pgfscope}%
\pgfpathrectangle{\pgfqpoint{0.538116in}{0.420833in}}{\pgfqpoint{2.991805in}{2.004431in}}%
\pgfusepath{clip}%
\pgfsetroundcap%
\pgfsetroundjoin%
\pgfsetlinewidth{0.803000pt}%
\definecolor{currentstroke}{rgb}{0.800000,0.800000,0.800000}%
\pgfsetstrokecolor{currentstroke}%
\pgfsetdash{}{0pt}%
\pgfpathmoveto{\pgfqpoint{2.407994in}{0.420833in}}%
\pgfpathlineto{\pgfqpoint{2.407994in}{2.425264in}}%
\pgfusepath{stroke}%
\end{pgfscope}%
\begin{pgfscope}%
\definecolor{textcolor}{rgb}{0.150000,0.150000,0.150000}%
\pgfsetstrokecolor{textcolor}%
\pgfsetfillcolor{textcolor}%
\pgftext[x=2.407994in,y=0.305556in,,top]{\color{textcolor}\rmfamily\fontsize{8.800000}{10.560000}\selectfont \(\displaystyle {0.6}\)}%
\end{pgfscope}%
\begin{pgfscope}%
\pgfpathrectangle{\pgfqpoint{0.538116in}{0.420833in}}{\pgfqpoint{2.991805in}{2.004431in}}%
\pgfusepath{clip}%
\pgfsetroundcap%
\pgfsetroundjoin%
\pgfsetlinewidth{0.803000pt}%
\definecolor{currentstroke}{rgb}{0.800000,0.800000,0.800000}%
\pgfsetstrokecolor{currentstroke}%
\pgfsetdash{}{0pt}%
\pgfpathmoveto{\pgfqpoint{2.781970in}{0.420833in}}%
\pgfpathlineto{\pgfqpoint{2.781970in}{2.425264in}}%
\pgfusepath{stroke}%
\end{pgfscope}%
\begin{pgfscope}%
\definecolor{textcolor}{rgb}{0.150000,0.150000,0.150000}%
\pgfsetstrokecolor{textcolor}%
\pgfsetfillcolor{textcolor}%
\pgftext[x=2.781970in,y=0.305556in,,top]{\color{textcolor}\rmfamily\fontsize{8.800000}{10.560000}\selectfont \(\displaystyle {0.7}\)}%
\end{pgfscope}%
\begin{pgfscope}%
\pgfpathrectangle{\pgfqpoint{0.538116in}{0.420833in}}{\pgfqpoint{2.991805in}{2.004431in}}%
\pgfusepath{clip}%
\pgfsetroundcap%
\pgfsetroundjoin%
\pgfsetlinewidth{0.803000pt}%
\definecolor{currentstroke}{rgb}{0.800000,0.800000,0.800000}%
\pgfsetstrokecolor{currentstroke}%
\pgfsetdash{}{0pt}%
\pgfpathmoveto{\pgfqpoint{3.155945in}{0.420833in}}%
\pgfpathlineto{\pgfqpoint{3.155945in}{2.425264in}}%
\pgfusepath{stroke}%
\end{pgfscope}%
\begin{pgfscope}%
\definecolor{textcolor}{rgb}{0.150000,0.150000,0.150000}%
\pgfsetstrokecolor{textcolor}%
\pgfsetfillcolor{textcolor}%
\pgftext[x=3.155945in,y=0.305556in,,top]{\color{textcolor}\rmfamily\fontsize{8.800000}{10.560000}\selectfont \(\displaystyle {0.8}\)}%
\end{pgfscope}%
\begin{pgfscope}%
\pgfpathrectangle{\pgfqpoint{0.538116in}{0.420833in}}{\pgfqpoint{2.991805in}{2.004431in}}%
\pgfusepath{clip}%
\pgfsetroundcap%
\pgfsetroundjoin%
\pgfsetlinewidth{0.803000pt}%
\definecolor{currentstroke}{rgb}{0.800000,0.800000,0.800000}%
\pgfsetstrokecolor{currentstroke}%
\pgfsetdash{}{0pt}%
\pgfpathmoveto{\pgfqpoint{3.529921in}{0.420833in}}%
\pgfpathlineto{\pgfqpoint{3.529921in}{2.425264in}}%
\pgfusepath{stroke}%
\end{pgfscope}%
\begin{pgfscope}%
\definecolor{textcolor}{rgb}{0.150000,0.150000,0.150000}%
\pgfsetstrokecolor{textcolor}%
\pgfsetfillcolor{textcolor}%
\pgftext[x=3.529921in,y=0.305556in,,top]{\color{textcolor}\rmfamily\fontsize{8.800000}{10.560000}\selectfont \(\displaystyle {0.9}\)}%
\end{pgfscope}%
\begin{pgfscope}%
\definecolor{textcolor}{rgb}{0.150000,0.150000,0.150000}%
\pgfsetstrokecolor{textcolor}%
\pgfsetfillcolor{textcolor}%
\pgftext[x=2.034018in,y=0.138889in,,top]{\color{textcolor}\rmfamily\fontsize{9.600000}{11.520000}\selectfont \(\displaystyle r_1\) (leaf-balanced noise parameter)}%
\end{pgfscope}%
\begin{pgfscope}%
\pgfpathrectangle{\pgfqpoint{0.538116in}{0.420833in}}{\pgfqpoint{2.991805in}{2.004431in}}%
\pgfusepath{clip}%
\pgfsetroundcap%
\pgfsetroundjoin%
\pgfsetlinewidth{0.803000pt}%
\definecolor{currentstroke}{rgb}{0.800000,0.800000,0.800000}%
\pgfsetstrokecolor{currentstroke}%
\pgfsetdash{}{0pt}%
\pgfpathmoveto{\pgfqpoint{0.538116in}{0.420833in}}%
\pgfpathlineto{\pgfqpoint{3.529921in}{0.420833in}}%
\pgfusepath{stroke}%
\end{pgfscope}%
\begin{pgfscope}%
\definecolor{textcolor}{rgb}{0.150000,0.150000,0.150000}%
\pgfsetstrokecolor{textcolor}%
\pgfsetfillcolor{textcolor}%
\pgftext[x=0.194444in, y=0.377431in, left, base]{\color{textcolor}\rmfamily\fontsize{8.800000}{10.560000}\selectfont \(\displaystyle {2.70}\)}%
\end{pgfscope}%
\begin{pgfscope}%
\pgfpathrectangle{\pgfqpoint{0.538116in}{0.420833in}}{\pgfqpoint{2.991805in}{2.004431in}}%
\pgfusepath{clip}%
\pgfsetroundcap%
\pgfsetroundjoin%
\pgfsetlinewidth{0.803000pt}%
\definecolor{currentstroke}{rgb}{0.800000,0.800000,0.800000}%
\pgfsetstrokecolor{currentstroke}%
\pgfsetdash{}{0pt}%
\pgfpathmoveto{\pgfqpoint{0.538116in}{0.671387in}}%
\pgfpathlineto{\pgfqpoint{3.529921in}{0.671387in}}%
\pgfusepath{stroke}%
\end{pgfscope}%
\begin{pgfscope}%
\definecolor{textcolor}{rgb}{0.150000,0.150000,0.150000}%
\pgfsetstrokecolor{textcolor}%
\pgfsetfillcolor{textcolor}%
\pgftext[x=0.194444in, y=0.627984in, left, base]{\color{textcolor}\rmfamily\fontsize{8.800000}{10.560000}\selectfont \(\displaystyle {2.75}\)}%
\end{pgfscope}%
\begin{pgfscope}%
\pgfpathrectangle{\pgfqpoint{0.538116in}{0.420833in}}{\pgfqpoint{2.991805in}{2.004431in}}%
\pgfusepath{clip}%
\pgfsetroundcap%
\pgfsetroundjoin%
\pgfsetlinewidth{0.803000pt}%
\definecolor{currentstroke}{rgb}{0.800000,0.800000,0.800000}%
\pgfsetstrokecolor{currentstroke}%
\pgfsetdash{}{0pt}%
\pgfpathmoveto{\pgfqpoint{0.538116in}{0.921941in}}%
\pgfpathlineto{\pgfqpoint{3.529921in}{0.921941in}}%
\pgfusepath{stroke}%
\end{pgfscope}%
\begin{pgfscope}%
\definecolor{textcolor}{rgb}{0.150000,0.150000,0.150000}%
\pgfsetstrokecolor{textcolor}%
\pgfsetfillcolor{textcolor}%
\pgftext[x=0.194444in, y=0.878538in, left, base]{\color{textcolor}\rmfamily\fontsize{8.800000}{10.560000}\selectfont \(\displaystyle {2.80}\)}%
\end{pgfscope}%
\begin{pgfscope}%
\pgfpathrectangle{\pgfqpoint{0.538116in}{0.420833in}}{\pgfqpoint{2.991805in}{2.004431in}}%
\pgfusepath{clip}%
\pgfsetroundcap%
\pgfsetroundjoin%
\pgfsetlinewidth{0.803000pt}%
\definecolor{currentstroke}{rgb}{0.800000,0.800000,0.800000}%
\pgfsetstrokecolor{currentstroke}%
\pgfsetdash{}{0pt}%
\pgfpathmoveto{\pgfqpoint{0.538116in}{1.172495in}}%
\pgfpathlineto{\pgfqpoint{3.529921in}{1.172495in}}%
\pgfusepath{stroke}%
\end{pgfscope}%
\begin{pgfscope}%
\definecolor{textcolor}{rgb}{0.150000,0.150000,0.150000}%
\pgfsetstrokecolor{textcolor}%
\pgfsetfillcolor{textcolor}%
\pgftext[x=0.194444in, y=1.129092in, left, base]{\color{textcolor}\rmfamily\fontsize{8.800000}{10.560000}\selectfont \(\displaystyle {2.85}\)}%
\end{pgfscope}%
\begin{pgfscope}%
\pgfpathrectangle{\pgfqpoint{0.538116in}{0.420833in}}{\pgfqpoint{2.991805in}{2.004431in}}%
\pgfusepath{clip}%
\pgfsetroundcap%
\pgfsetroundjoin%
\pgfsetlinewidth{0.803000pt}%
\definecolor{currentstroke}{rgb}{0.800000,0.800000,0.800000}%
\pgfsetstrokecolor{currentstroke}%
\pgfsetdash{}{0pt}%
\pgfpathmoveto{\pgfqpoint{0.538116in}{1.423049in}}%
\pgfpathlineto{\pgfqpoint{3.529921in}{1.423049in}}%
\pgfusepath{stroke}%
\end{pgfscope}%
\begin{pgfscope}%
\definecolor{textcolor}{rgb}{0.150000,0.150000,0.150000}%
\pgfsetstrokecolor{textcolor}%
\pgfsetfillcolor{textcolor}%
\pgftext[x=0.194444in, y=1.379646in, left, base]{\color{textcolor}\rmfamily\fontsize{8.800000}{10.560000}\selectfont \(\displaystyle {2.90}\)}%
\end{pgfscope}%
\begin{pgfscope}%
\pgfpathrectangle{\pgfqpoint{0.538116in}{0.420833in}}{\pgfqpoint{2.991805in}{2.004431in}}%
\pgfusepath{clip}%
\pgfsetroundcap%
\pgfsetroundjoin%
\pgfsetlinewidth{0.803000pt}%
\definecolor{currentstroke}{rgb}{0.800000,0.800000,0.800000}%
\pgfsetstrokecolor{currentstroke}%
\pgfsetdash{}{0pt}%
\pgfpathmoveto{\pgfqpoint{0.538116in}{1.673602in}}%
\pgfpathlineto{\pgfqpoint{3.529921in}{1.673602in}}%
\pgfusepath{stroke}%
\end{pgfscope}%
\begin{pgfscope}%
\definecolor{textcolor}{rgb}{0.150000,0.150000,0.150000}%
\pgfsetstrokecolor{textcolor}%
\pgfsetfillcolor{textcolor}%
\pgftext[x=0.194444in, y=1.630200in, left, base]{\color{textcolor}\rmfamily\fontsize{8.800000}{10.560000}\selectfont \(\displaystyle {2.95}\)}%
\end{pgfscope}%
\begin{pgfscope}%
\pgfpathrectangle{\pgfqpoint{0.538116in}{0.420833in}}{\pgfqpoint{2.991805in}{2.004431in}}%
\pgfusepath{clip}%
\pgfsetroundcap%
\pgfsetroundjoin%
\pgfsetlinewidth{0.803000pt}%
\definecolor{currentstroke}{rgb}{0.800000,0.800000,0.800000}%
\pgfsetstrokecolor{currentstroke}%
\pgfsetdash{}{0pt}%
\pgfpathmoveto{\pgfqpoint{0.538116in}{1.924156in}}%
\pgfpathlineto{\pgfqpoint{3.529921in}{1.924156in}}%
\pgfusepath{stroke}%
\end{pgfscope}%
\begin{pgfscope}%
\definecolor{textcolor}{rgb}{0.150000,0.150000,0.150000}%
\pgfsetstrokecolor{textcolor}%
\pgfsetfillcolor{textcolor}%
\pgftext[x=0.194444in, y=1.880753in, left, base]{\color{textcolor}\rmfamily\fontsize{8.800000}{10.560000}\selectfont \(\displaystyle {3.00}\)}%
\end{pgfscope}%
\begin{pgfscope}%
\pgfpathrectangle{\pgfqpoint{0.538116in}{0.420833in}}{\pgfqpoint{2.991805in}{2.004431in}}%
\pgfusepath{clip}%
\pgfsetroundcap%
\pgfsetroundjoin%
\pgfsetlinewidth{0.803000pt}%
\definecolor{currentstroke}{rgb}{0.800000,0.800000,0.800000}%
\pgfsetstrokecolor{currentstroke}%
\pgfsetdash{}{0pt}%
\pgfpathmoveto{\pgfqpoint{0.538116in}{2.174710in}}%
\pgfpathlineto{\pgfqpoint{3.529921in}{2.174710in}}%
\pgfusepath{stroke}%
\end{pgfscope}%
\begin{pgfscope}%
\definecolor{textcolor}{rgb}{0.150000,0.150000,0.150000}%
\pgfsetstrokecolor{textcolor}%
\pgfsetfillcolor{textcolor}%
\pgftext[x=0.194444in, y=2.131307in, left, base]{\color{textcolor}\rmfamily\fontsize{8.800000}{10.560000}\selectfont \(\displaystyle {3.05}\)}%
\end{pgfscope}%
\begin{pgfscope}%
\pgfpathrectangle{\pgfqpoint{0.538116in}{0.420833in}}{\pgfqpoint{2.991805in}{2.004431in}}%
\pgfusepath{clip}%
\pgfsetroundcap%
\pgfsetroundjoin%
\pgfsetlinewidth{0.803000pt}%
\definecolor{currentstroke}{rgb}{0.800000,0.800000,0.800000}%
\pgfsetstrokecolor{currentstroke}%
\pgfsetdash{}{0pt}%
\pgfpathmoveto{\pgfqpoint{0.538116in}{2.425264in}}%
\pgfpathlineto{\pgfqpoint{3.529921in}{2.425264in}}%
\pgfusepath{stroke}%
\end{pgfscope}%
\begin{pgfscope}%
\definecolor{textcolor}{rgb}{0.150000,0.150000,0.150000}%
\pgfsetstrokecolor{textcolor}%
\pgfsetfillcolor{textcolor}%
\pgftext[x=0.194444in, y=2.381861in, left, base]{\color{textcolor}\rmfamily\fontsize{8.800000}{10.560000}\selectfont \(\displaystyle {3.10}\)}%
\end{pgfscope}%
\begin{pgfscope}%
\definecolor{textcolor}{rgb}{0.150000,0.150000,0.150000}%
\pgfsetstrokecolor{textcolor}%
\pgfsetfillcolor{textcolor}%
\pgftext[x=0.138889in,y=1.423049in,,bottom,rotate=90.000000]{\color{textcolor}\rmfamily\fontsize{9.600000}{11.520000}\selectfont Mean test regression error (RMSE)}%
\end{pgfscope}%
\begin{pgfscope}%
\pgfpathrectangle{\pgfqpoint{0.538116in}{0.420833in}}{\pgfqpoint{2.991805in}{2.004431in}}%
\pgfusepath{clip}%
\pgfsetbuttcap%
\pgfsetroundjoin%
\definecolor{currentfill}{rgb}{0.121569,0.466667,0.705882}%
\pgfsetfillcolor{currentfill}%
\pgfsetlinewidth{0.803000pt}%
\definecolor{currentstroke}{rgb}{0.121569,0.466667,0.705882}%
\pgfsetstrokecolor{currentstroke}%
\pgfsetdash{}{0pt}%
\pgfsys@defobject{currentmarker}{\pgfqpoint{-0.038036in}{-0.038036in}}{\pgfqpoint{0.038036in}{0.038036in}}{%
\pgfpathmoveto{\pgfqpoint{0.000000in}{-0.038036in}}%
\pgfpathcurveto{\pgfqpoint{0.010087in}{-0.038036in}}{\pgfqpoint{0.019763in}{-0.034029in}}{\pgfqpoint{0.026896in}{-0.026896in}}%
\pgfpathcurveto{\pgfqpoint{0.034029in}{-0.019763in}}{\pgfqpoint{0.038036in}{-0.010087in}}{\pgfqpoint{0.038036in}{0.000000in}}%
\pgfpathcurveto{\pgfqpoint{0.038036in}{0.010087in}}{\pgfqpoint{0.034029in}{0.019763in}}{\pgfqpoint{0.026896in}{0.026896in}}%
\pgfpathcurveto{\pgfqpoint{0.019763in}{0.034029in}}{\pgfqpoint{0.010087in}{0.038036in}}{\pgfqpoint{0.000000in}{0.038036in}}%
\pgfpathcurveto{\pgfqpoint{-0.010087in}{0.038036in}}{\pgfqpoint{-0.019763in}{0.034029in}}{\pgfqpoint{-0.026896in}{0.026896in}}%
\pgfpathcurveto{\pgfqpoint{-0.034029in}{0.019763in}}{\pgfqpoint{-0.038036in}{0.010087in}}{\pgfqpoint{-0.038036in}{0.000000in}}%
\pgfpathcurveto{\pgfqpoint{-0.038036in}{-0.010087in}}{\pgfqpoint{-0.034029in}{-0.019763in}}{\pgfqpoint{-0.026896in}{-0.026896in}}%
\pgfpathcurveto{\pgfqpoint{-0.019763in}{-0.034029in}}{\pgfqpoint{-0.010087in}{-0.038036in}}{\pgfqpoint{0.000000in}{-0.038036in}}%
\pgfpathlineto{\pgfqpoint{0.000000in}{-0.038036in}}%
\pgfpathclose%
\pgfusepath{stroke,fill}%
}%
\begin{pgfscope}%
\pgfsys@transformshift{2.034018in}{1.618481in}%
\pgfsys@useobject{currentmarker}{}%
\end{pgfscope}%
\end{pgfscope}%
\begin{pgfscope}%
\pgfpathrectangle{\pgfqpoint{0.538116in}{0.420833in}}{\pgfqpoint{2.991805in}{2.004431in}}%
\pgfusepath{clip}%
\pgfsetbuttcap%
\pgfsetroundjoin%
\definecolor{currentfill}{rgb}{1.000000,0.498039,0.054902}%
\pgfsetfillcolor{currentfill}%
\pgfsetfillopacity{0.100000}%
\pgfsetlinewidth{0.803000pt}%
\definecolor{currentstroke}{rgb}{1.000000,0.498039,0.054902}%
\pgfsetstrokecolor{currentstroke}%
\pgfsetstrokeopacity{0.100000}%
\pgfsetdash{}{0pt}%
\pgfsys@defobject{currentmarker}{\pgfqpoint{0.313730in}{0.626287in}}{\pgfqpoint{3.529921in}{3.011560in}}{%
\pgfpathmoveto{\pgfqpoint{0.313730in}{0.916930in}}%
\pgfpathlineto{\pgfqpoint{0.313730in}{0.826731in}}%
\pgfpathlineto{\pgfqpoint{0.538116in}{0.736531in}}%
\pgfpathlineto{\pgfqpoint{0.912091in}{0.626287in}}%
\pgfpathlineto{\pgfqpoint{1.286067in}{0.681409in}}%
\pgfpathlineto{\pgfqpoint{1.660043in}{0.726509in}}%
\pgfpathlineto{\pgfqpoint{2.034018in}{0.826731in}}%
\pgfpathlineto{\pgfqpoint{2.407994in}{0.931963in}}%
\pgfpathlineto{\pgfqpoint{2.781970in}{1.092318in}}%
\pgfpathlineto{\pgfqpoint{3.155945in}{1.588414in}}%
\pgfpathlineto{\pgfqpoint{3.529921in}{2.841183in}}%
\pgfpathlineto{\pgfqpoint{3.529921in}{3.011560in}}%
\pgfpathlineto{\pgfqpoint{3.529921in}{3.011560in}}%
\pgfpathlineto{\pgfqpoint{3.155945in}{1.758791in}}%
\pgfpathlineto{\pgfqpoint{2.781970in}{1.252672in}}%
\pgfpathlineto{\pgfqpoint{2.407994in}{1.012140in}}%
\pgfpathlineto{\pgfqpoint{2.034018in}{0.916930in}}%
\pgfpathlineto{\pgfqpoint{1.660043in}{0.816708in}}%
\pgfpathlineto{\pgfqpoint{1.286067in}{0.761587in}}%
\pgfpathlineto{\pgfqpoint{0.912091in}{0.716487in}}%
\pgfpathlineto{\pgfqpoint{0.538116in}{0.806686in}}%
\pgfpathlineto{\pgfqpoint{0.313730in}{0.916930in}}%
\pgfpathlineto{\pgfqpoint{0.313730in}{0.916930in}}%
\pgfpathclose%
\pgfusepath{stroke,fill}%
}%
\begin{pgfscope}%
\pgfsys@transformshift{0.000000in}{0.000000in}%
\pgfsys@useobject{currentmarker}{}%
\end{pgfscope}%
\end{pgfscope}%
\begin{pgfscope}%
\pgfpathrectangle{\pgfqpoint{0.538116in}{0.420833in}}{\pgfqpoint{2.991805in}{2.004431in}}%
\pgfusepath{clip}%
\pgfsetroundcap%
\pgfsetroundjoin%
\pgfsetlinewidth{1.204500pt}%
\definecolor{currentstroke}{rgb}{1.000000,0.498039,0.054902}%
\pgfsetstrokecolor{currentstroke}%
\pgfsetdash{}{0pt}%
\pgfpathmoveto{\pgfqpoint{0.528116in}{0.776075in}}%
\pgfpathlineto{\pgfqpoint{0.538116in}{0.771609in}}%
\pgfpathlineto{\pgfqpoint{0.912091in}{0.671387in}}%
\pgfpathlineto{\pgfqpoint{1.286067in}{0.721498in}}%
\pgfpathlineto{\pgfqpoint{1.660043in}{0.771609in}}%
\pgfpathlineto{\pgfqpoint{2.034018in}{0.871830in}}%
\pgfpathlineto{\pgfqpoint{2.407994in}{0.972052in}}%
\pgfpathlineto{\pgfqpoint{2.781970in}{1.172495in}}%
\pgfpathlineto{\pgfqpoint{3.155945in}{1.673602in}}%
\pgfpathlineto{\pgfqpoint{3.383316in}{2.435264in}}%
\pgfusepath{stroke}%
\end{pgfscope}%
\begin{pgfscope}%
\pgfsetrectcap%
\pgfsetmiterjoin%
\pgfsetlinewidth{1.003750pt}%
\definecolor{currentstroke}{rgb}{0.800000,0.800000,0.800000}%
\pgfsetstrokecolor{currentstroke}%
\pgfsetdash{}{0pt}%
\pgfpathmoveto{\pgfqpoint{0.538116in}{0.420833in}}%
\pgfpathlineto{\pgfqpoint{0.538116in}{2.425264in}}%
\pgfusepath{stroke}%
\end{pgfscope}%
\begin{pgfscope}%
\pgfsetrectcap%
\pgfsetmiterjoin%
\pgfsetlinewidth{1.003750pt}%
\definecolor{currentstroke}{rgb}{0.800000,0.800000,0.800000}%
\pgfsetstrokecolor{currentstroke}%
\pgfsetdash{}{0pt}%
\pgfpathmoveto{\pgfqpoint{0.538116in}{0.420833in}}%
\pgfpathlineto{\pgfqpoint{3.529921in}{0.420833in}}%
\pgfusepath{stroke}%
\end{pgfscope}%
\begin{pgfscope}%
\pgfsetroundcap%
\pgfsetroundjoin%
\definecolor{currentfill}{rgb}{0.862745,0.862745,0.862745}%
\pgfsetfillcolor{currentfill}%
\pgfsetlinewidth{0.803000pt}%
\definecolor{currentstroke}{rgb}{1.000000,1.000000,1.000000}%
\pgfsetstrokecolor{currentstroke}%
\pgfsetdash{}{0pt}%
\pgfpathmoveto{\pgfqpoint{0.946814in}{1.636029in}}%
\pgfpathquadraticcurveto{\pgfqpoint{0.946814in}{1.381133in}}{\pgfqpoint{0.946814in}{1.126237in}}%
\pgfpathlineto{\pgfqpoint{0.995425in}{1.126237in}}%
\pgfpathquadraticcurveto{\pgfqpoint{0.953758in}{1.042896in}}{\pgfqpoint{0.912091in}{0.959554in}}%
\pgfpathquadraticcurveto{\pgfqpoint{0.870425in}{1.042896in}}{\pgfqpoint{0.828758in}{1.126237in}}%
\pgfpathlineto{\pgfqpoint{0.877369in}{1.126237in}}%
\pgfpathquadraticcurveto{\pgfqpoint{0.877369in}{1.381133in}}{\pgfqpoint{0.877369in}{1.636029in}}%
\pgfpathlineto{\pgfqpoint{0.946814in}{1.636029in}}%
\pgfpathlineto{\pgfqpoint{0.946814in}{1.636029in}}%
\pgfpathclose%
\pgfusepath{stroke,fill}%
\end{pgfscope}%
\begin{pgfscope}%
\definecolor{textcolor}{rgb}{0.862745,0.862745,0.862745}%
\pgfsetstrokecolor{textcolor}%
\pgfsetfillcolor{textcolor}%
\pgftext[x=1.099079in,y=1.423049in,left,]{\color{textcolor}\rmfamily\fontsize{12.000000}{14.400000}\selectfont better}%
\end{pgfscope}%
\begin{pgfscope}%
\pgfsetbuttcap%
\pgfsetmiterjoin%
\definecolor{currentfill}{rgb}{1.000000,1.000000,1.000000}%
\pgfsetfillcolor{currentfill}%
\pgfsetfillopacity{0.800000}%
\pgfsetlinewidth{0.803000pt}%
\definecolor{currentstroke}{rgb}{0.800000,0.800000,0.800000}%
\pgfsetstrokecolor{currentstroke}%
\pgfsetstrokeopacity{0.800000}%
\pgfsetdash{}{0pt}%
\pgfpathmoveto{\pgfqpoint{0.623671in}{1.983042in}}%
\pgfpathlineto{\pgfqpoint{1.857012in}{1.983042in}}%
\pgfpathquadraticcurveto{\pgfqpoint{1.881457in}{1.983042in}}{\pgfqpoint{1.881457in}{2.007486in}}%
\pgfpathlineto{\pgfqpoint{1.881457in}{2.339708in}}%
\pgfpathquadraticcurveto{\pgfqpoint{1.881457in}{2.364153in}}{\pgfqpoint{1.857012in}{2.364153in}}%
\pgfpathlineto{\pgfqpoint{0.623671in}{2.364153in}}%
\pgfpathquadraticcurveto{\pgfqpoint{0.599227in}{2.364153in}}{\pgfqpoint{0.599227in}{2.339708in}}%
\pgfpathlineto{\pgfqpoint{0.599227in}{2.007486in}}%
\pgfpathquadraticcurveto{\pgfqpoint{0.599227in}{1.983042in}}{\pgfqpoint{0.623671in}{1.983042in}}%
\pgfpathlineto{\pgfqpoint{0.623671in}{1.983042in}}%
\pgfpathclose%
\pgfusepath{stroke,fill}%
\end{pgfscope}%
\begin{pgfscope}%
\pgfsetbuttcap%
\pgfsetroundjoin%
\definecolor{currentfill}{rgb}{0.121569,0.466667,0.705882}%
\pgfsetfillcolor{currentfill}%
\pgfsetlinewidth{0.803000pt}%
\definecolor{currentstroke}{rgb}{0.121569,0.466667,0.705882}%
\pgfsetstrokecolor{currentstroke}%
\pgfsetdash{}{0pt}%
\pgfsys@defobject{currentmarker}{\pgfqpoint{-0.038036in}{-0.038036in}}{\pgfqpoint{0.038036in}{0.038036in}}{%
\pgfpathmoveto{\pgfqpoint{0.000000in}{-0.038036in}}%
\pgfpathcurveto{\pgfqpoint{0.010087in}{-0.038036in}}{\pgfqpoint{0.019763in}{-0.034029in}}{\pgfqpoint{0.026896in}{-0.026896in}}%
\pgfpathcurveto{\pgfqpoint{0.034029in}{-0.019763in}}{\pgfqpoint{0.038036in}{-0.010087in}}{\pgfqpoint{0.038036in}{0.000000in}}%
\pgfpathcurveto{\pgfqpoint{0.038036in}{0.010087in}}{\pgfqpoint{0.034029in}{0.019763in}}{\pgfqpoint{0.026896in}{0.026896in}}%
\pgfpathcurveto{\pgfqpoint{0.019763in}{0.034029in}}{\pgfqpoint{0.010087in}{0.038036in}}{\pgfqpoint{0.000000in}{0.038036in}}%
\pgfpathcurveto{\pgfqpoint{-0.010087in}{0.038036in}}{\pgfqpoint{-0.019763in}{0.034029in}}{\pgfqpoint{-0.026896in}{0.026896in}}%
\pgfpathcurveto{\pgfqpoint{-0.034029in}{0.019763in}}{\pgfqpoint{-0.038036in}{0.010087in}}{\pgfqpoint{-0.038036in}{0.000000in}}%
\pgfpathcurveto{\pgfqpoint{-0.038036in}{-0.010087in}}{\pgfqpoint{-0.034029in}{-0.019763in}}{\pgfqpoint{-0.026896in}{-0.026896in}}%
\pgfpathcurveto{\pgfqpoint{-0.019763in}{-0.034029in}}{\pgfqpoint{-0.010087in}{-0.038036in}}{\pgfqpoint{0.000000in}{-0.038036in}}%
\pgfpathlineto{\pgfqpoint{0.000000in}{-0.038036in}}%
\pgfpathclose%
\pgfusepath{stroke,fill}%
}%
\begin{pgfscope}%
\pgfsys@transformshift{0.770338in}{2.260542in}%
\pgfsys@useobject{currentmarker}{}%
\end{pgfscope}%
\end{pgfscope}%
\begin{pgfscope}%
\definecolor{textcolor}{rgb}{0.150000,0.150000,0.150000}%
\pgfsetstrokecolor{textcolor}%
\pgfsetfillcolor{textcolor}%
\pgftext[x=0.990338in,y=2.228458in,left,base]{\color{textcolor}\rmfamily\fontsize{8.800000}{10.560000}\selectfont Maddock et al.}%
\end{pgfscope}%
\begin{pgfscope}%
\pgfsetroundcap%
\pgfsetroundjoin%
\pgfsetlinewidth{1.204500pt}%
\definecolor{currentstroke}{rgb}{1.000000,0.498039,0.054902}%
\pgfsetstrokecolor{currentstroke}%
\pgfsetdash{}{0pt}%
\pgfpathmoveto{\pgfqpoint{0.648116in}{2.099014in}}%
\pgfpathlineto{\pgfqpoint{0.770338in}{2.099014in}}%
\pgfpathlineto{\pgfqpoint{0.892560in}{2.099014in}}%
\pgfusepath{stroke}%
\end{pgfscope}%
\begin{pgfscope}%
\definecolor{textcolor}{rgb}{0.150000,0.150000,0.150000}%
\pgfsetstrokecolor{textcolor}%
\pgfsetfillcolor{textcolor}%
\pgftext[x=0.990338in,y=2.056236in,left,base]{\color{textcolor}\rmfamily\fontsize{8.800000}{10.560000}\selectfont S-BDT}%
\end{pgfscope}%
\end{pgfpicture}%
\makeatother%
\endgroup%

%% file: images/abalone_subsampling_ablation.pgf
\begingroup%
\makeatletter%
\begin{pgfpicture}%
\pgfpathrectangle{\pgfpointorigin}{\pgfqpoint{3.612000in}{2.495774in}}%
\pgfusepath{use as bounding box, clip}%
\begin{pgfscope}%
\pgfsetbuttcap%
\pgfsetmiterjoin%
\definecolor{currentfill}{rgb}{1.000000,1.000000,1.000000}%
\pgfsetfillcolor{currentfill}%
\pgfsetlinewidth{0.000000pt}%
\definecolor{currentstroke}{rgb}{1.000000,1.000000,1.000000}%
\pgfsetstrokecolor{currentstroke}%
\pgfsetdash{}{0pt}%
\pgfpathmoveto{\pgfqpoint{0.000000in}{0.000000in}}%
\pgfpathlineto{\pgfqpoint{3.612000in}{0.000000in}}%
\pgfpathlineto{\pgfqpoint{3.612000in}{2.495774in}}%
\pgfpathlineto{\pgfqpoint{0.000000in}{2.495774in}}%
\pgfpathlineto{\pgfqpoint{0.000000in}{0.000000in}}%
\pgfpathclose%
\pgfusepath{fill}%
\end{pgfscope}%
\begin{pgfscope}%
\pgfsetbuttcap%
\pgfsetmiterjoin%
\definecolor{currentfill}{rgb}{1.000000,1.000000,1.000000}%
\pgfsetfillcolor{currentfill}%
\pgfsetlinewidth{0.000000pt}%
\definecolor{currentstroke}{rgb}{0.000000,0.000000,0.000000}%
\pgfsetstrokecolor{currentstroke}%
\pgfsetstrokeopacity{0.000000}%
\pgfsetdash{}{0pt}%
\pgfpathmoveto{\pgfqpoint{0.538116in}{0.420833in}}%
\pgfpathlineto{\pgfqpoint{3.497803in}{0.420833in}}%
\pgfpathlineto{\pgfqpoint{3.497803in}{2.425264in}}%
\pgfpathlineto{\pgfqpoint{0.538116in}{2.425264in}}%
\pgfpathlineto{\pgfqpoint{0.538116in}{0.420833in}}%
\pgfpathclose%
\pgfusepath{fill}%
\end{pgfscope}%
\begin{pgfscope}%
\pgfpathrectangle{\pgfqpoint{0.538116in}{0.420833in}}{\pgfqpoint{2.959687in}{2.004431in}}%
\pgfusepath{clip}%
\pgfsetroundcap%
\pgfsetroundjoin%
\pgfsetlinewidth{0.803000pt}%
\definecolor{currentstroke}{rgb}{0.800000,0.800000,0.800000}%
\pgfsetstrokecolor{currentstroke}%
\pgfsetdash{}{0pt}%
\pgfpathmoveto{\pgfqpoint{0.567713in}{0.420833in}}%
\pgfpathlineto{\pgfqpoint{0.567713in}{2.425264in}}%
\pgfusepath{stroke}%
\end{pgfscope}%
\begin{pgfscope}%
\definecolor{textcolor}{rgb}{0.150000,0.150000,0.150000}%
\pgfsetstrokecolor{textcolor}%
\pgfsetfillcolor{textcolor}%
\pgftext[x=0.567713in,y=0.305556in,,top]{\color{textcolor}\rmfamily\fontsize{8.800000}{10.560000}\selectfont \(\displaystyle {0.01}\)}%
\end{pgfscope}%
\begin{pgfscope}%
\pgfpathrectangle{\pgfqpoint{0.538116in}{0.420833in}}{\pgfqpoint{2.959687in}{2.004431in}}%
\pgfusepath{clip}%
\pgfsetroundcap%
\pgfsetroundjoin%
\pgfsetlinewidth{0.803000pt}%
\definecolor{currentstroke}{rgb}{0.800000,0.800000,0.800000}%
\pgfsetstrokecolor{currentstroke}%
\pgfsetdash{}{0pt}%
\pgfpathmoveto{\pgfqpoint{1.130053in}{0.420833in}}%
\pgfpathlineto{\pgfqpoint{1.130053in}{2.425264in}}%
\pgfusepath{stroke}%
\end{pgfscope}%
\begin{pgfscope}%
\definecolor{textcolor}{rgb}{0.150000,0.150000,0.150000}%
\pgfsetstrokecolor{textcolor}%
\pgfsetfillcolor{textcolor}%
\pgftext[x=1.130053in,y=0.305556in,,top]{\color{textcolor}\rmfamily\fontsize{8.800000}{10.560000}\selectfont \(\displaystyle {0.20}\)}%
\end{pgfscope}%
\begin{pgfscope}%
\pgfpathrectangle{\pgfqpoint{0.538116in}{0.420833in}}{\pgfqpoint{2.959687in}{2.004431in}}%
\pgfusepath{clip}%
\pgfsetroundcap%
\pgfsetroundjoin%
\pgfsetlinewidth{0.803000pt}%
\definecolor{currentstroke}{rgb}{0.800000,0.800000,0.800000}%
\pgfsetstrokecolor{currentstroke}%
\pgfsetdash{}{0pt}%
\pgfpathmoveto{\pgfqpoint{1.721991in}{0.420833in}}%
\pgfpathlineto{\pgfqpoint{1.721991in}{2.425264in}}%
\pgfusepath{stroke}%
\end{pgfscope}%
\begin{pgfscope}%
\definecolor{textcolor}{rgb}{0.150000,0.150000,0.150000}%
\pgfsetstrokecolor{textcolor}%
\pgfsetfillcolor{textcolor}%
\pgftext[x=1.721991in,y=0.305556in,,top]{\color{textcolor}\rmfamily\fontsize{8.800000}{10.560000}\selectfont \(\displaystyle {0.40}\)}%
\end{pgfscope}%
\begin{pgfscope}%
\pgfpathrectangle{\pgfqpoint{0.538116in}{0.420833in}}{\pgfqpoint{2.959687in}{2.004431in}}%
\pgfusepath{clip}%
\pgfsetroundcap%
\pgfsetroundjoin%
\pgfsetlinewidth{0.803000pt}%
\definecolor{currentstroke}{rgb}{0.800000,0.800000,0.800000}%
\pgfsetstrokecolor{currentstroke}%
\pgfsetdash{}{0pt}%
\pgfpathmoveto{\pgfqpoint{2.313928in}{0.420833in}}%
\pgfpathlineto{\pgfqpoint{2.313928in}{2.425264in}}%
\pgfusepath{stroke}%
\end{pgfscope}%
\begin{pgfscope}%
\definecolor{textcolor}{rgb}{0.150000,0.150000,0.150000}%
\pgfsetstrokecolor{textcolor}%
\pgfsetfillcolor{textcolor}%
\pgftext[x=2.313928in,y=0.305556in,,top]{\color{textcolor}\rmfamily\fontsize{8.800000}{10.560000}\selectfont \(\displaystyle {0.60}\)}%
\end{pgfscope}%
\begin{pgfscope}%
\pgfpathrectangle{\pgfqpoint{0.538116in}{0.420833in}}{\pgfqpoint{2.959687in}{2.004431in}}%
\pgfusepath{clip}%
\pgfsetroundcap%
\pgfsetroundjoin%
\pgfsetlinewidth{0.803000pt}%
\definecolor{currentstroke}{rgb}{0.800000,0.800000,0.800000}%
\pgfsetstrokecolor{currentstroke}%
\pgfsetdash{}{0pt}%
\pgfpathmoveto{\pgfqpoint{2.905866in}{0.420833in}}%
\pgfpathlineto{\pgfqpoint{2.905866in}{2.425264in}}%
\pgfusepath{stroke}%
\end{pgfscope}%
\begin{pgfscope}%
\definecolor{textcolor}{rgb}{0.150000,0.150000,0.150000}%
\pgfsetstrokecolor{textcolor}%
\pgfsetfillcolor{textcolor}%
\pgftext[x=2.905866in,y=0.305556in,,top]{\color{textcolor}\rmfamily\fontsize{8.800000}{10.560000}\selectfont \(\displaystyle {0.80}\)}%
\end{pgfscope}%
\begin{pgfscope}%
\pgfpathrectangle{\pgfqpoint{0.538116in}{0.420833in}}{\pgfqpoint{2.959687in}{2.004431in}}%
\pgfusepath{clip}%
\pgfsetroundcap%
\pgfsetroundjoin%
\pgfsetlinewidth{0.803000pt}%
\definecolor{currentstroke}{rgb}{0.800000,0.800000,0.800000}%
\pgfsetstrokecolor{currentstroke}%
\pgfsetdash{}{0pt}%
\pgfpathmoveto{\pgfqpoint{3.497803in}{0.420833in}}%
\pgfpathlineto{\pgfqpoint{3.497803in}{2.425264in}}%
\pgfusepath{stroke}%
\end{pgfscope}%
\begin{pgfscope}%
\definecolor{textcolor}{rgb}{0.150000,0.150000,0.150000}%
\pgfsetstrokecolor{textcolor}%
\pgfsetfillcolor{textcolor}%
\pgftext[x=3.497803in,y=0.305556in,,top]{\color{textcolor}\rmfamily\fontsize{8.800000}{10.560000}\selectfont \(\displaystyle {1.00}\)}%
\end{pgfscope}%
\begin{pgfscope}%
\definecolor{textcolor}{rgb}{0.150000,0.150000,0.150000}%
\pgfsetstrokecolor{textcolor}%
\pgfsetfillcolor{textcolor}%
\pgftext[x=2.017960in,y=0.138889in,,top]{\color{textcolor}\rmfamily\fontsize{9.600000}{11.520000}\selectfont \(\displaystyle \gamma\) (subsampling ratio)}%
\end{pgfscope}%
\begin{pgfscope}%
\pgfpathrectangle{\pgfqpoint{0.538116in}{0.420833in}}{\pgfqpoint{2.959687in}{2.004431in}}%
\pgfusepath{clip}%
\pgfsetroundcap%
\pgfsetroundjoin%
\pgfsetlinewidth{0.803000pt}%
\definecolor{currentstroke}{rgb}{0.800000,0.800000,0.800000}%
\pgfsetstrokecolor{currentstroke}%
\pgfsetdash{}{0pt}%
\pgfpathmoveto{\pgfqpoint{0.538116in}{0.420833in}}%
\pgfpathlineto{\pgfqpoint{3.497803in}{0.420833in}}%
\pgfusepath{stroke}%
\end{pgfscope}%
\begin{pgfscope}%
\definecolor{textcolor}{rgb}{0.150000,0.150000,0.150000}%
\pgfsetstrokecolor{textcolor}%
\pgfsetfillcolor{textcolor}%
\pgftext[x=0.194444in, y=0.377431in, left, base]{\color{textcolor}\rmfamily\fontsize{8.800000}{10.560000}\selectfont \(\displaystyle {2.70}\)}%
\end{pgfscope}%
\begin{pgfscope}%
\pgfpathrectangle{\pgfqpoint{0.538116in}{0.420833in}}{\pgfqpoint{2.959687in}{2.004431in}}%
\pgfusepath{clip}%
\pgfsetroundcap%
\pgfsetroundjoin%
\pgfsetlinewidth{0.803000pt}%
\definecolor{currentstroke}{rgb}{0.800000,0.800000,0.800000}%
\pgfsetstrokecolor{currentstroke}%
\pgfsetdash{}{0pt}%
\pgfpathmoveto{\pgfqpoint{0.538116in}{0.671387in}}%
\pgfpathlineto{\pgfqpoint{3.497803in}{0.671387in}}%
\pgfusepath{stroke}%
\end{pgfscope}%
\begin{pgfscope}%
\definecolor{textcolor}{rgb}{0.150000,0.150000,0.150000}%
\pgfsetstrokecolor{textcolor}%
\pgfsetfillcolor{textcolor}%
\pgftext[x=0.194444in, y=0.627984in, left, base]{\color{textcolor}\rmfamily\fontsize{8.800000}{10.560000}\selectfont \(\displaystyle {2.75}\)}%
\end{pgfscope}%
\begin{pgfscope}%
\pgfpathrectangle{\pgfqpoint{0.538116in}{0.420833in}}{\pgfqpoint{2.959687in}{2.004431in}}%
\pgfusepath{clip}%
\pgfsetroundcap%
\pgfsetroundjoin%
\pgfsetlinewidth{0.803000pt}%
\definecolor{currentstroke}{rgb}{0.800000,0.800000,0.800000}%
\pgfsetstrokecolor{currentstroke}%
\pgfsetdash{}{0pt}%
\pgfpathmoveto{\pgfqpoint{0.538116in}{0.921941in}}%
\pgfpathlineto{\pgfqpoint{3.497803in}{0.921941in}}%
\pgfusepath{stroke}%
\end{pgfscope}%
\begin{pgfscope}%
\definecolor{textcolor}{rgb}{0.150000,0.150000,0.150000}%
\pgfsetstrokecolor{textcolor}%
\pgfsetfillcolor{textcolor}%
\pgftext[x=0.194444in, y=0.878538in, left, base]{\color{textcolor}\rmfamily\fontsize{8.800000}{10.560000}\selectfont \(\displaystyle {2.80}\)}%
\end{pgfscope}%
\begin{pgfscope}%
\pgfpathrectangle{\pgfqpoint{0.538116in}{0.420833in}}{\pgfqpoint{2.959687in}{2.004431in}}%
\pgfusepath{clip}%
\pgfsetroundcap%
\pgfsetroundjoin%
\pgfsetlinewidth{0.803000pt}%
\definecolor{currentstroke}{rgb}{0.800000,0.800000,0.800000}%
\pgfsetstrokecolor{currentstroke}%
\pgfsetdash{}{0pt}%
\pgfpathmoveto{\pgfqpoint{0.538116in}{1.172495in}}%
\pgfpathlineto{\pgfqpoint{3.497803in}{1.172495in}}%
\pgfusepath{stroke}%
\end{pgfscope}%
\begin{pgfscope}%
\definecolor{textcolor}{rgb}{0.150000,0.150000,0.150000}%
\pgfsetstrokecolor{textcolor}%
\pgfsetfillcolor{textcolor}%
\pgftext[x=0.194444in, y=1.129092in, left, base]{\color{textcolor}\rmfamily\fontsize{8.800000}{10.560000}\selectfont \(\displaystyle {2.85}\)}%
\end{pgfscope}%
\begin{pgfscope}%
\pgfpathrectangle{\pgfqpoint{0.538116in}{0.420833in}}{\pgfqpoint{2.959687in}{2.004431in}}%
\pgfusepath{clip}%
\pgfsetroundcap%
\pgfsetroundjoin%
\pgfsetlinewidth{0.803000pt}%
\definecolor{currentstroke}{rgb}{0.800000,0.800000,0.800000}%
\pgfsetstrokecolor{currentstroke}%
\pgfsetdash{}{0pt}%
\pgfpathmoveto{\pgfqpoint{0.538116in}{1.423049in}}%
\pgfpathlineto{\pgfqpoint{3.497803in}{1.423049in}}%
\pgfusepath{stroke}%
\end{pgfscope}%
\begin{pgfscope}%
\definecolor{textcolor}{rgb}{0.150000,0.150000,0.150000}%
\pgfsetstrokecolor{textcolor}%
\pgfsetfillcolor{textcolor}%
\pgftext[x=0.194444in, y=1.379646in, left, base]{\color{textcolor}\rmfamily\fontsize{8.800000}{10.560000}\selectfont \(\displaystyle {2.90}\)}%
\end{pgfscope}%
\begin{pgfscope}%
\pgfpathrectangle{\pgfqpoint{0.538116in}{0.420833in}}{\pgfqpoint{2.959687in}{2.004431in}}%
\pgfusepath{clip}%
\pgfsetroundcap%
\pgfsetroundjoin%
\pgfsetlinewidth{0.803000pt}%
\definecolor{currentstroke}{rgb}{0.800000,0.800000,0.800000}%
\pgfsetstrokecolor{currentstroke}%
\pgfsetdash{}{0pt}%
\pgfpathmoveto{\pgfqpoint{0.538116in}{1.673602in}}%
\pgfpathlineto{\pgfqpoint{3.497803in}{1.673602in}}%
\pgfusepath{stroke}%
\end{pgfscope}%
\begin{pgfscope}%
\definecolor{textcolor}{rgb}{0.150000,0.150000,0.150000}%
\pgfsetstrokecolor{textcolor}%
\pgfsetfillcolor{textcolor}%
\pgftext[x=0.194444in, y=1.630200in, left, base]{\color{textcolor}\rmfamily\fontsize{8.800000}{10.560000}\selectfont \(\displaystyle {2.95}\)}%
\end{pgfscope}%
\begin{pgfscope}%
\pgfpathrectangle{\pgfqpoint{0.538116in}{0.420833in}}{\pgfqpoint{2.959687in}{2.004431in}}%
\pgfusepath{clip}%
\pgfsetroundcap%
\pgfsetroundjoin%
\pgfsetlinewidth{0.803000pt}%
\definecolor{currentstroke}{rgb}{0.800000,0.800000,0.800000}%
\pgfsetstrokecolor{currentstroke}%
\pgfsetdash{}{0pt}%
\pgfpathmoveto{\pgfqpoint{0.538116in}{1.924156in}}%
\pgfpathlineto{\pgfqpoint{3.497803in}{1.924156in}}%
\pgfusepath{stroke}%
\end{pgfscope}%
\begin{pgfscope}%
\definecolor{textcolor}{rgb}{0.150000,0.150000,0.150000}%
\pgfsetstrokecolor{textcolor}%
\pgfsetfillcolor{textcolor}%
\pgftext[x=0.194444in, y=1.880753in, left, base]{\color{textcolor}\rmfamily\fontsize{8.800000}{10.560000}\selectfont \(\displaystyle {3.00}\)}%
\end{pgfscope}%
\begin{pgfscope}%
\pgfpathrectangle{\pgfqpoint{0.538116in}{0.420833in}}{\pgfqpoint{2.959687in}{2.004431in}}%
\pgfusepath{clip}%
\pgfsetroundcap%
\pgfsetroundjoin%
\pgfsetlinewidth{0.803000pt}%
\definecolor{currentstroke}{rgb}{0.800000,0.800000,0.800000}%
\pgfsetstrokecolor{currentstroke}%
\pgfsetdash{}{0pt}%
\pgfpathmoveto{\pgfqpoint{0.538116in}{2.174710in}}%
\pgfpathlineto{\pgfqpoint{3.497803in}{2.174710in}}%
\pgfusepath{stroke}%
\end{pgfscope}%
\begin{pgfscope}%
\definecolor{textcolor}{rgb}{0.150000,0.150000,0.150000}%
\pgfsetstrokecolor{textcolor}%
\pgfsetfillcolor{textcolor}%
\pgftext[x=0.194444in, y=2.131307in, left, base]{\color{textcolor}\rmfamily\fontsize{8.800000}{10.560000}\selectfont \(\displaystyle {3.05}\)}%
\end{pgfscope}%
\begin{pgfscope}%
\pgfpathrectangle{\pgfqpoint{0.538116in}{0.420833in}}{\pgfqpoint{2.959687in}{2.004431in}}%
\pgfusepath{clip}%
\pgfsetroundcap%
\pgfsetroundjoin%
\pgfsetlinewidth{0.803000pt}%
\definecolor{currentstroke}{rgb}{0.800000,0.800000,0.800000}%
\pgfsetstrokecolor{currentstroke}%
\pgfsetdash{}{0pt}%
\pgfpathmoveto{\pgfqpoint{0.538116in}{2.425264in}}%
\pgfpathlineto{\pgfqpoint{3.497803in}{2.425264in}}%
\pgfusepath{stroke}%
\end{pgfscope}%
\begin{pgfscope}%
\definecolor{textcolor}{rgb}{0.150000,0.150000,0.150000}%
\pgfsetstrokecolor{textcolor}%
\pgfsetfillcolor{textcolor}%
\pgftext[x=0.194444in, y=2.381861in, left, base]{\color{textcolor}\rmfamily\fontsize{8.800000}{10.560000}\selectfont \(\displaystyle {3.10}\)}%
\end{pgfscope}%
\begin{pgfscope}%
\definecolor{textcolor}{rgb}{0.150000,0.150000,0.150000}%
\pgfsetstrokecolor{textcolor}%
\pgfsetfillcolor{textcolor}%
\pgftext[x=0.138889in,y=1.423049in,,bottom,rotate=90.000000]{\color{textcolor}\rmfamily\fontsize{9.600000}{11.520000}\selectfont Mean test regression error (RMSE)}%
\end{pgfscope}%
\begin{pgfscope}%
\pgfpathrectangle{\pgfqpoint{0.538116in}{0.420833in}}{\pgfqpoint{2.959687in}{2.004431in}}%
\pgfusepath{clip}%
\pgfsetbuttcap%
\pgfsetroundjoin%
\definecolor{currentfill}{rgb}{0.121569,0.466667,0.705882}%
\pgfsetfillcolor{currentfill}%
\pgfsetlinewidth{0.803000pt}%
\definecolor{currentstroke}{rgb}{0.121569,0.466667,0.705882}%
\pgfsetstrokecolor{currentstroke}%
\pgfsetdash{}{0pt}%
\pgfsys@defobject{currentmarker}{\pgfqpoint{-0.038036in}{-0.038036in}}{\pgfqpoint{0.038036in}{0.038036in}}{%
\pgfpathmoveto{\pgfqpoint{0.000000in}{-0.038036in}}%
\pgfpathcurveto{\pgfqpoint{0.010087in}{-0.038036in}}{\pgfqpoint{0.019763in}{-0.034029in}}{\pgfqpoint{0.026896in}{-0.026896in}}%
\pgfpathcurveto{\pgfqpoint{0.034029in}{-0.019763in}}{\pgfqpoint{0.038036in}{-0.010087in}}{\pgfqpoint{0.038036in}{0.000000in}}%
\pgfpathcurveto{\pgfqpoint{0.038036in}{0.010087in}}{\pgfqpoint{0.034029in}{0.019763in}}{\pgfqpoint{0.026896in}{0.026896in}}%
\pgfpathcurveto{\pgfqpoint{0.019763in}{0.034029in}}{\pgfqpoint{0.010087in}{0.038036in}}{\pgfqpoint{0.000000in}{0.038036in}}%
\pgfpathcurveto{\pgfqpoint{-0.010087in}{0.038036in}}{\pgfqpoint{-0.019763in}{0.034029in}}{\pgfqpoint{-0.026896in}{0.026896in}}%
\pgfpathcurveto{\pgfqpoint{-0.034029in}{0.019763in}}{\pgfqpoint{-0.038036in}{0.010087in}}{\pgfqpoint{-0.038036in}{0.000000in}}%
\pgfpathcurveto{\pgfqpoint{-0.038036in}{-0.010087in}}{\pgfqpoint{-0.034029in}{-0.019763in}}{\pgfqpoint{-0.026896in}{-0.026896in}}%
\pgfpathcurveto{\pgfqpoint{-0.019763in}{-0.034029in}}{\pgfqpoint{-0.010087in}{-0.038036in}}{\pgfqpoint{0.000000in}{-0.038036in}}%
\pgfpathlineto{\pgfqpoint{0.000000in}{-0.038036in}}%
\pgfpathclose%
\pgfusepath{stroke,fill}%
}%
\begin{pgfscope}%
\pgfsys@transformshift{3.468206in}{1.618481in}%
\pgfsys@useobject{currentmarker}{}%
\end{pgfscope}%
\end{pgfscope}%
\begin{pgfscope}%
\pgfpathrectangle{\pgfqpoint{0.538116in}{0.420833in}}{\pgfqpoint{2.959687in}{2.004431in}}%
\pgfusepath{clip}%
\pgfsetbuttcap%
\pgfsetroundjoin%
\definecolor{currentfill}{rgb}{1.000000,0.498039,0.054902}%
\pgfsetfillcolor{currentfill}%
\pgfsetfillopacity{0.100000}%
\pgfsetlinewidth{0.803000pt}%
\definecolor{currentstroke}{rgb}{1.000000,0.498039,0.054902}%
\pgfsetstrokecolor{currentstroke}%
\pgfsetstrokeopacity{0.100000}%
\pgfsetdash{}{0pt}%
\pgfsys@defobject{currentmarker}{\pgfqpoint{0.538116in}{0.616265in}}{\pgfqpoint{3.497803in}{1.708680in}}{%
\pgfpathmoveto{\pgfqpoint{0.538116in}{1.167484in}}%
\pgfpathlineto{\pgfqpoint{0.538116in}{1.077284in}}%
\pgfpathlineto{\pgfqpoint{0.567713in}{1.077284in}}%
\pgfpathlineto{\pgfqpoint{0.686100in}{0.636310in}}%
\pgfpathlineto{\pgfqpoint{0.834085in}{0.616265in}}%
\pgfpathlineto{\pgfqpoint{1.130053in}{0.856797in}}%
\pgfpathlineto{\pgfqpoint{1.426022in}{1.032185in}}%
\pgfpathlineto{\pgfqpoint{1.721991in}{1.127395in}}%
\pgfpathlineto{\pgfqpoint{2.017960in}{1.377949in}}%
\pgfpathlineto{\pgfqpoint{2.313928in}{1.433071in}}%
\pgfpathlineto{\pgfqpoint{2.609897in}{1.473159in}}%
\pgfpathlineto{\pgfqpoint{2.905866in}{1.538303in}}%
\pgfpathlineto{\pgfqpoint{3.201834in}{1.538303in}}%
\pgfpathlineto{\pgfqpoint{3.497803in}{1.538303in}}%
\pgfpathlineto{\pgfqpoint{3.497803in}{1.708680in}}%
\pgfpathlineto{\pgfqpoint{3.497803in}{1.708680in}}%
\pgfpathlineto{\pgfqpoint{3.201834in}{1.708680in}}%
\pgfpathlineto{\pgfqpoint{2.905866in}{1.708680in}}%
\pgfpathlineto{\pgfqpoint{2.609897in}{1.633514in}}%
\pgfpathlineto{\pgfqpoint{2.313928in}{1.513248in}}%
\pgfpathlineto{\pgfqpoint{2.017960in}{1.468148in}}%
\pgfpathlineto{\pgfqpoint{1.721991in}{1.217594in}}%
\pgfpathlineto{\pgfqpoint{1.426022in}{1.112362in}}%
\pgfpathlineto{\pgfqpoint{1.130053in}{0.946996in}}%
\pgfpathlineto{\pgfqpoint{0.834085in}{0.686420in}}%
\pgfpathlineto{\pgfqpoint{0.686100in}{0.706465in}}%
\pgfpathlineto{\pgfqpoint{0.567713in}{1.167484in}}%
\pgfpathlineto{\pgfqpoint{0.538116in}{1.167484in}}%
\pgfpathlineto{\pgfqpoint{0.538116in}{1.167484in}}%
\pgfpathclose%
\pgfusepath{stroke,fill}%
}%
\begin{pgfscope}%
\pgfsys@transformshift{0.000000in}{0.000000in}%
\pgfsys@useobject{currentmarker}{}%
\end{pgfscope}%
\end{pgfscope}%
\begin{pgfscope}%
\pgfpathrectangle{\pgfqpoint{0.538116in}{0.420833in}}{\pgfqpoint{2.959687in}{2.004431in}}%
\pgfusepath{clip}%
\pgfsetroundcap%
\pgfsetroundjoin%
\pgfsetlinewidth{1.204500pt}%
\definecolor{currentstroke}{rgb}{1.000000,0.498039,0.054902}%
\pgfsetstrokecolor{currentstroke}%
\pgfsetdash{}{0pt}%
\pgfpathmoveto{\pgfqpoint{0.538116in}{1.122384in}}%
\pgfpathlineto{\pgfqpoint{0.567713in}{1.122384in}}%
\pgfpathlineto{\pgfqpoint{0.686100in}{0.671387in}}%
\pgfpathlineto{\pgfqpoint{0.834085in}{0.651343in}}%
\pgfpathlineto{\pgfqpoint{1.130053in}{0.901897in}}%
\pgfpathlineto{\pgfqpoint{1.426022in}{1.072273in}}%
\pgfpathlineto{\pgfqpoint{1.721991in}{1.172495in}}%
\pgfpathlineto{\pgfqpoint{2.017960in}{1.423049in}}%
\pgfpathlineto{\pgfqpoint{2.313928in}{1.473159in}}%
\pgfpathlineto{\pgfqpoint{2.609897in}{1.553337in}}%
\pgfpathlineto{\pgfqpoint{2.905866in}{1.623492in}}%
\pgfpathlineto{\pgfqpoint{3.201834in}{1.623492in}}%
\pgfpathlineto{\pgfqpoint{3.497803in}{1.623492in}}%
\pgfusepath{stroke}%
\end{pgfscope}%
\begin{pgfscope}%
\pgfsetrectcap%
\pgfsetmiterjoin%
\pgfsetlinewidth{1.003750pt}%
\definecolor{currentstroke}{rgb}{0.800000,0.800000,0.800000}%
\pgfsetstrokecolor{currentstroke}%
\pgfsetdash{}{0pt}%
\pgfpathmoveto{\pgfqpoint{0.538116in}{0.420833in}}%
\pgfpathlineto{\pgfqpoint{0.538116in}{2.425264in}}%
\pgfusepath{stroke}%
\end{pgfscope}%
\begin{pgfscope}%
\pgfsetrectcap%
\pgfsetmiterjoin%
\pgfsetlinewidth{1.003750pt}%
\definecolor{currentstroke}{rgb}{0.800000,0.800000,0.800000}%
\pgfsetstrokecolor{currentstroke}%
\pgfsetdash{}{0pt}%
\pgfpathmoveto{\pgfqpoint{0.538116in}{0.420833in}}%
\pgfpathlineto{\pgfqpoint{3.497803in}{0.420833in}}%
\pgfusepath{stroke}%
\end{pgfscope}%
\begin{pgfscope}%
\pgfsetroundcap%
\pgfsetroundjoin%
\definecolor{currentfill}{rgb}{0.862745,0.862745,0.862745}%
\pgfsetfillcolor{currentfill}%
\pgfsetlinewidth{0.803000pt}%
\definecolor{currentstroke}{rgb}{1.000000,1.000000,1.000000}%
\pgfsetstrokecolor{currentstroke}%
\pgfsetdash{}{0pt}%
\pgfpathmoveto{\pgfqpoint{1.164776in}{1.886582in}}%
\pgfpathquadraticcurveto{\pgfqpoint{1.164776in}{1.631687in}}{\pgfqpoint{1.164776in}{1.376791in}}%
\pgfpathlineto{\pgfqpoint{1.213387in}{1.376791in}}%
\pgfpathquadraticcurveto{\pgfqpoint{1.171720in}{1.293449in}}{\pgfqpoint{1.130053in}{1.210108in}}%
\pgfpathquadraticcurveto{\pgfqpoint{1.088387in}{1.293449in}}{\pgfqpoint{1.046720in}{1.376791in}}%
\pgfpathlineto{\pgfqpoint{1.095331in}{1.376791in}}%
\pgfpathquadraticcurveto{\pgfqpoint{1.095331in}{1.631687in}}{\pgfqpoint{1.095331in}{1.886582in}}%
\pgfpathlineto{\pgfqpoint{1.164776in}{1.886582in}}%
\pgfpathlineto{\pgfqpoint{1.164776in}{1.886582in}}%
\pgfpathclose%
\pgfusepath{stroke,fill}%
\end{pgfscope}%
\begin{pgfscope}%
\definecolor{textcolor}{rgb}{0.862745,0.862745,0.862745}%
\pgfsetstrokecolor{textcolor}%
\pgfsetfillcolor{textcolor}%
\pgftext[x=1.278038in,y=1.673602in,left,]{\color{textcolor}\rmfamily\fontsize{12.000000}{14.400000}\selectfont better}%
\end{pgfscope}%
\begin{pgfscope}%
\pgfsetbuttcap%
\pgfsetmiterjoin%
\definecolor{currentfill}{rgb}{1.000000,1.000000,1.000000}%
\pgfsetfillcolor{currentfill}%
\pgfsetfillopacity{0.800000}%
\pgfsetlinewidth{0.803000pt}%
\definecolor{currentstroke}{rgb}{0.800000,0.800000,0.800000}%
\pgfsetstrokecolor{currentstroke}%
\pgfsetstrokeopacity{0.800000}%
\pgfsetdash{}{0pt}%
\pgfpathmoveto{\pgfqpoint{2.178907in}{1.983042in}}%
\pgfpathlineto{\pgfqpoint{3.412248in}{1.983042in}}%
\pgfpathquadraticcurveto{\pgfqpoint{3.436692in}{1.983042in}}{\pgfqpoint{3.436692in}{2.007486in}}%
\pgfpathlineto{\pgfqpoint{3.436692in}{2.339708in}}%
\pgfpathquadraticcurveto{\pgfqpoint{3.436692in}{2.364153in}}{\pgfqpoint{3.412248in}{2.364153in}}%
\pgfpathlineto{\pgfqpoint{2.178907in}{2.364153in}}%
\pgfpathquadraticcurveto{\pgfqpoint{2.154462in}{2.364153in}}{\pgfqpoint{2.154462in}{2.339708in}}%
\pgfpathlineto{\pgfqpoint{2.154462in}{2.007486in}}%
\pgfpathquadraticcurveto{\pgfqpoint{2.154462in}{1.983042in}}{\pgfqpoint{2.178907in}{1.983042in}}%
\pgfpathlineto{\pgfqpoint{2.178907in}{1.983042in}}%
\pgfpathclose%
\pgfusepath{stroke,fill}%
\end{pgfscope}%
\begin{pgfscope}%
\pgfsetbuttcap%
\pgfsetroundjoin%
\definecolor{currentfill}{rgb}{0.121569,0.466667,0.705882}%
\pgfsetfillcolor{currentfill}%
\pgfsetlinewidth{0.803000pt}%
\definecolor{currentstroke}{rgb}{0.121569,0.466667,0.705882}%
\pgfsetstrokecolor{currentstroke}%
\pgfsetdash{}{0pt}%
\pgfsys@defobject{currentmarker}{\pgfqpoint{-0.038036in}{-0.038036in}}{\pgfqpoint{0.038036in}{0.038036in}}{%
\pgfpathmoveto{\pgfqpoint{0.000000in}{-0.038036in}}%
\pgfpathcurveto{\pgfqpoint{0.010087in}{-0.038036in}}{\pgfqpoint{0.019763in}{-0.034029in}}{\pgfqpoint{0.026896in}{-0.026896in}}%
\pgfpathcurveto{\pgfqpoint{0.034029in}{-0.019763in}}{\pgfqpoint{0.038036in}{-0.010087in}}{\pgfqpoint{0.038036in}{0.000000in}}%
\pgfpathcurveto{\pgfqpoint{0.038036in}{0.010087in}}{\pgfqpoint{0.034029in}{0.019763in}}{\pgfqpoint{0.026896in}{0.026896in}}%
\pgfpathcurveto{\pgfqpoint{0.019763in}{0.034029in}}{\pgfqpoint{0.010087in}{0.038036in}}{\pgfqpoint{0.000000in}{0.038036in}}%
\pgfpathcurveto{\pgfqpoint{-0.010087in}{0.038036in}}{\pgfqpoint{-0.019763in}{0.034029in}}{\pgfqpoint{-0.026896in}{0.026896in}}%
\pgfpathcurveto{\pgfqpoint{-0.034029in}{0.019763in}}{\pgfqpoint{-0.038036in}{0.010087in}}{\pgfqpoint{-0.038036in}{0.000000in}}%
\pgfpathcurveto{\pgfqpoint{-0.038036in}{-0.010087in}}{\pgfqpoint{-0.034029in}{-0.019763in}}{\pgfqpoint{-0.026896in}{-0.026896in}}%
\pgfpathcurveto{\pgfqpoint{-0.019763in}{-0.034029in}}{\pgfqpoint{-0.010087in}{-0.038036in}}{\pgfqpoint{0.000000in}{-0.038036in}}%
\pgfpathlineto{\pgfqpoint{0.000000in}{-0.038036in}}%
\pgfpathclose%
\pgfusepath{stroke,fill}%
}%
\begin{pgfscope}%
\pgfsys@transformshift{2.325573in}{2.260542in}%
\pgfsys@useobject{currentmarker}{}%
\end{pgfscope}%
\end{pgfscope}%
\begin{pgfscope}%
\definecolor{textcolor}{rgb}{0.150000,0.150000,0.150000}%
\pgfsetstrokecolor{textcolor}%
\pgfsetfillcolor{textcolor}%
\pgftext[x=2.545573in,y=2.228458in,left,base]{\color{textcolor}\rmfamily\fontsize{8.800000}{10.560000}\selectfont Maddock et al.}%
\end{pgfscope}%
\begin{pgfscope}%
\pgfsetroundcap%
\pgfsetroundjoin%
\pgfsetlinewidth{1.204500pt}%
\definecolor{currentstroke}{rgb}{1.000000,0.498039,0.054902}%
\pgfsetstrokecolor{currentstroke}%
\pgfsetdash{}{0pt}%
\pgfpathmoveto{\pgfqpoint{2.203351in}{2.099014in}}%
\pgfpathlineto{\pgfqpoint{2.325573in}{2.099014in}}%
\pgfpathlineto{\pgfqpoint{2.447796in}{2.099014in}}%
\pgfusepath{stroke}%
\end{pgfscope}%
\begin{pgfscope}%
\definecolor{textcolor}{rgb}{0.150000,0.150000,0.150000}%
\pgfsetstrokecolor{textcolor}%
\pgfsetfillcolor{textcolor}%
\pgftext[x=2.545573in,y=2.056236in,left,base]{\color{textcolor}\rmfamily\fontsize{8.800000}{10.560000}\selectfont S-BDT}%
\end{pgfscope}%
\end{pgfpicture}%
\makeatother%
\endgroup%

%% file: images/abalone_init_ablation.pgf
\begingroup%
\makeatletter%
\begin{pgfpicture}%
\pgfpathrectangle{\pgfpointorigin}{\pgfqpoint{3.612000in}{2.495774in}}%
\pgfusepath{use as bounding box, clip}%
\begin{pgfscope}%
\pgfsetbuttcap%
\pgfsetmiterjoin%
\definecolor{currentfill}{rgb}{1.000000,1.000000,1.000000}%
\pgfsetfillcolor{currentfill}%
\pgfsetlinewidth{0.000000pt}%
\definecolor{currentstroke}{rgb}{1.000000,1.000000,1.000000}%
\pgfsetstrokecolor{currentstroke}%
\pgfsetdash{}{0pt}%
\pgfpathmoveto{\pgfqpoint{0.000000in}{0.000000in}}%
\pgfpathlineto{\pgfqpoint{3.612000in}{0.000000in}}%
\pgfpathlineto{\pgfqpoint{3.612000in}{2.495774in}}%
\pgfpathlineto{\pgfqpoint{0.000000in}{2.495774in}}%
\pgfpathlineto{\pgfqpoint{0.000000in}{0.000000in}}%
\pgfpathclose%
\pgfusepath{fill}%
\end{pgfscope}%
\begin{pgfscope}%
\pgfsetbuttcap%
\pgfsetmiterjoin%
\definecolor{currentfill}{rgb}{1.000000,1.000000,1.000000}%
\pgfsetfillcolor{currentfill}%
\pgfsetlinewidth{0.000000pt}%
\definecolor{currentstroke}{rgb}{0.000000,0.000000,0.000000}%
\pgfsetstrokecolor{currentstroke}%
\pgfsetstrokeopacity{0.000000}%
\pgfsetdash{}{0pt}%
\pgfpathmoveto{\pgfqpoint{0.538116in}{0.420833in}}%
\pgfpathlineto{\pgfqpoint{3.497803in}{0.420833in}}%
\pgfpathlineto{\pgfqpoint{3.497803in}{2.425264in}}%
\pgfpathlineto{\pgfqpoint{0.538116in}{2.425264in}}%
\pgfpathlineto{\pgfqpoint{0.538116in}{0.420833in}}%
\pgfpathclose%
\pgfusepath{fill}%
\end{pgfscope}%
\begin{pgfscope}%
\pgfpathrectangle{\pgfqpoint{0.538116in}{0.420833in}}{\pgfqpoint{2.959687in}{2.004431in}}%
\pgfusepath{clip}%
\pgfsetroundcap%
\pgfsetroundjoin%
\pgfsetlinewidth{0.803000pt}%
\definecolor{currentstroke}{rgb}{0.800000,0.800000,0.800000}%
\pgfsetstrokecolor{currentstroke}%
\pgfsetdash{}{0pt}%
\pgfpathmoveto{\pgfqpoint{0.597310in}{0.420833in}}%
\pgfpathlineto{\pgfqpoint{0.597310in}{2.425264in}}%
\pgfusepath{stroke}%
\end{pgfscope}%
\begin{pgfscope}%
\definecolor{textcolor}{rgb}{0.150000,0.150000,0.150000}%
\pgfsetstrokecolor{textcolor}%
\pgfsetfillcolor{textcolor}%
\pgftext[x=0.597310in,y=0.305556in,,top]{\color{textcolor}\rmfamily\fontsize{8.800000}{10.560000}\selectfont \(\displaystyle {0.01}\)}%
\end{pgfscope}%
\begin{pgfscope}%
\pgfpathrectangle{\pgfqpoint{0.538116in}{0.420833in}}{\pgfqpoint{2.959687in}{2.004431in}}%
\pgfusepath{clip}%
\pgfsetroundcap%
\pgfsetroundjoin%
\pgfsetlinewidth{0.803000pt}%
\definecolor{currentstroke}{rgb}{0.800000,0.800000,0.800000}%
\pgfsetstrokecolor{currentstroke}%
\pgfsetdash{}{0pt}%
\pgfpathmoveto{\pgfqpoint{1.130053in}{0.420833in}}%
\pgfpathlineto{\pgfqpoint{1.130053in}{2.425264in}}%
\pgfusepath{stroke}%
\end{pgfscope}%
\begin{pgfscope}%
\definecolor{textcolor}{rgb}{0.150000,0.150000,0.150000}%
\pgfsetstrokecolor{textcolor}%
\pgfsetfillcolor{textcolor}%
\pgftext[x=1.130053in,y=0.305556in,,top]{\color{textcolor}\rmfamily\fontsize{8.800000}{10.560000}\selectfont \(\displaystyle {0.10}\)}%
\end{pgfscope}%
\begin{pgfscope}%
\pgfpathrectangle{\pgfqpoint{0.538116in}{0.420833in}}{\pgfqpoint{2.959687in}{2.004431in}}%
\pgfusepath{clip}%
\pgfsetroundcap%
\pgfsetroundjoin%
\pgfsetlinewidth{0.803000pt}%
\definecolor{currentstroke}{rgb}{0.800000,0.800000,0.800000}%
\pgfsetstrokecolor{currentstroke}%
\pgfsetdash{}{0pt}%
\pgfpathmoveto{\pgfqpoint{1.721991in}{0.420833in}}%
\pgfpathlineto{\pgfqpoint{1.721991in}{2.425264in}}%
\pgfusepath{stroke}%
\end{pgfscope}%
\begin{pgfscope}%
\definecolor{textcolor}{rgb}{0.150000,0.150000,0.150000}%
\pgfsetstrokecolor{textcolor}%
\pgfsetfillcolor{textcolor}%
\pgftext[x=1.721991in,y=0.305556in,,top]{\color{textcolor}\rmfamily\fontsize{8.800000}{10.560000}\selectfont \(\displaystyle {0.20}\)}%
\end{pgfscope}%
\begin{pgfscope}%
\pgfpathrectangle{\pgfqpoint{0.538116in}{0.420833in}}{\pgfqpoint{2.959687in}{2.004431in}}%
\pgfusepath{clip}%
\pgfsetroundcap%
\pgfsetroundjoin%
\pgfsetlinewidth{0.803000pt}%
\definecolor{currentstroke}{rgb}{0.800000,0.800000,0.800000}%
\pgfsetstrokecolor{currentstroke}%
\pgfsetdash{}{0pt}%
\pgfpathmoveto{\pgfqpoint{2.313928in}{0.420833in}}%
\pgfpathlineto{\pgfqpoint{2.313928in}{2.425264in}}%
\pgfusepath{stroke}%
\end{pgfscope}%
\begin{pgfscope}%
\definecolor{textcolor}{rgb}{0.150000,0.150000,0.150000}%
\pgfsetstrokecolor{textcolor}%
\pgfsetfillcolor{textcolor}%
\pgftext[x=2.313928in,y=0.305556in,,top]{\color{textcolor}\rmfamily\fontsize{8.800000}{10.560000}\selectfont \(\displaystyle {0.30}\)}%
\end{pgfscope}%
\begin{pgfscope}%
\pgfpathrectangle{\pgfqpoint{0.538116in}{0.420833in}}{\pgfqpoint{2.959687in}{2.004431in}}%
\pgfusepath{clip}%
\pgfsetroundcap%
\pgfsetroundjoin%
\pgfsetlinewidth{0.803000pt}%
\definecolor{currentstroke}{rgb}{0.800000,0.800000,0.800000}%
\pgfsetstrokecolor{currentstroke}%
\pgfsetdash{}{0pt}%
\pgfpathmoveto{\pgfqpoint{2.905866in}{0.420833in}}%
\pgfpathlineto{\pgfqpoint{2.905866in}{2.425264in}}%
\pgfusepath{stroke}%
\end{pgfscope}%
\begin{pgfscope}%
\definecolor{textcolor}{rgb}{0.150000,0.150000,0.150000}%
\pgfsetstrokecolor{textcolor}%
\pgfsetfillcolor{textcolor}%
\pgftext[x=2.905866in,y=0.305556in,,top]{\color{textcolor}\rmfamily\fontsize{8.800000}{10.560000}\selectfont \(\displaystyle {0.40}\)}%
\end{pgfscope}%
\begin{pgfscope}%
\pgfpathrectangle{\pgfqpoint{0.538116in}{0.420833in}}{\pgfqpoint{2.959687in}{2.004431in}}%
\pgfusepath{clip}%
\pgfsetroundcap%
\pgfsetroundjoin%
\pgfsetlinewidth{0.803000pt}%
\definecolor{currentstroke}{rgb}{0.800000,0.800000,0.800000}%
\pgfsetstrokecolor{currentstroke}%
\pgfsetdash{}{0pt}%
\pgfpathmoveto{\pgfqpoint{3.497803in}{0.420833in}}%
\pgfpathlineto{\pgfqpoint{3.497803in}{2.425264in}}%
\pgfusepath{stroke}%
\end{pgfscope}%
\begin{pgfscope}%
\definecolor{textcolor}{rgb}{0.150000,0.150000,0.150000}%
\pgfsetstrokecolor{textcolor}%
\pgfsetfillcolor{textcolor}%
\pgftext[x=3.497803in,y=0.305556in,,top]{\color{textcolor}\rmfamily\fontsize{8.800000}{10.560000}\selectfont \(\displaystyle {0.50}\)}%
\end{pgfscope}%
\begin{pgfscope}%
\definecolor{textcolor}{rgb}{0.150000,0.150000,0.150000}%
\pgfsetstrokecolor{textcolor}%
\pgfsetfillcolor{textcolor}%
\pgftext[x=2.017960in,y=0.138889in,,top]{\color{textcolor}\rmfamily\fontsize{9.600000}{11.520000}\selectfont \(\displaystyle \varepsilon_\text{init}/\varepsilon\) (privacy budget ratio for initial score)}%
\end{pgfscope}%
\begin{pgfscope}%
\pgfpathrectangle{\pgfqpoint{0.538116in}{0.420833in}}{\pgfqpoint{2.959687in}{2.004431in}}%
\pgfusepath{clip}%
\pgfsetroundcap%
\pgfsetroundjoin%
\pgfsetlinewidth{0.803000pt}%
\definecolor{currentstroke}{rgb}{0.800000,0.800000,0.800000}%
\pgfsetstrokecolor{currentstroke}%
\pgfsetdash{}{0pt}%
\pgfpathmoveto{\pgfqpoint{0.538116in}{0.420833in}}%
\pgfpathlineto{\pgfqpoint{3.497803in}{0.420833in}}%
\pgfusepath{stroke}%
\end{pgfscope}%
\begin{pgfscope}%
\definecolor{textcolor}{rgb}{0.150000,0.150000,0.150000}%
\pgfsetstrokecolor{textcolor}%
\pgfsetfillcolor{textcolor}%
\pgftext[x=0.194444in, y=0.377431in, left, base]{\color{textcolor}\rmfamily\fontsize{8.800000}{10.560000}\selectfont \(\displaystyle {2.70}\)}%
\end{pgfscope}%
\begin{pgfscope}%
\pgfpathrectangle{\pgfqpoint{0.538116in}{0.420833in}}{\pgfqpoint{2.959687in}{2.004431in}}%
\pgfusepath{clip}%
\pgfsetroundcap%
\pgfsetroundjoin%
\pgfsetlinewidth{0.803000pt}%
\definecolor{currentstroke}{rgb}{0.800000,0.800000,0.800000}%
\pgfsetstrokecolor{currentstroke}%
\pgfsetdash{}{0pt}%
\pgfpathmoveto{\pgfqpoint{0.538116in}{0.671387in}}%
\pgfpathlineto{\pgfqpoint{3.497803in}{0.671387in}}%
\pgfusepath{stroke}%
\end{pgfscope}%
\begin{pgfscope}%
\definecolor{textcolor}{rgb}{0.150000,0.150000,0.150000}%
\pgfsetstrokecolor{textcolor}%
\pgfsetfillcolor{textcolor}%
\pgftext[x=0.194444in, y=0.627984in, left, base]{\color{textcolor}\rmfamily\fontsize{8.800000}{10.560000}\selectfont \(\displaystyle {2.75}\)}%
\end{pgfscope}%
\begin{pgfscope}%
\pgfpathrectangle{\pgfqpoint{0.538116in}{0.420833in}}{\pgfqpoint{2.959687in}{2.004431in}}%
\pgfusepath{clip}%
\pgfsetroundcap%
\pgfsetroundjoin%
\pgfsetlinewidth{0.803000pt}%
\definecolor{currentstroke}{rgb}{0.800000,0.800000,0.800000}%
\pgfsetstrokecolor{currentstroke}%
\pgfsetdash{}{0pt}%
\pgfpathmoveto{\pgfqpoint{0.538116in}{0.921941in}}%
\pgfpathlineto{\pgfqpoint{3.497803in}{0.921941in}}%
\pgfusepath{stroke}%
\end{pgfscope}%
\begin{pgfscope}%
\definecolor{textcolor}{rgb}{0.150000,0.150000,0.150000}%
\pgfsetstrokecolor{textcolor}%
\pgfsetfillcolor{textcolor}%
\pgftext[x=0.194444in, y=0.878538in, left, base]{\color{textcolor}\rmfamily\fontsize{8.800000}{10.560000}\selectfont \(\displaystyle {2.80}\)}%
\end{pgfscope}%
\begin{pgfscope}%
\pgfpathrectangle{\pgfqpoint{0.538116in}{0.420833in}}{\pgfqpoint{2.959687in}{2.004431in}}%
\pgfusepath{clip}%
\pgfsetroundcap%
\pgfsetroundjoin%
\pgfsetlinewidth{0.803000pt}%
\definecolor{currentstroke}{rgb}{0.800000,0.800000,0.800000}%
\pgfsetstrokecolor{currentstroke}%
\pgfsetdash{}{0pt}%
\pgfpathmoveto{\pgfqpoint{0.538116in}{1.172495in}}%
\pgfpathlineto{\pgfqpoint{3.497803in}{1.172495in}}%
\pgfusepath{stroke}%
\end{pgfscope}%
\begin{pgfscope}%
\definecolor{textcolor}{rgb}{0.150000,0.150000,0.150000}%
\pgfsetstrokecolor{textcolor}%
\pgfsetfillcolor{textcolor}%
\pgftext[x=0.194444in, y=1.129092in, left, base]{\color{textcolor}\rmfamily\fontsize{8.800000}{10.560000}\selectfont \(\displaystyle {2.85}\)}%
\end{pgfscope}%
\begin{pgfscope}%
\pgfpathrectangle{\pgfqpoint{0.538116in}{0.420833in}}{\pgfqpoint{2.959687in}{2.004431in}}%
\pgfusepath{clip}%
\pgfsetroundcap%
\pgfsetroundjoin%
\pgfsetlinewidth{0.803000pt}%
\definecolor{currentstroke}{rgb}{0.800000,0.800000,0.800000}%
\pgfsetstrokecolor{currentstroke}%
\pgfsetdash{}{0pt}%
\pgfpathmoveto{\pgfqpoint{0.538116in}{1.423049in}}%
\pgfpathlineto{\pgfqpoint{3.497803in}{1.423049in}}%
\pgfusepath{stroke}%
\end{pgfscope}%
\begin{pgfscope}%
\definecolor{textcolor}{rgb}{0.150000,0.150000,0.150000}%
\pgfsetstrokecolor{textcolor}%
\pgfsetfillcolor{textcolor}%
\pgftext[x=0.194444in, y=1.379646in, left, base]{\color{textcolor}\rmfamily\fontsize{8.800000}{10.560000}\selectfont \(\displaystyle {2.90}\)}%
\end{pgfscope}%
\begin{pgfscope}%
\pgfpathrectangle{\pgfqpoint{0.538116in}{0.420833in}}{\pgfqpoint{2.959687in}{2.004431in}}%
\pgfusepath{clip}%
\pgfsetroundcap%
\pgfsetroundjoin%
\pgfsetlinewidth{0.803000pt}%
\definecolor{currentstroke}{rgb}{0.800000,0.800000,0.800000}%
\pgfsetstrokecolor{currentstroke}%
\pgfsetdash{}{0pt}%
\pgfpathmoveto{\pgfqpoint{0.538116in}{1.673602in}}%
\pgfpathlineto{\pgfqpoint{3.497803in}{1.673602in}}%
\pgfusepath{stroke}%
\end{pgfscope}%
\begin{pgfscope}%
\definecolor{textcolor}{rgb}{0.150000,0.150000,0.150000}%
\pgfsetstrokecolor{textcolor}%
\pgfsetfillcolor{textcolor}%
\pgftext[x=0.194444in, y=1.630200in, left, base]{\color{textcolor}\rmfamily\fontsize{8.800000}{10.560000}\selectfont \(\displaystyle {2.95}\)}%
\end{pgfscope}%
\begin{pgfscope}%
\pgfpathrectangle{\pgfqpoint{0.538116in}{0.420833in}}{\pgfqpoint{2.959687in}{2.004431in}}%
\pgfusepath{clip}%
\pgfsetroundcap%
\pgfsetroundjoin%
\pgfsetlinewidth{0.803000pt}%
\definecolor{currentstroke}{rgb}{0.800000,0.800000,0.800000}%
\pgfsetstrokecolor{currentstroke}%
\pgfsetdash{}{0pt}%
\pgfpathmoveto{\pgfqpoint{0.538116in}{1.924156in}}%
\pgfpathlineto{\pgfqpoint{3.497803in}{1.924156in}}%
\pgfusepath{stroke}%
\end{pgfscope}%
\begin{pgfscope}%
\definecolor{textcolor}{rgb}{0.150000,0.150000,0.150000}%
\pgfsetstrokecolor{textcolor}%
\pgfsetfillcolor{textcolor}%
\pgftext[x=0.194444in, y=1.880753in, left, base]{\color{textcolor}\rmfamily\fontsize{8.800000}{10.560000}\selectfont \(\displaystyle {3.00}\)}%
\end{pgfscope}%
\begin{pgfscope}%
\pgfpathrectangle{\pgfqpoint{0.538116in}{0.420833in}}{\pgfqpoint{2.959687in}{2.004431in}}%
\pgfusepath{clip}%
\pgfsetroundcap%
\pgfsetroundjoin%
\pgfsetlinewidth{0.803000pt}%
\definecolor{currentstroke}{rgb}{0.800000,0.800000,0.800000}%
\pgfsetstrokecolor{currentstroke}%
\pgfsetdash{}{0pt}%
\pgfpathmoveto{\pgfqpoint{0.538116in}{2.174710in}}%
\pgfpathlineto{\pgfqpoint{3.497803in}{2.174710in}}%
\pgfusepath{stroke}%
\end{pgfscope}%
\begin{pgfscope}%
\definecolor{textcolor}{rgb}{0.150000,0.150000,0.150000}%
\pgfsetstrokecolor{textcolor}%
\pgfsetfillcolor{textcolor}%
\pgftext[x=0.194444in, y=2.131307in, left, base]{\color{textcolor}\rmfamily\fontsize{8.800000}{10.560000}\selectfont \(\displaystyle {3.05}\)}%
\end{pgfscope}%
\begin{pgfscope}%
\pgfpathrectangle{\pgfqpoint{0.538116in}{0.420833in}}{\pgfqpoint{2.959687in}{2.004431in}}%
\pgfusepath{clip}%
\pgfsetroundcap%
\pgfsetroundjoin%
\pgfsetlinewidth{0.803000pt}%
\definecolor{currentstroke}{rgb}{0.800000,0.800000,0.800000}%
\pgfsetstrokecolor{currentstroke}%
\pgfsetdash{}{0pt}%
\pgfpathmoveto{\pgfqpoint{0.538116in}{2.425264in}}%
\pgfpathlineto{\pgfqpoint{3.497803in}{2.425264in}}%
\pgfusepath{stroke}%
\end{pgfscope}%
\begin{pgfscope}%
\definecolor{textcolor}{rgb}{0.150000,0.150000,0.150000}%
\pgfsetstrokecolor{textcolor}%
\pgfsetfillcolor{textcolor}%
\pgftext[x=0.194444in, y=2.381861in, left, base]{\color{textcolor}\rmfamily\fontsize{8.800000}{10.560000}\selectfont \(\displaystyle {3.10}\)}%
\end{pgfscope}%
\begin{pgfscope}%
\definecolor{textcolor}{rgb}{0.150000,0.150000,0.150000}%
\pgfsetstrokecolor{textcolor}%
\pgfsetfillcolor{textcolor}%
\pgftext[x=0.138889in,y=1.423049in,,bottom,rotate=90.000000]{\color{textcolor}\rmfamily\fontsize{9.600000}{11.520000}\selectfont Mean test regression error (RMSE)}%
\end{pgfscope}%
\begin{pgfscope}%
\pgfpathrectangle{\pgfqpoint{0.538116in}{0.420833in}}{\pgfqpoint{2.959687in}{2.004431in}}%
\pgfusepath{clip}%
\pgfsetbuttcap%
\pgfsetroundjoin%
\definecolor{currentfill}{rgb}{0.121569,0.466667,0.705882}%
\pgfsetfillcolor{currentfill}%
\pgfsetlinewidth{0.803000pt}%
\definecolor{currentstroke}{rgb}{0.121569,0.466667,0.705882}%
\pgfsetstrokecolor{currentstroke}%
\pgfsetdash{}{0pt}%
\pgfsys@defobject{currentmarker}{\pgfqpoint{-0.038036in}{-0.038036in}}{\pgfqpoint{0.038036in}{0.038036in}}{%
\pgfpathmoveto{\pgfqpoint{0.000000in}{-0.038036in}}%
\pgfpathcurveto{\pgfqpoint{0.010087in}{-0.038036in}}{\pgfqpoint{0.019763in}{-0.034029in}}{\pgfqpoint{0.026896in}{-0.026896in}}%
\pgfpathcurveto{\pgfqpoint{0.034029in}{-0.019763in}}{\pgfqpoint{0.038036in}{-0.010087in}}{\pgfqpoint{0.038036in}{0.000000in}}%
\pgfpathcurveto{\pgfqpoint{0.038036in}{0.010087in}}{\pgfqpoint{0.034029in}{0.019763in}}{\pgfqpoint{0.026896in}{0.026896in}}%
\pgfpathcurveto{\pgfqpoint{0.019763in}{0.034029in}}{\pgfqpoint{0.010087in}{0.038036in}}{\pgfqpoint{0.000000in}{0.038036in}}%
\pgfpathcurveto{\pgfqpoint{-0.010087in}{0.038036in}}{\pgfqpoint{-0.019763in}{0.034029in}}{\pgfqpoint{-0.026896in}{0.026896in}}%
\pgfpathcurveto{\pgfqpoint{-0.034029in}{0.019763in}}{\pgfqpoint{-0.038036in}{0.010087in}}{\pgfqpoint{-0.038036in}{0.000000in}}%
\pgfpathcurveto{\pgfqpoint{-0.038036in}{-0.010087in}}{\pgfqpoint{-0.034029in}{-0.019763in}}{\pgfqpoint{-0.026896in}{-0.026896in}}%
\pgfpathcurveto{\pgfqpoint{-0.019763in}{-0.034029in}}{\pgfqpoint{-0.010087in}{-0.038036in}}{\pgfqpoint{0.000000in}{-0.038036in}}%
\pgfpathlineto{\pgfqpoint{0.000000in}{-0.038036in}}%
\pgfpathclose%
\pgfusepath{stroke,fill}%
}%
\begin{pgfscope}%
\pgfsys@transformshift{0.567713in}{1.618481in}%
\pgfsys@useobject{currentmarker}{}%
\end{pgfscope}%
\end{pgfscope}%
\begin{pgfscope}%
\pgfpathrectangle{\pgfqpoint{0.538116in}{0.420833in}}{\pgfqpoint{2.959687in}{2.004431in}}%
\pgfusepath{clip}%
\pgfsetbuttcap%
\pgfsetroundjoin%
\definecolor{currentfill}{rgb}{1.000000,0.498039,0.054902}%
\pgfsetfillcolor{currentfill}%
\pgfsetfillopacity{0.100000}%
\pgfsetlinewidth{0.803000pt}%
\definecolor{currentstroke}{rgb}{1.000000,0.498039,0.054902}%
\pgfsetstrokecolor{currentstroke}%
\pgfsetstrokeopacity{0.100000}%
\pgfsetdash{}{0pt}%
\pgfsys@defobject{currentmarker}{\pgfqpoint{0.538116in}{0.616265in}}{\pgfqpoint{3.497803in}{1.408015in}}{%
\pgfpathmoveto{\pgfqpoint{0.538116in}{1.017151in}}%
\pgfpathlineto{\pgfqpoint{0.538116in}{0.926952in}}%
\pgfpathlineto{\pgfqpoint{0.834085in}{0.616265in}}%
\pgfpathlineto{\pgfqpoint{1.130053in}{0.631299in}}%
\pgfpathlineto{\pgfqpoint{1.721991in}{0.701454in}}%
\pgfpathlineto{\pgfqpoint{2.313928in}{0.866819in}}%
\pgfpathlineto{\pgfqpoint{2.905866in}{1.077284in}}%
\pgfpathlineto{\pgfqpoint{3.497803in}{1.317816in}}%
\pgfpathlineto{\pgfqpoint{3.497803in}{1.408015in}}%
\pgfpathlineto{\pgfqpoint{3.497803in}{1.408015in}}%
\pgfpathlineto{\pgfqpoint{2.905866in}{1.167484in}}%
\pgfpathlineto{\pgfqpoint{2.313928in}{0.946996in}}%
\pgfpathlineto{\pgfqpoint{1.721991in}{0.791653in}}%
\pgfpathlineto{\pgfqpoint{1.130053in}{0.701454in}}%
\pgfpathlineto{\pgfqpoint{0.834085in}{0.676398in}}%
\pgfpathlineto{\pgfqpoint{0.538116in}{1.017151in}}%
\pgfpathlineto{\pgfqpoint{0.538116in}{1.017151in}}%
\pgfpathclose%
\pgfusepath{stroke,fill}%
}%
\begin{pgfscope}%
\pgfsys@transformshift{0.000000in}{0.000000in}%
\pgfsys@useobject{currentmarker}{}%
\end{pgfscope}%
\end{pgfscope}%
\begin{pgfscope}%
\pgfpathrectangle{\pgfqpoint{0.538116in}{0.420833in}}{\pgfqpoint{2.959687in}{2.004431in}}%
\pgfusepath{clip}%
\pgfsetroundcap%
\pgfsetroundjoin%
\pgfsetlinewidth{1.204500pt}%
\definecolor{currentstroke}{rgb}{1.000000,0.498039,0.054902}%
\pgfsetstrokecolor{currentstroke}%
\pgfsetdash{}{0pt}%
\pgfpathmoveto{\pgfqpoint{0.538116in}{0.972052in}}%
\pgfpathlineto{\pgfqpoint{0.834085in}{0.646332in}}%
\pgfpathlineto{\pgfqpoint{1.130053in}{0.666376in}}%
\pgfpathlineto{\pgfqpoint{1.721991in}{0.746553in}}%
\pgfpathlineto{\pgfqpoint{2.313928in}{0.906908in}}%
\pgfpathlineto{\pgfqpoint{2.905866in}{1.122384in}}%
\pgfpathlineto{\pgfqpoint{3.497803in}{1.362916in}}%
\pgfusepath{stroke}%
\end{pgfscope}%
\begin{pgfscope}%
\pgfsetrectcap%
\pgfsetmiterjoin%
\pgfsetlinewidth{1.003750pt}%
\definecolor{currentstroke}{rgb}{0.800000,0.800000,0.800000}%
\pgfsetstrokecolor{currentstroke}%
\pgfsetdash{}{0pt}%
\pgfpathmoveto{\pgfqpoint{0.538116in}{0.420833in}}%
\pgfpathlineto{\pgfqpoint{0.538116in}{2.425264in}}%
\pgfusepath{stroke}%
\end{pgfscope}%
\begin{pgfscope}%
\pgfsetrectcap%
\pgfsetmiterjoin%
\pgfsetlinewidth{1.003750pt}%
\definecolor{currentstroke}{rgb}{0.800000,0.800000,0.800000}%
\pgfsetstrokecolor{currentstroke}%
\pgfsetdash{}{0pt}%
\pgfpathmoveto{\pgfqpoint{0.538116in}{0.420833in}}%
\pgfpathlineto{\pgfqpoint{3.497803in}{0.420833in}}%
\pgfusepath{stroke}%
\end{pgfscope}%
\begin{pgfscope}%
\pgfsetroundcap%
\pgfsetroundjoin%
\definecolor{currentfill}{rgb}{0.862745,0.862745,0.862745}%
\pgfsetfillcolor{currentfill}%
\pgfsetlinewidth{0.803000pt}%
\definecolor{currentstroke}{rgb}{1.000000,1.000000,1.000000}%
\pgfsetstrokecolor{currentstroke}%
\pgfsetdash{}{0pt}%
\pgfpathmoveto{\pgfqpoint{1.342357in}{1.886582in}}%
\pgfpathquadraticcurveto{\pgfqpoint{1.342357in}{1.631687in}}{\pgfqpoint{1.342357in}{1.376791in}}%
\pgfpathlineto{\pgfqpoint{1.390968in}{1.376791in}}%
\pgfpathquadraticcurveto{\pgfqpoint{1.349301in}{1.293449in}}{\pgfqpoint{1.307635in}{1.210108in}}%
\pgfpathquadraticcurveto{\pgfqpoint{1.265968in}{1.293449in}}{\pgfqpoint{1.224301in}{1.376791in}}%
\pgfpathlineto{\pgfqpoint{1.272912in}{1.376791in}}%
\pgfpathquadraticcurveto{\pgfqpoint{1.272912in}{1.631687in}}{\pgfqpoint{1.272912in}{1.886582in}}%
\pgfpathlineto{\pgfqpoint{1.342357in}{1.886582in}}%
\pgfpathlineto{\pgfqpoint{1.342357in}{1.886582in}}%
\pgfpathclose%
\pgfusepath{stroke,fill}%
\end{pgfscope}%
\begin{pgfscope}%
\definecolor{textcolor}{rgb}{0.862745,0.862745,0.862745}%
\pgfsetstrokecolor{textcolor}%
\pgfsetfillcolor{textcolor}%
\pgftext[x=1.603603in,y=1.673602in,left,]{\color{textcolor}\rmfamily\fontsize{12.000000}{14.400000}\selectfont better}%
\end{pgfscope}%
\begin{pgfscope}%
\pgfsetbuttcap%
\pgfsetmiterjoin%
\definecolor{currentfill}{rgb}{1.000000,1.000000,1.000000}%
\pgfsetfillcolor{currentfill}%
\pgfsetfillopacity{0.800000}%
\pgfsetlinewidth{0.803000pt}%
\definecolor{currentstroke}{rgb}{0.800000,0.800000,0.800000}%
\pgfsetstrokecolor{currentstroke}%
\pgfsetstrokeopacity{0.800000}%
\pgfsetdash{}{0pt}%
\pgfpathmoveto{\pgfqpoint{2.178907in}{1.983042in}}%
\pgfpathlineto{\pgfqpoint{3.412248in}{1.983042in}}%
\pgfpathquadraticcurveto{\pgfqpoint{3.436692in}{1.983042in}}{\pgfqpoint{3.436692in}{2.007486in}}%
\pgfpathlineto{\pgfqpoint{3.436692in}{2.339708in}}%
\pgfpathquadraticcurveto{\pgfqpoint{3.436692in}{2.364153in}}{\pgfqpoint{3.412248in}{2.364153in}}%
\pgfpathlineto{\pgfqpoint{2.178907in}{2.364153in}}%
\pgfpathquadraticcurveto{\pgfqpoint{2.154462in}{2.364153in}}{\pgfqpoint{2.154462in}{2.339708in}}%
\pgfpathlineto{\pgfqpoint{2.154462in}{2.007486in}}%
\pgfpathquadraticcurveto{\pgfqpoint{2.154462in}{1.983042in}}{\pgfqpoint{2.178907in}{1.983042in}}%
\pgfpathlineto{\pgfqpoint{2.178907in}{1.983042in}}%
\pgfpathclose%
\pgfusepath{stroke,fill}%
\end{pgfscope}%
\begin{pgfscope}%
\pgfsetbuttcap%
\pgfsetroundjoin%
\definecolor{currentfill}{rgb}{0.121569,0.466667,0.705882}%
\pgfsetfillcolor{currentfill}%
\pgfsetlinewidth{0.803000pt}%
\definecolor{currentstroke}{rgb}{0.121569,0.466667,0.705882}%
\pgfsetstrokecolor{currentstroke}%
\pgfsetdash{}{0pt}%
\pgfsys@defobject{currentmarker}{\pgfqpoint{-0.038036in}{-0.038036in}}{\pgfqpoint{0.038036in}{0.038036in}}{%
\pgfpathmoveto{\pgfqpoint{0.000000in}{-0.038036in}}%
\pgfpathcurveto{\pgfqpoint{0.010087in}{-0.038036in}}{\pgfqpoint{0.019763in}{-0.034029in}}{\pgfqpoint{0.026896in}{-0.026896in}}%
\pgfpathcurveto{\pgfqpoint{0.034029in}{-0.019763in}}{\pgfqpoint{0.038036in}{-0.010087in}}{\pgfqpoint{0.038036in}{0.000000in}}%
\pgfpathcurveto{\pgfqpoint{0.038036in}{0.010087in}}{\pgfqpoint{0.034029in}{0.019763in}}{\pgfqpoint{0.026896in}{0.026896in}}%
\pgfpathcurveto{\pgfqpoint{0.019763in}{0.034029in}}{\pgfqpoint{0.010087in}{0.038036in}}{\pgfqpoint{0.000000in}{0.038036in}}%
\pgfpathcurveto{\pgfqpoint{-0.010087in}{0.038036in}}{\pgfqpoint{-0.019763in}{0.034029in}}{\pgfqpoint{-0.026896in}{0.026896in}}%
\pgfpathcurveto{\pgfqpoint{-0.034029in}{0.019763in}}{\pgfqpoint{-0.038036in}{0.010087in}}{\pgfqpoint{-0.038036in}{0.000000in}}%
\pgfpathcurveto{\pgfqpoint{-0.038036in}{-0.010087in}}{\pgfqpoint{-0.034029in}{-0.019763in}}{\pgfqpoint{-0.026896in}{-0.026896in}}%
\pgfpathcurveto{\pgfqpoint{-0.019763in}{-0.034029in}}{\pgfqpoint{-0.010087in}{-0.038036in}}{\pgfqpoint{0.000000in}{-0.038036in}}%
\pgfpathlineto{\pgfqpoint{0.000000in}{-0.038036in}}%
\pgfpathclose%
\pgfusepath{stroke,fill}%
}%
\begin{pgfscope}%
\pgfsys@transformshift{2.325573in}{2.260542in}%
\pgfsys@useobject{currentmarker}{}%
\end{pgfscope}%
\end{pgfscope}%
\begin{pgfscope}%
\definecolor{textcolor}{rgb}{0.150000,0.150000,0.150000}%
\pgfsetstrokecolor{textcolor}%
\pgfsetfillcolor{textcolor}%
\pgftext[x=2.545573in,y=2.228458in,left,base]{\color{textcolor}\rmfamily\fontsize{8.800000}{10.560000}\selectfont Maddock et al.}%
\end{pgfscope}%
\begin{pgfscope}%
\pgfsetroundcap%
\pgfsetroundjoin%
\pgfsetlinewidth{1.204500pt}%
\definecolor{currentstroke}{rgb}{1.000000,0.498039,0.054902}%
\pgfsetstrokecolor{currentstroke}%
\pgfsetdash{}{0pt}%
\pgfpathmoveto{\pgfqpoint{2.203351in}{2.099014in}}%
\pgfpathlineto{\pgfqpoint{2.325573in}{2.099014in}}%
\pgfpathlineto{\pgfqpoint{2.447796in}{2.099014in}}%
\pgfusepath{stroke}%
\end{pgfscope}%
\begin{pgfscope}%
\definecolor{textcolor}{rgb}{0.150000,0.150000,0.150000}%
\pgfsetstrokecolor{textcolor}%
\pgfsetfillcolor{textcolor}%
\pgftext[x=2.545573in,y=2.056236in,left,base]{\color{textcolor}\rmfamily\fontsize{8.800000}{10.560000}\selectfont S-BDT}%
\end{pgfscope}%
\end{pgfpicture}%
\makeatother%
\endgroup%

%% file: images/adult_lbn_ablation.pgf
\begingroup%
\makeatletter%
\begin{pgfpicture}%
\pgfpathrectangle{\pgfpointorigin}{\pgfqpoint{3.612000in}{2.468667in}}%
\pgfusepath{use as bounding box, clip}%
\begin{pgfscope}%
\pgfsetbuttcap%
\pgfsetmiterjoin%
\definecolor{currentfill}{rgb}{1.000000,1.000000,1.000000}%
\pgfsetfillcolor{currentfill}%
\pgfsetlinewidth{0.000000pt}%
\definecolor{currentstroke}{rgb}{1.000000,1.000000,1.000000}%
\pgfsetstrokecolor{currentstroke}%
\pgfsetdash{}{0pt}%
\pgfpathmoveto{\pgfqpoint{0.000000in}{0.000000in}}%
\pgfpathlineto{\pgfqpoint{3.612000in}{0.000000in}}%
\pgfpathlineto{\pgfqpoint{3.612000in}{2.468667in}}%
\pgfpathlineto{\pgfqpoint{0.000000in}{2.468667in}}%
\pgfpathlineto{\pgfqpoint{0.000000in}{0.000000in}}%
\pgfpathclose%
\pgfusepath{fill}%
\end{pgfscope}%
\begin{pgfscope}%
\pgfsetbuttcap%
\pgfsetmiterjoin%
\definecolor{currentfill}{rgb}{1.000000,1.000000,1.000000}%
\pgfsetfillcolor{currentfill}%
\pgfsetlinewidth{0.000000pt}%
\definecolor{currentstroke}{rgb}{0.000000,0.000000,0.000000}%
\pgfsetstrokecolor{currentstroke}%
\pgfsetstrokeopacity{0.000000}%
\pgfsetdash{}{0pt}%
\pgfpathmoveto{\pgfqpoint{0.522684in}{0.420833in}}%
\pgfpathlineto{\pgfqpoint{3.529921in}{0.420833in}}%
\pgfpathlineto{\pgfqpoint{3.529921in}{2.468667in}}%
\pgfpathlineto{\pgfqpoint{0.522684in}{2.468667in}}%
\pgfpathlineto{\pgfqpoint{0.522684in}{0.420833in}}%
\pgfpathclose%
\pgfusepath{fill}%
\end{pgfscope}%
\begin{pgfscope}%
\pgfpathrectangle{\pgfqpoint{0.522684in}{0.420833in}}{\pgfqpoint{3.007237in}{2.047833in}}%
\pgfusepath{clip}%
\pgfsetroundcap%
\pgfsetroundjoin%
\pgfsetlinewidth{0.803000pt}%
\definecolor{currentstroke}{rgb}{0.800000,0.800000,0.800000}%
\pgfsetstrokecolor{currentstroke}%
\pgfsetdash{}{0pt}%
\pgfpathmoveto{\pgfqpoint{0.522684in}{0.420833in}}%
\pgfpathlineto{\pgfqpoint{0.522684in}{2.468667in}}%
\pgfusepath{stroke}%
\end{pgfscope}%
\begin{pgfscope}%
\definecolor{textcolor}{rgb}{0.150000,0.150000,0.150000}%
\pgfsetstrokecolor{textcolor}%
\pgfsetfillcolor{textcolor}%
\pgftext[x=0.522684in,y=0.305556in,,top]{\color{textcolor}\rmfamily\fontsize{8.800000}{10.560000}\selectfont \(\displaystyle {0.1}\)}%
\end{pgfscope}%
\begin{pgfscope}%
\pgfpathrectangle{\pgfqpoint{0.522684in}{0.420833in}}{\pgfqpoint{3.007237in}{2.047833in}}%
\pgfusepath{clip}%
\pgfsetroundcap%
\pgfsetroundjoin%
\pgfsetlinewidth{0.803000pt}%
\definecolor{currentstroke}{rgb}{0.800000,0.800000,0.800000}%
\pgfsetstrokecolor{currentstroke}%
\pgfsetdash{}{0pt}%
\pgfpathmoveto{\pgfqpoint{0.898588in}{0.420833in}}%
\pgfpathlineto{\pgfqpoint{0.898588in}{2.468667in}}%
\pgfusepath{stroke}%
\end{pgfscope}%
\begin{pgfscope}%
\definecolor{textcolor}{rgb}{0.150000,0.150000,0.150000}%
\pgfsetstrokecolor{textcolor}%
\pgfsetfillcolor{textcolor}%
\pgftext[x=0.898588in,y=0.305556in,,top]{\color{textcolor}\rmfamily\fontsize{8.800000}{10.560000}\selectfont \(\displaystyle {0.2}\)}%
\end{pgfscope}%
\begin{pgfscope}%
\pgfpathrectangle{\pgfqpoint{0.522684in}{0.420833in}}{\pgfqpoint{3.007237in}{2.047833in}}%
\pgfusepath{clip}%
\pgfsetroundcap%
\pgfsetroundjoin%
\pgfsetlinewidth{0.803000pt}%
\definecolor{currentstroke}{rgb}{0.800000,0.800000,0.800000}%
\pgfsetstrokecolor{currentstroke}%
\pgfsetdash{}{0pt}%
\pgfpathmoveto{\pgfqpoint{1.274493in}{0.420833in}}%
\pgfpathlineto{\pgfqpoint{1.274493in}{2.468667in}}%
\pgfusepath{stroke}%
\end{pgfscope}%
\begin{pgfscope}%
\definecolor{textcolor}{rgb}{0.150000,0.150000,0.150000}%
\pgfsetstrokecolor{textcolor}%
\pgfsetfillcolor{textcolor}%
\pgftext[x=1.274493in,y=0.305556in,,top]{\color{textcolor}\rmfamily\fontsize{8.800000}{10.560000}\selectfont \(\displaystyle {0.3}\)}%
\end{pgfscope}%
\begin{pgfscope}%
\pgfpathrectangle{\pgfqpoint{0.522684in}{0.420833in}}{\pgfqpoint{3.007237in}{2.047833in}}%
\pgfusepath{clip}%
\pgfsetroundcap%
\pgfsetroundjoin%
\pgfsetlinewidth{0.803000pt}%
\definecolor{currentstroke}{rgb}{0.800000,0.800000,0.800000}%
\pgfsetstrokecolor{currentstroke}%
\pgfsetdash{}{0pt}%
\pgfpathmoveto{\pgfqpoint{1.650398in}{0.420833in}}%
\pgfpathlineto{\pgfqpoint{1.650398in}{2.468667in}}%
\pgfusepath{stroke}%
\end{pgfscope}%
\begin{pgfscope}%
\definecolor{textcolor}{rgb}{0.150000,0.150000,0.150000}%
\pgfsetstrokecolor{textcolor}%
\pgfsetfillcolor{textcolor}%
\pgftext[x=1.650398in,y=0.305556in,,top]{\color{textcolor}\rmfamily\fontsize{8.800000}{10.560000}\selectfont \(\displaystyle {0.4}\)}%
\end{pgfscope}%
\begin{pgfscope}%
\pgfpathrectangle{\pgfqpoint{0.522684in}{0.420833in}}{\pgfqpoint{3.007237in}{2.047833in}}%
\pgfusepath{clip}%
\pgfsetroundcap%
\pgfsetroundjoin%
\pgfsetlinewidth{0.803000pt}%
\definecolor{currentstroke}{rgb}{0.800000,0.800000,0.800000}%
\pgfsetstrokecolor{currentstroke}%
\pgfsetdash{}{0pt}%
\pgfpathmoveto{\pgfqpoint{2.026302in}{0.420833in}}%
\pgfpathlineto{\pgfqpoint{2.026302in}{2.468667in}}%
\pgfusepath{stroke}%
\end{pgfscope}%
\begin{pgfscope}%
\definecolor{textcolor}{rgb}{0.150000,0.150000,0.150000}%
\pgfsetstrokecolor{textcolor}%
\pgfsetfillcolor{textcolor}%
\pgftext[x=2.026302in,y=0.305556in,,top]{\color{textcolor}\rmfamily\fontsize{8.800000}{10.560000}\selectfont \(\displaystyle {0.5}\)}%
\end{pgfscope}%
\begin{pgfscope}%
\pgfpathrectangle{\pgfqpoint{0.522684in}{0.420833in}}{\pgfqpoint{3.007237in}{2.047833in}}%
\pgfusepath{clip}%
\pgfsetroundcap%
\pgfsetroundjoin%
\pgfsetlinewidth{0.803000pt}%
\definecolor{currentstroke}{rgb}{0.800000,0.800000,0.800000}%
\pgfsetstrokecolor{currentstroke}%
\pgfsetdash{}{0pt}%
\pgfpathmoveto{\pgfqpoint{2.402207in}{0.420833in}}%
\pgfpathlineto{\pgfqpoint{2.402207in}{2.468667in}}%
\pgfusepath{stroke}%
\end{pgfscope}%
\begin{pgfscope}%
\definecolor{textcolor}{rgb}{0.150000,0.150000,0.150000}%
\pgfsetstrokecolor{textcolor}%
\pgfsetfillcolor{textcolor}%
\pgftext[x=2.402207in,y=0.305556in,,top]{\color{textcolor}\rmfamily\fontsize{8.800000}{10.560000}\selectfont \(\displaystyle {0.6}\)}%
\end{pgfscope}%
\begin{pgfscope}%
\pgfpathrectangle{\pgfqpoint{0.522684in}{0.420833in}}{\pgfqpoint{3.007237in}{2.047833in}}%
\pgfusepath{clip}%
\pgfsetroundcap%
\pgfsetroundjoin%
\pgfsetlinewidth{0.803000pt}%
\definecolor{currentstroke}{rgb}{0.800000,0.800000,0.800000}%
\pgfsetstrokecolor{currentstroke}%
\pgfsetdash{}{0pt}%
\pgfpathmoveto{\pgfqpoint{2.778112in}{0.420833in}}%
\pgfpathlineto{\pgfqpoint{2.778112in}{2.468667in}}%
\pgfusepath{stroke}%
\end{pgfscope}%
\begin{pgfscope}%
\definecolor{textcolor}{rgb}{0.150000,0.150000,0.150000}%
\pgfsetstrokecolor{textcolor}%
\pgfsetfillcolor{textcolor}%
\pgftext[x=2.778112in,y=0.305556in,,top]{\color{textcolor}\rmfamily\fontsize{8.800000}{10.560000}\selectfont \(\displaystyle {0.7}\)}%
\end{pgfscope}%
\begin{pgfscope}%
\pgfpathrectangle{\pgfqpoint{0.522684in}{0.420833in}}{\pgfqpoint{3.007237in}{2.047833in}}%
\pgfusepath{clip}%
\pgfsetroundcap%
\pgfsetroundjoin%
\pgfsetlinewidth{0.803000pt}%
\definecolor{currentstroke}{rgb}{0.800000,0.800000,0.800000}%
\pgfsetstrokecolor{currentstroke}%
\pgfsetdash{}{0pt}%
\pgfpathmoveto{\pgfqpoint{3.154016in}{0.420833in}}%
\pgfpathlineto{\pgfqpoint{3.154016in}{2.468667in}}%
\pgfusepath{stroke}%
\end{pgfscope}%
\begin{pgfscope}%
\definecolor{textcolor}{rgb}{0.150000,0.150000,0.150000}%
\pgfsetstrokecolor{textcolor}%
\pgfsetfillcolor{textcolor}%
\pgftext[x=3.154016in,y=0.305556in,,top]{\color{textcolor}\rmfamily\fontsize{8.800000}{10.560000}\selectfont \(\displaystyle {0.8}\)}%
\end{pgfscope}%
\begin{pgfscope}%
\pgfpathrectangle{\pgfqpoint{0.522684in}{0.420833in}}{\pgfqpoint{3.007237in}{2.047833in}}%
\pgfusepath{clip}%
\pgfsetroundcap%
\pgfsetroundjoin%
\pgfsetlinewidth{0.803000pt}%
\definecolor{currentstroke}{rgb}{0.800000,0.800000,0.800000}%
\pgfsetstrokecolor{currentstroke}%
\pgfsetdash{}{0pt}%
\pgfpathmoveto{\pgfqpoint{3.529921in}{0.420833in}}%
\pgfpathlineto{\pgfqpoint{3.529921in}{2.468667in}}%
\pgfusepath{stroke}%
\end{pgfscope}%
\begin{pgfscope}%
\definecolor{textcolor}{rgb}{0.150000,0.150000,0.150000}%
\pgfsetstrokecolor{textcolor}%
\pgfsetfillcolor{textcolor}%
\pgftext[x=3.529921in,y=0.305556in,,top]{\color{textcolor}\rmfamily\fontsize{8.800000}{10.560000}\selectfont \(\displaystyle {0.9}\)}%
\end{pgfscope}%
\begin{pgfscope}%
\definecolor{textcolor}{rgb}{0.150000,0.150000,0.150000}%
\pgfsetstrokecolor{textcolor}%
\pgfsetfillcolor{textcolor}%
\pgftext[x=2.026302in,y=0.138889in,,top]{\color{textcolor}\rmfamily\fontsize{9.600000}{11.520000}\selectfont \(\displaystyle r_1\) (leaf-balanced noise parameter)}%
\end{pgfscope}%
\begin{pgfscope}%
\pgfpathrectangle{\pgfqpoint{0.522684in}{0.420833in}}{\pgfqpoint{3.007237in}{2.047833in}}%
\pgfusepath{clip}%
\pgfsetroundcap%
\pgfsetroundjoin%
\pgfsetlinewidth{0.803000pt}%
\definecolor{currentstroke}{rgb}{0.800000,0.800000,0.800000}%
\pgfsetstrokecolor{currentstroke}%
\pgfsetdash{}{0pt}%
\pgfpathmoveto{\pgfqpoint{0.522684in}{0.625617in}}%
\pgfpathlineto{\pgfqpoint{3.529921in}{0.625617in}}%
\pgfusepath{stroke}%
\end{pgfscope}%
\begin{pgfscope}%
\definecolor{textcolor}{rgb}{0.150000,0.150000,0.150000}%
\pgfsetstrokecolor{textcolor}%
\pgfsetfillcolor{textcolor}%
\pgftext[x=0.179012in, y=0.582214in, left, base]{\color{textcolor}\rmfamily\fontsize{8.800000}{10.560000}\selectfont \(\displaystyle {0.76}\)}%
\end{pgfscope}%
\begin{pgfscope}%
\pgfpathrectangle{\pgfqpoint{0.522684in}{0.420833in}}{\pgfqpoint{3.007237in}{2.047833in}}%
\pgfusepath{clip}%
\pgfsetroundcap%
\pgfsetroundjoin%
\pgfsetlinewidth{0.803000pt}%
\definecolor{currentstroke}{rgb}{0.800000,0.800000,0.800000}%
\pgfsetstrokecolor{currentstroke}%
\pgfsetdash{}{0pt}%
\pgfpathmoveto{\pgfqpoint{0.522684in}{1.035183in}}%
\pgfpathlineto{\pgfqpoint{3.529921in}{1.035183in}}%
\pgfusepath{stroke}%
\end{pgfscope}%
\begin{pgfscope}%
\definecolor{textcolor}{rgb}{0.150000,0.150000,0.150000}%
\pgfsetstrokecolor{textcolor}%
\pgfsetfillcolor{textcolor}%
\pgftext[x=0.179012in, y=0.991781in, left, base]{\color{textcolor}\rmfamily\fontsize{8.800000}{10.560000}\selectfont \(\displaystyle {0.78}\)}%
\end{pgfscope}%
\begin{pgfscope}%
\pgfpathrectangle{\pgfqpoint{0.522684in}{0.420833in}}{\pgfqpoint{3.007237in}{2.047833in}}%
\pgfusepath{clip}%
\pgfsetroundcap%
\pgfsetroundjoin%
\pgfsetlinewidth{0.803000pt}%
\definecolor{currentstroke}{rgb}{0.800000,0.800000,0.800000}%
\pgfsetstrokecolor{currentstroke}%
\pgfsetdash{}{0pt}%
\pgfpathmoveto{\pgfqpoint{0.522684in}{1.444750in}}%
\pgfpathlineto{\pgfqpoint{3.529921in}{1.444750in}}%
\pgfusepath{stroke}%
\end{pgfscope}%
\begin{pgfscope}%
\definecolor{textcolor}{rgb}{0.150000,0.150000,0.150000}%
\pgfsetstrokecolor{textcolor}%
\pgfsetfillcolor{textcolor}%
\pgftext[x=0.179012in, y=1.401347in, left, base]{\color{textcolor}\rmfamily\fontsize{8.800000}{10.560000}\selectfont \(\displaystyle {0.80}\)}%
\end{pgfscope}%
\begin{pgfscope}%
\pgfpathrectangle{\pgfqpoint{0.522684in}{0.420833in}}{\pgfqpoint{3.007237in}{2.047833in}}%
\pgfusepath{clip}%
\pgfsetroundcap%
\pgfsetroundjoin%
\pgfsetlinewidth{0.803000pt}%
\definecolor{currentstroke}{rgb}{0.800000,0.800000,0.800000}%
\pgfsetstrokecolor{currentstroke}%
\pgfsetdash{}{0pt}%
\pgfpathmoveto{\pgfqpoint{0.522684in}{1.854317in}}%
\pgfpathlineto{\pgfqpoint{3.529921in}{1.854317in}}%
\pgfusepath{stroke}%
\end{pgfscope}%
\begin{pgfscope}%
\definecolor{textcolor}{rgb}{0.150000,0.150000,0.150000}%
\pgfsetstrokecolor{textcolor}%
\pgfsetfillcolor{textcolor}%
\pgftext[x=0.179012in, y=1.810914in, left, base]{\color{textcolor}\rmfamily\fontsize{8.800000}{10.560000}\selectfont \(\displaystyle {0.82}\)}%
\end{pgfscope}%
\begin{pgfscope}%
\pgfpathrectangle{\pgfqpoint{0.522684in}{0.420833in}}{\pgfqpoint{3.007237in}{2.047833in}}%
\pgfusepath{clip}%
\pgfsetroundcap%
\pgfsetroundjoin%
\pgfsetlinewidth{0.803000pt}%
\definecolor{currentstroke}{rgb}{0.800000,0.800000,0.800000}%
\pgfsetstrokecolor{currentstroke}%
\pgfsetdash{}{0pt}%
\pgfpathmoveto{\pgfqpoint{0.522684in}{2.263883in}}%
\pgfpathlineto{\pgfqpoint{3.529921in}{2.263883in}}%
\pgfusepath{stroke}%
\end{pgfscope}%
\begin{pgfscope}%
\definecolor{textcolor}{rgb}{0.150000,0.150000,0.150000}%
\pgfsetstrokecolor{textcolor}%
\pgfsetfillcolor{textcolor}%
\pgftext[x=0.179012in, y=2.220481in, left, base]{\color{textcolor}\rmfamily\fontsize{8.800000}{10.560000}\selectfont \(\displaystyle {0.84}\)}%
\end{pgfscope}%
\begin{pgfscope}%
\definecolor{textcolor}{rgb}{0.150000,0.150000,0.150000}%
\pgfsetstrokecolor{textcolor}%
\pgfsetfillcolor{textcolor}%
\pgftext[x=0.123457in,y=1.444750in,,bottom,rotate=90.000000]{\color{textcolor}\rmfamily\fontsize{9.600000}{11.520000}\selectfont Mean test AUC}%
\end{pgfscope}%
\begin{pgfscope}%
\pgfpathrectangle{\pgfqpoint{0.522684in}{0.420833in}}{\pgfqpoint{3.007237in}{2.047833in}}%
\pgfusepath{clip}%
\pgfsetbuttcap%
\pgfsetroundjoin%
\definecolor{currentfill}{rgb}{0.121569,0.466667,0.705882}%
\pgfsetfillcolor{currentfill}%
\pgfsetlinewidth{0.803000pt}%
\definecolor{currentstroke}{rgb}{0.121569,0.466667,0.705882}%
\pgfsetstrokecolor{currentstroke}%
\pgfsetdash{}{0pt}%
\pgfsys@defobject{currentmarker}{\pgfqpoint{-0.038036in}{-0.038036in}}{\pgfqpoint{0.038036in}{0.038036in}}{%
\pgfpathmoveto{\pgfqpoint{0.000000in}{-0.038036in}}%
\pgfpathcurveto{\pgfqpoint{0.010087in}{-0.038036in}}{\pgfqpoint{0.019763in}{-0.034029in}}{\pgfqpoint{0.026896in}{-0.026896in}}%
\pgfpathcurveto{\pgfqpoint{0.034029in}{-0.019763in}}{\pgfqpoint{0.038036in}{-0.010087in}}{\pgfqpoint{0.038036in}{0.000000in}}%
\pgfpathcurveto{\pgfqpoint{0.038036in}{0.010087in}}{\pgfqpoint{0.034029in}{0.019763in}}{\pgfqpoint{0.026896in}{0.026896in}}%
\pgfpathcurveto{\pgfqpoint{0.019763in}{0.034029in}}{\pgfqpoint{0.010087in}{0.038036in}}{\pgfqpoint{0.000000in}{0.038036in}}%
\pgfpathcurveto{\pgfqpoint{-0.010087in}{0.038036in}}{\pgfqpoint{-0.019763in}{0.034029in}}{\pgfqpoint{-0.026896in}{0.026896in}}%
\pgfpathcurveto{\pgfqpoint{-0.034029in}{0.019763in}}{\pgfqpoint{-0.038036in}{0.010087in}}{\pgfqpoint{-0.038036in}{0.000000in}}%
\pgfpathcurveto{\pgfqpoint{-0.038036in}{-0.010087in}}{\pgfqpoint{-0.034029in}{-0.019763in}}{\pgfqpoint{-0.026896in}{-0.026896in}}%
\pgfpathcurveto{\pgfqpoint{-0.019763in}{-0.034029in}}{\pgfqpoint{-0.010087in}{-0.038036in}}{\pgfqpoint{0.000000in}{-0.038036in}}%
\pgfpathlineto{\pgfqpoint{0.000000in}{-0.038036in}}%
\pgfpathclose%
\pgfusepath{stroke,fill}%
}%
\begin{pgfscope}%
\pgfsys@transformshift{2.026302in}{1.260445in}%
\pgfsys@useobject{currentmarker}{}%
\end{pgfscope}%
\end{pgfscope}%
\begin{pgfscope}%
\pgfpathrectangle{\pgfqpoint{0.522684in}{0.420833in}}{\pgfqpoint{3.007237in}{2.047833in}}%
\pgfusepath{clip}%
\pgfsetbuttcap%
\pgfsetroundjoin%
\definecolor{currentfill}{rgb}{1.000000,0.498039,0.054902}%
\pgfsetfillcolor{currentfill}%
\pgfsetfillopacity{0.100000}%
\pgfsetlinewidth{0.803000pt}%
\definecolor{currentstroke}{rgb}{1.000000,0.498039,0.054902}%
\pgfsetstrokecolor{currentstroke}%
\pgfsetstrokeopacity{0.100000}%
\pgfsetdash{}{0pt}%
\pgfpathmoveto{\pgfqpoint{0.297141in}{1.813360in}}%
\pgfpathlineto{\pgfqpoint{0.297141in}{1.731447in}}%
\pgfpathlineto{\pgfqpoint{0.522684in}{1.895273in}}%
\pgfpathlineto{\pgfqpoint{0.898588in}{1.751925in}}%
\pgfpathlineto{\pgfqpoint{1.274493in}{1.608577in}}%
\pgfpathlineto{\pgfqpoint{1.650398in}{1.526663in}}%
\pgfpathlineto{\pgfqpoint{2.026302in}{1.321880in}}%
\pgfpathlineto{\pgfqpoint{2.402207in}{1.199010in}}%
\pgfpathlineto{\pgfqpoint{2.778112in}{0.728008in}}%
\pgfpathlineto{\pgfqpoint{3.154016in}{0.277485in}}%
\pgfpathlineto{\pgfqpoint{3.529921in}{-0.684997in}}%
\pgfpathlineto{\pgfqpoint{3.529921in}{-0.316387in}}%
\pgfpathlineto{\pgfqpoint{3.529921in}{-0.316387in}}%
\pgfpathlineto{\pgfqpoint{3.154016in}{0.564182in}}%
\pgfpathlineto{\pgfqpoint{2.778112in}{0.932792in}}%
\pgfpathlineto{\pgfqpoint{2.402207in}{1.321880in}}%
\pgfpathlineto{\pgfqpoint{2.026302in}{1.485707in}}%
\pgfpathlineto{\pgfqpoint{1.650398in}{1.649533in}}%
\pgfpathlineto{\pgfqpoint{1.274493in}{1.731447in}}%
\pgfpathlineto{\pgfqpoint{0.898588in}{1.833838in}}%
\pgfpathlineto{\pgfqpoint{0.522684in}{2.018143in}}%
\pgfpathlineto{\pgfqpoint{0.297141in}{1.813360in}}%
\pgfpathlineto{\pgfqpoint{0.297141in}{1.813360in}}%
\pgfpathclose%
\pgfusepath{stroke,fill}%
\end{pgfscope}%
\begin{pgfscope}%
\pgfpathrectangle{\pgfqpoint{0.522684in}{0.420833in}}{\pgfqpoint{3.007237in}{2.047833in}}%
\pgfusepath{clip}%
\pgfsetroundcap%
\pgfsetroundjoin%
\pgfsetlinewidth{1.204500pt}%
\definecolor{currentstroke}{rgb}{1.000000,0.498039,0.054902}%
\pgfsetstrokecolor{currentstroke}%
\pgfsetdash{}{0pt}%
\pgfpathmoveto{\pgfqpoint{0.512684in}{1.948537in}}%
\pgfpathlineto{\pgfqpoint{0.522684in}{1.956708in}}%
\pgfpathlineto{\pgfqpoint{0.898588in}{1.792882in}}%
\pgfpathlineto{\pgfqpoint{1.274493in}{1.670012in}}%
\pgfpathlineto{\pgfqpoint{1.650398in}{1.588098in}}%
\pgfpathlineto{\pgfqpoint{2.026302in}{1.403793in}}%
\pgfpathlineto{\pgfqpoint{2.402207in}{1.260445in}}%
\pgfpathlineto{\pgfqpoint{2.778112in}{0.830400in}}%
\pgfpathlineto{\pgfqpoint{3.154016in}{0.420833in}}%
\pgfpathlineto{\pgfqpoint{3.158096in}{0.410833in}}%
\pgfusepath{stroke}%
\end{pgfscope}%
\begin{pgfscope}%
\pgfsetrectcap%
\pgfsetmiterjoin%
\pgfsetlinewidth{1.003750pt}%
\definecolor{currentstroke}{rgb}{0.800000,0.800000,0.800000}%
\pgfsetstrokecolor{currentstroke}%
\pgfsetdash{}{0pt}%
\pgfpathmoveto{\pgfqpoint{0.522684in}{0.420833in}}%
\pgfpathlineto{\pgfqpoint{0.522684in}{2.468667in}}%
\pgfusepath{stroke}%
\end{pgfscope}%
\begin{pgfscope}%
\pgfsetrectcap%
\pgfsetmiterjoin%
\pgfsetlinewidth{1.003750pt}%
\definecolor{currentstroke}{rgb}{0.800000,0.800000,0.800000}%
\pgfsetstrokecolor{currentstroke}%
\pgfsetdash{}{0pt}%
\pgfpathmoveto{\pgfqpoint{0.522684in}{0.420833in}}%
\pgfpathlineto{\pgfqpoint{3.529921in}{0.420833in}}%
\pgfusepath{stroke}%
\end{pgfscope}%
\begin{pgfscope}%
\pgfsetroundcap%
\pgfsetroundjoin%
\definecolor{currentfill}{rgb}{0.862745,0.862745,0.862745}%
\pgfsetfillcolor{currentfill}%
\pgfsetlinewidth{0.803000pt}%
\definecolor{currentstroke}{rgb}{1.000000,1.000000,1.000000}%
\pgfsetstrokecolor{currentstroke}%
\pgfsetdash{}{0pt}%
\pgfpathmoveto{\pgfqpoint{0.788685in}{0.666563in}}%
\pgfpathquadraticcurveto{\pgfqpoint{0.788685in}{0.951809in}}{\pgfqpoint{0.788685in}{1.237054in}}%
\pgfpathlineto{\pgfqpoint{0.740074in}{1.237054in}}%
\pgfpathquadraticcurveto{\pgfqpoint{0.781741in}{1.320407in}}{\pgfqpoint{0.823407in}{1.403760in}}%
\pgfpathquadraticcurveto{\pgfqpoint{0.865074in}{1.320407in}}{\pgfqpoint{0.906741in}{1.237054in}}%
\pgfpathlineto{\pgfqpoint{0.858130in}{1.237054in}}%
\pgfpathquadraticcurveto{\pgfqpoint{0.858130in}{0.951809in}}{\pgfqpoint{0.858130in}{0.666563in}}%
\pgfpathlineto{\pgfqpoint{0.788685in}{0.666563in}}%
\pgfpathlineto{\pgfqpoint{0.788685in}{0.666563in}}%
\pgfpathclose%
\pgfusepath{stroke,fill}%
\end{pgfscope}%
\begin{pgfscope}%
\definecolor{textcolor}{rgb}{0.862745,0.862745,0.862745}%
\pgfsetstrokecolor{textcolor}%
\pgfsetfillcolor{textcolor}%
\pgftext[x=1.011360in,y=1.035183in,left,]{\color{textcolor}\rmfamily\fontsize{12.000000}{14.400000}\selectfont better}%
\end{pgfscope}%
\begin{pgfscope}%
\pgfsetbuttcap%
\pgfsetmiterjoin%
\definecolor{currentfill}{rgb}{1.000000,1.000000,1.000000}%
\pgfsetfillcolor{currentfill}%
\pgfsetfillopacity{0.800000}%
\pgfsetlinewidth{0.803000pt}%
\definecolor{currentstroke}{rgb}{0.800000,0.800000,0.800000}%
\pgfsetstrokecolor{currentstroke}%
\pgfsetstrokeopacity{0.800000}%
\pgfsetdash{}{0pt}%
\pgfpathmoveto{\pgfqpoint{2.211025in}{2.026444in}}%
\pgfpathlineto{\pgfqpoint{3.444365in}{2.026444in}}%
\pgfpathquadraticcurveto{\pgfqpoint{3.468810in}{2.026444in}}{\pgfqpoint{3.468810in}{2.050889in}}%
\pgfpathlineto{\pgfqpoint{3.468810in}{2.383111in}}%
\pgfpathquadraticcurveto{\pgfqpoint{3.468810in}{2.407556in}}{\pgfqpoint{3.444365in}{2.407556in}}%
\pgfpathlineto{\pgfqpoint{2.211025in}{2.407556in}}%
\pgfpathquadraticcurveto{\pgfqpoint{2.186580in}{2.407556in}}{\pgfqpoint{2.186580in}{2.383111in}}%
\pgfpathlineto{\pgfqpoint{2.186580in}{2.050889in}}%
\pgfpathquadraticcurveto{\pgfqpoint{2.186580in}{2.026444in}}{\pgfqpoint{2.211025in}{2.026444in}}%
\pgfpathlineto{\pgfqpoint{2.211025in}{2.026444in}}%
\pgfpathclose%
\pgfusepath{stroke,fill}%
\end{pgfscope}%
\begin{pgfscope}%
\pgfsetbuttcap%
\pgfsetroundjoin%
\definecolor{currentfill}{rgb}{0.121569,0.466667,0.705882}%
\pgfsetfillcolor{currentfill}%
\pgfsetlinewidth{0.803000pt}%
\definecolor{currentstroke}{rgb}{0.121569,0.466667,0.705882}%
\pgfsetstrokecolor{currentstroke}%
\pgfsetdash{}{0pt}%
\pgfsys@defobject{currentmarker}{\pgfqpoint{-0.038036in}{-0.038036in}}{\pgfqpoint{0.038036in}{0.038036in}}{%
\pgfpathmoveto{\pgfqpoint{0.000000in}{-0.038036in}}%
\pgfpathcurveto{\pgfqpoint{0.010087in}{-0.038036in}}{\pgfqpoint{0.019763in}{-0.034029in}}{\pgfqpoint{0.026896in}{-0.026896in}}%
\pgfpathcurveto{\pgfqpoint{0.034029in}{-0.019763in}}{\pgfqpoint{0.038036in}{-0.010087in}}{\pgfqpoint{0.038036in}{0.000000in}}%
\pgfpathcurveto{\pgfqpoint{0.038036in}{0.010087in}}{\pgfqpoint{0.034029in}{0.019763in}}{\pgfqpoint{0.026896in}{0.026896in}}%
\pgfpathcurveto{\pgfqpoint{0.019763in}{0.034029in}}{\pgfqpoint{0.010087in}{0.038036in}}{\pgfqpoint{0.000000in}{0.038036in}}%
\pgfpathcurveto{\pgfqpoint{-0.010087in}{0.038036in}}{\pgfqpoint{-0.019763in}{0.034029in}}{\pgfqpoint{-0.026896in}{0.026896in}}%
\pgfpathcurveto{\pgfqpoint{-0.034029in}{0.019763in}}{\pgfqpoint{-0.038036in}{0.010087in}}{\pgfqpoint{-0.038036in}{0.000000in}}%
\pgfpathcurveto{\pgfqpoint{-0.038036in}{-0.010087in}}{\pgfqpoint{-0.034029in}{-0.019763in}}{\pgfqpoint{-0.026896in}{-0.026896in}}%
\pgfpathcurveto{\pgfqpoint{-0.019763in}{-0.034029in}}{\pgfqpoint{-0.010087in}{-0.038036in}}{\pgfqpoint{0.000000in}{-0.038036in}}%
\pgfpathlineto{\pgfqpoint{0.000000in}{-0.038036in}}%
\pgfpathclose%
\pgfusepath{stroke,fill}%
}%
\begin{pgfscope}%
\pgfsys@transformshift{2.357691in}{2.303944in}%
\pgfsys@useobject{currentmarker}{}%
\end{pgfscope}%
\end{pgfscope}%
\begin{pgfscope}%
\definecolor{textcolor}{rgb}{0.150000,0.150000,0.150000}%
\pgfsetstrokecolor{textcolor}%
\pgfsetfillcolor{textcolor}%
\pgftext[x=2.577691in,y=2.271861in,left,base]{\color{textcolor}\rmfamily\fontsize{8.800000}{10.560000}\selectfont Maddock et al.}%
\end{pgfscope}%
\begin{pgfscope}%
\pgfsetroundcap%
\pgfsetroundjoin%
\pgfsetlinewidth{1.204500pt}%
\definecolor{currentstroke}{rgb}{1.000000,0.498039,0.054902}%
\pgfsetstrokecolor{currentstroke}%
\pgfsetdash{}{0pt}%
\pgfpathmoveto{\pgfqpoint{2.235469in}{2.142417in}}%
\pgfpathlineto{\pgfqpoint{2.357691in}{2.142417in}}%
\pgfpathlineto{\pgfqpoint{2.479913in}{2.142417in}}%
\pgfusepath{stroke}%
\end{pgfscope}%
\begin{pgfscope}%
\definecolor{textcolor}{rgb}{0.150000,0.150000,0.150000}%
\pgfsetstrokecolor{textcolor}%
\pgfsetfillcolor{textcolor}%
\pgftext[x=2.577691in,y=2.099639in,left,base]{\color{textcolor}\rmfamily\fontsize{8.800000}{10.560000}\selectfont S-BDT}%
\end{pgfscope}%
\end{pgfpicture}%
\makeatother%
\endgroup%

%% file: images/adult_subsampling_ablation.pgf
\begingroup%
\makeatletter%
\begin{pgfpicture}%
\pgfpathrectangle{\pgfpointorigin}{\pgfqpoint{3.612000in}{2.468667in}}%
\pgfusepath{use as bounding box, clip}%
\begin{pgfscope}%
\pgfsetbuttcap%
\pgfsetmiterjoin%
\definecolor{currentfill}{rgb}{1.000000,1.000000,1.000000}%
\pgfsetfillcolor{currentfill}%
\pgfsetlinewidth{0.000000pt}%
\definecolor{currentstroke}{rgb}{1.000000,1.000000,1.000000}%
\pgfsetstrokecolor{currentstroke}%
\pgfsetdash{}{0pt}%
\pgfpathmoveto{\pgfqpoint{0.000000in}{0.000000in}}%
\pgfpathlineto{\pgfqpoint{3.612000in}{0.000000in}}%
\pgfpathlineto{\pgfqpoint{3.612000in}{2.468667in}}%
\pgfpathlineto{\pgfqpoint{0.000000in}{2.468667in}}%
\pgfpathlineto{\pgfqpoint{0.000000in}{0.000000in}}%
\pgfpathclose%
\pgfusepath{fill}%
\end{pgfscope}%
\begin{pgfscope}%
\pgfsetbuttcap%
\pgfsetmiterjoin%
\definecolor{currentfill}{rgb}{1.000000,1.000000,1.000000}%
\pgfsetfillcolor{currentfill}%
\pgfsetlinewidth{0.000000pt}%
\definecolor{currentstroke}{rgb}{0.000000,0.000000,0.000000}%
\pgfsetstrokecolor{currentstroke}%
\pgfsetstrokeopacity{0.000000}%
\pgfsetdash{}{0pt}%
\pgfpathmoveto{\pgfqpoint{0.522684in}{0.420833in}}%
\pgfpathlineto{\pgfqpoint{3.465685in}{0.420833in}}%
\pgfpathlineto{\pgfqpoint{3.465685in}{2.468667in}}%
\pgfpathlineto{\pgfqpoint{0.522684in}{2.468667in}}%
\pgfpathlineto{\pgfqpoint{0.522684in}{0.420833in}}%
\pgfpathclose%
\pgfusepath{fill}%
\end{pgfscope}%
\begin{pgfscope}%
\pgfpathrectangle{\pgfqpoint{0.522684in}{0.420833in}}{\pgfqpoint{2.943002in}{2.047833in}}%
\pgfusepath{clip}%
\pgfsetroundcap%
\pgfsetroundjoin%
\pgfsetlinewidth{0.803000pt}%
\definecolor{currentstroke}{rgb}{0.800000,0.800000,0.800000}%
\pgfsetstrokecolor{currentstroke}%
\pgfsetdash{}{0pt}%
\pgfpathmoveto{\pgfqpoint{0.537399in}{0.420833in}}%
\pgfpathlineto{\pgfqpoint{0.537399in}{2.468667in}}%
\pgfusepath{stroke}%
\end{pgfscope}%
\begin{pgfscope}%
\definecolor{textcolor}{rgb}{0.150000,0.150000,0.150000}%
\pgfsetstrokecolor{textcolor}%
\pgfsetfillcolor{textcolor}%
\pgftext[x=0.537399in,y=0.305556in,,top]{\color{textcolor}\rmfamily\fontsize{8.800000}{10.560000}\selectfont \(\displaystyle {0.005}\)}%
\end{pgfscope}%
\begin{pgfscope}%
\pgfpathrectangle{\pgfqpoint{0.522684in}{0.420833in}}{\pgfqpoint{2.943002in}{2.047833in}}%
\pgfusepath{clip}%
\pgfsetroundcap%
\pgfsetroundjoin%
\pgfsetlinewidth{0.803000pt}%
\definecolor{currentstroke}{rgb}{0.800000,0.800000,0.800000}%
\pgfsetstrokecolor{currentstroke}%
\pgfsetdash{}{0pt}%
\pgfpathmoveto{\pgfqpoint{1.111284in}{0.420833in}}%
\pgfpathlineto{\pgfqpoint{1.111284in}{2.468667in}}%
\pgfusepath{stroke}%
\end{pgfscope}%
\begin{pgfscope}%
\definecolor{textcolor}{rgb}{0.150000,0.150000,0.150000}%
\pgfsetstrokecolor{textcolor}%
\pgfsetfillcolor{textcolor}%
\pgftext[x=1.111284in,y=0.305556in,,top]{\color{textcolor}\rmfamily\fontsize{8.800000}{10.560000}\selectfont \(\displaystyle {0.200}\)}%
\end{pgfscope}%
\begin{pgfscope}%
\pgfpathrectangle{\pgfqpoint{0.522684in}{0.420833in}}{\pgfqpoint{2.943002in}{2.047833in}}%
\pgfusepath{clip}%
\pgfsetroundcap%
\pgfsetroundjoin%
\pgfsetlinewidth{0.803000pt}%
\definecolor{currentstroke}{rgb}{0.800000,0.800000,0.800000}%
\pgfsetstrokecolor{currentstroke}%
\pgfsetdash{}{0pt}%
\pgfpathmoveto{\pgfqpoint{1.699884in}{0.420833in}}%
\pgfpathlineto{\pgfqpoint{1.699884in}{2.468667in}}%
\pgfusepath{stroke}%
\end{pgfscope}%
\begin{pgfscope}%
\definecolor{textcolor}{rgb}{0.150000,0.150000,0.150000}%
\pgfsetstrokecolor{textcolor}%
\pgfsetfillcolor{textcolor}%
\pgftext[x=1.699884in,y=0.305556in,,top]{\color{textcolor}\rmfamily\fontsize{8.800000}{10.560000}\selectfont \(\displaystyle {0.400}\)}%
\end{pgfscope}%
\begin{pgfscope}%
\pgfpathrectangle{\pgfqpoint{0.522684in}{0.420833in}}{\pgfqpoint{2.943002in}{2.047833in}}%
\pgfusepath{clip}%
\pgfsetroundcap%
\pgfsetroundjoin%
\pgfsetlinewidth{0.803000pt}%
\definecolor{currentstroke}{rgb}{0.800000,0.800000,0.800000}%
\pgfsetstrokecolor{currentstroke}%
\pgfsetdash{}{0pt}%
\pgfpathmoveto{\pgfqpoint{2.288485in}{0.420833in}}%
\pgfpathlineto{\pgfqpoint{2.288485in}{2.468667in}}%
\pgfusepath{stroke}%
\end{pgfscope}%
\begin{pgfscope}%
\definecolor{textcolor}{rgb}{0.150000,0.150000,0.150000}%
\pgfsetstrokecolor{textcolor}%
\pgfsetfillcolor{textcolor}%
\pgftext[x=2.288485in,y=0.305556in,,top]{\color{textcolor}\rmfamily\fontsize{8.800000}{10.560000}\selectfont \(\displaystyle {0.600}\)}%
\end{pgfscope}%
\begin{pgfscope}%
\pgfpathrectangle{\pgfqpoint{0.522684in}{0.420833in}}{\pgfqpoint{2.943002in}{2.047833in}}%
\pgfusepath{clip}%
\pgfsetroundcap%
\pgfsetroundjoin%
\pgfsetlinewidth{0.803000pt}%
\definecolor{currentstroke}{rgb}{0.800000,0.800000,0.800000}%
\pgfsetstrokecolor{currentstroke}%
\pgfsetdash{}{0pt}%
\pgfpathmoveto{\pgfqpoint{2.877085in}{0.420833in}}%
\pgfpathlineto{\pgfqpoint{2.877085in}{2.468667in}}%
\pgfusepath{stroke}%
\end{pgfscope}%
\begin{pgfscope}%
\definecolor{textcolor}{rgb}{0.150000,0.150000,0.150000}%
\pgfsetstrokecolor{textcolor}%
\pgfsetfillcolor{textcolor}%
\pgftext[x=2.877085in,y=0.305556in,,top]{\color{textcolor}\rmfamily\fontsize{8.800000}{10.560000}\selectfont \(\displaystyle {0.800}\)}%
\end{pgfscope}%
\begin{pgfscope}%
\pgfpathrectangle{\pgfqpoint{0.522684in}{0.420833in}}{\pgfqpoint{2.943002in}{2.047833in}}%
\pgfusepath{clip}%
\pgfsetroundcap%
\pgfsetroundjoin%
\pgfsetlinewidth{0.803000pt}%
\definecolor{currentstroke}{rgb}{0.800000,0.800000,0.800000}%
\pgfsetstrokecolor{currentstroke}%
\pgfsetdash{}{0pt}%
\pgfpathmoveto{\pgfqpoint{3.465685in}{0.420833in}}%
\pgfpathlineto{\pgfqpoint{3.465685in}{2.468667in}}%
\pgfusepath{stroke}%
\end{pgfscope}%
\begin{pgfscope}%
\definecolor{textcolor}{rgb}{0.150000,0.150000,0.150000}%
\pgfsetstrokecolor{textcolor}%
\pgfsetfillcolor{textcolor}%
\pgftext[x=3.465685in,y=0.305556in,,top]{\color{textcolor}\rmfamily\fontsize{8.800000}{10.560000}\selectfont \(\displaystyle {1.000}\)}%
\end{pgfscope}%
\begin{pgfscope}%
\definecolor{textcolor}{rgb}{0.150000,0.150000,0.150000}%
\pgfsetstrokecolor{textcolor}%
\pgfsetfillcolor{textcolor}%
\pgftext[x=1.994185in,y=0.138889in,,top]{\color{textcolor}\rmfamily\fontsize{9.600000}{11.520000}\selectfont \(\displaystyle \gamma\) (subsampling ratio)}%
\end{pgfscope}%
\begin{pgfscope}%
\pgfpathrectangle{\pgfqpoint{0.522684in}{0.420833in}}{\pgfqpoint{2.943002in}{2.047833in}}%
\pgfusepath{clip}%
\pgfsetroundcap%
\pgfsetroundjoin%
\pgfsetlinewidth{0.803000pt}%
\definecolor{currentstroke}{rgb}{0.800000,0.800000,0.800000}%
\pgfsetstrokecolor{currentstroke}%
\pgfsetdash{}{0pt}%
\pgfpathmoveto{\pgfqpoint{0.522684in}{0.625617in}}%
\pgfpathlineto{\pgfqpoint{3.465685in}{0.625617in}}%
\pgfusepath{stroke}%
\end{pgfscope}%
\begin{pgfscope}%
\definecolor{textcolor}{rgb}{0.150000,0.150000,0.150000}%
\pgfsetstrokecolor{textcolor}%
\pgfsetfillcolor{textcolor}%
\pgftext[x=0.179012in, y=0.582214in, left, base]{\color{textcolor}\rmfamily\fontsize{8.800000}{10.560000}\selectfont \(\displaystyle {0.76}\)}%
\end{pgfscope}%
\begin{pgfscope}%
\pgfpathrectangle{\pgfqpoint{0.522684in}{0.420833in}}{\pgfqpoint{2.943002in}{2.047833in}}%
\pgfusepath{clip}%
\pgfsetroundcap%
\pgfsetroundjoin%
\pgfsetlinewidth{0.803000pt}%
\definecolor{currentstroke}{rgb}{0.800000,0.800000,0.800000}%
\pgfsetstrokecolor{currentstroke}%
\pgfsetdash{}{0pt}%
\pgfpathmoveto{\pgfqpoint{0.522684in}{1.035183in}}%
\pgfpathlineto{\pgfqpoint{3.465685in}{1.035183in}}%
\pgfusepath{stroke}%
\end{pgfscope}%
\begin{pgfscope}%
\definecolor{textcolor}{rgb}{0.150000,0.150000,0.150000}%
\pgfsetstrokecolor{textcolor}%
\pgfsetfillcolor{textcolor}%
\pgftext[x=0.179012in, y=0.991781in, left, base]{\color{textcolor}\rmfamily\fontsize{8.800000}{10.560000}\selectfont \(\displaystyle {0.78}\)}%
\end{pgfscope}%
\begin{pgfscope}%
\pgfpathrectangle{\pgfqpoint{0.522684in}{0.420833in}}{\pgfqpoint{2.943002in}{2.047833in}}%
\pgfusepath{clip}%
\pgfsetroundcap%
\pgfsetroundjoin%
\pgfsetlinewidth{0.803000pt}%
\definecolor{currentstroke}{rgb}{0.800000,0.800000,0.800000}%
\pgfsetstrokecolor{currentstroke}%
\pgfsetdash{}{0pt}%
\pgfpathmoveto{\pgfqpoint{0.522684in}{1.444750in}}%
\pgfpathlineto{\pgfqpoint{3.465685in}{1.444750in}}%
\pgfusepath{stroke}%
\end{pgfscope}%
\begin{pgfscope}%
\definecolor{textcolor}{rgb}{0.150000,0.150000,0.150000}%
\pgfsetstrokecolor{textcolor}%
\pgfsetfillcolor{textcolor}%
\pgftext[x=0.179012in, y=1.401347in, left, base]{\color{textcolor}\rmfamily\fontsize{8.800000}{10.560000}\selectfont \(\displaystyle {0.80}\)}%
\end{pgfscope}%
\begin{pgfscope}%
\pgfpathrectangle{\pgfqpoint{0.522684in}{0.420833in}}{\pgfqpoint{2.943002in}{2.047833in}}%
\pgfusepath{clip}%
\pgfsetroundcap%
\pgfsetroundjoin%
\pgfsetlinewidth{0.803000pt}%
\definecolor{currentstroke}{rgb}{0.800000,0.800000,0.800000}%
\pgfsetstrokecolor{currentstroke}%
\pgfsetdash{}{0pt}%
\pgfpathmoveto{\pgfqpoint{0.522684in}{1.854317in}}%
\pgfpathlineto{\pgfqpoint{3.465685in}{1.854317in}}%
\pgfusepath{stroke}%
\end{pgfscope}%
\begin{pgfscope}%
\definecolor{textcolor}{rgb}{0.150000,0.150000,0.150000}%
\pgfsetstrokecolor{textcolor}%
\pgfsetfillcolor{textcolor}%
\pgftext[x=0.179012in, y=1.810914in, left, base]{\color{textcolor}\rmfamily\fontsize{8.800000}{10.560000}\selectfont \(\displaystyle {0.82}\)}%
\end{pgfscope}%
\begin{pgfscope}%
\pgfpathrectangle{\pgfqpoint{0.522684in}{0.420833in}}{\pgfqpoint{2.943002in}{2.047833in}}%
\pgfusepath{clip}%
\pgfsetroundcap%
\pgfsetroundjoin%
\pgfsetlinewidth{0.803000pt}%
\definecolor{currentstroke}{rgb}{0.800000,0.800000,0.800000}%
\pgfsetstrokecolor{currentstroke}%
\pgfsetdash{}{0pt}%
\pgfpathmoveto{\pgfqpoint{0.522684in}{2.263883in}}%
\pgfpathlineto{\pgfqpoint{3.465685in}{2.263883in}}%
\pgfusepath{stroke}%
\end{pgfscope}%
\begin{pgfscope}%
\definecolor{textcolor}{rgb}{0.150000,0.150000,0.150000}%
\pgfsetstrokecolor{textcolor}%
\pgfsetfillcolor{textcolor}%
\pgftext[x=0.179012in, y=2.220481in, left, base]{\color{textcolor}\rmfamily\fontsize{8.800000}{10.560000}\selectfont \(\displaystyle {0.84}\)}%
\end{pgfscope}%
\begin{pgfscope}%
\definecolor{textcolor}{rgb}{0.150000,0.150000,0.150000}%
\pgfsetstrokecolor{textcolor}%
\pgfsetfillcolor{textcolor}%
\pgftext[x=0.123457in,y=1.444750in,,bottom,rotate=90.000000]{\color{textcolor}\rmfamily\fontsize{9.600000}{11.520000}\selectfont Mean test AUC}%
\end{pgfscope}%
\begin{pgfscope}%
\pgfpathrectangle{\pgfqpoint{0.522684in}{0.420833in}}{\pgfqpoint{2.943002in}{2.047833in}}%
\pgfusepath{clip}%
\pgfsetbuttcap%
\pgfsetroundjoin%
\definecolor{currentfill}{rgb}{0.121569,0.466667,0.705882}%
\pgfsetfillcolor{currentfill}%
\pgfsetlinewidth{0.803000pt}%
\definecolor{currentstroke}{rgb}{0.121569,0.466667,0.705882}%
\pgfsetstrokecolor{currentstroke}%
\pgfsetdash{}{0pt}%
\pgfsys@defobject{currentmarker}{\pgfqpoint{-0.038036in}{-0.038036in}}{\pgfqpoint{0.038036in}{0.038036in}}{%
\pgfpathmoveto{\pgfqpoint{0.000000in}{-0.038036in}}%
\pgfpathcurveto{\pgfqpoint{0.010087in}{-0.038036in}}{\pgfqpoint{0.019763in}{-0.034029in}}{\pgfqpoint{0.026896in}{-0.026896in}}%
\pgfpathcurveto{\pgfqpoint{0.034029in}{-0.019763in}}{\pgfqpoint{0.038036in}{-0.010087in}}{\pgfqpoint{0.038036in}{0.000000in}}%
\pgfpathcurveto{\pgfqpoint{0.038036in}{0.010087in}}{\pgfqpoint{0.034029in}{0.019763in}}{\pgfqpoint{0.026896in}{0.026896in}}%
\pgfpathcurveto{\pgfqpoint{0.019763in}{0.034029in}}{\pgfqpoint{0.010087in}{0.038036in}}{\pgfqpoint{0.000000in}{0.038036in}}%
\pgfpathcurveto{\pgfqpoint{-0.010087in}{0.038036in}}{\pgfqpoint{-0.019763in}{0.034029in}}{\pgfqpoint{-0.026896in}{0.026896in}}%
\pgfpathcurveto{\pgfqpoint{-0.034029in}{0.019763in}}{\pgfqpoint{-0.038036in}{0.010087in}}{\pgfqpoint{-0.038036in}{0.000000in}}%
\pgfpathcurveto{\pgfqpoint{-0.038036in}{-0.010087in}}{\pgfqpoint{-0.034029in}{-0.019763in}}{\pgfqpoint{-0.026896in}{-0.026896in}}%
\pgfpathcurveto{\pgfqpoint{-0.019763in}{-0.034029in}}{\pgfqpoint{-0.010087in}{-0.038036in}}{\pgfqpoint{0.000000in}{-0.038036in}}%
\pgfpathlineto{\pgfqpoint{0.000000in}{-0.038036in}}%
\pgfpathclose%
\pgfusepath{stroke,fill}%
}%
\begin{pgfscope}%
\pgfsys@transformshift{3.436255in}{1.260445in}%
\pgfsys@useobject{currentmarker}{}%
\end{pgfscope}%
\end{pgfscope}%
\begin{pgfscope}%
\pgfpathrectangle{\pgfqpoint{0.522684in}{0.420833in}}{\pgfqpoint{2.943002in}{2.047833in}}%
\pgfusepath{clip}%
\pgfsetbuttcap%
\pgfsetroundjoin%
\definecolor{currentfill}{rgb}{1.000000,0.498039,0.054902}%
\pgfsetfillcolor{currentfill}%
\pgfsetfillopacity{0.100000}%
\pgfsetlinewidth{0.803000pt}%
\definecolor{currentstroke}{rgb}{1.000000,0.498039,0.054902}%
\pgfsetstrokecolor{currentstroke}%
\pgfsetstrokeopacity{0.100000}%
\pgfsetdash{}{0pt}%
\pgfsys@defobject{currentmarker}{\pgfqpoint{0.522684in}{1.526663in}}{\pgfqpoint{3.465685in}{1.977187in}}{%
\pgfpathmoveto{\pgfqpoint{0.522684in}{1.977187in}}%
\pgfpathlineto{\pgfqpoint{0.522684in}{1.936230in}}%
\pgfpathlineto{\pgfqpoint{0.537399in}{1.936230in}}%
\pgfpathlineto{\pgfqpoint{0.552114in}{1.833838in}}%
\pgfpathlineto{\pgfqpoint{0.816984in}{1.670012in}}%
\pgfpathlineto{\pgfqpoint{1.111284in}{1.649533in}}%
\pgfpathlineto{\pgfqpoint{1.405584in}{1.588098in}}%
\pgfpathlineto{\pgfqpoint{1.699884in}{1.588098in}}%
\pgfpathlineto{\pgfqpoint{1.994185in}{1.547142in}}%
\pgfpathlineto{\pgfqpoint{2.288485in}{1.526663in}}%
\pgfpathlineto{\pgfqpoint{2.582785in}{1.588098in}}%
\pgfpathlineto{\pgfqpoint{2.877085in}{1.567620in}}%
\pgfpathlineto{\pgfqpoint{3.171385in}{1.547142in}}%
\pgfpathlineto{\pgfqpoint{3.465685in}{1.588098in}}%
\pgfpathlineto{\pgfqpoint{3.465685in}{1.629055in}}%
\pgfpathlineto{\pgfqpoint{3.465685in}{1.629055in}}%
\pgfpathlineto{\pgfqpoint{3.171385in}{1.588098in}}%
\pgfpathlineto{\pgfqpoint{2.877085in}{1.608577in}}%
\pgfpathlineto{\pgfqpoint{2.582785in}{1.629055in}}%
\pgfpathlineto{\pgfqpoint{2.288485in}{1.608577in}}%
\pgfpathlineto{\pgfqpoint{1.994185in}{1.588098in}}%
\pgfpathlineto{\pgfqpoint{1.699884in}{1.629055in}}%
\pgfpathlineto{\pgfqpoint{1.405584in}{1.629055in}}%
\pgfpathlineto{\pgfqpoint{1.111284in}{1.690490in}}%
\pgfpathlineto{\pgfqpoint{0.816984in}{1.710968in}}%
\pgfpathlineto{\pgfqpoint{0.552114in}{1.874795in}}%
\pgfpathlineto{\pgfqpoint{0.537399in}{1.977187in}}%
\pgfpathlineto{\pgfqpoint{0.522684in}{1.977187in}}%
\pgfpathlineto{\pgfqpoint{0.522684in}{1.977187in}}%
\pgfpathclose%
\pgfusepath{stroke,fill}%
}%
\begin{pgfscope}%
\pgfsys@transformshift{0.000000in}{0.000000in}%
\pgfsys@useobject{currentmarker}{}%
\end{pgfscope}%
\end{pgfscope}%
\begin{pgfscope}%
\pgfpathrectangle{\pgfqpoint{0.522684in}{0.420833in}}{\pgfqpoint{2.943002in}{2.047833in}}%
\pgfusepath{clip}%
\pgfsetroundcap%
\pgfsetroundjoin%
\pgfsetlinewidth{1.204500pt}%
\definecolor{currentstroke}{rgb}{1.000000,0.498039,0.054902}%
\pgfsetstrokecolor{currentstroke}%
\pgfsetdash{}{0pt}%
\pgfpathmoveto{\pgfqpoint{0.522684in}{1.956708in}}%
\pgfpathlineto{\pgfqpoint{0.537399in}{1.956708in}}%
\pgfpathlineto{\pgfqpoint{0.552114in}{1.854317in}}%
\pgfpathlineto{\pgfqpoint{0.816984in}{1.690490in}}%
\pgfpathlineto{\pgfqpoint{1.111284in}{1.670012in}}%
\pgfpathlineto{\pgfqpoint{1.405584in}{1.608577in}}%
\pgfpathlineto{\pgfqpoint{1.699884in}{1.608577in}}%
\pgfpathlineto{\pgfqpoint{1.994185in}{1.567620in}}%
\pgfpathlineto{\pgfqpoint{2.288485in}{1.567620in}}%
\pgfpathlineto{\pgfqpoint{2.582785in}{1.608577in}}%
\pgfpathlineto{\pgfqpoint{2.877085in}{1.588098in}}%
\pgfpathlineto{\pgfqpoint{3.171385in}{1.567620in}}%
\pgfpathlineto{\pgfqpoint{3.465685in}{1.608577in}}%
\pgfusepath{stroke}%
\end{pgfscope}%
\begin{pgfscope}%
\pgfsetrectcap%
\pgfsetmiterjoin%
\pgfsetlinewidth{1.003750pt}%
\definecolor{currentstroke}{rgb}{0.800000,0.800000,0.800000}%
\pgfsetstrokecolor{currentstroke}%
\pgfsetdash{}{0pt}%
\pgfpathmoveto{\pgfqpoint{0.522684in}{0.420833in}}%
\pgfpathlineto{\pgfqpoint{0.522684in}{2.468667in}}%
\pgfusepath{stroke}%
\end{pgfscope}%
\begin{pgfscope}%
\pgfsetrectcap%
\pgfsetmiterjoin%
\pgfsetlinewidth{1.003750pt}%
\definecolor{currentstroke}{rgb}{0.800000,0.800000,0.800000}%
\pgfsetstrokecolor{currentstroke}%
\pgfsetdash{}{0pt}%
\pgfpathmoveto{\pgfqpoint{0.522684in}{0.420833in}}%
\pgfpathlineto{\pgfqpoint{3.465685in}{0.420833in}}%
\pgfusepath{stroke}%
\end{pgfscope}%
\begin{pgfscope}%
\pgfsetroundcap%
\pgfsetroundjoin%
\definecolor{currentfill}{rgb}{0.862745,0.862745,0.862745}%
\pgfsetfillcolor{currentfill}%
\pgfsetlinewidth{0.803000pt}%
\definecolor{currentstroke}{rgb}{1.000000,1.000000,1.000000}%
\pgfsetstrokecolor{currentstroke}%
\pgfsetdash{}{0pt}%
\pgfpathmoveto{\pgfqpoint{1.017702in}{0.666563in}}%
\pgfpathquadraticcurveto{\pgfqpoint{1.017702in}{0.951809in}}{\pgfqpoint{1.017702in}{1.237054in}}%
\pgfpathlineto{\pgfqpoint{0.969091in}{1.237054in}}%
\pgfpathquadraticcurveto{\pgfqpoint{1.010757in}{1.320407in}}{\pgfqpoint{1.052424in}{1.403760in}}%
\pgfpathquadraticcurveto{\pgfqpoint{1.094091in}{1.320407in}}{\pgfqpoint{1.135757in}{1.237054in}}%
\pgfpathlineto{\pgfqpoint{1.087146in}{1.237054in}}%
\pgfpathquadraticcurveto{\pgfqpoint{1.087146in}{0.951809in}}{\pgfqpoint{1.087146in}{0.666563in}}%
\pgfpathlineto{\pgfqpoint{1.017702in}{0.666563in}}%
\pgfpathlineto{\pgfqpoint{1.017702in}{0.666563in}}%
\pgfpathclose%
\pgfusepath{stroke,fill}%
\end{pgfscope}%
\begin{pgfscope}%
\definecolor{textcolor}{rgb}{0.862745,0.862745,0.862745}%
\pgfsetstrokecolor{textcolor}%
\pgfsetfillcolor{textcolor}%
\pgftext[x=1.199574in,y=1.035183in,left,]{\color{textcolor}\rmfamily\fontsize{12.000000}{14.400000}\selectfont better}%
\end{pgfscope}%
\begin{pgfscope}%
\pgfsetbuttcap%
\pgfsetmiterjoin%
\definecolor{currentfill}{rgb}{1.000000,1.000000,1.000000}%
\pgfsetfillcolor{currentfill}%
\pgfsetfillopacity{0.800000}%
\pgfsetlinewidth{0.803000pt}%
\definecolor{currentstroke}{rgb}{0.800000,0.800000,0.800000}%
\pgfsetstrokecolor{currentstroke}%
\pgfsetstrokeopacity{0.800000}%
\pgfsetdash{}{0pt}%
\pgfpathmoveto{\pgfqpoint{2.146789in}{2.026444in}}%
\pgfpathlineto{\pgfqpoint{3.380130in}{2.026444in}}%
\pgfpathquadraticcurveto{\pgfqpoint{3.404574in}{2.026444in}}{\pgfqpoint{3.404574in}{2.050889in}}%
\pgfpathlineto{\pgfqpoint{3.404574in}{2.383111in}}%
\pgfpathquadraticcurveto{\pgfqpoint{3.404574in}{2.407556in}}{\pgfqpoint{3.380130in}{2.407556in}}%
\pgfpathlineto{\pgfqpoint{2.146789in}{2.407556in}}%
\pgfpathquadraticcurveto{\pgfqpoint{2.122344in}{2.407556in}}{\pgfqpoint{2.122344in}{2.383111in}}%
\pgfpathlineto{\pgfqpoint{2.122344in}{2.050889in}}%
\pgfpathquadraticcurveto{\pgfqpoint{2.122344in}{2.026444in}}{\pgfqpoint{2.146789in}{2.026444in}}%
\pgfpathlineto{\pgfqpoint{2.146789in}{2.026444in}}%
\pgfpathclose%
\pgfusepath{stroke,fill}%
\end{pgfscope}%
\begin{pgfscope}%
\pgfsetbuttcap%
\pgfsetroundjoin%
\definecolor{currentfill}{rgb}{0.121569,0.466667,0.705882}%
\pgfsetfillcolor{currentfill}%
\pgfsetlinewidth{0.803000pt}%
\definecolor{currentstroke}{rgb}{0.121569,0.466667,0.705882}%
\pgfsetstrokecolor{currentstroke}%
\pgfsetdash{}{0pt}%
\pgfsys@defobject{currentmarker}{\pgfqpoint{-0.038036in}{-0.038036in}}{\pgfqpoint{0.038036in}{0.038036in}}{%
\pgfpathmoveto{\pgfqpoint{0.000000in}{-0.038036in}}%
\pgfpathcurveto{\pgfqpoint{0.010087in}{-0.038036in}}{\pgfqpoint{0.019763in}{-0.034029in}}{\pgfqpoint{0.026896in}{-0.026896in}}%
\pgfpathcurveto{\pgfqpoint{0.034029in}{-0.019763in}}{\pgfqpoint{0.038036in}{-0.010087in}}{\pgfqpoint{0.038036in}{0.000000in}}%
\pgfpathcurveto{\pgfqpoint{0.038036in}{0.010087in}}{\pgfqpoint{0.034029in}{0.019763in}}{\pgfqpoint{0.026896in}{0.026896in}}%
\pgfpathcurveto{\pgfqpoint{0.019763in}{0.034029in}}{\pgfqpoint{0.010087in}{0.038036in}}{\pgfqpoint{0.000000in}{0.038036in}}%
\pgfpathcurveto{\pgfqpoint{-0.010087in}{0.038036in}}{\pgfqpoint{-0.019763in}{0.034029in}}{\pgfqpoint{-0.026896in}{0.026896in}}%
\pgfpathcurveto{\pgfqpoint{-0.034029in}{0.019763in}}{\pgfqpoint{-0.038036in}{0.010087in}}{\pgfqpoint{-0.038036in}{0.000000in}}%
\pgfpathcurveto{\pgfqpoint{-0.038036in}{-0.010087in}}{\pgfqpoint{-0.034029in}{-0.019763in}}{\pgfqpoint{-0.026896in}{-0.026896in}}%
\pgfpathcurveto{\pgfqpoint{-0.019763in}{-0.034029in}}{\pgfqpoint{-0.010087in}{-0.038036in}}{\pgfqpoint{0.000000in}{-0.038036in}}%
\pgfpathlineto{\pgfqpoint{0.000000in}{-0.038036in}}%
\pgfpathclose%
\pgfusepath{stroke,fill}%
}%
\begin{pgfscope}%
\pgfsys@transformshift{2.293456in}{2.303944in}%
\pgfsys@useobject{currentmarker}{}%
\end{pgfscope}%
\end{pgfscope}%
\begin{pgfscope}%
\definecolor{textcolor}{rgb}{0.150000,0.150000,0.150000}%
\pgfsetstrokecolor{textcolor}%
\pgfsetfillcolor{textcolor}%
\pgftext[x=2.513456in,y=2.271861in,left,base]{\color{textcolor}\rmfamily\fontsize{8.800000}{10.560000}\selectfont Maddock et al.}%
\end{pgfscope}%
\begin{pgfscope}%
\pgfsetroundcap%
\pgfsetroundjoin%
\pgfsetlinewidth{1.204500pt}%
\definecolor{currentstroke}{rgb}{1.000000,0.498039,0.054902}%
\pgfsetstrokecolor{currentstroke}%
\pgfsetdash{}{0pt}%
\pgfpathmoveto{\pgfqpoint{2.171233in}{2.142417in}}%
\pgfpathlineto{\pgfqpoint{2.293456in}{2.142417in}}%
\pgfpathlineto{\pgfqpoint{2.415678in}{2.142417in}}%
\pgfusepath{stroke}%
\end{pgfscope}%
\begin{pgfscope}%
\definecolor{textcolor}{rgb}{0.150000,0.150000,0.150000}%
\pgfsetstrokecolor{textcolor}%
\pgfsetfillcolor{textcolor}%
\pgftext[x=2.513456in,y=2.099639in,left,base]{\color{textcolor}\rmfamily\fontsize{8.800000}{10.560000}\selectfont S-BDT}%
\end{pgfscope}%
\end{pgfpicture}%
\makeatother%
\endgroup%

%% file: images/adult_init_ablation.pgf
\begingroup%
\makeatletter%
\begin{pgfpicture}%
\pgfpathrectangle{\pgfpointorigin}{\pgfqpoint{3.612000in}{2.468667in}}%
\pgfusepath{use as bounding box, clip}%
\begin{pgfscope}%
\pgfsetbuttcap%
\pgfsetmiterjoin%
\definecolor{currentfill}{rgb}{1.000000,1.000000,1.000000}%
\pgfsetfillcolor{currentfill}%
\pgfsetlinewidth{0.000000pt}%
\definecolor{currentstroke}{rgb}{1.000000,1.000000,1.000000}%
\pgfsetstrokecolor{currentstroke}%
\pgfsetdash{}{0pt}%
\pgfpathmoveto{\pgfqpoint{0.000000in}{0.000000in}}%
\pgfpathlineto{\pgfqpoint{3.612000in}{0.000000in}}%
\pgfpathlineto{\pgfqpoint{3.612000in}{2.468667in}}%
\pgfpathlineto{\pgfqpoint{0.000000in}{2.468667in}}%
\pgfpathlineto{\pgfqpoint{0.000000in}{0.000000in}}%
\pgfpathclose%
\pgfusepath{fill}%
\end{pgfscope}%
\begin{pgfscope}%
\pgfsetbuttcap%
\pgfsetmiterjoin%
\definecolor{currentfill}{rgb}{1.000000,1.000000,1.000000}%
\pgfsetfillcolor{currentfill}%
\pgfsetlinewidth{0.000000pt}%
\definecolor{currentstroke}{rgb}{0.000000,0.000000,0.000000}%
\pgfsetstrokecolor{currentstroke}%
\pgfsetstrokeopacity{0.000000}%
\pgfsetdash{}{0pt}%
\pgfpathmoveto{\pgfqpoint{0.522684in}{0.420833in}}%
\pgfpathlineto{\pgfqpoint{3.497803in}{0.420833in}}%
\pgfpathlineto{\pgfqpoint{3.497803in}{2.468667in}}%
\pgfpathlineto{\pgfqpoint{0.522684in}{2.468667in}}%
\pgfpathlineto{\pgfqpoint{0.522684in}{0.420833in}}%
\pgfpathclose%
\pgfusepath{fill}%
\end{pgfscope}%
\begin{pgfscope}%
\pgfpathrectangle{\pgfqpoint{0.522684in}{0.420833in}}{\pgfqpoint{2.975120in}{2.047833in}}%
\pgfusepath{clip}%
\pgfsetroundcap%
\pgfsetroundjoin%
\pgfsetlinewidth{0.803000pt}%
\definecolor{currentstroke}{rgb}{0.800000,0.800000,0.800000}%
\pgfsetstrokecolor{currentstroke}%
\pgfsetdash{}{0pt}%
\pgfpathmoveto{\pgfqpoint{0.582186in}{0.420833in}}%
\pgfpathlineto{\pgfqpoint{0.582186in}{2.468667in}}%
\pgfusepath{stroke}%
\end{pgfscope}%
\begin{pgfscope}%
\definecolor{textcolor}{rgb}{0.150000,0.150000,0.150000}%
\pgfsetstrokecolor{textcolor}%
\pgfsetfillcolor{textcolor}%
\pgftext[x=0.582186in,y=0.305556in,,top]{\color{textcolor}\rmfamily\fontsize{8.800000}{10.560000}\selectfont \(\displaystyle {0.01}\)}%
\end{pgfscope}%
\begin{pgfscope}%
\pgfpathrectangle{\pgfqpoint{0.522684in}{0.420833in}}{\pgfqpoint{2.975120in}{2.047833in}}%
\pgfusepath{clip}%
\pgfsetroundcap%
\pgfsetroundjoin%
\pgfsetlinewidth{0.803000pt}%
\definecolor{currentstroke}{rgb}{0.800000,0.800000,0.800000}%
\pgfsetstrokecolor{currentstroke}%
\pgfsetdash{}{0pt}%
\pgfpathmoveto{\pgfqpoint{1.117708in}{0.420833in}}%
\pgfpathlineto{\pgfqpoint{1.117708in}{2.468667in}}%
\pgfusepath{stroke}%
\end{pgfscope}%
\begin{pgfscope}%
\definecolor{textcolor}{rgb}{0.150000,0.150000,0.150000}%
\pgfsetstrokecolor{textcolor}%
\pgfsetfillcolor{textcolor}%
\pgftext[x=1.117708in,y=0.305556in,,top]{\color{textcolor}\rmfamily\fontsize{8.800000}{10.560000}\selectfont \(\displaystyle {0.10}\)}%
\end{pgfscope}%
\begin{pgfscope}%
\pgfpathrectangle{\pgfqpoint{0.522684in}{0.420833in}}{\pgfqpoint{2.975120in}{2.047833in}}%
\pgfusepath{clip}%
\pgfsetroundcap%
\pgfsetroundjoin%
\pgfsetlinewidth{0.803000pt}%
\definecolor{currentstroke}{rgb}{0.800000,0.800000,0.800000}%
\pgfsetstrokecolor{currentstroke}%
\pgfsetdash{}{0pt}%
\pgfpathmoveto{\pgfqpoint{1.712731in}{0.420833in}}%
\pgfpathlineto{\pgfqpoint{1.712731in}{2.468667in}}%
\pgfusepath{stroke}%
\end{pgfscope}%
\begin{pgfscope}%
\definecolor{textcolor}{rgb}{0.150000,0.150000,0.150000}%
\pgfsetstrokecolor{textcolor}%
\pgfsetfillcolor{textcolor}%
\pgftext[x=1.712731in,y=0.305556in,,top]{\color{textcolor}\rmfamily\fontsize{8.800000}{10.560000}\selectfont \(\displaystyle {0.20}\)}%
\end{pgfscope}%
\begin{pgfscope}%
\pgfpathrectangle{\pgfqpoint{0.522684in}{0.420833in}}{\pgfqpoint{2.975120in}{2.047833in}}%
\pgfusepath{clip}%
\pgfsetroundcap%
\pgfsetroundjoin%
\pgfsetlinewidth{0.803000pt}%
\definecolor{currentstroke}{rgb}{0.800000,0.800000,0.800000}%
\pgfsetstrokecolor{currentstroke}%
\pgfsetdash{}{0pt}%
\pgfpathmoveto{\pgfqpoint{2.307755in}{0.420833in}}%
\pgfpathlineto{\pgfqpoint{2.307755in}{2.468667in}}%
\pgfusepath{stroke}%
\end{pgfscope}%
\begin{pgfscope}%
\definecolor{textcolor}{rgb}{0.150000,0.150000,0.150000}%
\pgfsetstrokecolor{textcolor}%
\pgfsetfillcolor{textcolor}%
\pgftext[x=2.307755in,y=0.305556in,,top]{\color{textcolor}\rmfamily\fontsize{8.800000}{10.560000}\selectfont \(\displaystyle {0.30}\)}%
\end{pgfscope}%
\begin{pgfscope}%
\pgfpathrectangle{\pgfqpoint{0.522684in}{0.420833in}}{\pgfqpoint{2.975120in}{2.047833in}}%
\pgfusepath{clip}%
\pgfsetroundcap%
\pgfsetroundjoin%
\pgfsetlinewidth{0.803000pt}%
\definecolor{currentstroke}{rgb}{0.800000,0.800000,0.800000}%
\pgfsetstrokecolor{currentstroke}%
\pgfsetdash{}{0pt}%
\pgfpathmoveto{\pgfqpoint{2.902779in}{0.420833in}}%
\pgfpathlineto{\pgfqpoint{2.902779in}{2.468667in}}%
\pgfusepath{stroke}%
\end{pgfscope}%
\begin{pgfscope}%
\definecolor{textcolor}{rgb}{0.150000,0.150000,0.150000}%
\pgfsetstrokecolor{textcolor}%
\pgfsetfillcolor{textcolor}%
\pgftext[x=2.902779in,y=0.305556in,,top]{\color{textcolor}\rmfamily\fontsize{8.800000}{10.560000}\selectfont \(\displaystyle {0.40}\)}%
\end{pgfscope}%
\begin{pgfscope}%
\pgfpathrectangle{\pgfqpoint{0.522684in}{0.420833in}}{\pgfqpoint{2.975120in}{2.047833in}}%
\pgfusepath{clip}%
\pgfsetroundcap%
\pgfsetroundjoin%
\pgfsetlinewidth{0.803000pt}%
\definecolor{currentstroke}{rgb}{0.800000,0.800000,0.800000}%
\pgfsetstrokecolor{currentstroke}%
\pgfsetdash{}{0pt}%
\pgfpathmoveto{\pgfqpoint{3.497803in}{0.420833in}}%
\pgfpathlineto{\pgfqpoint{3.497803in}{2.468667in}}%
\pgfusepath{stroke}%
\end{pgfscope}%
\begin{pgfscope}%
\definecolor{textcolor}{rgb}{0.150000,0.150000,0.150000}%
\pgfsetstrokecolor{textcolor}%
\pgfsetfillcolor{textcolor}%
\pgftext[x=3.497803in,y=0.305556in,,top]{\color{textcolor}\rmfamily\fontsize{8.800000}{10.560000}\selectfont \(\displaystyle {0.50}\)}%
\end{pgfscope}%
\begin{pgfscope}%
\definecolor{textcolor}{rgb}{0.150000,0.150000,0.150000}%
\pgfsetstrokecolor{textcolor}%
\pgfsetfillcolor{textcolor}%
\pgftext[x=2.010243in,y=0.138889in,,top]{\color{textcolor}\rmfamily\fontsize{9.600000}{11.520000}\selectfont \(\displaystyle \varepsilon_\text{init}/\varepsilon)\) (privacy budget ratio for initial score)}%
\end{pgfscope}%
\begin{pgfscope}%
\pgfpathrectangle{\pgfqpoint{0.522684in}{0.420833in}}{\pgfqpoint{2.975120in}{2.047833in}}%
\pgfusepath{clip}%
\pgfsetroundcap%
\pgfsetroundjoin%
\pgfsetlinewidth{0.803000pt}%
\definecolor{currentstroke}{rgb}{0.800000,0.800000,0.800000}%
\pgfsetstrokecolor{currentstroke}%
\pgfsetdash{}{0pt}%
\pgfpathmoveto{\pgfqpoint{0.522684in}{0.625617in}}%
\pgfpathlineto{\pgfqpoint{3.497803in}{0.625617in}}%
\pgfusepath{stroke}%
\end{pgfscope}%
\begin{pgfscope}%
\definecolor{textcolor}{rgb}{0.150000,0.150000,0.150000}%
\pgfsetstrokecolor{textcolor}%
\pgfsetfillcolor{textcolor}%
\pgftext[x=0.179012in, y=0.582214in, left, base]{\color{textcolor}\rmfamily\fontsize{8.800000}{10.560000}\selectfont \(\displaystyle {0.76}\)}%
\end{pgfscope}%
\begin{pgfscope}%
\pgfpathrectangle{\pgfqpoint{0.522684in}{0.420833in}}{\pgfqpoint{2.975120in}{2.047833in}}%
\pgfusepath{clip}%
\pgfsetroundcap%
\pgfsetroundjoin%
\pgfsetlinewidth{0.803000pt}%
\definecolor{currentstroke}{rgb}{0.800000,0.800000,0.800000}%
\pgfsetstrokecolor{currentstroke}%
\pgfsetdash{}{0pt}%
\pgfpathmoveto{\pgfqpoint{0.522684in}{1.035183in}}%
\pgfpathlineto{\pgfqpoint{3.497803in}{1.035183in}}%
\pgfusepath{stroke}%
\end{pgfscope}%
\begin{pgfscope}%
\definecolor{textcolor}{rgb}{0.150000,0.150000,0.150000}%
\pgfsetstrokecolor{textcolor}%
\pgfsetfillcolor{textcolor}%
\pgftext[x=0.179012in, y=0.991781in, left, base]{\color{textcolor}\rmfamily\fontsize{8.800000}{10.560000}\selectfont \(\displaystyle {0.78}\)}%
\end{pgfscope}%
\begin{pgfscope}%
\pgfpathrectangle{\pgfqpoint{0.522684in}{0.420833in}}{\pgfqpoint{2.975120in}{2.047833in}}%
\pgfusepath{clip}%
\pgfsetroundcap%
\pgfsetroundjoin%
\pgfsetlinewidth{0.803000pt}%
\definecolor{currentstroke}{rgb}{0.800000,0.800000,0.800000}%
\pgfsetstrokecolor{currentstroke}%
\pgfsetdash{}{0pt}%
\pgfpathmoveto{\pgfqpoint{0.522684in}{1.444750in}}%
\pgfpathlineto{\pgfqpoint{3.497803in}{1.444750in}}%
\pgfusepath{stroke}%
\end{pgfscope}%
\begin{pgfscope}%
\definecolor{textcolor}{rgb}{0.150000,0.150000,0.150000}%
\pgfsetstrokecolor{textcolor}%
\pgfsetfillcolor{textcolor}%
\pgftext[x=0.179012in, y=1.401347in, left, base]{\color{textcolor}\rmfamily\fontsize{8.800000}{10.560000}\selectfont \(\displaystyle {0.80}\)}%
\end{pgfscope}%
\begin{pgfscope}%
\pgfpathrectangle{\pgfqpoint{0.522684in}{0.420833in}}{\pgfqpoint{2.975120in}{2.047833in}}%
\pgfusepath{clip}%
\pgfsetroundcap%
\pgfsetroundjoin%
\pgfsetlinewidth{0.803000pt}%
\definecolor{currentstroke}{rgb}{0.800000,0.800000,0.800000}%
\pgfsetstrokecolor{currentstroke}%
\pgfsetdash{}{0pt}%
\pgfpathmoveto{\pgfqpoint{0.522684in}{1.854317in}}%
\pgfpathlineto{\pgfqpoint{3.497803in}{1.854317in}}%
\pgfusepath{stroke}%
\end{pgfscope}%
\begin{pgfscope}%
\definecolor{textcolor}{rgb}{0.150000,0.150000,0.150000}%
\pgfsetstrokecolor{textcolor}%
\pgfsetfillcolor{textcolor}%
\pgftext[x=0.179012in, y=1.810914in, left, base]{\color{textcolor}\rmfamily\fontsize{8.800000}{10.560000}\selectfont \(\displaystyle {0.82}\)}%
\end{pgfscope}%
\begin{pgfscope}%
\pgfpathrectangle{\pgfqpoint{0.522684in}{0.420833in}}{\pgfqpoint{2.975120in}{2.047833in}}%
\pgfusepath{clip}%
\pgfsetroundcap%
\pgfsetroundjoin%
\pgfsetlinewidth{0.803000pt}%
\definecolor{currentstroke}{rgb}{0.800000,0.800000,0.800000}%
\pgfsetstrokecolor{currentstroke}%
\pgfsetdash{}{0pt}%
\pgfpathmoveto{\pgfqpoint{0.522684in}{2.263883in}}%
\pgfpathlineto{\pgfqpoint{3.497803in}{2.263883in}}%
\pgfusepath{stroke}%
\end{pgfscope}%
\begin{pgfscope}%
\definecolor{textcolor}{rgb}{0.150000,0.150000,0.150000}%
\pgfsetstrokecolor{textcolor}%
\pgfsetfillcolor{textcolor}%
\pgftext[x=0.179012in, y=2.220481in, left, base]{\color{textcolor}\rmfamily\fontsize{8.800000}{10.560000}\selectfont \(\displaystyle {0.84}\)}%
\end{pgfscope}%
\begin{pgfscope}%
\definecolor{textcolor}{rgb}{0.150000,0.150000,0.150000}%
\pgfsetstrokecolor{textcolor}%
\pgfsetfillcolor{textcolor}%
\pgftext[x=0.123457in,y=1.444750in,,bottom,rotate=90.000000]{\color{textcolor}\rmfamily\fontsize{9.600000}{11.520000}\selectfont Mean test AUC}%
\end{pgfscope}%
\begin{pgfscope}%
\pgfpathrectangle{\pgfqpoint{0.522684in}{0.420833in}}{\pgfqpoint{2.975120in}{2.047833in}}%
\pgfusepath{clip}%
\pgfsetbuttcap%
\pgfsetroundjoin%
\definecolor{currentfill}{rgb}{0.121569,0.466667,0.705882}%
\pgfsetfillcolor{currentfill}%
\pgfsetlinewidth{0.803000pt}%
\definecolor{currentstroke}{rgb}{0.121569,0.466667,0.705882}%
\pgfsetstrokecolor{currentstroke}%
\pgfsetdash{}{0pt}%
\pgfsys@defobject{currentmarker}{\pgfqpoint{-0.038036in}{-0.038036in}}{\pgfqpoint{0.038036in}{0.038036in}}{%
\pgfpathmoveto{\pgfqpoint{0.000000in}{-0.038036in}}%
\pgfpathcurveto{\pgfqpoint{0.010087in}{-0.038036in}}{\pgfqpoint{0.019763in}{-0.034029in}}{\pgfqpoint{0.026896in}{-0.026896in}}%
\pgfpathcurveto{\pgfqpoint{0.034029in}{-0.019763in}}{\pgfqpoint{0.038036in}{-0.010087in}}{\pgfqpoint{0.038036in}{0.000000in}}%
\pgfpathcurveto{\pgfqpoint{0.038036in}{0.010087in}}{\pgfqpoint{0.034029in}{0.019763in}}{\pgfqpoint{0.026896in}{0.026896in}}%
\pgfpathcurveto{\pgfqpoint{0.019763in}{0.034029in}}{\pgfqpoint{0.010087in}{0.038036in}}{\pgfqpoint{0.000000in}{0.038036in}}%
\pgfpathcurveto{\pgfqpoint{-0.010087in}{0.038036in}}{\pgfqpoint{-0.019763in}{0.034029in}}{\pgfqpoint{-0.026896in}{0.026896in}}%
\pgfpathcurveto{\pgfqpoint{-0.034029in}{0.019763in}}{\pgfqpoint{-0.038036in}{0.010087in}}{\pgfqpoint{-0.038036in}{0.000000in}}%
\pgfpathcurveto{\pgfqpoint{-0.038036in}{-0.010087in}}{\pgfqpoint{-0.034029in}{-0.019763in}}{\pgfqpoint{-0.026896in}{-0.026896in}}%
\pgfpathcurveto{\pgfqpoint{-0.019763in}{-0.034029in}}{\pgfqpoint{-0.010087in}{-0.038036in}}{\pgfqpoint{0.000000in}{-0.038036in}}%
\pgfpathlineto{\pgfqpoint{0.000000in}{-0.038036in}}%
\pgfpathclose%
\pgfusepath{stroke,fill}%
}%
\begin{pgfscope}%
\pgfsys@transformshift{0.552435in}{1.260445in}%
\pgfsys@useobject{currentmarker}{}%
\end{pgfscope}%
\end{pgfscope}%
\begin{pgfscope}%
\pgfpathrectangle{\pgfqpoint{0.522684in}{0.420833in}}{\pgfqpoint{2.975120in}{2.047833in}}%
\pgfusepath{clip}%
\pgfsetbuttcap%
\pgfsetroundjoin%
\definecolor{currentfill}{rgb}{1.000000,0.498039,0.054902}%
\pgfsetfillcolor{currentfill}%
\pgfsetfillopacity{0.100000}%
\pgfsetlinewidth{0.803000pt}%
\definecolor{currentstroke}{rgb}{1.000000,0.498039,0.054902}%
\pgfsetstrokecolor{currentstroke}%
\pgfsetstrokeopacity{0.100000}%
\pgfsetdash{}{0pt}%
\pgfsys@defobject{currentmarker}{\pgfqpoint{0.522684in}{0.748487in}}{\pgfqpoint{3.497803in}{1.977187in}}{%
\pgfpathmoveto{\pgfqpoint{0.522684in}{1.977187in}}%
\pgfpathlineto{\pgfqpoint{0.522684in}{1.936230in}}%
\pgfpathlineto{\pgfqpoint{0.582186in}{1.936230in}}%
\pgfpathlineto{\pgfqpoint{1.117708in}{1.833838in}}%
\pgfpathlineto{\pgfqpoint{1.712731in}{1.690490in}}%
\pgfpathlineto{\pgfqpoint{2.307755in}{1.444750in}}%
\pgfpathlineto{\pgfqpoint{2.902779in}{1.096618in}}%
\pgfpathlineto{\pgfqpoint{3.497803in}{0.748487in}}%
\pgfpathlineto{\pgfqpoint{3.497803in}{0.789443in}}%
\pgfpathlineto{\pgfqpoint{3.497803in}{0.789443in}}%
\pgfpathlineto{\pgfqpoint{2.902779in}{1.137575in}}%
\pgfpathlineto{\pgfqpoint{2.307755in}{1.485707in}}%
\pgfpathlineto{\pgfqpoint{1.712731in}{1.731447in}}%
\pgfpathlineto{\pgfqpoint{1.117708in}{1.874795in}}%
\pgfpathlineto{\pgfqpoint{0.582186in}{1.977187in}}%
\pgfpathlineto{\pgfqpoint{0.522684in}{1.977187in}}%
\pgfpathlineto{\pgfqpoint{0.522684in}{1.977187in}}%
\pgfpathclose%
\pgfusepath{stroke,fill}%
}%
\begin{pgfscope}%
\pgfsys@transformshift{0.000000in}{0.000000in}%
\pgfsys@useobject{currentmarker}{}%
\end{pgfscope}%
\end{pgfscope}%
\begin{pgfscope}%
\pgfpathrectangle{\pgfqpoint{0.522684in}{0.420833in}}{\pgfqpoint{2.975120in}{2.047833in}}%
\pgfusepath{clip}%
\pgfsetroundcap%
\pgfsetroundjoin%
\pgfsetlinewidth{1.204500pt}%
\definecolor{currentstroke}{rgb}{1.000000,0.498039,0.054902}%
\pgfsetstrokecolor{currentstroke}%
\pgfsetdash{}{0pt}%
\pgfpathmoveto{\pgfqpoint{0.522684in}{1.956708in}}%
\pgfpathlineto{\pgfqpoint{0.582186in}{1.956708in}}%
\pgfpathlineto{\pgfqpoint{1.117708in}{1.854317in}}%
\pgfpathlineto{\pgfqpoint{1.712731in}{1.710968in}}%
\pgfpathlineto{\pgfqpoint{2.307755in}{1.465228in}}%
\pgfpathlineto{\pgfqpoint{2.902779in}{1.117097in}}%
\pgfpathlineto{\pgfqpoint{3.497803in}{0.768965in}}%
\pgfusepath{stroke}%
\end{pgfscope}%
\begin{pgfscope}%
\pgfsetrectcap%
\pgfsetmiterjoin%
\pgfsetlinewidth{1.003750pt}%
\definecolor{currentstroke}{rgb}{0.800000,0.800000,0.800000}%
\pgfsetstrokecolor{currentstroke}%
\pgfsetdash{}{0pt}%
\pgfpathmoveto{\pgfqpoint{0.522684in}{0.420833in}}%
\pgfpathlineto{\pgfqpoint{0.522684in}{2.468667in}}%
\pgfusepath{stroke}%
\end{pgfscope}%
\begin{pgfscope}%
\pgfsetrectcap%
\pgfsetmiterjoin%
\pgfsetlinewidth{1.003750pt}%
\definecolor{currentstroke}{rgb}{0.800000,0.800000,0.800000}%
\pgfsetstrokecolor{currentstroke}%
\pgfsetdash{}{0pt}%
\pgfpathmoveto{\pgfqpoint{0.522684in}{0.420833in}}%
\pgfpathlineto{\pgfqpoint{3.497803in}{0.420833in}}%
\pgfusepath{stroke}%
\end{pgfscope}%
\begin{pgfscope}%
\pgfsetroundcap%
\pgfsetroundjoin%
\definecolor{currentfill}{rgb}{0.862745,0.862745,0.862745}%
\pgfsetfillcolor{currentfill}%
\pgfsetlinewidth{0.803000pt}%
\definecolor{currentstroke}{rgb}{1.000000,1.000000,1.000000}%
\pgfsetstrokecolor{currentstroke}%
\pgfsetdash{}{0pt}%
\pgfpathmoveto{\pgfqpoint{1.201990in}{0.666563in}}%
\pgfpathquadraticcurveto{\pgfqpoint{1.201990in}{0.951809in}}{\pgfqpoint{1.201990in}{1.237054in}}%
\pgfpathlineto{\pgfqpoint{1.153379in}{1.237054in}}%
\pgfpathquadraticcurveto{\pgfqpoint{1.195046in}{1.320407in}}{\pgfqpoint{1.236712in}{1.403760in}}%
\pgfpathquadraticcurveto{\pgfqpoint{1.278379in}{1.320407in}}{\pgfqpoint{1.320046in}{1.237054in}}%
\pgfpathlineto{\pgfqpoint{1.271435in}{1.237054in}}%
\pgfpathquadraticcurveto{\pgfqpoint{1.271435in}{0.951809in}}{\pgfqpoint{1.271435in}{0.666563in}}%
\pgfpathlineto{\pgfqpoint{1.201990in}{0.666563in}}%
\pgfpathlineto{\pgfqpoint{1.201990in}{0.666563in}}%
\pgfpathclose%
\pgfusepath{stroke,fill}%
\end{pgfscope}%
\begin{pgfscope}%
\definecolor{textcolor}{rgb}{0.862745,0.862745,0.862745}%
\pgfsetstrokecolor{textcolor}%
\pgfsetfillcolor{textcolor}%
\pgftext[x=1.534224in,y=1.035183in,left,]{\color{textcolor}\rmfamily\fontsize{12.000000}{14.400000}\selectfont better}%
\end{pgfscope}%
\begin{pgfscope}%
\pgfsetbuttcap%
\pgfsetmiterjoin%
\definecolor{currentfill}{rgb}{1.000000,1.000000,1.000000}%
\pgfsetfillcolor{currentfill}%
\pgfsetfillopacity{0.800000}%
\pgfsetlinewidth{0.803000pt}%
\definecolor{currentstroke}{rgb}{0.800000,0.800000,0.800000}%
\pgfsetstrokecolor{currentstroke}%
\pgfsetstrokeopacity{0.800000}%
\pgfsetdash{}{0pt}%
\pgfpathmoveto{\pgfqpoint{2.178907in}{2.026444in}}%
\pgfpathlineto{\pgfqpoint{3.412248in}{2.026444in}}%
\pgfpathquadraticcurveto{\pgfqpoint{3.436692in}{2.026444in}}{\pgfqpoint{3.436692in}{2.050889in}}%
\pgfpathlineto{\pgfqpoint{3.436692in}{2.383111in}}%
\pgfpathquadraticcurveto{\pgfqpoint{3.436692in}{2.407556in}}{\pgfqpoint{3.412248in}{2.407556in}}%
\pgfpathlineto{\pgfqpoint{2.178907in}{2.407556in}}%
\pgfpathquadraticcurveto{\pgfqpoint{2.154462in}{2.407556in}}{\pgfqpoint{2.154462in}{2.383111in}}%
\pgfpathlineto{\pgfqpoint{2.154462in}{2.050889in}}%
\pgfpathquadraticcurveto{\pgfqpoint{2.154462in}{2.026444in}}{\pgfqpoint{2.178907in}{2.026444in}}%
\pgfpathlineto{\pgfqpoint{2.178907in}{2.026444in}}%
\pgfpathclose%
\pgfusepath{stroke,fill}%
\end{pgfscope}%
\begin{pgfscope}%
\pgfsetbuttcap%
\pgfsetroundjoin%
\definecolor{currentfill}{rgb}{0.121569,0.466667,0.705882}%
\pgfsetfillcolor{currentfill}%
\pgfsetlinewidth{0.803000pt}%
\definecolor{currentstroke}{rgb}{0.121569,0.466667,0.705882}%
\pgfsetstrokecolor{currentstroke}%
\pgfsetdash{}{0pt}%
\pgfsys@defobject{currentmarker}{\pgfqpoint{-0.038036in}{-0.038036in}}{\pgfqpoint{0.038036in}{0.038036in}}{%
\pgfpathmoveto{\pgfqpoint{0.000000in}{-0.038036in}}%
\pgfpathcurveto{\pgfqpoint{0.010087in}{-0.038036in}}{\pgfqpoint{0.019763in}{-0.034029in}}{\pgfqpoint{0.026896in}{-0.026896in}}%
\pgfpathcurveto{\pgfqpoint{0.034029in}{-0.019763in}}{\pgfqpoint{0.038036in}{-0.010087in}}{\pgfqpoint{0.038036in}{0.000000in}}%
\pgfpathcurveto{\pgfqpoint{0.038036in}{0.010087in}}{\pgfqpoint{0.034029in}{0.019763in}}{\pgfqpoint{0.026896in}{0.026896in}}%
\pgfpathcurveto{\pgfqpoint{0.019763in}{0.034029in}}{\pgfqpoint{0.010087in}{0.038036in}}{\pgfqpoint{0.000000in}{0.038036in}}%
\pgfpathcurveto{\pgfqpoint{-0.010087in}{0.038036in}}{\pgfqpoint{-0.019763in}{0.034029in}}{\pgfqpoint{-0.026896in}{0.026896in}}%
\pgfpathcurveto{\pgfqpoint{-0.034029in}{0.019763in}}{\pgfqpoint{-0.038036in}{0.010087in}}{\pgfqpoint{-0.038036in}{0.000000in}}%
\pgfpathcurveto{\pgfqpoint{-0.038036in}{-0.010087in}}{\pgfqpoint{-0.034029in}{-0.019763in}}{\pgfqpoint{-0.026896in}{-0.026896in}}%
\pgfpathcurveto{\pgfqpoint{-0.019763in}{-0.034029in}}{\pgfqpoint{-0.010087in}{-0.038036in}}{\pgfqpoint{0.000000in}{-0.038036in}}%
\pgfpathlineto{\pgfqpoint{0.000000in}{-0.038036in}}%
\pgfpathclose%
\pgfusepath{stroke,fill}%
}%
\begin{pgfscope}%
\pgfsys@transformshift{2.325573in}{2.303944in}%
\pgfsys@useobject{currentmarker}{}%
\end{pgfscope}%
\end{pgfscope}%
\begin{pgfscope}%
\definecolor{textcolor}{rgb}{0.150000,0.150000,0.150000}%
\pgfsetstrokecolor{textcolor}%
\pgfsetfillcolor{textcolor}%
\pgftext[x=2.545573in,y=2.271861in,left,base]{\color{textcolor}\rmfamily\fontsize{8.800000}{10.560000}\selectfont Maddock et al.}%
\end{pgfscope}%
\begin{pgfscope}%
\pgfsetroundcap%
\pgfsetroundjoin%
\pgfsetlinewidth{1.204500pt}%
\definecolor{currentstroke}{rgb}{1.000000,0.498039,0.054902}%
\pgfsetstrokecolor{currentstroke}%
\pgfsetdash{}{0pt}%
\pgfpathmoveto{\pgfqpoint{2.203351in}{2.142417in}}%
\pgfpathlineto{\pgfqpoint{2.325573in}{2.142417in}}%
\pgfpathlineto{\pgfqpoint{2.447796in}{2.142417in}}%
\pgfusepath{stroke}%
\end{pgfscope}%
\begin{pgfscope}%
\definecolor{textcolor}{rgb}{0.150000,0.150000,0.150000}%
\pgfsetstrokecolor{textcolor}%
\pgfsetfillcolor{textcolor}%
\pgftext[x=2.545573in,y=2.099639in,left,base]{\color{textcolor}\rmfamily\fontsize{8.800000}{10.560000}\selectfont S-BDT}%
\end{pgfscope}%
\end{pgfpicture}%
\makeatother%
\endgroup%

%% file: experiments_streams.tex
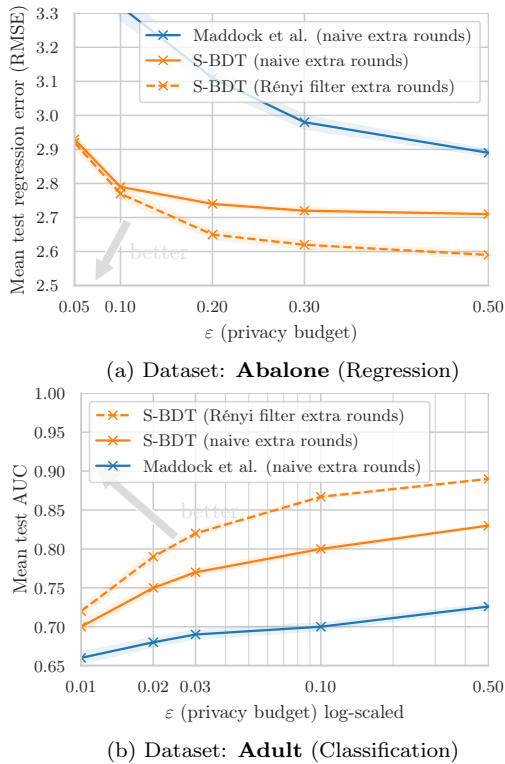
\begin{figure}[!ht]
  \centering
  \begin{subfigure}[h]{0.89\columnwidth}
      \begin{subfigure}[h]{\columnwidth}
        \resizebox{0.9\textwidth}{!}{\input{images/abalone_streams.pgf}\unskip}
          \caption{\blue{Dataset: \textbf{Abalone} (Regression)}}
          \label{fig:abalone_streams}
      \end{subfigure}\\
      \begin{subfigure}[h]{\columnwidth}
        \resizebox{0.9\textwidth}{!}{\input{images/adult_streams.pgf}\unskip}
        \caption{D\blue{ataset: \textbf{Adult} (Classification)}}
          \label{fig:adult_streams}
      \end{subfigure}
  \end{subfigure}
  \caption{\textbf{Learning a stream of non-IID data:} Regression error (RMSE) (Abalone) and AUC (Adult) of 200 runs vs. privacy budget $\varepsilon$. The transparent area is the standard error.}
  \label{fig:all_streams}
\end{figure}

%% file: images/abalone_streams.pgf
\begingroup%
\makeatletter%
\begin{pgfpicture}%
\pgfpathrectangle{\pgfpointorigin}{\pgfqpoint{3.612000in}{2.495774in}}%
\pgfusepath{use as bounding box, clip}%
\begin{pgfscope}%
\pgfsetbuttcap%
\pgfsetmiterjoin%
\definecolor{currentfill}{rgb}{1.000000,1.000000,1.000000}%
\pgfsetfillcolor{currentfill}%
\pgfsetlinewidth{0.000000pt}%
\definecolor{currentstroke}{rgb}{1.000000,1.000000,1.000000}%
\pgfsetstrokecolor{currentstroke}%
\pgfsetdash{}{0pt}%
\pgfpathmoveto{\pgfqpoint{0.000000in}{0.000000in}}%
\pgfpathlineto{\pgfqpoint{3.612000in}{0.000000in}}%
\pgfpathlineto{\pgfqpoint{3.612000in}{2.495774in}}%
\pgfpathlineto{\pgfqpoint{0.000000in}{2.495774in}}%
\pgfpathclose%
\pgfusepath{fill}%
\end{pgfscope}%
\begin{pgfscope}%
\pgfsetbuttcap%
\pgfsetmiterjoin%
\definecolor{currentfill}{rgb}{1.000000,1.000000,1.000000}%
\pgfsetfillcolor{currentfill}%
\pgfsetlinewidth{0.000000pt}%
\definecolor{currentstroke}{rgb}{0.000000,0.000000,0.000000}%
\pgfsetstrokecolor{currentstroke}%
\pgfsetstrokeopacity{0.000000}%
\pgfsetdash{}{0pt}%
\pgfpathmoveto{\pgfqpoint{0.473880in}{0.420833in}}%
\pgfpathlineto{\pgfqpoint{3.497803in}{0.420833in}}%
\pgfpathlineto{\pgfqpoint{3.497803in}{2.425264in}}%
\pgfpathlineto{\pgfqpoint{0.473880in}{2.425264in}}%
\pgfpathclose%
\pgfusepath{fill}%
\end{pgfscope}%
\begin{pgfscope}%
\pgfpathrectangle{\pgfqpoint{0.473880in}{0.420833in}}{\pgfqpoint{3.023923in}{2.004431in}}%
\pgfusepath{clip}%
\pgfsetroundcap%
\pgfsetroundjoin%
\pgfsetlinewidth{0.803000pt}%
\definecolor{currentstroke}{rgb}{0.800000,0.800000,0.800000}%
\pgfsetstrokecolor{currentstroke}%
\pgfsetdash{}{0pt}%
\pgfpathmoveto{\pgfqpoint{0.473880in}{0.420833in}}%
\pgfpathlineto{\pgfqpoint{0.473880in}{2.425264in}}%
\pgfusepath{stroke}%
\end{pgfscope}%
\begin{pgfscope}%
\definecolor{textcolor}{rgb}{0.150000,0.150000,0.150000}%
\pgfsetstrokecolor{textcolor}%
\pgfsetfillcolor{textcolor}%
\pgftext[x=0.473880in,y=0.305556in,,top]{\color{textcolor}\rmfamily\fontsize{8.800000}{10.560000}\selectfont \(\displaystyle {0.05}\)}%
\end{pgfscope}%
\begin{pgfscope}%
\pgfpathrectangle{\pgfqpoint{0.473880in}{0.420833in}}{\pgfqpoint{3.023923in}{2.004431in}}%
\pgfusepath{clip}%
\pgfsetroundcap%
\pgfsetroundjoin%
\pgfsetlinewidth{0.803000pt}%
\definecolor{currentstroke}{rgb}{0.800000,0.800000,0.800000}%
\pgfsetstrokecolor{currentstroke}%
\pgfsetdash{}{0pt}%
\pgfpathmoveto{\pgfqpoint{0.809872in}{0.420833in}}%
\pgfpathlineto{\pgfqpoint{0.809872in}{2.425264in}}%
\pgfusepath{stroke}%
\end{pgfscope}%
\begin{pgfscope}%
\definecolor{textcolor}{rgb}{0.150000,0.150000,0.150000}%
\pgfsetstrokecolor{textcolor}%
\pgfsetfillcolor{textcolor}%
\pgftext[x=0.809872in,y=0.305556in,,top]{\color{textcolor}\rmfamily\fontsize{8.800000}{10.560000}\selectfont \(\displaystyle {0.10}\)}%
\end{pgfscope}%
\begin{pgfscope}%
\pgfpathrectangle{\pgfqpoint{0.473880in}{0.420833in}}{\pgfqpoint{3.023923in}{2.004431in}}%
\pgfusepath{clip}%
\pgfsetroundcap%
\pgfsetroundjoin%
\pgfsetlinewidth{0.803000pt}%
\definecolor{currentstroke}{rgb}{0.800000,0.800000,0.800000}%
\pgfsetstrokecolor{currentstroke}%
\pgfsetdash{}{0pt}%
\pgfpathmoveto{\pgfqpoint{1.481854in}{0.420833in}}%
\pgfpathlineto{\pgfqpoint{1.481854in}{2.425264in}}%
\pgfusepath{stroke}%
\end{pgfscope}%
\begin{pgfscope}%
\definecolor{textcolor}{rgb}{0.150000,0.150000,0.150000}%
\pgfsetstrokecolor{textcolor}%
\pgfsetfillcolor{textcolor}%
\pgftext[x=1.481854in,y=0.305556in,,top]{\color{textcolor}\rmfamily\fontsize{8.800000}{10.560000}\selectfont \(\displaystyle {0.20}\)}%
\end{pgfscope}%
\begin{pgfscope}%
\pgfpathrectangle{\pgfqpoint{0.473880in}{0.420833in}}{\pgfqpoint{3.023923in}{2.004431in}}%
\pgfusepath{clip}%
\pgfsetroundcap%
\pgfsetroundjoin%
\pgfsetlinewidth{0.803000pt}%
\definecolor{currentstroke}{rgb}{0.800000,0.800000,0.800000}%
\pgfsetstrokecolor{currentstroke}%
\pgfsetdash{}{0pt}%
\pgfpathmoveto{\pgfqpoint{2.153837in}{0.420833in}}%
\pgfpathlineto{\pgfqpoint{2.153837in}{2.425264in}}%
\pgfusepath{stroke}%
\end{pgfscope}%
\begin{pgfscope}%
\definecolor{textcolor}{rgb}{0.150000,0.150000,0.150000}%
\pgfsetstrokecolor{textcolor}%
\pgfsetfillcolor{textcolor}%
\pgftext[x=2.153837in,y=0.305556in,,top]{\color{textcolor}\rmfamily\fontsize{8.800000}{10.560000}\selectfont \(\displaystyle {0.30}\)}%
\end{pgfscope}%
\begin{pgfscope}%
\pgfpathrectangle{\pgfqpoint{0.473880in}{0.420833in}}{\pgfqpoint{3.023923in}{2.004431in}}%
\pgfusepath{clip}%
\pgfsetroundcap%
\pgfsetroundjoin%
\pgfsetlinewidth{0.803000pt}%
\definecolor{currentstroke}{rgb}{0.800000,0.800000,0.800000}%
\pgfsetstrokecolor{currentstroke}%
\pgfsetdash{}{0pt}%
\pgfpathmoveto{\pgfqpoint{3.497803in}{0.420833in}}%
\pgfpathlineto{\pgfqpoint{3.497803in}{2.425264in}}%
\pgfusepath{stroke}%
\end{pgfscope}%
\begin{pgfscope}%
\definecolor{textcolor}{rgb}{0.150000,0.150000,0.150000}%
\pgfsetstrokecolor{textcolor}%
\pgfsetfillcolor{textcolor}%
\pgftext[x=3.497803in,y=0.305556in,,top]{\color{textcolor}\rmfamily\fontsize{8.800000}{10.560000}\selectfont \(\displaystyle {0.50}\)}%
\end{pgfscope}%
\begin{pgfscope}%
\definecolor{textcolor}{rgb}{0.150000,0.150000,0.150000}%
\pgfsetstrokecolor{textcolor}%
\pgfsetfillcolor{textcolor}%
\pgftext[x=1.985842in,y=0.138889in,,top]{\color{textcolor}\rmfamily\fontsize{9.600000}{11.520000}\selectfont \(\displaystyle \varepsilon\) (privacy budget)}%
\end{pgfscope}%
\begin{pgfscope}%
\pgfpathrectangle{\pgfqpoint{0.473880in}{0.420833in}}{\pgfqpoint{3.023923in}{2.004431in}}%
\pgfusepath{clip}%
\pgfsetroundcap%
\pgfsetroundjoin%
\pgfsetlinewidth{0.803000pt}%
\definecolor{currentstroke}{rgb}{0.800000,0.800000,0.800000}%
\pgfsetstrokecolor{currentstroke}%
\pgfsetdash{}{0pt}%
\pgfpathmoveto{\pgfqpoint{0.473880in}{0.420833in}}%
\pgfpathlineto{\pgfqpoint{3.497803in}{0.420833in}}%
\pgfusepath{stroke}%
\end{pgfscope}%
\begin{pgfscope}%
\definecolor{textcolor}{rgb}{0.150000,0.150000,0.150000}%
\pgfsetstrokecolor{textcolor}%
\pgfsetfillcolor{textcolor}%
\pgftext[x=0.194444in, y=0.377431in, left, base]{\color{textcolor}\rmfamily\fontsize{8.800000}{10.560000}\selectfont \(\displaystyle {2.5}\)}%
\end{pgfscope}%
\begin{pgfscope}%
\pgfpathrectangle{\pgfqpoint{0.473880in}{0.420833in}}{\pgfqpoint{3.023923in}{2.004431in}}%
\pgfusepath{clip}%
\pgfsetroundcap%
\pgfsetroundjoin%
\pgfsetlinewidth{0.803000pt}%
\definecolor{currentstroke}{rgb}{0.800000,0.800000,0.800000}%
\pgfsetstrokecolor{currentstroke}%
\pgfsetdash{}{0pt}%
\pgfpathmoveto{\pgfqpoint{0.473880in}{0.671387in}}%
\pgfpathlineto{\pgfqpoint{3.497803in}{0.671387in}}%
\pgfusepath{stroke}%
\end{pgfscope}%
\begin{pgfscope}%
\definecolor{textcolor}{rgb}{0.150000,0.150000,0.150000}%
\pgfsetstrokecolor{textcolor}%
\pgfsetfillcolor{textcolor}%
\pgftext[x=0.194444in, y=0.627984in, left, base]{\color{textcolor}\rmfamily\fontsize{8.800000}{10.560000}\selectfont \(\displaystyle {2.6}\)}%
\end{pgfscope}%
\begin{pgfscope}%
\pgfpathrectangle{\pgfqpoint{0.473880in}{0.420833in}}{\pgfqpoint{3.023923in}{2.004431in}}%
\pgfusepath{clip}%
\pgfsetroundcap%
\pgfsetroundjoin%
\pgfsetlinewidth{0.803000pt}%
\definecolor{currentstroke}{rgb}{0.800000,0.800000,0.800000}%
\pgfsetstrokecolor{currentstroke}%
\pgfsetdash{}{0pt}%
\pgfpathmoveto{\pgfqpoint{0.473880in}{0.921941in}}%
\pgfpathlineto{\pgfqpoint{3.497803in}{0.921941in}}%
\pgfusepath{stroke}%
\end{pgfscope}%
\begin{pgfscope}%
\definecolor{textcolor}{rgb}{0.150000,0.150000,0.150000}%
\pgfsetstrokecolor{textcolor}%
\pgfsetfillcolor{textcolor}%
\pgftext[x=0.194444in, y=0.878538in, left, base]{\color{textcolor}\rmfamily\fontsize{8.800000}{10.560000}\selectfont \(\displaystyle {2.7}\)}%
\end{pgfscope}%
\begin{pgfscope}%
\pgfpathrectangle{\pgfqpoint{0.473880in}{0.420833in}}{\pgfqpoint{3.023923in}{2.004431in}}%
\pgfusepath{clip}%
\pgfsetroundcap%
\pgfsetroundjoin%
\pgfsetlinewidth{0.803000pt}%
\definecolor{currentstroke}{rgb}{0.800000,0.800000,0.800000}%
\pgfsetstrokecolor{currentstroke}%
\pgfsetdash{}{0pt}%
\pgfpathmoveto{\pgfqpoint{0.473880in}{1.172495in}}%
\pgfpathlineto{\pgfqpoint{3.497803in}{1.172495in}}%
\pgfusepath{stroke}%
\end{pgfscope}%
\begin{pgfscope}%
\definecolor{textcolor}{rgb}{0.150000,0.150000,0.150000}%
\pgfsetstrokecolor{textcolor}%
\pgfsetfillcolor{textcolor}%
\pgftext[x=0.194444in, y=1.129092in, left, base]{\color{textcolor}\rmfamily\fontsize{8.800000}{10.560000}\selectfont \(\displaystyle {2.8}\)}%
\end{pgfscope}%
\begin{pgfscope}%
\pgfpathrectangle{\pgfqpoint{0.473880in}{0.420833in}}{\pgfqpoint{3.023923in}{2.004431in}}%
\pgfusepath{clip}%
\pgfsetroundcap%
\pgfsetroundjoin%
\pgfsetlinewidth{0.803000pt}%
\definecolor{currentstroke}{rgb}{0.800000,0.800000,0.800000}%
\pgfsetstrokecolor{currentstroke}%
\pgfsetdash{}{0pt}%
\pgfpathmoveto{\pgfqpoint{0.473880in}{1.423049in}}%
\pgfpathlineto{\pgfqpoint{3.497803in}{1.423049in}}%
\pgfusepath{stroke}%
\end{pgfscope}%
\begin{pgfscope}%
\definecolor{textcolor}{rgb}{0.150000,0.150000,0.150000}%
\pgfsetstrokecolor{textcolor}%
\pgfsetfillcolor{textcolor}%
\pgftext[x=0.194444in, y=1.379646in, left, base]{\color{textcolor}\rmfamily\fontsize{8.800000}{10.560000}\selectfont \(\displaystyle {2.9}\)}%
\end{pgfscope}%
\begin{pgfscope}%
\pgfpathrectangle{\pgfqpoint{0.473880in}{0.420833in}}{\pgfqpoint{3.023923in}{2.004431in}}%
\pgfusepath{clip}%
\pgfsetroundcap%
\pgfsetroundjoin%
\pgfsetlinewidth{0.803000pt}%
\definecolor{currentstroke}{rgb}{0.800000,0.800000,0.800000}%
\pgfsetstrokecolor{currentstroke}%
\pgfsetdash{}{0pt}%
\pgfpathmoveto{\pgfqpoint{0.473880in}{1.673602in}}%
\pgfpathlineto{\pgfqpoint{3.497803in}{1.673602in}}%
\pgfusepath{stroke}%
\end{pgfscope}%
\begin{pgfscope}%
\definecolor{textcolor}{rgb}{0.150000,0.150000,0.150000}%
\pgfsetstrokecolor{textcolor}%
\pgfsetfillcolor{textcolor}%
\pgftext[x=0.194444in, y=1.630200in, left, base]{\color{textcolor}\rmfamily\fontsize{8.800000}{10.560000}\selectfont \(\displaystyle {3.0}\)}%
\end{pgfscope}%
\begin{pgfscope}%
\pgfpathrectangle{\pgfqpoint{0.473880in}{0.420833in}}{\pgfqpoint{3.023923in}{2.004431in}}%
\pgfusepath{clip}%
\pgfsetroundcap%
\pgfsetroundjoin%
\pgfsetlinewidth{0.803000pt}%
\definecolor{currentstroke}{rgb}{0.800000,0.800000,0.800000}%
\pgfsetstrokecolor{currentstroke}%
\pgfsetdash{}{0pt}%
\pgfpathmoveto{\pgfqpoint{0.473880in}{1.924156in}}%
\pgfpathlineto{\pgfqpoint{3.497803in}{1.924156in}}%
\pgfusepath{stroke}%
\end{pgfscope}%
\begin{pgfscope}%
\definecolor{textcolor}{rgb}{0.150000,0.150000,0.150000}%
\pgfsetstrokecolor{textcolor}%
\pgfsetfillcolor{textcolor}%
\pgftext[x=0.194444in, y=1.880753in, left, base]{\color{textcolor}\rmfamily\fontsize{8.800000}{10.560000}\selectfont \(\displaystyle {3.1}\)}%
\end{pgfscope}%
\begin{pgfscope}%
\pgfpathrectangle{\pgfqpoint{0.473880in}{0.420833in}}{\pgfqpoint{3.023923in}{2.004431in}}%
\pgfusepath{clip}%
\pgfsetroundcap%
\pgfsetroundjoin%
\pgfsetlinewidth{0.803000pt}%
\definecolor{currentstroke}{rgb}{0.800000,0.800000,0.800000}%
\pgfsetstrokecolor{currentstroke}%
\pgfsetdash{}{0pt}%
\pgfpathmoveto{\pgfqpoint{0.473880in}{2.174710in}}%
\pgfpathlineto{\pgfqpoint{3.497803in}{2.174710in}}%
\pgfusepath{stroke}%
\end{pgfscope}%
\begin{pgfscope}%
\definecolor{textcolor}{rgb}{0.150000,0.150000,0.150000}%
\pgfsetstrokecolor{textcolor}%
\pgfsetfillcolor{textcolor}%
\pgftext[x=0.194444in, y=2.131307in, left, base]{\color{textcolor}\rmfamily\fontsize{8.800000}{10.560000}\selectfont \(\displaystyle {3.2}\)}%
\end{pgfscope}%
\begin{pgfscope}%
\pgfpathrectangle{\pgfqpoint{0.473880in}{0.420833in}}{\pgfqpoint{3.023923in}{2.004431in}}%
\pgfusepath{clip}%
\pgfsetroundcap%
\pgfsetroundjoin%
\pgfsetlinewidth{0.803000pt}%
\definecolor{currentstroke}{rgb}{0.800000,0.800000,0.800000}%
\pgfsetstrokecolor{currentstroke}%
\pgfsetdash{}{0pt}%
\pgfpathmoveto{\pgfqpoint{0.473880in}{2.425264in}}%
\pgfpathlineto{\pgfqpoint{3.497803in}{2.425264in}}%
\pgfusepath{stroke}%
\end{pgfscope}%
\begin{pgfscope}%
\definecolor{textcolor}{rgb}{0.150000,0.150000,0.150000}%
\pgfsetstrokecolor{textcolor}%
\pgfsetfillcolor{textcolor}%
\pgftext[x=0.194444in, y=2.381861in, left, base]{\color{textcolor}\rmfamily\fontsize{8.800000}{10.560000}\selectfont \(\displaystyle {3.3}\)}%
\end{pgfscope}%
\begin{pgfscope}%
\definecolor{textcolor}{rgb}{0.150000,0.150000,0.150000}%
\pgfsetstrokecolor{textcolor}%
\pgfsetfillcolor{textcolor}%
\pgftext[x=0.138889in,y=1.423049in,,bottom,rotate=90.000000]{\color{textcolor}\rmfamily\fontsize{9.600000}{11.520000}\selectfont Mean test regression error (RMSE)}%
\end{pgfscope}%
\begin{pgfscope}%
\pgfpathrectangle{\pgfqpoint{0.473880in}{0.420833in}}{\pgfqpoint{3.023923in}{2.004431in}}%
\pgfusepath{clip}%
\pgfsetbuttcap%
\pgfsetroundjoin%
\definecolor{currentfill}{rgb}{0.121569,0.466667,0.705882}%
\pgfsetfillcolor{currentfill}%
\pgfsetfillopacity{0.100000}%
\pgfsetlinewidth{0.803000pt}%
\definecolor{currentstroke}{rgb}{0.121569,0.466667,0.705882}%
\pgfsetstrokecolor{currentstroke}%
\pgfsetstrokeopacity{0.100000}%
\pgfsetdash{}{0pt}%
\pgfsys@defobject{currentmarker}{\pgfqpoint{0.473880in}{1.372938in}}{\pgfqpoint{3.497803in}{3.552756in}}{%
\pgfpathmoveto{\pgfqpoint{0.473880in}{3.552756in}}%
\pgfpathlineto{\pgfqpoint{0.473880in}{3.302202in}}%
\pgfpathlineto{\pgfqpoint{0.809872in}{2.375153in}}%
\pgfpathlineto{\pgfqpoint{1.481854in}{1.874045in}}%
\pgfpathlineto{\pgfqpoint{2.153837in}{1.573381in}}%
\pgfpathlineto{\pgfqpoint{3.497803in}{1.372938in}}%
\pgfpathlineto{\pgfqpoint{3.497803in}{1.423049in}}%
\pgfpathlineto{\pgfqpoint{3.497803in}{1.423049in}}%
\pgfpathlineto{\pgfqpoint{2.153837in}{1.673602in}}%
\pgfpathlineto{\pgfqpoint{1.481854in}{2.024378in}}%
\pgfpathlineto{\pgfqpoint{0.809872in}{2.575596in}}%
\pgfpathlineto{\pgfqpoint{0.473880in}{3.552756in}}%
\pgfpathclose%
\pgfusepath{stroke,fill}%
}%
\begin{pgfscope}%
\pgfsys@transformshift{0.000000in}{0.000000in}%
\pgfsys@useobject{currentmarker}{}%
\end{pgfscope}%
\end{pgfscope}%
\begin{pgfscope}%
\pgfpathrectangle{\pgfqpoint{0.473880in}{0.420833in}}{\pgfqpoint{3.023923in}{2.004431in}}%
\pgfusepath{clip}%
\pgfsetbuttcap%
\pgfsetroundjoin%
\definecolor{currentfill}{rgb}{1.000000,0.498039,0.054902}%
\pgfsetfillcolor{currentfill}%
\pgfsetfillopacity{0.100000}%
\pgfsetlinewidth{0.803000pt}%
\definecolor{currentstroke}{rgb}{1.000000,0.498039,0.054902}%
\pgfsetstrokecolor{currentstroke}%
\pgfsetstrokeopacity{0.100000}%
\pgfsetdash{}{0pt}%
\pgfsys@defobject{currentmarker}{\pgfqpoint{0.473880in}{0.931963in}}{\pgfqpoint{3.497803in}{1.528281in}}{%
\pgfpathmoveto{\pgfqpoint{0.473880in}{1.528281in}}%
\pgfpathlineto{\pgfqpoint{0.473880in}{1.468148in}}%
\pgfpathlineto{\pgfqpoint{0.809872in}{1.124890in}}%
\pgfpathlineto{\pgfqpoint{1.481854in}{1.002118in}}%
\pgfpathlineto{\pgfqpoint{2.153837in}{0.952007in}}%
\pgfpathlineto{\pgfqpoint{3.497803in}{0.931963in}}%
\pgfpathlineto{\pgfqpoint{3.497803in}{0.962030in}}%
\pgfpathlineto{\pgfqpoint{3.497803in}{0.962030in}}%
\pgfpathlineto{\pgfqpoint{2.153837in}{0.992096in}}%
\pgfpathlineto{\pgfqpoint{1.481854in}{1.042207in}}%
\pgfpathlineto{\pgfqpoint{0.809872in}{1.169989in}}%
\pgfpathlineto{\pgfqpoint{0.473880in}{1.528281in}}%
\pgfpathclose%
\pgfusepath{stroke,fill}%
}%
\begin{pgfscope}%
\pgfsys@transformshift{0.000000in}{0.000000in}%
\pgfsys@useobject{currentmarker}{}%
\end{pgfscope}%
\end{pgfscope}%
\begin{pgfscope}%
\pgfpathrectangle{\pgfqpoint{0.473880in}{0.420833in}}{\pgfqpoint{3.023923in}{2.004431in}}%
\pgfusepath{clip}%
\pgfsetbuttcap%
\pgfsetroundjoin%
\definecolor{currentfill}{rgb}{1.000000,0.498039,0.054902}%
\pgfsetfillcolor{currentfill}%
\pgfsetfillopacity{0.100000}%
\pgfsetlinewidth{0.803000pt}%
\definecolor{currentstroke}{rgb}{1.000000,0.498039,0.054902}%
\pgfsetstrokecolor{currentstroke}%
\pgfsetstrokeopacity{0.100000}%
\pgfsetdash{}{0pt}%
\pgfsys@defobject{currentmarker}{\pgfqpoint{0.473880in}{0.621276in}}{\pgfqpoint{3.497803in}{1.513248in}}{%
\pgfpathmoveto{\pgfqpoint{0.473880in}{1.513248in}}%
\pgfpathlineto{\pgfqpoint{0.473880in}{1.433071in}}%
\pgfpathlineto{\pgfqpoint{0.809872in}{1.059746in}}%
\pgfpathlineto{\pgfqpoint{1.481854in}{0.761587in}}%
\pgfpathlineto{\pgfqpoint{2.153837in}{0.691431in}}%
\pgfpathlineto{\pgfqpoint{3.497803in}{0.621276in}}%
\pgfpathlineto{\pgfqpoint{3.497803in}{0.671387in}}%
\pgfpathlineto{\pgfqpoint{3.497803in}{0.671387in}}%
\pgfpathlineto{\pgfqpoint{2.153837in}{0.751564in}}%
\pgfpathlineto{\pgfqpoint{1.481854in}{0.831742in}}%
\pgfpathlineto{\pgfqpoint{0.809872in}{1.134912in}}%
\pgfpathlineto{\pgfqpoint{0.473880in}{1.513248in}}%
\pgfpathclose%
\pgfusepath{stroke,fill}%
}%
\begin{pgfscope}%
\pgfsys@transformshift{0.000000in}{0.000000in}%
\pgfsys@useobject{currentmarker}{}%
\end{pgfscope}%
\end{pgfscope}%
\begin{pgfscope}%
\pgfpathrectangle{\pgfqpoint{0.473880in}{0.420833in}}{\pgfqpoint{3.023923in}{2.004431in}}%
\pgfusepath{clip}%
\pgfsetroundcap%
\pgfsetroundjoin%
\pgfsetlinewidth{1.204500pt}%
\definecolor{currentstroke}{rgb}{0.121569,0.466667,0.705882}%
\pgfsetstrokecolor{currentstroke}%
\pgfsetdash{}{0pt}%
\pgfpathmoveto{\pgfqpoint{0.861099in}{2.435264in}}%
\pgfpathlineto{\pgfqpoint{1.481854in}{1.949212in}}%
\pgfpathlineto{\pgfqpoint{2.153837in}{1.623492in}}%
\pgfpathlineto{\pgfqpoint{3.497803in}{1.397993in}}%
\pgfusepath{stroke}%
\end{pgfscope}%
\begin{pgfscope}%
\pgfpathrectangle{\pgfqpoint{0.473880in}{0.420833in}}{\pgfqpoint{3.023923in}{2.004431in}}%
\pgfusepath{clip}%
\pgfsetbuttcap%
\pgfsetroundjoin%
\definecolor{currentfill}{rgb}{0.121569,0.466667,0.705882}%
\pgfsetfillcolor{currentfill}%
\pgfsetlinewidth{0.752812pt}%
\definecolor{currentstroke}{rgb}{0.121569,0.466667,0.705882}%
\pgfsetstrokecolor{currentstroke}%
\pgfsetdash{}{0pt}%
\pgfsys@defobject{currentmarker}{\pgfqpoint{-0.033333in}{-0.033333in}}{\pgfqpoint{0.033333in}{0.033333in}}{%
\pgfpathmoveto{\pgfqpoint{-0.033333in}{-0.033333in}}%
\pgfpathlineto{\pgfqpoint{0.033333in}{0.033333in}}%
\pgfpathmoveto{\pgfqpoint{-0.033333in}{0.033333in}}%
\pgfpathlineto{\pgfqpoint{0.033333in}{-0.033333in}}%
\pgfusepath{stroke,fill}%
}%
\begin{pgfscope}%
\pgfsys@transformshift{0.473880in}{3.427479in}%
\pgfsys@useobject{currentmarker}{}%
\end{pgfscope}%
\begin{pgfscope}%
\pgfsys@transformshift{0.809872in}{2.475375in}%
\pgfsys@useobject{currentmarker}{}%
\end{pgfscope}%
\begin{pgfscope}%
\pgfsys@transformshift{1.481854in}{1.949212in}%
\pgfsys@useobject{currentmarker}{}%
\end{pgfscope}%
\begin{pgfscope}%
\pgfsys@transformshift{2.153837in}{1.623492in}%
\pgfsys@useobject{currentmarker}{}%
\end{pgfscope}%
\begin{pgfscope}%
\pgfsys@transformshift{3.497803in}{1.397993in}%
\pgfsys@useobject{currentmarker}{}%
\end{pgfscope}%
\end{pgfscope}%
\begin{pgfscope}%
\pgfpathrectangle{\pgfqpoint{0.473880in}{0.420833in}}{\pgfqpoint{3.023923in}{2.004431in}}%
\pgfusepath{clip}%
\pgfsetroundcap%
\pgfsetroundjoin%
\pgfsetlinewidth{1.204500pt}%
\definecolor{currentstroke}{rgb}{1.000000,0.498039,0.054902}%
\pgfsetstrokecolor{currentstroke}%
\pgfsetdash{}{0pt}%
\pgfpathmoveto{\pgfqpoint{0.473880in}{1.498215in}}%
\pgfpathlineto{\pgfqpoint{0.809872in}{1.147439in}}%
\pgfpathlineto{\pgfqpoint{1.481854in}{1.022163in}}%
\pgfpathlineto{\pgfqpoint{2.153837in}{0.972052in}}%
\pgfpathlineto{\pgfqpoint{3.497803in}{0.946996in}}%
\pgfusepath{stroke}%
\end{pgfscope}%
\begin{pgfscope}%
\pgfpathrectangle{\pgfqpoint{0.473880in}{0.420833in}}{\pgfqpoint{3.023923in}{2.004431in}}%
\pgfusepath{clip}%
\pgfsetbuttcap%
\pgfsetroundjoin%
\definecolor{currentfill}{rgb}{1.000000,0.498039,0.054902}%
\pgfsetfillcolor{currentfill}%
\pgfsetlinewidth{0.752812pt}%
\definecolor{currentstroke}{rgb}{1.000000,0.498039,0.054902}%
\pgfsetstrokecolor{currentstroke}%
\pgfsetdash{}{0pt}%
\pgfsys@defobject{currentmarker}{\pgfqpoint{-0.033333in}{-0.033333in}}{\pgfqpoint{0.033333in}{0.033333in}}{%
\pgfpathmoveto{\pgfqpoint{-0.033333in}{-0.033333in}}%
\pgfpathlineto{\pgfqpoint{0.033333in}{0.033333in}}%
\pgfpathmoveto{\pgfqpoint{-0.033333in}{0.033333in}}%
\pgfpathlineto{\pgfqpoint{0.033333in}{-0.033333in}}%
\pgfusepath{stroke,fill}%
}%
\begin{pgfscope}%
\pgfsys@transformshift{0.473880in}{1.498215in}%
\pgfsys@useobject{currentmarker}{}%
\end{pgfscope}%
\begin{pgfscope}%
\pgfsys@transformshift{0.809872in}{1.147439in}%
\pgfsys@useobject{currentmarker}{}%
\end{pgfscope}%
\begin{pgfscope}%
\pgfsys@transformshift{1.481854in}{1.022163in}%
\pgfsys@useobject{currentmarker}{}%
\end{pgfscope}%
\begin{pgfscope}%
\pgfsys@transformshift{2.153837in}{0.972052in}%
\pgfsys@useobject{currentmarker}{}%
\end{pgfscope}%
\begin{pgfscope}%
\pgfsys@transformshift{3.497803in}{0.946996in}%
\pgfsys@useobject{currentmarker}{}%
\end{pgfscope}%
\end{pgfscope}%
\begin{pgfscope}%
\pgfpathrectangle{\pgfqpoint{0.473880in}{0.420833in}}{\pgfqpoint{3.023923in}{2.004431in}}%
\pgfusepath{clip}%
\pgfsetbuttcap%
\pgfsetroundjoin%
\pgfsetlinewidth{1.204500pt}%
\definecolor{currentstroke}{rgb}{1.000000,0.498039,0.054902}%
\pgfsetstrokecolor{currentstroke}%
\pgfsetdash{{4.440000pt}{1.920000pt}}{0.000000pt}%
\pgfpathmoveto{\pgfqpoint{0.473880in}{1.473159in}}%
\pgfpathlineto{\pgfqpoint{0.809872in}{1.097329in}}%
\pgfpathlineto{\pgfqpoint{1.481854in}{0.796664in}}%
\pgfpathlineto{\pgfqpoint{2.153837in}{0.721498in}}%
\pgfpathlineto{\pgfqpoint{3.497803in}{0.646332in}}%
\pgfusepath{stroke}%
\end{pgfscope}%
\begin{pgfscope}%
\pgfpathrectangle{\pgfqpoint{0.473880in}{0.420833in}}{\pgfqpoint{3.023923in}{2.004431in}}%
\pgfusepath{clip}%
\pgfsetbuttcap%
\pgfsetroundjoin%
\definecolor{currentfill}{rgb}{1.000000,0.498039,0.054902}%
\pgfsetfillcolor{currentfill}%
\pgfsetlinewidth{0.752812pt}%
\definecolor{currentstroke}{rgb}{1.000000,0.498039,0.054902}%
\pgfsetstrokecolor{currentstroke}%
\pgfsetdash{}{0pt}%
\pgfsys@defobject{currentmarker}{\pgfqpoint{-0.033333in}{-0.033333in}}{\pgfqpoint{0.033333in}{0.033333in}}{%
\pgfpathmoveto{\pgfqpoint{-0.033333in}{-0.033333in}}%
\pgfpathlineto{\pgfqpoint{0.033333in}{0.033333in}}%
\pgfpathmoveto{\pgfqpoint{-0.033333in}{0.033333in}}%
\pgfpathlineto{\pgfqpoint{0.033333in}{-0.033333in}}%
\pgfusepath{stroke,fill}%
}%
\begin{pgfscope}%
\pgfsys@transformshift{0.473880in}{1.473159in}%
\pgfsys@useobject{currentmarker}{}%
\end{pgfscope}%
\begin{pgfscope}%
\pgfsys@transformshift{0.809872in}{1.097329in}%
\pgfsys@useobject{currentmarker}{}%
\end{pgfscope}%
\begin{pgfscope}%
\pgfsys@transformshift{1.481854in}{0.796664in}%
\pgfsys@useobject{currentmarker}{}%
\end{pgfscope}%
\begin{pgfscope}%
\pgfsys@transformshift{2.153837in}{0.721498in}%
\pgfsys@useobject{currentmarker}{}%
\end{pgfscope}%
\begin{pgfscope}%
\pgfsys@transformshift{3.497803in}{0.646332in}%
\pgfsys@useobject{currentmarker}{}%
\end{pgfscope}%
\end{pgfscope}%
\begin{pgfscope}%
\pgfsetrectcap%
\pgfsetmiterjoin%
\pgfsetlinewidth{1.003750pt}%
\definecolor{currentstroke}{rgb}{0.800000,0.800000,0.800000}%
\pgfsetstrokecolor{currentstroke}%
\pgfsetdash{}{0pt}%
\pgfpathmoveto{\pgfqpoint{0.473880in}{0.420833in}}%
\pgfpathlineto{\pgfqpoint{0.473880in}{2.425264in}}%
\pgfusepath{stroke}%
\end{pgfscope}%
\begin{pgfscope}%
\pgfsetrectcap%
\pgfsetmiterjoin%
\pgfsetlinewidth{1.003750pt}%
\definecolor{currentstroke}{rgb}{0.800000,0.800000,0.800000}%
\pgfsetstrokecolor{currentstroke}%
\pgfsetdash{}{0pt}%
\pgfpathmoveto{\pgfqpoint{0.473880in}{0.420833in}}%
\pgfpathlineto{\pgfqpoint{3.497803in}{0.420833in}}%
\pgfusepath{stroke}%
\end{pgfscope}%
\begin{pgfscope}%
\pgfsetroundcap%
\pgfsetroundjoin%
\definecolor{currentfill}{rgb}{0.862745,0.862745,0.862745}%
\pgfsetfillcolor{currentfill}%
\pgfsetlinewidth{0.803000pt}%
\definecolor{currentstroke}{rgb}{1.000000,1.000000,1.000000}%
\pgfsetstrokecolor{currentstroke}%
\pgfsetdash{}{0pt}%
\pgfpathmoveto{\pgfqpoint{0.894232in}{0.880479in}}%
\pgfpathquadraticcurveto{\pgfqpoint{0.812664in}{0.728413in}}{\pgfqpoint{0.731096in}{0.576346in}}%
\pgfpathlineto{\pgfqpoint{0.773933in}{0.553368in}}%
\pgfpathquadraticcurveto{\pgfqpoint{0.697830in}{0.499639in}}{\pgfqpoint{0.621727in}{0.445909in}}%
\pgfpathquadraticcurveto{\pgfqpoint{0.624394in}{0.539029in}}{\pgfqpoint{0.627062in}{0.632150in}}%
\pgfpathlineto{\pgfqpoint{0.669899in}{0.609172in}}%
\pgfpathquadraticcurveto{\pgfqpoint{0.751467in}{0.761238in}}{\pgfqpoint{0.833035in}{0.913305in}}%
\pgfpathlineto{\pgfqpoint{0.894232in}{0.880479in}}%
\pgfpathclose%
\pgfusepath{stroke,fill}%
\end{pgfscope}%
\begin{pgfscope}%
\definecolor{textcolor}{rgb}{0.862745,0.862745,0.862745}%
\pgfsetstrokecolor{textcolor}%
\pgfsetfillcolor{textcolor}%
\pgftext[x=0.877070in,y=0.671387in,left,]{\color{textcolor}\rmfamily\fontsize{12.000000}{14.400000}\selectfont better}%
\end{pgfscope}%
\begin{pgfscope}%
\pgfsetbuttcap%
\pgfsetmiterjoin%
\definecolor{currentfill}{rgb}{1.000000,1.000000,1.000000}%
\pgfsetfillcolor{currentfill}%
\pgfsetfillopacity{0.800000}%
\pgfsetlinewidth{0.803000pt}%
\definecolor{currentstroke}{rgb}{0.800000,0.800000,0.800000}%
\pgfsetstrokecolor{currentstroke}%
\pgfsetstrokeopacity{0.800000}%
\pgfsetdash{}{0pt}%
\pgfpathmoveto{\pgfqpoint{0.974883in}{1.769153in}}%
\pgfpathlineto{\pgfqpoint{3.412248in}{1.769153in}}%
\pgfpathquadraticcurveto{\pgfqpoint{3.436692in}{1.769153in}}{\pgfqpoint{3.436692in}{1.793597in}}%
\pgfpathlineto{\pgfqpoint{3.436692in}{2.339708in}}%
\pgfpathquadraticcurveto{\pgfqpoint{3.436692in}{2.364153in}}{\pgfqpoint{3.412248in}{2.364153in}}%
\pgfpathlineto{\pgfqpoint{0.974883in}{2.364153in}}%
\pgfpathquadraticcurveto{\pgfqpoint{0.950439in}{2.364153in}}{\pgfqpoint{0.950439in}{2.339708in}}%
\pgfpathlineto{\pgfqpoint{0.950439in}{1.793597in}}%
\pgfpathquadraticcurveto{\pgfqpoint{0.950439in}{1.769153in}}{\pgfqpoint{0.974883in}{1.769153in}}%
\pgfpathclose%
\pgfusepath{stroke,fill}%
\end{pgfscope}%
\begin{pgfscope}%
\pgfsetroundcap%
\pgfsetroundjoin%
\pgfsetlinewidth{1.204500pt}%
\definecolor{currentstroke}{rgb}{0.121569,0.466667,0.705882}%
\pgfsetstrokecolor{currentstroke}%
\pgfsetdash{}{0pt}%
\pgfpathmoveto{\pgfqpoint{0.999328in}{2.264292in}}%
\pgfpathlineto{\pgfqpoint{1.243772in}{2.264292in}}%
\pgfusepath{stroke}%
\end{pgfscope}%
\begin{pgfscope}%
\pgfsetbuttcap%
\pgfsetroundjoin%
\definecolor{currentfill}{rgb}{0.121569,0.466667,0.705882}%
\pgfsetfillcolor{currentfill}%
\pgfsetlinewidth{0.752812pt}%
\definecolor{currentstroke}{rgb}{0.121569,0.466667,0.705882}%
\pgfsetstrokecolor{currentstroke}%
\pgfsetdash{}{0pt}%
\pgfsys@defobject{currentmarker}{\pgfqpoint{-0.033333in}{-0.033333in}}{\pgfqpoint{0.033333in}{0.033333in}}{%
\pgfpathmoveto{\pgfqpoint{-0.033333in}{-0.033333in}}%
\pgfpathlineto{\pgfqpoint{0.033333in}{0.033333in}}%
\pgfpathmoveto{\pgfqpoint{-0.033333in}{0.033333in}}%
\pgfpathlineto{\pgfqpoint{0.033333in}{-0.033333in}}%
\pgfusepath{stroke,fill}%
}%
\begin{pgfscope}%
\pgfsys@transformshift{1.121550in}{2.264292in}%
\pgfsys@useobject{currentmarker}{}%
\end{pgfscope}%
\end{pgfscope}%
\begin{pgfscope}%
\definecolor{textcolor}{rgb}{0.150000,0.150000,0.150000}%
\pgfsetstrokecolor{textcolor}%
\pgfsetfillcolor{textcolor}%
\pgftext[x=1.341550in,y=2.221514in,left,base]{\color{textcolor}\rmfamily\fontsize{8.800000}{10.560000}\selectfont Maddock et al. (naive extra rounds)}%
\end{pgfscope}%
\begin{pgfscope}%
\pgfsetroundcap%
\pgfsetroundjoin%
\pgfsetlinewidth{1.204500pt}%
\definecolor{currentstroke}{rgb}{1.000000,0.498039,0.054902}%
\pgfsetstrokecolor{currentstroke}%
\pgfsetdash{}{0pt}%
\pgfpathmoveto{\pgfqpoint{0.999328in}{2.078181in}}%
\pgfpathlineto{\pgfqpoint{1.243772in}{2.078181in}}%
\pgfusepath{stroke}%
\end{pgfscope}%
\begin{pgfscope}%
\pgfsetbuttcap%
\pgfsetroundjoin%
\definecolor{currentfill}{rgb}{1.000000,0.498039,0.054902}%
\pgfsetfillcolor{currentfill}%
\pgfsetlinewidth{0.752812pt}%
\definecolor{currentstroke}{rgb}{1.000000,0.498039,0.054902}%
\pgfsetstrokecolor{currentstroke}%
\pgfsetdash{}{0pt}%
\pgfsys@defobject{currentmarker}{\pgfqpoint{-0.033333in}{-0.033333in}}{\pgfqpoint{0.033333in}{0.033333in}}{%
\pgfpathmoveto{\pgfqpoint{-0.033333in}{-0.033333in}}%
\pgfpathlineto{\pgfqpoint{0.033333in}{0.033333in}}%
\pgfpathmoveto{\pgfqpoint{-0.033333in}{0.033333in}}%
\pgfpathlineto{\pgfqpoint{0.033333in}{-0.033333in}}%
\pgfusepath{stroke,fill}%
}%
\begin{pgfscope}%
\pgfsys@transformshift{1.121550in}{2.078181in}%
\pgfsys@useobject{currentmarker}{}%
\end{pgfscope}%
\end{pgfscope}%
\begin{pgfscope}%
\definecolor{textcolor}{rgb}{0.150000,0.150000,0.150000}%
\pgfsetstrokecolor{textcolor}%
\pgfsetfillcolor{textcolor}%
\pgftext[x=1.341550in,y=2.035403in,left,base]{\color{textcolor}\rmfamily\fontsize{8.800000}{10.560000}\selectfont S-BDT (naive extra rounds)}%
\end{pgfscope}%
\begin{pgfscope}%
\pgfsetbuttcap%
\pgfsetroundjoin%
\pgfsetlinewidth{1.204500pt}%
\definecolor{currentstroke}{rgb}{1.000000,0.498039,0.054902}%
\pgfsetstrokecolor{currentstroke}%
\pgfsetdash{{4.440000pt}{1.920000pt}}{0.000000pt}%
\pgfpathmoveto{\pgfqpoint{0.999328in}{1.892069in}}%
\pgfpathlineto{\pgfqpoint{1.243772in}{1.892069in}}%
\pgfusepath{stroke}%
\end{pgfscope}%
\begin{pgfscope}%
\pgfsetbuttcap%
\pgfsetroundjoin%
\definecolor{currentfill}{rgb}{1.000000,0.498039,0.054902}%
\pgfsetfillcolor{currentfill}%
\pgfsetlinewidth{0.752812pt}%
\definecolor{currentstroke}{rgb}{1.000000,0.498039,0.054902}%
\pgfsetstrokecolor{currentstroke}%
\pgfsetdash{}{0pt}%
\pgfsys@defobject{currentmarker}{\pgfqpoint{-0.033333in}{-0.033333in}}{\pgfqpoint{0.033333in}{0.033333in}}{%
\pgfpathmoveto{\pgfqpoint{-0.033333in}{-0.033333in}}%
\pgfpathlineto{\pgfqpoint{0.033333in}{0.033333in}}%
\pgfpathmoveto{\pgfqpoint{-0.033333in}{0.033333in}}%
\pgfpathlineto{\pgfqpoint{0.033333in}{-0.033333in}}%
\pgfusepath{stroke,fill}%
}%
\begin{pgfscope}%
\pgfsys@transformshift{1.121550in}{1.892069in}%
\pgfsys@useobject{currentmarker}{}%
\end{pgfscope}%
\end{pgfscope}%
\begin{pgfscope}%
\definecolor{textcolor}{rgb}{0.150000,0.150000,0.150000}%
\pgfsetstrokecolor{textcolor}%
\pgfsetfillcolor{textcolor}%
\pgftext[x=1.341550in,y=1.849292in,left,base]{\color{textcolor}\rmfamily\fontsize{8.800000}{10.560000}\selectfont S-BDT (Rényi filter extra rounds)}%
\end{pgfscope}%
\end{pgfpicture}%
\makeatother%
\endgroup%

%% file: images/adult_streams.pgf
\begingroup%
\makeatletter%
\begin{pgfpicture}%
\pgfpathrectangle{\pgfpointorigin}{\pgfqpoint{3.612000in}{2.468667in}}%
\pgfusepath{use as bounding box, clip}%
\begin{pgfscope}%
\pgfsetbuttcap%
\pgfsetmiterjoin%
\definecolor{currentfill}{rgb}{1.000000,1.000000,1.000000}%
\pgfsetfillcolor{currentfill}%
\pgfsetlinewidth{0.000000pt}%
\definecolor{currentstroke}{rgb}{1.000000,1.000000,1.000000}%
\pgfsetstrokecolor{currentstroke}%
\pgfsetdash{}{0pt}%
\pgfpathmoveto{\pgfqpoint{0.000000in}{0.000000in}}%
\pgfpathlineto{\pgfqpoint{3.612000in}{0.000000in}}%
\pgfpathlineto{\pgfqpoint{3.612000in}{2.468667in}}%
\pgfpathlineto{\pgfqpoint{0.000000in}{2.468667in}}%
\pgfpathclose%
\pgfusepath{fill}%
\end{pgfscope}%
\begin{pgfscope}%
\pgfsetbuttcap%
\pgfsetmiterjoin%
\definecolor{currentfill}{rgb}{1.000000,1.000000,1.000000}%
\pgfsetfillcolor{currentfill}%
\pgfsetlinewidth{0.000000pt}%
\definecolor{currentstroke}{rgb}{0.000000,0.000000,0.000000}%
\pgfsetstrokecolor{currentstroke}%
\pgfsetstrokeopacity{0.000000}%
\pgfsetdash{}{0pt}%
\pgfpathmoveto{\pgfqpoint{0.522684in}{0.420833in}}%
\pgfpathlineto{\pgfqpoint{3.497803in}{0.420833in}}%
\pgfpathlineto{\pgfqpoint{3.497803in}{2.425264in}}%
\pgfpathlineto{\pgfqpoint{0.522684in}{2.425264in}}%
\pgfpathclose%
\pgfusepath{fill}%
\end{pgfscope}%
\begin{pgfscope}%
\pgfpathrectangle{\pgfqpoint{0.522684in}{0.420833in}}{\pgfqpoint{2.975120in}{2.004431in}}%
\pgfusepath{clip}%
\pgfsetroundcap%
\pgfsetroundjoin%
\pgfsetlinewidth{0.803000pt}%
\definecolor{currentstroke}{rgb}{0.800000,0.800000,0.800000}%
\pgfsetstrokecolor{currentstroke}%
\pgfsetdash{}{0pt}%
\pgfpathmoveto{\pgfqpoint{0.522684in}{0.420833in}}%
\pgfpathlineto{\pgfqpoint{0.522684in}{2.425264in}}%
\pgfusepath{stroke}%
\end{pgfscope}%
\begin{pgfscope}%
\definecolor{textcolor}{rgb}{0.150000,0.150000,0.150000}%
\pgfsetstrokecolor{textcolor}%
\pgfsetfillcolor{textcolor}%
\pgftext[x=0.522684in,y=0.305556in,,top]{\color{textcolor}\rmfamily\fontsize{8.800000}{10.560000}\selectfont \(\displaystyle {0.01}\)}%
\end{pgfscope}%
\begin{pgfscope}%
\pgfpathrectangle{\pgfqpoint{0.522684in}{0.420833in}}{\pgfqpoint{2.975120in}{2.004431in}}%
\pgfusepath{clip}%
\pgfsetroundcap%
\pgfsetroundjoin%
\pgfsetlinewidth{0.803000pt}%
\definecolor{currentstroke}{rgb}{0.800000,0.800000,0.800000}%
\pgfsetstrokecolor{currentstroke}%
\pgfsetdash{}{0pt}%
\pgfpathmoveto{\pgfqpoint{1.049827in}{0.420833in}}%
\pgfpathlineto{\pgfqpoint{1.049827in}{2.425264in}}%
\pgfusepath{stroke}%
\end{pgfscope}%
\begin{pgfscope}%
\definecolor{textcolor}{rgb}{0.150000,0.150000,0.150000}%
\pgfsetstrokecolor{textcolor}%
\pgfsetfillcolor{textcolor}%
\pgftext[x=1.049827in,y=0.305556in,,top]{\color{textcolor}\rmfamily\fontsize{8.800000}{10.560000}\selectfont \(\displaystyle {0.02}\)}%
\end{pgfscope}%
\begin{pgfscope}%
\pgfpathrectangle{\pgfqpoint{0.522684in}{0.420833in}}{\pgfqpoint{2.975120in}{2.004431in}}%
\pgfusepath{clip}%
\pgfsetroundcap%
\pgfsetroundjoin%
\pgfsetlinewidth{0.803000pt}%
\definecolor{currentstroke}{rgb}{0.800000,0.800000,0.800000}%
\pgfsetstrokecolor{currentstroke}%
\pgfsetdash{}{0pt}%
\pgfpathmoveto{\pgfqpoint{1.358186in}{0.420833in}}%
\pgfpathlineto{\pgfqpoint{1.358186in}{2.425264in}}%
\pgfusepath{stroke}%
\end{pgfscope}%
\begin{pgfscope}%
\definecolor{textcolor}{rgb}{0.150000,0.150000,0.150000}%
\pgfsetstrokecolor{textcolor}%
\pgfsetfillcolor{textcolor}%
\pgftext[x=1.358186in,y=0.305556in,,top]{\color{textcolor}\rmfamily\fontsize{8.800000}{10.560000}\selectfont \(\displaystyle {0.03}\)}%
\end{pgfscope}%
\begin{pgfscope}%
\pgfpathrectangle{\pgfqpoint{0.522684in}{0.420833in}}{\pgfqpoint{2.975120in}{2.004431in}}%
\pgfusepath{clip}%
\pgfsetroundcap%
\pgfsetroundjoin%
\pgfsetlinewidth{0.803000pt}%
\definecolor{currentstroke}{rgb}{0.800000,0.800000,0.800000}%
\pgfsetstrokecolor{currentstroke}%
\pgfsetdash{}{0pt}%
\pgfpathmoveto{\pgfqpoint{2.273815in}{0.420833in}}%
\pgfpathlineto{\pgfqpoint{2.273815in}{2.425264in}}%
\pgfusepath{stroke}%
\end{pgfscope}%
\begin{pgfscope}%
\definecolor{textcolor}{rgb}{0.150000,0.150000,0.150000}%
\pgfsetstrokecolor{textcolor}%
\pgfsetfillcolor{textcolor}%
\pgftext[x=2.273815in,y=0.305556in,,top]{\color{textcolor}\rmfamily\fontsize{8.800000}{10.560000}\selectfont \(\displaystyle {0.10}\)}%
\end{pgfscope}%
\begin{pgfscope}%
\pgfpathrectangle{\pgfqpoint{0.522684in}{0.420833in}}{\pgfqpoint{2.975120in}{2.004431in}}%
\pgfusepath{clip}%
\pgfsetroundcap%
\pgfsetroundjoin%
\pgfsetlinewidth{0.803000pt}%
\definecolor{currentstroke}{rgb}{0.800000,0.800000,0.800000}%
\pgfsetstrokecolor{currentstroke}%
\pgfsetdash{}{0pt}%
\pgfpathmoveto{\pgfqpoint{3.497803in}{0.420833in}}%
\pgfpathlineto{\pgfqpoint{3.497803in}{2.425264in}}%
\pgfusepath{stroke}%
\end{pgfscope}%
\begin{pgfscope}%
\definecolor{textcolor}{rgb}{0.150000,0.150000,0.150000}%
\pgfsetstrokecolor{textcolor}%
\pgfsetfillcolor{textcolor}%
\pgftext[x=3.497803in,y=0.305556in,,top]{\color{textcolor}\rmfamily\fontsize{8.800000}{10.560000}\selectfont \(\displaystyle {0.50}\)}%
\end{pgfscope}%
\begin{pgfscope}%
\pgfpathrectangle{\pgfqpoint{0.522684in}{0.420833in}}{\pgfqpoint{2.975120in}{2.004431in}}%
\pgfusepath{clip}%
\pgfsetroundcap%
\pgfsetroundjoin%
\pgfsetlinewidth{0.075281pt}%
\definecolor{currentstroke}{rgb}{0.827451,0.827451,0.827451}%
\pgfsetstrokecolor{currentstroke}%
\pgfsetdash{}{0pt}%
\pgfpathmoveto{\pgfqpoint{1.576970in}{0.420833in}}%
\pgfpathlineto{\pgfqpoint{1.576970in}{2.425264in}}%
\pgfusepath{stroke}%
\end{pgfscope}%
\begin{pgfscope}%
\pgfpathrectangle{\pgfqpoint{0.522684in}{0.420833in}}{\pgfqpoint{2.975120in}{2.004431in}}%
\pgfusepath{clip}%
\pgfsetroundcap%
\pgfsetroundjoin%
\pgfsetlinewidth{0.075281pt}%
\definecolor{currentstroke}{rgb}{0.827451,0.827451,0.827451}%
\pgfsetstrokecolor{currentstroke}%
\pgfsetdash{}{0pt}%
\pgfpathmoveto{\pgfqpoint{1.746672in}{0.420833in}}%
\pgfpathlineto{\pgfqpoint{1.746672in}{2.425264in}}%
\pgfusepath{stroke}%
\end{pgfscope}%
\begin{pgfscope}%
\pgfpathrectangle{\pgfqpoint{0.522684in}{0.420833in}}{\pgfqpoint{2.975120in}{2.004431in}}%
\pgfusepath{clip}%
\pgfsetroundcap%
\pgfsetroundjoin%
\pgfsetlinewidth{0.075281pt}%
\definecolor{currentstroke}{rgb}{0.827451,0.827451,0.827451}%
\pgfsetstrokecolor{currentstroke}%
\pgfsetdash{}{0pt}%
\pgfpathmoveto{\pgfqpoint{1.885329in}{0.420833in}}%
\pgfpathlineto{\pgfqpoint{1.885329in}{2.425264in}}%
\pgfusepath{stroke}%
\end{pgfscope}%
\begin{pgfscope}%
\pgfpathrectangle{\pgfqpoint{0.522684in}{0.420833in}}{\pgfqpoint{2.975120in}{2.004431in}}%
\pgfusepath{clip}%
\pgfsetroundcap%
\pgfsetroundjoin%
\pgfsetlinewidth{0.075281pt}%
\definecolor{currentstroke}{rgb}{0.827451,0.827451,0.827451}%
\pgfsetstrokecolor{currentstroke}%
\pgfsetdash{}{0pt}%
\pgfpathmoveto{\pgfqpoint{2.002561in}{0.420833in}}%
\pgfpathlineto{\pgfqpoint{2.002561in}{2.425264in}}%
\pgfusepath{stroke}%
\end{pgfscope}%
\begin{pgfscope}%
\pgfpathrectangle{\pgfqpoint{0.522684in}{0.420833in}}{\pgfqpoint{2.975120in}{2.004431in}}%
\pgfusepath{clip}%
\pgfsetroundcap%
\pgfsetroundjoin%
\pgfsetlinewidth{0.075281pt}%
\definecolor{currentstroke}{rgb}{0.827451,0.827451,0.827451}%
\pgfsetstrokecolor{currentstroke}%
\pgfsetdash{}{0pt}%
\pgfpathmoveto{\pgfqpoint{2.104113in}{0.420833in}}%
\pgfpathlineto{\pgfqpoint{2.104113in}{2.425264in}}%
\pgfusepath{stroke}%
\end{pgfscope}%
\begin{pgfscope}%
\pgfpathrectangle{\pgfqpoint{0.522684in}{0.420833in}}{\pgfqpoint{2.975120in}{2.004431in}}%
\pgfusepath{clip}%
\pgfsetroundcap%
\pgfsetroundjoin%
\pgfsetlinewidth{0.075281pt}%
\definecolor{currentstroke}{rgb}{0.827451,0.827451,0.827451}%
\pgfsetstrokecolor{currentstroke}%
\pgfsetdash{}{0pt}%
\pgfpathmoveto{\pgfqpoint{2.193688in}{0.420833in}}%
\pgfpathlineto{\pgfqpoint{2.193688in}{2.425264in}}%
\pgfusepath{stroke}%
\end{pgfscope}%
\begin{pgfscope}%
\pgfpathrectangle{\pgfqpoint{0.522684in}{0.420833in}}{\pgfqpoint{2.975120in}{2.004431in}}%
\pgfusepath{clip}%
\pgfsetroundcap%
\pgfsetroundjoin%
\pgfsetlinewidth{0.075281pt}%
\definecolor{currentstroke}{rgb}{0.827451,0.827451,0.827451}%
\pgfsetstrokecolor{currentstroke}%
\pgfsetdash{}{0pt}%
\pgfpathmoveto{\pgfqpoint{2.800958in}{0.420833in}}%
\pgfpathlineto{\pgfqpoint{2.800958in}{2.425264in}}%
\pgfusepath{stroke}%
\end{pgfscope}%
\begin{pgfscope}%
\pgfpathrectangle{\pgfqpoint{0.522684in}{0.420833in}}{\pgfqpoint{2.975120in}{2.004431in}}%
\pgfusepath{clip}%
\pgfsetroundcap%
\pgfsetroundjoin%
\pgfsetlinewidth{0.075281pt}%
\definecolor{currentstroke}{rgb}{0.827451,0.827451,0.827451}%
\pgfsetstrokecolor{currentstroke}%
\pgfsetdash{}{0pt}%
\pgfpathmoveto{\pgfqpoint{3.109317in}{0.420833in}}%
\pgfpathlineto{\pgfqpoint{3.109317in}{2.425264in}}%
\pgfusepath{stroke}%
\end{pgfscope}%
\begin{pgfscope}%
\pgfpathrectangle{\pgfqpoint{0.522684in}{0.420833in}}{\pgfqpoint{2.975120in}{2.004431in}}%
\pgfusepath{clip}%
\pgfsetroundcap%
\pgfsetroundjoin%
\pgfsetlinewidth{0.075281pt}%
\definecolor{currentstroke}{rgb}{0.827451,0.827451,0.827451}%
\pgfsetstrokecolor{currentstroke}%
\pgfsetdash{}{0pt}%
\pgfpathmoveto{\pgfqpoint{3.328101in}{0.420833in}}%
\pgfpathlineto{\pgfqpoint{3.328101in}{2.425264in}}%
\pgfusepath{stroke}%
\end{pgfscope}%
\begin{pgfscope}%
\definecolor{textcolor}{rgb}{0.150000,0.150000,0.150000}%
\pgfsetstrokecolor{textcolor}%
\pgfsetfillcolor{textcolor}%
\pgftext[x=2.010243in,y=0.138889in,,top]{\color{textcolor}\rmfamily\fontsize{9.600000}{11.520000}\selectfont \(\displaystyle \varepsilon\) (privacy budget) log-scaled}%
\end{pgfscope}%
\begin{pgfscope}%
\pgfpathrectangle{\pgfqpoint{0.522684in}{0.420833in}}{\pgfqpoint{2.975120in}{2.004431in}}%
\pgfusepath{clip}%
\pgfsetroundcap%
\pgfsetroundjoin%
\pgfsetlinewidth{0.803000pt}%
\definecolor{currentstroke}{rgb}{0.800000,0.800000,0.800000}%
\pgfsetstrokecolor{currentstroke}%
\pgfsetdash{}{0pt}%
\pgfpathmoveto{\pgfqpoint{0.522684in}{0.420833in}}%
\pgfpathlineto{\pgfqpoint{3.497803in}{0.420833in}}%
\pgfusepath{stroke}%
\end{pgfscope}%
\begin{pgfscope}%
\definecolor{textcolor}{rgb}{0.150000,0.150000,0.150000}%
\pgfsetstrokecolor{textcolor}%
\pgfsetfillcolor{textcolor}%
\pgftext[x=0.179012in, y=0.377431in, left, base]{\color{textcolor}\rmfamily\fontsize{8.800000}{10.560000}\selectfont \(\displaystyle {0.65}\)}%
\end{pgfscope}%
\begin{pgfscope}%
\pgfpathrectangle{\pgfqpoint{0.522684in}{0.420833in}}{\pgfqpoint{2.975120in}{2.004431in}}%
\pgfusepath{clip}%
\pgfsetroundcap%
\pgfsetroundjoin%
\pgfsetlinewidth{0.803000pt}%
\definecolor{currentstroke}{rgb}{0.800000,0.800000,0.800000}%
\pgfsetstrokecolor{currentstroke}%
\pgfsetdash{}{0pt}%
\pgfpathmoveto{\pgfqpoint{0.522684in}{0.707181in}}%
\pgfpathlineto{\pgfqpoint{3.497803in}{0.707181in}}%
\pgfusepath{stroke}%
\end{pgfscope}%
\begin{pgfscope}%
\definecolor{textcolor}{rgb}{0.150000,0.150000,0.150000}%
\pgfsetstrokecolor{textcolor}%
\pgfsetfillcolor{textcolor}%
\pgftext[x=0.179012in, y=0.663778in, left, base]{\color{textcolor}\rmfamily\fontsize{8.800000}{10.560000}\selectfont \(\displaystyle {0.70}\)}%
\end{pgfscope}%
\begin{pgfscope}%
\pgfpathrectangle{\pgfqpoint{0.522684in}{0.420833in}}{\pgfqpoint{2.975120in}{2.004431in}}%
\pgfusepath{clip}%
\pgfsetroundcap%
\pgfsetroundjoin%
\pgfsetlinewidth{0.803000pt}%
\definecolor{currentstroke}{rgb}{0.800000,0.800000,0.800000}%
\pgfsetstrokecolor{currentstroke}%
\pgfsetdash{}{0pt}%
\pgfpathmoveto{\pgfqpoint{0.522684in}{0.993528in}}%
\pgfpathlineto{\pgfqpoint{3.497803in}{0.993528in}}%
\pgfusepath{stroke}%
\end{pgfscope}%
\begin{pgfscope}%
\definecolor{textcolor}{rgb}{0.150000,0.150000,0.150000}%
\pgfsetstrokecolor{textcolor}%
\pgfsetfillcolor{textcolor}%
\pgftext[x=0.179012in, y=0.950125in, left, base]{\color{textcolor}\rmfamily\fontsize{8.800000}{10.560000}\selectfont \(\displaystyle {0.75}\)}%
\end{pgfscope}%
\begin{pgfscope}%
\pgfpathrectangle{\pgfqpoint{0.522684in}{0.420833in}}{\pgfqpoint{2.975120in}{2.004431in}}%
\pgfusepath{clip}%
\pgfsetroundcap%
\pgfsetroundjoin%
\pgfsetlinewidth{0.803000pt}%
\definecolor{currentstroke}{rgb}{0.800000,0.800000,0.800000}%
\pgfsetstrokecolor{currentstroke}%
\pgfsetdash{}{0pt}%
\pgfpathmoveto{\pgfqpoint{0.522684in}{1.279875in}}%
\pgfpathlineto{\pgfqpoint{3.497803in}{1.279875in}}%
\pgfusepath{stroke}%
\end{pgfscope}%
\begin{pgfscope}%
\definecolor{textcolor}{rgb}{0.150000,0.150000,0.150000}%
\pgfsetstrokecolor{textcolor}%
\pgfsetfillcolor{textcolor}%
\pgftext[x=0.179012in, y=1.236472in, left, base]{\color{textcolor}\rmfamily\fontsize{8.800000}{10.560000}\selectfont \(\displaystyle {0.80}\)}%
\end{pgfscope}%
\begin{pgfscope}%
\pgfpathrectangle{\pgfqpoint{0.522684in}{0.420833in}}{\pgfqpoint{2.975120in}{2.004431in}}%
\pgfusepath{clip}%
\pgfsetroundcap%
\pgfsetroundjoin%
\pgfsetlinewidth{0.803000pt}%
\definecolor{currentstroke}{rgb}{0.800000,0.800000,0.800000}%
\pgfsetstrokecolor{currentstroke}%
\pgfsetdash{}{0pt}%
\pgfpathmoveto{\pgfqpoint{0.522684in}{1.566222in}}%
\pgfpathlineto{\pgfqpoint{3.497803in}{1.566222in}}%
\pgfusepath{stroke}%
\end{pgfscope}%
\begin{pgfscope}%
\definecolor{textcolor}{rgb}{0.150000,0.150000,0.150000}%
\pgfsetstrokecolor{textcolor}%
\pgfsetfillcolor{textcolor}%
\pgftext[x=0.179012in, y=1.522819in, left, base]{\color{textcolor}\rmfamily\fontsize{8.800000}{10.560000}\selectfont \(\displaystyle {0.85}\)}%
\end{pgfscope}%
\begin{pgfscope}%
\pgfpathrectangle{\pgfqpoint{0.522684in}{0.420833in}}{\pgfqpoint{2.975120in}{2.004431in}}%
\pgfusepath{clip}%
\pgfsetroundcap%
\pgfsetroundjoin%
\pgfsetlinewidth{0.803000pt}%
\definecolor{currentstroke}{rgb}{0.800000,0.800000,0.800000}%
\pgfsetstrokecolor{currentstroke}%
\pgfsetdash{}{0pt}%
\pgfpathmoveto{\pgfqpoint{0.522684in}{1.852569in}}%
\pgfpathlineto{\pgfqpoint{3.497803in}{1.852569in}}%
\pgfusepath{stroke}%
\end{pgfscope}%
\begin{pgfscope}%
\definecolor{textcolor}{rgb}{0.150000,0.150000,0.150000}%
\pgfsetstrokecolor{textcolor}%
\pgfsetfillcolor{textcolor}%
\pgftext[x=0.179012in, y=1.809167in, left, base]{\color{textcolor}\rmfamily\fontsize{8.800000}{10.560000}\selectfont \(\displaystyle {0.90}\)}%
\end{pgfscope}%
\begin{pgfscope}%
\pgfpathrectangle{\pgfqpoint{0.522684in}{0.420833in}}{\pgfqpoint{2.975120in}{2.004431in}}%
\pgfusepath{clip}%
\pgfsetroundcap%
\pgfsetroundjoin%
\pgfsetlinewidth{0.803000pt}%
\definecolor{currentstroke}{rgb}{0.800000,0.800000,0.800000}%
\pgfsetstrokecolor{currentstroke}%
\pgfsetdash{}{0pt}%
\pgfpathmoveto{\pgfqpoint{0.522684in}{2.138917in}}%
\pgfpathlineto{\pgfqpoint{3.497803in}{2.138917in}}%
\pgfusepath{stroke}%
\end{pgfscope}%
\begin{pgfscope}%
\definecolor{textcolor}{rgb}{0.150000,0.150000,0.150000}%
\pgfsetstrokecolor{textcolor}%
\pgfsetfillcolor{textcolor}%
\pgftext[x=0.179012in, y=2.095514in, left, base]{\color{textcolor}\rmfamily\fontsize{8.800000}{10.560000}\selectfont \(\displaystyle {0.95}\)}%
\end{pgfscope}%
\begin{pgfscope}%
\pgfpathrectangle{\pgfqpoint{0.522684in}{0.420833in}}{\pgfqpoint{2.975120in}{2.004431in}}%
\pgfusepath{clip}%
\pgfsetroundcap%
\pgfsetroundjoin%
\pgfsetlinewidth{0.803000pt}%
\definecolor{currentstroke}{rgb}{0.800000,0.800000,0.800000}%
\pgfsetstrokecolor{currentstroke}%
\pgfsetdash{}{0pt}%
\pgfpathmoveto{\pgfqpoint{0.522684in}{2.425264in}}%
\pgfpathlineto{\pgfqpoint{3.497803in}{2.425264in}}%
\pgfusepath{stroke}%
\end{pgfscope}%
\begin{pgfscope}%
\definecolor{textcolor}{rgb}{0.150000,0.150000,0.150000}%
\pgfsetstrokecolor{textcolor}%
\pgfsetfillcolor{textcolor}%
\pgftext[x=0.179012in, y=2.381861in, left, base]{\color{textcolor}\rmfamily\fontsize{8.800000}{10.560000}\selectfont \(\displaystyle {1.00}\)}%
\end{pgfscope}%
\begin{pgfscope}%
\definecolor{textcolor}{rgb}{0.150000,0.150000,0.150000}%
\pgfsetstrokecolor{textcolor}%
\pgfsetfillcolor{textcolor}%
\pgftext[x=0.123457in,y=1.423049in,,bottom,rotate=90.000000]{\color{textcolor}\rmfamily\fontsize{9.600000}{11.520000}\selectfont Mean test AUC}%
\end{pgfscope}%
\begin{pgfscope}%
\pgfpathrectangle{\pgfqpoint{0.522684in}{0.420833in}}{\pgfqpoint{2.975120in}{2.004431in}}%
\pgfusepath{clip}%
\pgfsetbuttcap%
\pgfsetroundjoin%
\definecolor{currentfill}{rgb}{1.000000,0.498039,0.054902}%
\pgfsetfillcolor{currentfill}%
\pgfsetfillopacity{0.100000}%
\pgfsetlinewidth{0.803000pt}%
\definecolor{currentstroke}{rgb}{1.000000,0.498039,0.054902}%
\pgfsetstrokecolor{currentstroke}%
\pgfsetstrokeopacity{0.100000}%
\pgfsetdash{}{0pt}%
\pgfsys@defobject{currentmarker}{\pgfqpoint{0.522684in}{0.781631in}}{\pgfqpoint{3.497803in}{1.801027in}}{%
\pgfpathmoveto{\pgfqpoint{0.522684in}{0.861808in}}%
\pgfpathlineto{\pgfqpoint{0.522684in}{0.781631in}}%
\pgfpathlineto{\pgfqpoint{1.049827in}{1.188244in}}%
\pgfpathlineto{\pgfqpoint{1.358186in}{1.371506in}}%
\pgfpathlineto{\pgfqpoint{2.273815in}{1.652126in}}%
\pgfpathlineto{\pgfqpoint{3.497803in}{1.789573in}}%
\pgfpathlineto{\pgfqpoint{3.497803in}{1.801027in}}%
\pgfpathlineto{\pgfqpoint{3.497803in}{1.801027in}}%
\pgfpathlineto{\pgfqpoint{2.273815in}{1.675034in}}%
\pgfpathlineto{\pgfqpoint{1.358186in}{1.417322in}}%
\pgfpathlineto{\pgfqpoint{1.049827in}{1.256967in}}%
\pgfpathlineto{\pgfqpoint{0.522684in}{0.861808in}}%
\pgfpathclose%
\pgfusepath{stroke,fill}%
}%
\begin{pgfscope}%
\pgfsys@transformshift{0.000000in}{0.000000in}%
\pgfsys@useobject{currentmarker}{}%
\end{pgfscope}%
\end{pgfscope}%
\begin{pgfscope}%
\pgfpathrectangle{\pgfqpoint{0.522684in}{0.420833in}}{\pgfqpoint{2.975120in}{2.004431in}}%
\pgfusepath{clip}%
\pgfsetbuttcap%
\pgfsetroundjoin%
\definecolor{currentfill}{rgb}{1.000000,0.498039,0.054902}%
\pgfsetfillcolor{currentfill}%
\pgfsetfillopacity{0.100000}%
\pgfsetlinewidth{0.803000pt}%
\definecolor{currentstroke}{rgb}{1.000000,0.498039,0.054902}%
\pgfsetstrokecolor{currentstroke}%
\pgfsetstrokeopacity{0.100000}%
\pgfsetdash{}{0pt}%
\pgfsys@defobject{currentmarker}{\pgfqpoint{0.522684in}{0.667092in}}{\pgfqpoint{3.497803in}{1.457410in}}{%
\pgfpathmoveto{\pgfqpoint{0.522684in}{0.747269in}}%
\pgfpathlineto{\pgfqpoint{0.522684in}{0.667092in}}%
\pgfpathlineto{\pgfqpoint{1.049827in}{0.959166in}}%
\pgfpathlineto{\pgfqpoint{1.358186in}{1.079432in}}%
\pgfpathlineto{\pgfqpoint{2.273815in}{1.256967in}}%
\pgfpathlineto{\pgfqpoint{3.497803in}{1.445956in}}%
\pgfpathlineto{\pgfqpoint{3.497803in}{1.457410in}}%
\pgfpathlineto{\pgfqpoint{3.497803in}{1.457410in}}%
\pgfpathlineto{\pgfqpoint{2.273815in}{1.302783in}}%
\pgfpathlineto{\pgfqpoint{1.358186in}{1.136701in}}%
\pgfpathlineto{\pgfqpoint{1.049827in}{1.027889in}}%
\pgfpathlineto{\pgfqpoint{0.522684in}{0.747269in}}%
\pgfpathclose%
\pgfusepath{stroke,fill}%
}%
\begin{pgfscope}%
\pgfsys@transformshift{0.000000in}{0.000000in}%
\pgfsys@useobject{currentmarker}{}%
\end{pgfscope}%
\end{pgfscope}%
\begin{pgfscope}%
\pgfpathrectangle{\pgfqpoint{0.522684in}{0.420833in}}{\pgfqpoint{2.975120in}{2.004431in}}%
\pgfusepath{clip}%
\pgfsetbuttcap%
\pgfsetroundjoin%
\definecolor{currentfill}{rgb}{0.121569,0.466667,0.705882}%
\pgfsetfillcolor{currentfill}%
\pgfsetfillopacity{0.100000}%
\pgfsetlinewidth{0.803000pt}%
\definecolor{currentstroke}{rgb}{0.121569,0.466667,0.705882}%
\pgfsetstrokecolor{currentstroke}%
\pgfsetstrokeopacity{0.100000}%
\pgfsetdash{}{0pt}%
\pgfsys@defobject{currentmarker}{\pgfqpoint{0.522684in}{0.438014in}}{\pgfqpoint{3.497803in}{0.878989in}}{%
\pgfpathmoveto{\pgfqpoint{0.522684in}{0.518191in}}%
\pgfpathlineto{\pgfqpoint{0.522684in}{0.438014in}}%
\pgfpathlineto{\pgfqpoint{1.049827in}{0.558280in}}%
\pgfpathlineto{\pgfqpoint{1.358186in}{0.615549in}}%
\pgfpathlineto{\pgfqpoint{2.273815in}{0.678546in}}%
\pgfpathlineto{\pgfqpoint{3.497803in}{0.833173in}}%
\pgfpathlineto{\pgfqpoint{3.497803in}{0.878989in}}%
\pgfpathlineto{\pgfqpoint{3.497803in}{0.878989in}}%
\pgfpathlineto{\pgfqpoint{2.273815in}{0.735815in}}%
\pgfpathlineto{\pgfqpoint{1.358186in}{0.684273in}}%
\pgfpathlineto{\pgfqpoint{1.049827in}{0.627003in}}%
\pgfpathlineto{\pgfqpoint{0.522684in}{0.518191in}}%
\pgfpathclose%
\pgfusepath{stroke,fill}%
}%
\begin{pgfscope}%
\pgfsys@transformshift{0.000000in}{0.000000in}%
\pgfsys@useobject{currentmarker}{}%
\end{pgfscope}%
\end{pgfscope}%
\begin{pgfscope}%
\pgfpathrectangle{\pgfqpoint{0.522684in}{0.420833in}}{\pgfqpoint{2.975120in}{2.004431in}}%
\pgfusepath{clip}%
\pgfsetbuttcap%
\pgfsetroundjoin%
\pgfsetlinewidth{1.204500pt}%
\definecolor{currentstroke}{rgb}{1.000000,0.498039,0.054902}%
\pgfsetstrokecolor{currentstroke}%
\pgfsetdash{{4.440000pt}{1.920000pt}}{0.000000pt}%
\pgfpathmoveto{\pgfqpoint{0.522684in}{0.821719in}}%
\pgfpathlineto{\pgfqpoint{1.049827in}{1.222606in}}%
\pgfpathlineto{\pgfqpoint{1.358186in}{1.394414in}}%
\pgfpathlineto{\pgfqpoint{2.273815in}{1.663580in}}%
\pgfpathlineto{\pgfqpoint{3.497803in}{1.795300in}}%
\pgfusepath{stroke}%
\end{pgfscope}%
\begin{pgfscope}%
\pgfpathrectangle{\pgfqpoint{0.522684in}{0.420833in}}{\pgfqpoint{2.975120in}{2.004431in}}%
\pgfusepath{clip}%
\pgfsetbuttcap%
\pgfsetroundjoin%
\definecolor{currentfill}{rgb}{1.000000,0.498039,0.054902}%
\pgfsetfillcolor{currentfill}%
\pgfsetlinewidth{0.752812pt}%
\definecolor{currentstroke}{rgb}{1.000000,0.498039,0.054902}%
\pgfsetstrokecolor{currentstroke}%
\pgfsetdash{}{0pt}%
\pgfsys@defobject{currentmarker}{\pgfqpoint{-0.033333in}{-0.033333in}}{\pgfqpoint{0.033333in}{0.033333in}}{%
\pgfpathmoveto{\pgfqpoint{-0.033333in}{-0.033333in}}%
\pgfpathlineto{\pgfqpoint{0.033333in}{0.033333in}}%
\pgfpathmoveto{\pgfqpoint{-0.033333in}{0.033333in}}%
\pgfpathlineto{\pgfqpoint{0.033333in}{-0.033333in}}%
\pgfusepath{stroke,fill}%
}%
\begin{pgfscope}%
\pgfsys@transformshift{0.522684in}{0.821719in}%
\pgfsys@useobject{currentmarker}{}%
\end{pgfscope}%
\begin{pgfscope}%
\pgfsys@transformshift{1.049827in}{1.222606in}%
\pgfsys@useobject{currentmarker}{}%
\end{pgfscope}%
\begin{pgfscope}%
\pgfsys@transformshift{1.358186in}{1.394414in}%
\pgfsys@useobject{currentmarker}{}%
\end{pgfscope}%
\begin{pgfscope}%
\pgfsys@transformshift{2.273815in}{1.663580in}%
\pgfsys@useobject{currentmarker}{}%
\end{pgfscope}%
\begin{pgfscope}%
\pgfsys@transformshift{3.497803in}{1.795300in}%
\pgfsys@useobject{currentmarker}{}%
\end{pgfscope}%
\end{pgfscope}%
\begin{pgfscope}%
\pgfpathrectangle{\pgfqpoint{0.522684in}{0.420833in}}{\pgfqpoint{2.975120in}{2.004431in}}%
\pgfusepath{clip}%
\pgfsetroundcap%
\pgfsetroundjoin%
\pgfsetlinewidth{1.204500pt}%
\definecolor{currentstroke}{rgb}{1.000000,0.498039,0.054902}%
\pgfsetstrokecolor{currentstroke}%
\pgfsetdash{}{0pt}%
\pgfpathmoveto{\pgfqpoint{0.522684in}{0.707181in}}%
\pgfpathlineto{\pgfqpoint{1.049827in}{0.993528in}}%
\pgfpathlineto{\pgfqpoint{1.358186in}{1.108067in}}%
\pgfpathlineto{\pgfqpoint{2.273815in}{1.279875in}}%
\pgfpathlineto{\pgfqpoint{3.497803in}{1.451683in}}%
\pgfusepath{stroke}%
\end{pgfscope}%
\begin{pgfscope}%
\pgfpathrectangle{\pgfqpoint{0.522684in}{0.420833in}}{\pgfqpoint{2.975120in}{2.004431in}}%
\pgfusepath{clip}%
\pgfsetbuttcap%
\pgfsetroundjoin%
\definecolor{currentfill}{rgb}{1.000000,0.498039,0.054902}%
\pgfsetfillcolor{currentfill}%
\pgfsetlinewidth{0.752812pt}%
\definecolor{currentstroke}{rgb}{1.000000,0.498039,0.054902}%
\pgfsetstrokecolor{currentstroke}%
\pgfsetdash{}{0pt}%
\pgfsys@defobject{currentmarker}{\pgfqpoint{-0.033333in}{-0.033333in}}{\pgfqpoint{0.033333in}{0.033333in}}{%
\pgfpathmoveto{\pgfqpoint{-0.033333in}{-0.033333in}}%
\pgfpathlineto{\pgfqpoint{0.033333in}{0.033333in}}%
\pgfpathmoveto{\pgfqpoint{-0.033333in}{0.033333in}}%
\pgfpathlineto{\pgfqpoint{0.033333in}{-0.033333in}}%
\pgfusepath{stroke,fill}%
}%
\begin{pgfscope}%
\pgfsys@transformshift{0.522684in}{0.707181in}%
\pgfsys@useobject{currentmarker}{}%
\end{pgfscope}%
\begin{pgfscope}%
\pgfsys@transformshift{1.049827in}{0.993528in}%
\pgfsys@useobject{currentmarker}{}%
\end{pgfscope}%
\begin{pgfscope}%
\pgfsys@transformshift{1.358186in}{1.108067in}%
\pgfsys@useobject{currentmarker}{}%
\end{pgfscope}%
\begin{pgfscope}%
\pgfsys@transformshift{2.273815in}{1.279875in}%
\pgfsys@useobject{currentmarker}{}%
\end{pgfscope}%
\begin{pgfscope}%
\pgfsys@transformshift{3.497803in}{1.451683in}%
\pgfsys@useobject{currentmarker}{}%
\end{pgfscope}%
\end{pgfscope}%
\begin{pgfscope}%
\pgfpathrectangle{\pgfqpoint{0.522684in}{0.420833in}}{\pgfqpoint{2.975120in}{2.004431in}}%
\pgfusepath{clip}%
\pgfsetroundcap%
\pgfsetroundjoin%
\pgfsetlinewidth{1.204500pt}%
\definecolor{currentstroke}{rgb}{0.121569,0.466667,0.705882}%
\pgfsetstrokecolor{currentstroke}%
\pgfsetdash{}{0pt}%
\pgfpathmoveto{\pgfqpoint{0.522684in}{0.478103in}}%
\pgfpathlineto{\pgfqpoint{1.049827in}{0.592642in}}%
\pgfpathlineto{\pgfqpoint{1.358186in}{0.649911in}}%
\pgfpathlineto{\pgfqpoint{2.273815in}{0.707181in}}%
\pgfpathlineto{\pgfqpoint{3.497803in}{0.856081in}}%
\pgfusepath{stroke}%
\end{pgfscope}%
\begin{pgfscope}%
\pgfpathrectangle{\pgfqpoint{0.522684in}{0.420833in}}{\pgfqpoint{2.975120in}{2.004431in}}%
\pgfusepath{clip}%
\pgfsetbuttcap%
\pgfsetroundjoin%
\definecolor{currentfill}{rgb}{0.121569,0.466667,0.705882}%
\pgfsetfillcolor{currentfill}%
\pgfsetlinewidth{0.752812pt}%
\definecolor{currentstroke}{rgb}{0.121569,0.466667,0.705882}%
\pgfsetstrokecolor{currentstroke}%
\pgfsetdash{}{0pt}%
\pgfsys@defobject{currentmarker}{\pgfqpoint{-0.033333in}{-0.033333in}}{\pgfqpoint{0.033333in}{0.033333in}}{%
\pgfpathmoveto{\pgfqpoint{-0.033333in}{-0.033333in}}%
\pgfpathlineto{\pgfqpoint{0.033333in}{0.033333in}}%
\pgfpathmoveto{\pgfqpoint{-0.033333in}{0.033333in}}%
\pgfpathlineto{\pgfqpoint{0.033333in}{-0.033333in}}%
\pgfusepath{stroke,fill}%
}%
\begin{pgfscope}%
\pgfsys@transformshift{0.522684in}{0.478103in}%
\pgfsys@useobject{currentmarker}{}%
\end{pgfscope}%
\begin{pgfscope}%
\pgfsys@transformshift{1.049827in}{0.592642in}%
\pgfsys@useobject{currentmarker}{}%
\end{pgfscope}%
\begin{pgfscope}%
\pgfsys@transformshift{1.358186in}{0.649911in}%
\pgfsys@useobject{currentmarker}{}%
\end{pgfscope}%
\begin{pgfscope}%
\pgfsys@transformshift{2.273815in}{0.707181in}%
\pgfsys@useobject{currentmarker}{}%
\end{pgfscope}%
\begin{pgfscope}%
\pgfsys@transformshift{3.497803in}{0.856081in}%
\pgfsys@useobject{currentmarker}{}%
\end{pgfscope}%
\end{pgfscope}%
\begin{pgfscope}%
\pgfsetrectcap%
\pgfsetmiterjoin%
\pgfsetlinewidth{1.003750pt}%
\definecolor{currentstroke}{rgb}{0.800000,0.800000,0.800000}%
\pgfsetstrokecolor{currentstroke}%
\pgfsetdash{}{0pt}%
\pgfpathmoveto{\pgfqpoint{0.522684in}{0.420833in}}%
\pgfpathlineto{\pgfqpoint{0.522684in}{2.425264in}}%
\pgfusepath{stroke}%
\end{pgfscope}%
\begin{pgfscope}%
\pgfsetrectcap%
\pgfsetmiterjoin%
\pgfsetlinewidth{1.003750pt}%
\definecolor{currentstroke}{rgb}{0.800000,0.800000,0.800000}%
\pgfsetstrokecolor{currentstroke}%
\pgfsetdash{}{0pt}%
\pgfpathmoveto{\pgfqpoint{0.522684in}{0.420833in}}%
\pgfpathlineto{\pgfqpoint{3.497803in}{0.420833in}}%
\pgfusepath{stroke}%
\end{pgfscope}%
\begin{pgfscope}%
\pgfsetroundcap%
\pgfsetroundjoin%
\definecolor{currentfill}{rgb}{0.862745,0.862745,0.862745}%
\pgfsetfillcolor{currentfill}%
\pgfsetlinewidth{0.803000pt}%
\definecolor{currentstroke}{rgb}{1.000000,1.000000,1.000000}%
\pgfsetstrokecolor{currentstroke}%
\pgfsetdash{}{0pt}%
\pgfpathmoveto{\pgfqpoint{1.193764in}{1.339664in}}%
\pgfpathquadraticcurveto{\pgfqpoint{0.962076in}{1.542490in}}{\pgfqpoint{0.730388in}{1.745316in}}%
\pgfpathlineto{\pgfqpoint{0.698369in}{1.708740in}}%
\pgfpathquadraticcurveto{\pgfqpoint{0.663126in}{1.794969in}}{\pgfqpoint{0.627883in}{1.881199in}}%
\pgfpathquadraticcurveto{\pgfqpoint{0.718017in}{1.857671in}}{\pgfqpoint{0.808150in}{1.834143in}}%
\pgfpathlineto{\pgfqpoint{0.776130in}{1.797567in}}%
\pgfpathquadraticcurveto{\pgfqpoint{1.007818in}{1.594741in}}{\pgfqpoint{1.239506in}{1.391915in}}%
\pgfpathlineto{\pgfqpoint{1.193764in}{1.339664in}}%
\pgfpathclose%
\pgfusepath{stroke,fill}%
\end{pgfscope}%
\begin{pgfscope}%
\definecolor{textcolor}{rgb}{0.862745,0.862745,0.862745}%
\pgfsetstrokecolor{textcolor}%
\pgfsetfillcolor{textcolor}%
\pgftext[x=1.249356in,y=1.566222in,left,]{\color{textcolor}\rmfamily\fontsize{12.000000}{14.400000}\selectfont better}%
\end{pgfscope}%
\begin{pgfscope}%
\pgfsetbuttcap%
\pgfsetmiterjoin%
\definecolor{currentfill}{rgb}{1.000000,1.000000,1.000000}%
\pgfsetfillcolor{currentfill}%
\pgfsetfillopacity{0.800000}%
\pgfsetlinewidth{0.803000pt}%
\definecolor{currentstroke}{rgb}{0.800000,0.800000,0.800000}%
\pgfsetstrokecolor{currentstroke}%
\pgfsetstrokeopacity{0.800000}%
\pgfsetdash{}{0pt}%
\pgfpathmoveto{\pgfqpoint{0.608239in}{1.769153in}}%
\pgfpathlineto{\pgfqpoint{3.045603in}{1.769153in}}%
\pgfpathquadraticcurveto{\pgfqpoint{3.070048in}{1.769153in}}{\pgfqpoint{3.070048in}{1.793597in}}%
\pgfpathlineto{\pgfqpoint{3.070048in}{2.339708in}}%
\pgfpathquadraticcurveto{\pgfqpoint{3.070048in}{2.364153in}}{\pgfqpoint{3.045603in}{2.364153in}}%
\pgfpathlineto{\pgfqpoint{0.608239in}{2.364153in}}%
\pgfpathquadraticcurveto{\pgfqpoint{0.583795in}{2.364153in}}{\pgfqpoint{0.583795in}{2.339708in}}%
\pgfpathlineto{\pgfqpoint{0.583795in}{1.793597in}}%
\pgfpathquadraticcurveto{\pgfqpoint{0.583795in}{1.769153in}}{\pgfqpoint{0.608239in}{1.769153in}}%
\pgfpathclose%
\pgfusepath{stroke,fill}%
\end{pgfscope}%
\begin{pgfscope}%
\pgfsetbuttcap%
\pgfsetroundjoin%
\pgfsetlinewidth{1.204500pt}%
\definecolor{currentstroke}{rgb}{1.000000,0.498039,0.054902}%
\pgfsetstrokecolor{currentstroke}%
\pgfsetdash{{4.440000pt}{1.920000pt}}{0.000000pt}%
\pgfpathmoveto{\pgfqpoint{0.632684in}{2.264292in}}%
\pgfpathlineto{\pgfqpoint{0.877128in}{2.264292in}}%
\pgfusepath{stroke}%
\end{pgfscope}%
\begin{pgfscope}%
\pgfsetbuttcap%
\pgfsetroundjoin%
\definecolor{currentfill}{rgb}{1.000000,0.498039,0.054902}%
\pgfsetfillcolor{currentfill}%
\pgfsetlinewidth{0.752812pt}%
\definecolor{currentstroke}{rgb}{1.000000,0.498039,0.054902}%
\pgfsetstrokecolor{currentstroke}%
\pgfsetdash{}{0pt}%
\pgfsys@defobject{currentmarker}{\pgfqpoint{-0.033333in}{-0.033333in}}{\pgfqpoint{0.033333in}{0.033333in}}{%
\pgfpathmoveto{\pgfqpoint{-0.033333in}{-0.033333in}}%
\pgfpathlineto{\pgfqpoint{0.033333in}{0.033333in}}%
\pgfpathmoveto{\pgfqpoint{-0.033333in}{0.033333in}}%
\pgfpathlineto{\pgfqpoint{0.033333in}{-0.033333in}}%
\pgfusepath{stroke,fill}%
}%
\begin{pgfscope}%
\pgfsys@transformshift{0.754906in}{2.264292in}%
\pgfsys@useobject{currentmarker}{}%
\end{pgfscope}%
\end{pgfscope}%
\begin{pgfscope}%
\definecolor{textcolor}{rgb}{0.150000,0.150000,0.150000}%
\pgfsetstrokecolor{textcolor}%
\pgfsetfillcolor{textcolor}%
\pgftext[x=0.974906in,y=2.221514in,left,base]{\color{textcolor}\rmfamily\fontsize{8.800000}{10.560000}\selectfont S-BDT (Rényi filter extra rounds)}%
\end{pgfscope}%
\begin{pgfscope}%
\pgfsetroundcap%
\pgfsetroundjoin%
\pgfsetlinewidth{1.204500pt}%
\definecolor{currentstroke}{rgb}{1.000000,0.498039,0.054902}%
\pgfsetstrokecolor{currentstroke}%
\pgfsetdash{}{0pt}%
\pgfpathmoveto{\pgfqpoint{0.632684in}{2.078181in}}%
\pgfpathlineto{\pgfqpoint{0.877128in}{2.078181in}}%
\pgfusepath{stroke}%
\end{pgfscope}%
\begin{pgfscope}%
\pgfsetbuttcap%
\pgfsetroundjoin%
\definecolor{currentfill}{rgb}{1.000000,0.498039,0.054902}%
\pgfsetfillcolor{currentfill}%
\pgfsetlinewidth{0.752812pt}%
\definecolor{currentstroke}{rgb}{1.000000,0.498039,0.054902}%
\pgfsetstrokecolor{currentstroke}%
\pgfsetdash{}{0pt}%
\pgfsys@defobject{currentmarker}{\pgfqpoint{-0.033333in}{-0.033333in}}{\pgfqpoint{0.033333in}{0.033333in}}{%
\pgfpathmoveto{\pgfqpoint{-0.033333in}{-0.033333in}}%
\pgfpathlineto{\pgfqpoint{0.033333in}{0.033333in}}%
\pgfpathmoveto{\pgfqpoint{-0.033333in}{0.033333in}}%
\pgfpathlineto{\pgfqpoint{0.033333in}{-0.033333in}}%
\pgfusepath{stroke,fill}%
}%
\begin{pgfscope}%
\pgfsys@transformshift{0.754906in}{2.078181in}%
\pgfsys@useobject{currentmarker}{}%
\end{pgfscope}%
\end{pgfscope}%
\begin{pgfscope}%
\definecolor{textcolor}{rgb}{0.150000,0.150000,0.150000}%
\pgfsetstrokecolor{textcolor}%
\pgfsetfillcolor{textcolor}%
\pgftext[x=0.974906in,y=2.035403in,left,base]{\color{textcolor}\rmfamily\fontsize{8.800000}{10.560000}\selectfont S-BDT (naive extra rounds)}%
\end{pgfscope}%
\begin{pgfscope}%
\pgfsetroundcap%
\pgfsetroundjoin%
\pgfsetlinewidth{1.204500pt}%
\definecolor{currentstroke}{rgb}{0.121569,0.466667,0.705882}%
\pgfsetstrokecolor{currentstroke}%
\pgfsetdash{}{0pt}%
\pgfpathmoveto{\pgfqpoint{0.632684in}{1.892069in}}%
\pgfpathlineto{\pgfqpoint{0.877128in}{1.892069in}}%
\pgfusepath{stroke}%
\end{pgfscope}%
\begin{pgfscope}%
\pgfsetbuttcap%
\pgfsetroundjoin%
\definecolor{currentfill}{rgb}{0.121569,0.466667,0.705882}%
\pgfsetfillcolor{currentfill}%
\pgfsetlinewidth{0.752812pt}%
\definecolor{currentstroke}{rgb}{0.121569,0.466667,0.705882}%
\pgfsetstrokecolor{currentstroke}%
\pgfsetdash{}{0pt}%
\pgfsys@defobject{currentmarker}{\pgfqpoint{-0.033333in}{-0.033333in}}{\pgfqpoint{0.033333in}{0.033333in}}{%
\pgfpathmoveto{\pgfqpoint{-0.033333in}{-0.033333in}}%
\pgfpathlineto{\pgfqpoint{0.033333in}{0.033333in}}%
\pgfpathmoveto{\pgfqpoint{-0.033333in}{0.033333in}}%
\pgfpathlineto{\pgfqpoint{0.033333in}{-0.033333in}}%
\pgfusepath{stroke,fill}%
}%
\begin{pgfscope}%
\pgfsys@transformshift{0.754906in}{1.892069in}%
\pgfsys@useobject{currentmarker}{}%
\end{pgfscope}%
\end{pgfscope}%
\begin{pgfscope}%
\definecolor{textcolor}{rgb}{0.150000,0.150000,0.150000}%
\pgfsetstrokecolor{textcolor}%
\pgfsetfillcolor{textcolor}%
\pgftext[x=0.974906in,y=1.849292in,left,base]{\color{textcolor}\rmfamily\fontsize{8.800000}{10.560000}\selectfont Maddock et al. (naive extra rounds)}%
\end{pgfscope}%
\end{pgfpicture}%
\makeatother%
\endgroup%

%% file: experiments_rf_ablation.tex
\begin{figure}[!b]
  \centering
  \begin{subfigure}[h]{\columnwidth}
    \resizebox{\textwidth}{!}{\input{images/abalone_rf_ablation.pgf}\unskip}
      \caption{Dataset: \textbf{Abalone} (Regression)}
      \label{fig:abalone_rf_ablation}
  \end{subfigure}\\
  \begin{subfigure}[h]{\columnwidth}
    \resizebox{\textwidth}{!}{\input{images/adult_rf_ablation.pgf}\unskip}
    \caption{Dataset: \textbf{Adult} (Classification)}
      \label{fig:adult_rf_ablation}
  \end{subfigure}
  \caption{\textbf{Ablation study of an individual Rényi filter tailored to \dpgbdt{}}. Regression error (RMSE) (Abalone) and AUC (Adult) of 200 runs. The transparent area is the standard error. For Abalone we set $\varepsilon=0.1$, number of trees 150, depth 2, subsampling ratio $\gamma=0.1$, leaf-balanced noise parameter $r_1=0.2$, privacy budget ratio for initial score $\varepsilon_\text{init}/\varepsilon=0.1$ and clipping bound $g^*=0.1$. For Adult we set $\varepsilon=0.02$, number of trees 200, depth 5, subsampling ratio $\gamma=0.005$, leaf-balanced noise parameter $r_1=0.1$, privacy budget ratio for initial score $\varepsilon_\text{init}/\varepsilon=0.1$ and clipping bounds $g^*=0.5, h^*=0.1$.}
  \label{fig:rf_ablations}
\end{figure}
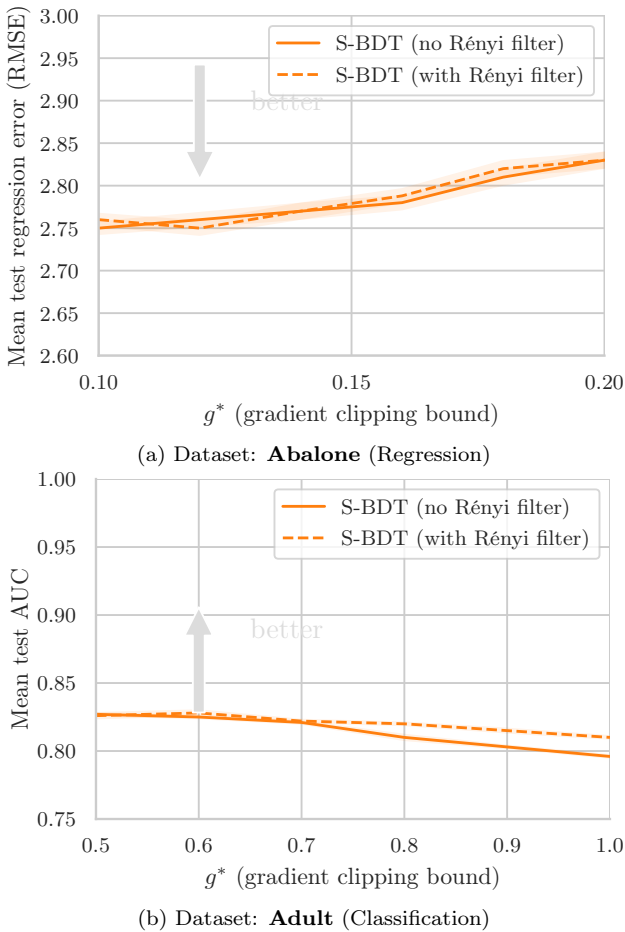

%% file: images/abalone_rf_ablation.pgf
\begingroup%
\makeatletter%
\begin{pgfpicture}%
\pgfpathrectangle{\pgfpointorigin}{\pgfqpoint{3.612000in}{2.495774in}}%
\pgfusepath{use as bounding box, clip}%
\begin{pgfscope}%
\pgfsetbuttcap%
\pgfsetmiterjoin%
\definecolor{currentfill}{rgb}{1.000000,1.000000,1.000000}%
\pgfsetfillcolor{currentfill}%
\pgfsetlinewidth{0.000000pt}%
\definecolor{currentstroke}{rgb}{1.000000,1.000000,1.000000}%
\pgfsetstrokecolor{currentstroke}%
\pgfsetdash{}{0pt}%
\pgfpathmoveto{\pgfqpoint{0.000000in}{0.000000in}}%
\pgfpathlineto{\pgfqpoint{3.612000in}{0.000000in}}%
\pgfpathlineto{\pgfqpoint{3.612000in}{2.495774in}}%
\pgfpathlineto{\pgfqpoint{0.000000in}{2.495774in}}%
\pgfpathlineto{\pgfqpoint{0.000000in}{0.000000in}}%
\pgfpathclose%
\pgfusepath{fill}%
\end{pgfscope}%
\begin{pgfscope}%
\pgfsetbuttcap%
\pgfsetmiterjoin%
\definecolor{currentfill}{rgb}{1.000000,1.000000,1.000000}%
\pgfsetfillcolor{currentfill}%
\pgfsetlinewidth{0.000000pt}%
\definecolor{currentstroke}{rgb}{0.000000,0.000000,0.000000}%
\pgfsetstrokecolor{currentstroke}%
\pgfsetstrokeopacity{0.000000}%
\pgfsetdash{}{0pt}%
\pgfpathmoveto{\pgfqpoint{0.538116in}{0.420833in}}%
\pgfpathlineto{\pgfqpoint{3.497803in}{0.420833in}}%
\pgfpathlineto{\pgfqpoint{3.497803in}{2.425264in}}%
\pgfpathlineto{\pgfqpoint{0.538116in}{2.425264in}}%
\pgfpathlineto{\pgfqpoint{0.538116in}{0.420833in}}%
\pgfpathclose%
\pgfusepath{fill}%
\end{pgfscope}%
\begin{pgfscope}%
\pgfpathrectangle{\pgfqpoint{0.538116in}{0.420833in}}{\pgfqpoint{2.959687in}{2.004431in}}%
\pgfusepath{clip}%
\pgfsetroundcap%
\pgfsetroundjoin%
\pgfsetlinewidth{0.803000pt}%
\definecolor{currentstroke}{rgb}{0.800000,0.800000,0.800000}%
\pgfsetstrokecolor{currentstroke}%
\pgfsetdash{}{0pt}%
\pgfpathmoveto{\pgfqpoint{0.538116in}{0.420833in}}%
\pgfpathlineto{\pgfqpoint{0.538116in}{2.425264in}}%
\pgfusepath{stroke}%
\end{pgfscope}%
\begin{pgfscope}%
\definecolor{textcolor}{rgb}{0.150000,0.150000,0.150000}%
\pgfsetstrokecolor{textcolor}%
\pgfsetfillcolor{textcolor}%
\pgftext[x=0.538116in,y=0.305556in,,top]{\color{textcolor}\rmfamily\fontsize{8.800000}{10.560000}\selectfont \(\displaystyle {0.10}\)}%
\end{pgfscope}%
\begin{pgfscope}%
\pgfpathrectangle{\pgfqpoint{0.538116in}{0.420833in}}{\pgfqpoint{2.959687in}{2.004431in}}%
\pgfusepath{clip}%
\pgfsetroundcap%
\pgfsetroundjoin%
\pgfsetlinewidth{0.803000pt}%
\definecolor{currentstroke}{rgb}{0.800000,0.800000,0.800000}%
\pgfsetstrokecolor{currentstroke}%
\pgfsetdash{}{0pt}%
\pgfpathmoveto{\pgfqpoint{2.017960in}{0.420833in}}%
\pgfpathlineto{\pgfqpoint{2.017960in}{2.425264in}}%
\pgfusepath{stroke}%
\end{pgfscope}%
\begin{pgfscope}%
\definecolor{textcolor}{rgb}{0.150000,0.150000,0.150000}%
\pgfsetstrokecolor{textcolor}%
\pgfsetfillcolor{textcolor}%
\pgftext[x=2.017960in,y=0.305556in,,top]{\color{textcolor}\rmfamily\fontsize{8.800000}{10.560000}\selectfont \(\displaystyle {0.15}\)}%
\end{pgfscope}%
\begin{pgfscope}%
\pgfpathrectangle{\pgfqpoint{0.538116in}{0.420833in}}{\pgfqpoint{2.959687in}{2.004431in}}%
\pgfusepath{clip}%
\pgfsetroundcap%
\pgfsetroundjoin%
\pgfsetlinewidth{0.803000pt}%
\definecolor{currentstroke}{rgb}{0.800000,0.800000,0.800000}%
\pgfsetstrokecolor{currentstroke}%
\pgfsetdash{}{0pt}%
\pgfpathmoveto{\pgfqpoint{3.497803in}{0.420833in}}%
\pgfpathlineto{\pgfqpoint{3.497803in}{2.425264in}}%
\pgfusepath{stroke}%
\end{pgfscope}%
\begin{pgfscope}%
\definecolor{textcolor}{rgb}{0.150000,0.150000,0.150000}%
\pgfsetstrokecolor{textcolor}%
\pgfsetfillcolor{textcolor}%
\pgftext[x=3.497803in,y=0.305556in,,top]{\color{textcolor}\rmfamily\fontsize{8.800000}{10.560000}\selectfont \(\displaystyle {0.20}\)}%
\end{pgfscope}%
\begin{pgfscope}%
\definecolor{textcolor}{rgb}{0.150000,0.150000,0.150000}%
\pgfsetstrokecolor{textcolor}%
\pgfsetfillcolor{textcolor}%
\pgftext[x=2.017960in,y=0.138889in,,top]{\color{textcolor}\rmfamily\fontsize{9.600000}{11.520000}\selectfont \(\displaystyle g^*\) (gradient clipping bound)}%
\end{pgfscope}%
\begin{pgfscope}%
\pgfpathrectangle{\pgfqpoint{0.538116in}{0.420833in}}{\pgfqpoint{2.959687in}{2.004431in}}%
\pgfusepath{clip}%
\pgfsetroundcap%
\pgfsetroundjoin%
\pgfsetlinewidth{0.803000pt}%
\definecolor{currentstroke}{rgb}{0.800000,0.800000,0.800000}%
\pgfsetstrokecolor{currentstroke}%
\pgfsetdash{}{0pt}%
\pgfpathmoveto{\pgfqpoint{0.538116in}{0.420833in}}%
\pgfpathlineto{\pgfqpoint{3.497803in}{0.420833in}}%
\pgfusepath{stroke}%
\end{pgfscope}%
\begin{pgfscope}%
\definecolor{textcolor}{rgb}{0.150000,0.150000,0.150000}%
\pgfsetstrokecolor{textcolor}%
\pgfsetfillcolor{textcolor}%
\pgftext[x=0.194444in, y=0.377431in, left, base]{\color{textcolor}\rmfamily\fontsize{8.800000}{10.560000}\selectfont \(\displaystyle {2.60}\)}%
\end{pgfscope}%
\begin{pgfscope}%
\pgfpathrectangle{\pgfqpoint{0.538116in}{0.420833in}}{\pgfqpoint{2.959687in}{2.004431in}}%
\pgfusepath{clip}%
\pgfsetroundcap%
\pgfsetroundjoin%
\pgfsetlinewidth{0.803000pt}%
\definecolor{currentstroke}{rgb}{0.800000,0.800000,0.800000}%
\pgfsetstrokecolor{currentstroke}%
\pgfsetdash{}{0pt}%
\pgfpathmoveto{\pgfqpoint{0.538116in}{0.671387in}}%
\pgfpathlineto{\pgfqpoint{3.497803in}{0.671387in}}%
\pgfusepath{stroke}%
\end{pgfscope}%
\begin{pgfscope}%
\definecolor{textcolor}{rgb}{0.150000,0.150000,0.150000}%
\pgfsetstrokecolor{textcolor}%
\pgfsetfillcolor{textcolor}%
\pgftext[x=0.194444in, y=0.627984in, left, base]{\color{textcolor}\rmfamily\fontsize{8.800000}{10.560000}\selectfont \(\displaystyle {2.65}\)}%
\end{pgfscope}%
\begin{pgfscope}%
\pgfpathrectangle{\pgfqpoint{0.538116in}{0.420833in}}{\pgfqpoint{2.959687in}{2.004431in}}%
\pgfusepath{clip}%
\pgfsetroundcap%
\pgfsetroundjoin%
\pgfsetlinewidth{0.803000pt}%
\definecolor{currentstroke}{rgb}{0.800000,0.800000,0.800000}%
\pgfsetstrokecolor{currentstroke}%
\pgfsetdash{}{0pt}%
\pgfpathmoveto{\pgfqpoint{0.538116in}{0.921941in}}%
\pgfpathlineto{\pgfqpoint{3.497803in}{0.921941in}}%
\pgfusepath{stroke}%
\end{pgfscope}%
\begin{pgfscope}%
\definecolor{textcolor}{rgb}{0.150000,0.150000,0.150000}%
\pgfsetstrokecolor{textcolor}%
\pgfsetfillcolor{textcolor}%
\pgftext[x=0.194444in, y=0.878538in, left, base]{\color{textcolor}\rmfamily\fontsize{8.800000}{10.560000}\selectfont \(\displaystyle {2.70}\)}%
\end{pgfscope}%
\begin{pgfscope}%
\pgfpathrectangle{\pgfqpoint{0.538116in}{0.420833in}}{\pgfqpoint{2.959687in}{2.004431in}}%
\pgfusepath{clip}%
\pgfsetroundcap%
\pgfsetroundjoin%
\pgfsetlinewidth{0.803000pt}%
\definecolor{currentstroke}{rgb}{0.800000,0.800000,0.800000}%
\pgfsetstrokecolor{currentstroke}%
\pgfsetdash{}{0pt}%
\pgfpathmoveto{\pgfqpoint{0.538116in}{1.172495in}}%
\pgfpathlineto{\pgfqpoint{3.497803in}{1.172495in}}%
\pgfusepath{stroke}%
\end{pgfscope}%
\begin{pgfscope}%
\definecolor{textcolor}{rgb}{0.150000,0.150000,0.150000}%
\pgfsetstrokecolor{textcolor}%
\pgfsetfillcolor{textcolor}%
\pgftext[x=0.194444in, y=1.129092in, left, base]{\color{textcolor}\rmfamily\fontsize{8.800000}{10.560000}\selectfont \(\displaystyle {2.75}\)}%
\end{pgfscope}%
\begin{pgfscope}%
\pgfpathrectangle{\pgfqpoint{0.538116in}{0.420833in}}{\pgfqpoint{2.959687in}{2.004431in}}%
\pgfusepath{clip}%
\pgfsetroundcap%
\pgfsetroundjoin%
\pgfsetlinewidth{0.803000pt}%
\definecolor{currentstroke}{rgb}{0.800000,0.800000,0.800000}%
\pgfsetstrokecolor{currentstroke}%
\pgfsetdash{}{0pt}%
\pgfpathmoveto{\pgfqpoint{0.538116in}{1.423049in}}%
\pgfpathlineto{\pgfqpoint{3.497803in}{1.423049in}}%
\pgfusepath{stroke}%
\end{pgfscope}%
\begin{pgfscope}%
\definecolor{textcolor}{rgb}{0.150000,0.150000,0.150000}%
\pgfsetstrokecolor{textcolor}%
\pgfsetfillcolor{textcolor}%
\pgftext[x=0.194444in, y=1.379646in, left, base]{\color{textcolor}\rmfamily\fontsize{8.800000}{10.560000}\selectfont \(\displaystyle {2.80}\)}%
\end{pgfscope}%
\begin{pgfscope}%
\pgfpathrectangle{\pgfqpoint{0.538116in}{0.420833in}}{\pgfqpoint{2.959687in}{2.004431in}}%
\pgfusepath{clip}%
\pgfsetroundcap%
\pgfsetroundjoin%
\pgfsetlinewidth{0.803000pt}%
\definecolor{currentstroke}{rgb}{0.800000,0.800000,0.800000}%
\pgfsetstrokecolor{currentstroke}%
\pgfsetdash{}{0pt}%
\pgfpathmoveto{\pgfqpoint{0.538116in}{1.673602in}}%
\pgfpathlineto{\pgfqpoint{3.497803in}{1.673602in}}%
\pgfusepath{stroke}%
\end{pgfscope}%
\begin{pgfscope}%
\definecolor{textcolor}{rgb}{0.150000,0.150000,0.150000}%
\pgfsetstrokecolor{textcolor}%
\pgfsetfillcolor{textcolor}%
\pgftext[x=0.194444in, y=1.630200in, left, base]{\color{textcolor}\rmfamily\fontsize{8.800000}{10.560000}\selectfont \(\displaystyle {2.85}\)}%
\end{pgfscope}%
\begin{pgfscope}%
\pgfpathrectangle{\pgfqpoint{0.538116in}{0.420833in}}{\pgfqpoint{2.959687in}{2.004431in}}%
\pgfusepath{clip}%
\pgfsetroundcap%
\pgfsetroundjoin%
\pgfsetlinewidth{0.803000pt}%
\definecolor{currentstroke}{rgb}{0.800000,0.800000,0.800000}%
\pgfsetstrokecolor{currentstroke}%
\pgfsetdash{}{0pt}%
\pgfpathmoveto{\pgfqpoint{0.538116in}{1.924156in}}%
\pgfpathlineto{\pgfqpoint{3.497803in}{1.924156in}}%
\pgfusepath{stroke}%
\end{pgfscope}%
\begin{pgfscope}%
\definecolor{textcolor}{rgb}{0.150000,0.150000,0.150000}%
\pgfsetstrokecolor{textcolor}%
\pgfsetfillcolor{textcolor}%
\pgftext[x=0.194444in, y=1.880753in, left, base]{\color{textcolor}\rmfamily\fontsize{8.800000}{10.560000}\selectfont \(\displaystyle {2.90}\)}%
\end{pgfscope}%
\begin{pgfscope}%
\pgfpathrectangle{\pgfqpoint{0.538116in}{0.420833in}}{\pgfqpoint{2.959687in}{2.004431in}}%
\pgfusepath{clip}%
\pgfsetroundcap%
\pgfsetroundjoin%
\pgfsetlinewidth{0.803000pt}%
\definecolor{currentstroke}{rgb}{0.800000,0.800000,0.800000}%
\pgfsetstrokecolor{currentstroke}%
\pgfsetdash{}{0pt}%
\pgfpathmoveto{\pgfqpoint{0.538116in}{2.174710in}}%
\pgfpathlineto{\pgfqpoint{3.497803in}{2.174710in}}%
\pgfusepath{stroke}%
\end{pgfscope}%
\begin{pgfscope}%
\definecolor{textcolor}{rgb}{0.150000,0.150000,0.150000}%
\pgfsetstrokecolor{textcolor}%
\pgfsetfillcolor{textcolor}%
\pgftext[x=0.194444in, y=2.131307in, left, base]{\color{textcolor}\rmfamily\fontsize{8.800000}{10.560000}\selectfont \(\displaystyle {2.95}\)}%
\end{pgfscope}%
\begin{pgfscope}%
\pgfpathrectangle{\pgfqpoint{0.538116in}{0.420833in}}{\pgfqpoint{2.959687in}{2.004431in}}%
\pgfusepath{clip}%
\pgfsetroundcap%
\pgfsetroundjoin%
\pgfsetlinewidth{0.803000pt}%
\definecolor{currentstroke}{rgb}{0.800000,0.800000,0.800000}%
\pgfsetstrokecolor{currentstroke}%
\pgfsetdash{}{0pt}%
\pgfpathmoveto{\pgfqpoint{0.538116in}{2.425264in}}%
\pgfpathlineto{\pgfqpoint{3.497803in}{2.425264in}}%
\pgfusepath{stroke}%
\end{pgfscope}%
\begin{pgfscope}%
\definecolor{textcolor}{rgb}{0.150000,0.150000,0.150000}%
\pgfsetstrokecolor{textcolor}%
\pgfsetfillcolor{textcolor}%
\pgftext[x=0.194444in, y=2.381861in, left, base]{\color{textcolor}\rmfamily\fontsize{8.800000}{10.560000}\selectfont \(\displaystyle {3.00}\)}%
\end{pgfscope}%
\begin{pgfscope}%
\definecolor{textcolor}{rgb}{0.150000,0.150000,0.150000}%
\pgfsetstrokecolor{textcolor}%
\pgfsetfillcolor{textcolor}%
\pgftext[x=0.138889in,y=1.423049in,,bottom,rotate=90.000000]{\color{textcolor}\rmfamily\fontsize{9.600000}{11.520000}\selectfont Mean test regression error (RMSE)}%
\end{pgfscope}%
\begin{pgfscope}%
\pgfpathrectangle{\pgfqpoint{0.538116in}{0.420833in}}{\pgfqpoint{2.959687in}{2.004431in}}%
\pgfusepath{clip}%
\pgfsetbuttcap%
\pgfsetroundjoin%
\definecolor{currentfill}{rgb}{1.000000,0.498039,0.054902}%
\pgfsetfillcolor{currentfill}%
\pgfsetfillopacity{0.100000}%
\pgfsetlinewidth{0.803000pt}%
\definecolor{currentstroke}{rgb}{1.000000,0.498039,0.054902}%
\pgfsetstrokecolor{currentstroke}%
\pgfsetstrokeopacity{0.100000}%
\pgfsetdash{}{0pt}%
\pgfsys@defobject{currentmarker}{\pgfqpoint{0.538116in}{1.137417in}}{\pgfqpoint{3.497803in}{1.618481in}}{%
\pgfpathmoveto{\pgfqpoint{0.538116in}{1.207572in}}%
\pgfpathlineto{\pgfqpoint{0.538116in}{1.137417in}}%
\pgfpathlineto{\pgfqpoint{1.130053in}{1.182517in}}%
\pgfpathlineto{\pgfqpoint{1.721991in}{1.227617in}}%
\pgfpathlineto{\pgfqpoint{2.313928in}{1.282738in}}%
\pgfpathlineto{\pgfqpoint{2.905866in}{1.428060in}}%
\pgfpathlineto{\pgfqpoint{3.497803in}{1.528281in}}%
\pgfpathlineto{\pgfqpoint{3.497803in}{1.618481in}}%
\pgfpathlineto{\pgfqpoint{3.497803in}{1.618481in}}%
\pgfpathlineto{\pgfqpoint{2.905866in}{1.518259in}}%
\pgfpathlineto{\pgfqpoint{2.313928in}{1.362916in}}%
\pgfpathlineto{\pgfqpoint{1.721991in}{1.317816in}}%
\pgfpathlineto{\pgfqpoint{1.130053in}{1.262694in}}%
\pgfpathlineto{\pgfqpoint{0.538116in}{1.207572in}}%
\pgfpathlineto{\pgfqpoint{0.538116in}{1.207572in}}%
\pgfpathclose%
\pgfusepath{stroke,fill}%
}%
\begin{pgfscope}%
\pgfsys@transformshift{0.000000in}{0.000000in}%
\pgfsys@useobject{currentmarker}{}%
\end{pgfscope}%
\end{pgfscope}%
\begin{pgfscope}%
\pgfpathrectangle{\pgfqpoint{0.538116in}{0.420833in}}{\pgfqpoint{2.959687in}{2.004431in}}%
\pgfusepath{clip}%
\pgfsetbuttcap%
\pgfsetroundjoin%
\definecolor{currentfill}{rgb}{1.000000,0.498039,0.054902}%
\pgfsetfillcolor{currentfill}%
\pgfsetfillopacity{0.100000}%
\pgfsetlinewidth{0.803000pt}%
\definecolor{currentstroke}{rgb}{1.000000,0.498039,0.054902}%
\pgfsetstrokecolor{currentstroke}%
\pgfsetstrokeopacity{0.100000}%
\pgfsetdash{}{0pt}%
\pgfsys@defobject{currentmarker}{\pgfqpoint{0.538116in}{1.132406in}}{\pgfqpoint{3.497803in}{1.618481in}}{%
\pgfpathmoveto{\pgfqpoint{0.538116in}{1.257683in}}%
\pgfpathlineto{\pgfqpoint{0.538116in}{1.187528in}}%
\pgfpathlineto{\pgfqpoint{1.130053in}{1.132406in}}%
\pgfpathlineto{\pgfqpoint{1.721991in}{1.227617in}}%
\pgfpathlineto{\pgfqpoint{2.313928in}{1.322827in}}%
\pgfpathlineto{\pgfqpoint{2.905866in}{1.478170in}}%
\pgfpathlineto{\pgfqpoint{3.497803in}{1.528281in}}%
\pgfpathlineto{\pgfqpoint{3.497803in}{1.618481in}}%
\pgfpathlineto{\pgfqpoint{3.497803in}{1.618481in}}%
\pgfpathlineto{\pgfqpoint{2.905866in}{1.568370in}}%
\pgfpathlineto{\pgfqpoint{2.313928in}{1.403004in}}%
\pgfpathlineto{\pgfqpoint{1.721991in}{1.317816in}}%
\pgfpathlineto{\pgfqpoint{1.130053in}{1.212583in}}%
\pgfpathlineto{\pgfqpoint{0.538116in}{1.257683in}}%
\pgfpathlineto{\pgfqpoint{0.538116in}{1.257683in}}%
\pgfpathclose%
\pgfusepath{stroke,fill}%
}%
\begin{pgfscope}%
\pgfsys@transformshift{0.000000in}{0.000000in}%
\pgfsys@useobject{currentmarker}{}%
\end{pgfscope}%
\end{pgfscope}%
\begin{pgfscope}%
\pgfpathrectangle{\pgfqpoint{0.538116in}{0.420833in}}{\pgfqpoint{2.959687in}{2.004431in}}%
\pgfusepath{clip}%
\pgfsetroundcap%
\pgfsetroundjoin%
\pgfsetlinewidth{1.204500pt}%
\definecolor{currentstroke}{rgb}{1.000000,0.498039,0.054902}%
\pgfsetstrokecolor{currentstroke}%
\pgfsetdash{}{0pt}%
\pgfpathmoveto{\pgfqpoint{0.538116in}{1.172495in}}%
\pgfpathlineto{\pgfqpoint{1.130053in}{1.222606in}}%
\pgfpathlineto{\pgfqpoint{1.721991in}{1.272716in}}%
\pgfpathlineto{\pgfqpoint{2.313928in}{1.322827in}}%
\pgfpathlineto{\pgfqpoint{2.905866in}{1.473159in}}%
\pgfpathlineto{\pgfqpoint{3.497803in}{1.573381in}}%
\pgfusepath{stroke}%
\end{pgfscope}%
\begin{pgfscope}%
\pgfpathrectangle{\pgfqpoint{0.538116in}{0.420833in}}{\pgfqpoint{2.959687in}{2.004431in}}%
\pgfusepath{clip}%
\pgfsetbuttcap%
\pgfsetroundjoin%
\pgfsetlinewidth{1.204500pt}%
\definecolor{currentstroke}{rgb}{1.000000,0.498039,0.054902}%
\pgfsetstrokecolor{currentstroke}%
\pgfsetdash{{4.440000pt}{1.920000pt}}{0.000000pt}%
\pgfpathmoveto{\pgfqpoint{0.538116in}{1.222606in}}%
\pgfpathlineto{\pgfqpoint{1.130053in}{1.172495in}}%
\pgfpathlineto{\pgfqpoint{1.721991in}{1.272716in}}%
\pgfpathlineto{\pgfqpoint{2.313928in}{1.362916in}}%
\pgfpathlineto{\pgfqpoint{2.905866in}{1.523270in}}%
\pgfpathlineto{\pgfqpoint{3.497803in}{1.573381in}}%
\pgfusepath{stroke}%
\end{pgfscope}%
\begin{pgfscope}%
\pgfsetrectcap%
\pgfsetmiterjoin%
\pgfsetlinewidth{1.003750pt}%
\definecolor{currentstroke}{rgb}{0.800000,0.800000,0.800000}%
\pgfsetstrokecolor{currentstroke}%
\pgfsetdash{}{0pt}%
\pgfpathmoveto{\pgfqpoint{0.538116in}{0.420833in}}%
\pgfpathlineto{\pgfqpoint{0.538116in}{2.425264in}}%
\pgfusepath{stroke}%
\end{pgfscope}%
\begin{pgfscope}%
\pgfsetrectcap%
\pgfsetmiterjoin%
\pgfsetlinewidth{1.003750pt}%
\definecolor{currentstroke}{rgb}{0.800000,0.800000,0.800000}%
\pgfsetstrokecolor{currentstroke}%
\pgfsetdash{}{0pt}%
\pgfpathmoveto{\pgfqpoint{0.538116in}{0.420833in}}%
\pgfpathlineto{\pgfqpoint{3.497803in}{0.420833in}}%
\pgfusepath{stroke}%
\end{pgfscope}%
\begin{pgfscope}%
\pgfsetroundcap%
\pgfsetroundjoin%
\definecolor{currentfill}{rgb}{0.862745,0.862745,0.862745}%
\pgfsetfillcolor{currentfill}%
\pgfsetlinewidth{0.803000pt}%
\definecolor{currentstroke}{rgb}{1.000000,1.000000,1.000000}%
\pgfsetstrokecolor{currentstroke}%
\pgfsetdash{}{0pt}%
\pgfpathmoveto{\pgfqpoint{1.164776in}{2.137136in}}%
\pgfpathquadraticcurveto{\pgfqpoint{1.164776in}{1.882240in}}{\pgfqpoint{1.164776in}{1.627345in}}%
\pgfpathlineto{\pgfqpoint{1.213387in}{1.627345in}}%
\pgfpathquadraticcurveto{\pgfqpoint{1.171720in}{1.544003in}}{\pgfqpoint{1.130053in}{1.460662in}}%
\pgfpathquadraticcurveto{\pgfqpoint{1.088387in}{1.544003in}}{\pgfqpoint{1.046720in}{1.627345in}}%
\pgfpathlineto{\pgfqpoint{1.095331in}{1.627345in}}%
\pgfpathquadraticcurveto{\pgfqpoint{1.095331in}{1.882240in}}{\pgfqpoint{1.095331in}{2.137136in}}%
\pgfpathlineto{\pgfqpoint{1.164776in}{2.137136in}}%
\pgfpathlineto{\pgfqpoint{1.164776in}{2.137136in}}%
\pgfpathclose%
\pgfusepath{stroke,fill}%
\end{pgfscope}%
\begin{pgfscope}%
\definecolor{textcolor}{rgb}{0.862745,0.862745,0.862745}%
\pgfsetstrokecolor{textcolor}%
\pgfsetfillcolor{textcolor}%
\pgftext[x=1.426022in,y=1.924156in,left,]{\color{textcolor}\rmfamily\fontsize{12.000000}{14.400000}\selectfont better}%
\end{pgfscope}%
\begin{pgfscope}%
\pgfsetbuttcap%
\pgfsetmiterjoin%
\definecolor{currentfill}{rgb}{1.000000,1.000000,1.000000}%
\pgfsetfillcolor{currentfill}%
\pgfsetfillopacity{0.800000}%
\pgfsetlinewidth{0.803000pt}%
\definecolor{currentstroke}{rgb}{0.800000,0.800000,0.800000}%
\pgfsetstrokecolor{currentstroke}%
\pgfsetstrokeopacity{0.800000}%
\pgfsetdash{}{0pt}%
\pgfpathmoveto{\pgfqpoint{1.559419in}{1.955264in}}%
\pgfpathlineto{\pgfqpoint{3.412248in}{1.955264in}}%
\pgfpathquadraticcurveto{\pgfqpoint{3.436692in}{1.955264in}}{\pgfqpoint{3.436692in}{1.979708in}}%
\pgfpathlineto{\pgfqpoint{3.436692in}{2.339708in}}%
\pgfpathquadraticcurveto{\pgfqpoint{3.436692in}{2.364153in}}{\pgfqpoint{3.412248in}{2.364153in}}%
\pgfpathlineto{\pgfqpoint{1.559419in}{2.364153in}}%
\pgfpathquadraticcurveto{\pgfqpoint{1.534974in}{2.364153in}}{\pgfqpoint{1.534974in}{2.339708in}}%
\pgfpathlineto{\pgfqpoint{1.534974in}{1.979708in}}%
\pgfpathquadraticcurveto{\pgfqpoint{1.534974in}{1.955264in}}{\pgfqpoint{1.559419in}{1.955264in}}%
\pgfpathlineto{\pgfqpoint{1.559419in}{1.955264in}}%
\pgfpathclose%
\pgfusepath{stroke,fill}%
\end{pgfscope}%
\begin{pgfscope}%
\pgfsetroundcap%
\pgfsetroundjoin%
\pgfsetlinewidth{1.204500pt}%
\definecolor{currentstroke}{rgb}{1.000000,0.498039,0.054902}%
\pgfsetstrokecolor{currentstroke}%
\pgfsetdash{}{0pt}%
\pgfpathmoveto{\pgfqpoint{1.583863in}{2.264292in}}%
\pgfpathlineto{\pgfqpoint{1.706085in}{2.264292in}}%
\pgfpathlineto{\pgfqpoint{1.828308in}{2.264292in}}%
\pgfusepath{stroke}%
\end{pgfscope}%
\begin{pgfscope}%
\definecolor{textcolor}{rgb}{0.150000,0.150000,0.150000}%
\pgfsetstrokecolor{textcolor}%
\pgfsetfillcolor{textcolor}%
\pgftext[x=1.926085in,y=2.221514in,left,base]{\color{textcolor}\rmfamily\fontsize{8.800000}{10.560000}\selectfont S-BDT (no Rényi filter)}%
\end{pgfscope}%
\begin{pgfscope}%
\pgfsetbuttcap%
\pgfsetroundjoin%
\pgfsetlinewidth{1.204500pt}%
\definecolor{currentstroke}{rgb}{1.000000,0.498039,0.054902}%
\pgfsetstrokecolor{currentstroke}%
\pgfsetdash{{4.440000pt}{1.920000pt}}{0.000000pt}%
\pgfpathmoveto{\pgfqpoint{1.583863in}{2.078181in}}%
\pgfpathlineto{\pgfqpoint{1.706085in}{2.078181in}}%
\pgfpathlineto{\pgfqpoint{1.828308in}{2.078181in}}%
\pgfusepath{stroke}%
\end{pgfscope}%
\begin{pgfscope}%
\definecolor{textcolor}{rgb}{0.150000,0.150000,0.150000}%
\pgfsetstrokecolor{textcolor}%
\pgfsetfillcolor{textcolor}%
\pgftext[x=1.926085in,y=2.035403in,left,base]{\color{textcolor}\rmfamily\fontsize{8.800000}{10.560000}\selectfont S-BDT (with Rényi filter)}%
\end{pgfscope}%
\end{pgfpicture}%
\makeatother%
\endgroup%

%% file: images/adult_rf_ablation.pgf
\begingroup%
\makeatletter%
\begin{pgfpicture}%
\pgfpathrectangle{\pgfpointorigin}{\pgfqpoint{3.612000in}{2.468667in}}%
\pgfusepath{use as bounding box, clip}%
\begin{pgfscope}%
\pgfsetbuttcap%
\pgfsetmiterjoin%
\definecolor{currentfill}{rgb}{1.000000,1.000000,1.000000}%
\pgfsetfillcolor{currentfill}%
\pgfsetlinewidth{0.000000pt}%
\definecolor{currentstroke}{rgb}{1.000000,1.000000,1.000000}%
\pgfsetstrokecolor{currentstroke}%
\pgfsetdash{}{0pt}%
\pgfpathmoveto{\pgfqpoint{0.000000in}{0.000000in}}%
\pgfpathlineto{\pgfqpoint{3.612000in}{0.000000in}}%
\pgfpathlineto{\pgfqpoint{3.612000in}{2.468667in}}%
\pgfpathlineto{\pgfqpoint{0.000000in}{2.468667in}}%
\pgfpathlineto{\pgfqpoint{0.000000in}{0.000000in}}%
\pgfpathclose%
\pgfusepath{fill}%
\end{pgfscope}%
\begin{pgfscope}%
\pgfsetbuttcap%
\pgfsetmiterjoin%
\definecolor{currentfill}{rgb}{1.000000,1.000000,1.000000}%
\pgfsetfillcolor{currentfill}%
\pgfsetlinewidth{0.000000pt}%
\definecolor{currentstroke}{rgb}{0.000000,0.000000,0.000000}%
\pgfsetstrokecolor{currentstroke}%
\pgfsetstrokeopacity{0.000000}%
\pgfsetdash{}{0pt}%
\pgfpathmoveto{\pgfqpoint{0.522684in}{0.420833in}}%
\pgfpathlineto{\pgfqpoint{3.529921in}{0.420833in}}%
\pgfpathlineto{\pgfqpoint{3.529921in}{2.425264in}}%
\pgfpathlineto{\pgfqpoint{0.522684in}{2.425264in}}%
\pgfpathlineto{\pgfqpoint{0.522684in}{0.420833in}}%
\pgfpathclose%
\pgfusepath{fill}%
\end{pgfscope}%
\begin{pgfscope}%
\pgfpathrectangle{\pgfqpoint{0.522684in}{0.420833in}}{\pgfqpoint{3.007237in}{2.004431in}}%
\pgfusepath{clip}%
\pgfsetroundcap%
\pgfsetroundjoin%
\pgfsetlinewidth{0.803000pt}%
\definecolor{currentstroke}{rgb}{0.800000,0.800000,0.800000}%
\pgfsetstrokecolor{currentstroke}%
\pgfsetdash{}{0pt}%
\pgfpathmoveto{\pgfqpoint{0.522684in}{0.420833in}}%
\pgfpathlineto{\pgfqpoint{0.522684in}{2.425264in}}%
\pgfusepath{stroke}%
\end{pgfscope}%
\begin{pgfscope}%
\definecolor{textcolor}{rgb}{0.150000,0.150000,0.150000}%
\pgfsetstrokecolor{textcolor}%
\pgfsetfillcolor{textcolor}%
\pgftext[x=0.522684in,y=0.305556in,,top]{\color{textcolor}\rmfamily\fontsize{8.800000}{10.560000}\selectfont \(\displaystyle {0.5}\)}%
\end{pgfscope}%
\begin{pgfscope}%
\pgfpathrectangle{\pgfqpoint{0.522684in}{0.420833in}}{\pgfqpoint{3.007237in}{2.004431in}}%
\pgfusepath{clip}%
\pgfsetroundcap%
\pgfsetroundjoin%
\pgfsetlinewidth{0.803000pt}%
\definecolor{currentstroke}{rgb}{0.800000,0.800000,0.800000}%
\pgfsetstrokecolor{currentstroke}%
\pgfsetdash{}{0pt}%
\pgfpathmoveto{\pgfqpoint{1.124131in}{0.420833in}}%
\pgfpathlineto{\pgfqpoint{1.124131in}{2.425264in}}%
\pgfusepath{stroke}%
\end{pgfscope}%
\begin{pgfscope}%
\definecolor{textcolor}{rgb}{0.150000,0.150000,0.150000}%
\pgfsetstrokecolor{textcolor}%
\pgfsetfillcolor{textcolor}%
\pgftext[x=1.124131in,y=0.305556in,,top]{\color{textcolor}\rmfamily\fontsize{8.800000}{10.560000}\selectfont \(\displaystyle {0.6}\)}%
\end{pgfscope}%
\begin{pgfscope}%
\pgfpathrectangle{\pgfqpoint{0.522684in}{0.420833in}}{\pgfqpoint{3.007237in}{2.004431in}}%
\pgfusepath{clip}%
\pgfsetroundcap%
\pgfsetroundjoin%
\pgfsetlinewidth{0.803000pt}%
\definecolor{currentstroke}{rgb}{0.800000,0.800000,0.800000}%
\pgfsetstrokecolor{currentstroke}%
\pgfsetdash{}{0pt}%
\pgfpathmoveto{\pgfqpoint{1.725579in}{0.420833in}}%
\pgfpathlineto{\pgfqpoint{1.725579in}{2.425264in}}%
\pgfusepath{stroke}%
\end{pgfscope}%
\begin{pgfscope}%
\definecolor{textcolor}{rgb}{0.150000,0.150000,0.150000}%
\pgfsetstrokecolor{textcolor}%
\pgfsetfillcolor{textcolor}%
\pgftext[x=1.725579in,y=0.305556in,,top]{\color{textcolor}\rmfamily\fontsize{8.800000}{10.560000}\selectfont \(\displaystyle {0.7}\)}%
\end{pgfscope}%
\begin{pgfscope}%
\pgfpathrectangle{\pgfqpoint{0.522684in}{0.420833in}}{\pgfqpoint{3.007237in}{2.004431in}}%
\pgfusepath{clip}%
\pgfsetroundcap%
\pgfsetroundjoin%
\pgfsetlinewidth{0.803000pt}%
\definecolor{currentstroke}{rgb}{0.800000,0.800000,0.800000}%
\pgfsetstrokecolor{currentstroke}%
\pgfsetdash{}{0pt}%
\pgfpathmoveto{\pgfqpoint{2.327026in}{0.420833in}}%
\pgfpathlineto{\pgfqpoint{2.327026in}{2.425264in}}%
\pgfusepath{stroke}%
\end{pgfscope}%
\begin{pgfscope}%
\definecolor{textcolor}{rgb}{0.150000,0.150000,0.150000}%
\pgfsetstrokecolor{textcolor}%
\pgfsetfillcolor{textcolor}%
\pgftext[x=2.327026in,y=0.305556in,,top]{\color{textcolor}\rmfamily\fontsize{8.800000}{10.560000}\selectfont \(\displaystyle {0.8}\)}%
\end{pgfscope}%
\begin{pgfscope}%
\pgfpathrectangle{\pgfqpoint{0.522684in}{0.420833in}}{\pgfqpoint{3.007237in}{2.004431in}}%
\pgfusepath{clip}%
\pgfsetroundcap%
\pgfsetroundjoin%
\pgfsetlinewidth{0.803000pt}%
\definecolor{currentstroke}{rgb}{0.800000,0.800000,0.800000}%
\pgfsetstrokecolor{currentstroke}%
\pgfsetdash{}{0pt}%
\pgfpathmoveto{\pgfqpoint{2.928474in}{0.420833in}}%
\pgfpathlineto{\pgfqpoint{2.928474in}{2.425264in}}%
\pgfusepath{stroke}%
\end{pgfscope}%
\begin{pgfscope}%
\definecolor{textcolor}{rgb}{0.150000,0.150000,0.150000}%
\pgfsetstrokecolor{textcolor}%
\pgfsetfillcolor{textcolor}%
\pgftext[x=2.928474in,y=0.305556in,,top]{\color{textcolor}\rmfamily\fontsize{8.800000}{10.560000}\selectfont \(\displaystyle {0.9}\)}%
\end{pgfscope}%
\begin{pgfscope}%
\pgfpathrectangle{\pgfqpoint{0.522684in}{0.420833in}}{\pgfqpoint{3.007237in}{2.004431in}}%
\pgfusepath{clip}%
\pgfsetroundcap%
\pgfsetroundjoin%
\pgfsetlinewidth{0.803000pt}%
\definecolor{currentstroke}{rgb}{0.800000,0.800000,0.800000}%
\pgfsetstrokecolor{currentstroke}%
\pgfsetdash{}{0pt}%
\pgfpathmoveto{\pgfqpoint{3.529921in}{0.420833in}}%
\pgfpathlineto{\pgfqpoint{3.529921in}{2.425264in}}%
\pgfusepath{stroke}%
\end{pgfscope}%
\begin{pgfscope}%
\definecolor{textcolor}{rgb}{0.150000,0.150000,0.150000}%
\pgfsetstrokecolor{textcolor}%
\pgfsetfillcolor{textcolor}%
\pgftext[x=3.529921in,y=0.305556in,,top]{\color{textcolor}\rmfamily\fontsize{8.800000}{10.560000}\selectfont \(\displaystyle {1.0}\)}%
\end{pgfscope}%
\begin{pgfscope}%
\definecolor{textcolor}{rgb}{0.150000,0.150000,0.150000}%
\pgfsetstrokecolor{textcolor}%
\pgfsetfillcolor{textcolor}%
\pgftext[x=2.026302in,y=0.138889in,,top]{\color{textcolor}\rmfamily\fontsize{9.600000}{11.520000}\selectfont \(\displaystyle g^*\) (gradient clipping bound)}%
\end{pgfscope}%
\begin{pgfscope}%
\pgfpathrectangle{\pgfqpoint{0.522684in}{0.420833in}}{\pgfqpoint{3.007237in}{2.004431in}}%
\pgfusepath{clip}%
\pgfsetroundcap%
\pgfsetroundjoin%
\pgfsetlinewidth{0.803000pt}%
\definecolor{currentstroke}{rgb}{0.800000,0.800000,0.800000}%
\pgfsetstrokecolor{currentstroke}%
\pgfsetdash{}{0pt}%
\pgfpathmoveto{\pgfqpoint{0.522684in}{0.420833in}}%
\pgfpathlineto{\pgfqpoint{3.529921in}{0.420833in}}%
\pgfusepath{stroke}%
\end{pgfscope}%
\begin{pgfscope}%
\definecolor{textcolor}{rgb}{0.150000,0.150000,0.150000}%
\pgfsetstrokecolor{textcolor}%
\pgfsetfillcolor{textcolor}%
\pgftext[x=0.179012in, y=0.377431in, left, base]{\color{textcolor}\rmfamily\fontsize{8.800000}{10.560000}\selectfont \(\displaystyle {0.75}\)}%
\end{pgfscope}%
\begin{pgfscope}%
\pgfpathrectangle{\pgfqpoint{0.522684in}{0.420833in}}{\pgfqpoint{3.007237in}{2.004431in}}%
\pgfusepath{clip}%
\pgfsetroundcap%
\pgfsetroundjoin%
\pgfsetlinewidth{0.803000pt}%
\definecolor{currentstroke}{rgb}{0.800000,0.800000,0.800000}%
\pgfsetstrokecolor{currentstroke}%
\pgfsetdash{}{0pt}%
\pgfpathmoveto{\pgfqpoint{0.522684in}{0.821719in}}%
\pgfpathlineto{\pgfqpoint{3.529921in}{0.821719in}}%
\pgfusepath{stroke}%
\end{pgfscope}%
\begin{pgfscope}%
\definecolor{textcolor}{rgb}{0.150000,0.150000,0.150000}%
\pgfsetstrokecolor{textcolor}%
\pgfsetfillcolor{textcolor}%
\pgftext[x=0.179012in, y=0.778317in, left, base]{\color{textcolor}\rmfamily\fontsize{8.800000}{10.560000}\selectfont \(\displaystyle {0.80}\)}%
\end{pgfscope}%
\begin{pgfscope}%
\pgfpathrectangle{\pgfqpoint{0.522684in}{0.420833in}}{\pgfqpoint{3.007237in}{2.004431in}}%
\pgfusepath{clip}%
\pgfsetroundcap%
\pgfsetroundjoin%
\pgfsetlinewidth{0.803000pt}%
\definecolor{currentstroke}{rgb}{0.800000,0.800000,0.800000}%
\pgfsetstrokecolor{currentstroke}%
\pgfsetdash{}{0pt}%
\pgfpathmoveto{\pgfqpoint{0.522684in}{1.222606in}}%
\pgfpathlineto{\pgfqpoint{3.529921in}{1.222606in}}%
\pgfusepath{stroke}%
\end{pgfscope}%
\begin{pgfscope}%
\definecolor{textcolor}{rgb}{0.150000,0.150000,0.150000}%
\pgfsetstrokecolor{textcolor}%
\pgfsetfillcolor{textcolor}%
\pgftext[x=0.179012in, y=1.179203in, left, base]{\color{textcolor}\rmfamily\fontsize{8.800000}{10.560000}\selectfont \(\displaystyle {0.85}\)}%
\end{pgfscope}%
\begin{pgfscope}%
\pgfpathrectangle{\pgfqpoint{0.522684in}{0.420833in}}{\pgfqpoint{3.007237in}{2.004431in}}%
\pgfusepath{clip}%
\pgfsetroundcap%
\pgfsetroundjoin%
\pgfsetlinewidth{0.803000pt}%
\definecolor{currentstroke}{rgb}{0.800000,0.800000,0.800000}%
\pgfsetstrokecolor{currentstroke}%
\pgfsetdash{}{0pt}%
\pgfpathmoveto{\pgfqpoint{0.522684in}{1.623492in}}%
\pgfpathlineto{\pgfqpoint{3.529921in}{1.623492in}}%
\pgfusepath{stroke}%
\end{pgfscope}%
\begin{pgfscope}%
\definecolor{textcolor}{rgb}{0.150000,0.150000,0.150000}%
\pgfsetstrokecolor{textcolor}%
\pgfsetfillcolor{textcolor}%
\pgftext[x=0.179012in, y=1.580089in, left, base]{\color{textcolor}\rmfamily\fontsize{8.800000}{10.560000}\selectfont \(\displaystyle {0.90}\)}%
\end{pgfscope}%
\begin{pgfscope}%
\pgfpathrectangle{\pgfqpoint{0.522684in}{0.420833in}}{\pgfqpoint{3.007237in}{2.004431in}}%
\pgfusepath{clip}%
\pgfsetroundcap%
\pgfsetroundjoin%
\pgfsetlinewidth{0.803000pt}%
\definecolor{currentstroke}{rgb}{0.800000,0.800000,0.800000}%
\pgfsetstrokecolor{currentstroke}%
\pgfsetdash{}{0pt}%
\pgfpathmoveto{\pgfqpoint{0.522684in}{2.024378in}}%
\pgfpathlineto{\pgfqpoint{3.529921in}{2.024378in}}%
\pgfusepath{stroke}%
\end{pgfscope}%
\begin{pgfscope}%
\definecolor{textcolor}{rgb}{0.150000,0.150000,0.150000}%
\pgfsetstrokecolor{textcolor}%
\pgfsetfillcolor{textcolor}%
\pgftext[x=0.179012in, y=1.980975in, left, base]{\color{textcolor}\rmfamily\fontsize{8.800000}{10.560000}\selectfont \(\displaystyle {0.95}\)}%
\end{pgfscope}%
\begin{pgfscope}%
\pgfpathrectangle{\pgfqpoint{0.522684in}{0.420833in}}{\pgfqpoint{3.007237in}{2.004431in}}%
\pgfusepath{clip}%
\pgfsetroundcap%
\pgfsetroundjoin%
\pgfsetlinewidth{0.803000pt}%
\definecolor{currentstroke}{rgb}{0.800000,0.800000,0.800000}%
\pgfsetstrokecolor{currentstroke}%
\pgfsetdash{}{0pt}%
\pgfpathmoveto{\pgfqpoint{0.522684in}{2.425264in}}%
\pgfpathlineto{\pgfqpoint{3.529921in}{2.425264in}}%
\pgfusepath{stroke}%
\end{pgfscope}%
\begin{pgfscope}%
\definecolor{textcolor}{rgb}{0.150000,0.150000,0.150000}%
\pgfsetstrokecolor{textcolor}%
\pgfsetfillcolor{textcolor}%
\pgftext[x=0.179012in, y=2.381861in, left, base]{\color{textcolor}\rmfamily\fontsize{8.800000}{10.560000}\selectfont \(\displaystyle {1.00}\)}%
\end{pgfscope}%
\begin{pgfscope}%
\definecolor{textcolor}{rgb}{0.150000,0.150000,0.150000}%
\pgfsetstrokecolor{textcolor}%
\pgfsetfillcolor{textcolor}%
\pgftext[x=0.123457in,y=1.423049in,,bottom,rotate=90.000000]{\color{textcolor}\rmfamily\fontsize{9.600000}{11.520000}\selectfont Mean test AUC}%
\end{pgfscope}%
\begin{pgfscope}%
\pgfpathrectangle{\pgfqpoint{0.522684in}{0.420833in}}{\pgfqpoint{3.007237in}{2.004431in}}%
\pgfusepath{clip}%
\pgfsetbuttcap%
\pgfsetroundjoin%
\definecolor{currentfill}{rgb}{1.000000,0.498039,0.054902}%
\pgfsetfillcolor{currentfill}%
\pgfsetfillopacity{0.100000}%
\pgfsetlinewidth{0.803000pt}%
\definecolor{currentstroke}{rgb}{1.000000,0.498039,0.054902}%
\pgfsetstrokecolor{currentstroke}%
\pgfsetstrokeopacity{0.100000}%
\pgfsetdash{}{0pt}%
\pgfsys@defobject{currentmarker}{\pgfqpoint{0.522684in}{0.781631in}}{\pgfqpoint{3.529921in}{1.046216in}}{%
\pgfpathmoveto{\pgfqpoint{0.522684in}{1.046216in}}%
\pgfpathlineto{\pgfqpoint{0.522684in}{1.030180in}}%
\pgfpathlineto{\pgfqpoint{1.124131in}{1.014145in}}%
\pgfpathlineto{\pgfqpoint{1.725579in}{0.982074in}}%
\pgfpathlineto{\pgfqpoint{2.327026in}{0.885861in}}%
\pgfpathlineto{\pgfqpoint{2.928474in}{0.837755in}}%
\pgfpathlineto{\pgfqpoint{3.529921in}{0.781631in}}%
\pgfpathlineto{\pgfqpoint{3.529921in}{0.797666in}}%
\pgfpathlineto{\pgfqpoint{3.529921in}{0.797666in}}%
\pgfpathlineto{\pgfqpoint{2.928474in}{0.853790in}}%
\pgfpathlineto{\pgfqpoint{2.327026in}{0.917932in}}%
\pgfpathlineto{\pgfqpoint{1.725579in}{0.998109in}}%
\pgfpathlineto{\pgfqpoint{1.124131in}{1.030180in}}%
\pgfpathlineto{\pgfqpoint{0.522684in}{1.046216in}}%
\pgfpathlineto{\pgfqpoint{0.522684in}{1.046216in}}%
\pgfpathclose%
\pgfusepath{stroke,fill}%
}%
\begin{pgfscope}%
\pgfsys@transformshift{0.000000in}{0.000000in}%
\pgfsys@useobject{currentmarker}{}%
\end{pgfscope}%
\end{pgfscope}%
\begin{pgfscope}%
\pgfpathrectangle{\pgfqpoint{0.522684in}{0.420833in}}{\pgfqpoint{3.007237in}{2.004431in}}%
\pgfusepath{clip}%
\pgfsetbuttcap%
\pgfsetroundjoin%
\definecolor{currentfill}{rgb}{1.000000,0.498039,0.054902}%
\pgfsetfillcolor{currentfill}%
\pgfsetfillopacity{0.100000}%
\pgfsetlinewidth{0.803000pt}%
\definecolor{currentstroke}{rgb}{1.000000,0.498039,0.054902}%
\pgfsetstrokecolor{currentstroke}%
\pgfsetstrokeopacity{0.100000}%
\pgfsetdash{}{0pt}%
\pgfsys@defobject{currentmarker}{\pgfqpoint{0.522684in}{0.893879in}}{\pgfqpoint{3.529921in}{1.062251in}}{%
\pgfpathmoveto{\pgfqpoint{0.522684in}{1.046216in}}%
\pgfpathlineto{\pgfqpoint{0.522684in}{1.014145in}}%
\pgfpathlineto{\pgfqpoint{1.124131in}{1.030180in}}%
\pgfpathlineto{\pgfqpoint{1.725579in}{0.990092in}}%
\pgfpathlineto{\pgfqpoint{2.327026in}{0.974056in}}%
\pgfpathlineto{\pgfqpoint{2.928474in}{0.929959in}}%
\pgfpathlineto{\pgfqpoint{3.529921in}{0.893879in}}%
\pgfpathlineto{\pgfqpoint{3.529921in}{0.909914in}}%
\pgfpathlineto{\pgfqpoint{3.529921in}{0.909914in}}%
\pgfpathlineto{\pgfqpoint{2.928474in}{0.954012in}}%
\pgfpathlineto{\pgfqpoint{2.327026in}{0.990092in}}%
\pgfpathlineto{\pgfqpoint{1.725579in}{1.006127in}}%
\pgfpathlineto{\pgfqpoint{1.124131in}{1.062251in}}%
\pgfpathlineto{\pgfqpoint{0.522684in}{1.046216in}}%
\pgfpathlineto{\pgfqpoint{0.522684in}{1.046216in}}%
\pgfpathclose%
\pgfusepath{stroke,fill}%
}%
\begin{pgfscope}%
\pgfsys@transformshift{0.000000in}{0.000000in}%
\pgfsys@useobject{currentmarker}{}%
\end{pgfscope}%
\end{pgfscope}%
\begin{pgfscope}%
\pgfpathrectangle{\pgfqpoint{0.522684in}{0.420833in}}{\pgfqpoint{3.007237in}{2.004431in}}%
\pgfusepath{clip}%
\pgfsetroundcap%
\pgfsetroundjoin%
\pgfsetlinewidth{1.204500pt}%
\definecolor{currentstroke}{rgb}{1.000000,0.498039,0.054902}%
\pgfsetstrokecolor{currentstroke}%
\pgfsetdash{}{0pt}%
\pgfpathmoveto{\pgfqpoint{0.522684in}{1.038198in}}%
\pgfpathlineto{\pgfqpoint{1.124131in}{1.022162in}}%
\pgfpathlineto{\pgfqpoint{1.725579in}{0.990092in}}%
\pgfpathlineto{\pgfqpoint{2.327026in}{0.901897in}}%
\pgfpathlineto{\pgfqpoint{2.928474in}{0.845773in}}%
\pgfpathlineto{\pgfqpoint{3.529921in}{0.789649in}}%
\pgfusepath{stroke}%
\end{pgfscope}%
\begin{pgfscope}%
\pgfpathrectangle{\pgfqpoint{0.522684in}{0.420833in}}{\pgfqpoint{3.007237in}{2.004431in}}%
\pgfusepath{clip}%
\pgfsetbuttcap%
\pgfsetroundjoin%
\pgfsetlinewidth{1.204500pt}%
\definecolor{currentstroke}{rgb}{1.000000,0.498039,0.054902}%
\pgfsetstrokecolor{currentstroke}%
\pgfsetdash{{4.440000pt}{1.920000pt}}{0.000000pt}%
\pgfpathmoveto{\pgfqpoint{0.522684in}{1.030180in}}%
\pgfpathlineto{\pgfqpoint{1.124131in}{1.046216in}}%
\pgfpathlineto{\pgfqpoint{1.725579in}{0.998109in}}%
\pgfpathlineto{\pgfqpoint{2.327026in}{0.982074in}}%
\pgfpathlineto{\pgfqpoint{2.928474in}{0.941985in}}%
\pgfpathlineto{\pgfqpoint{3.529921in}{0.901897in}}%
\pgfusepath{stroke}%
\end{pgfscope}%
\begin{pgfscope}%
\pgfsetrectcap%
\pgfsetmiterjoin%
\pgfsetlinewidth{1.003750pt}%
\definecolor{currentstroke}{rgb}{0.800000,0.800000,0.800000}%
\pgfsetstrokecolor{currentstroke}%
\pgfsetdash{}{0pt}%
\pgfpathmoveto{\pgfqpoint{0.522684in}{0.420833in}}%
\pgfpathlineto{\pgfqpoint{0.522684in}{2.425264in}}%
\pgfusepath{stroke}%
\end{pgfscope}%
\begin{pgfscope}%
\pgfsetrectcap%
\pgfsetmiterjoin%
\pgfsetlinewidth{1.003750pt}%
\definecolor{currentstroke}{rgb}{0.800000,0.800000,0.800000}%
\pgfsetstrokecolor{currentstroke}%
\pgfsetdash{}{0pt}%
\pgfpathmoveto{\pgfqpoint{0.522684in}{0.420833in}}%
\pgfpathlineto{\pgfqpoint{3.529921in}{0.420833in}}%
\pgfusepath{stroke}%
\end{pgfscope}%
\begin{pgfscope}%
\pgfsetroundcap%
\pgfsetroundjoin%
\definecolor{currentfill}{rgb}{0.862745,0.862745,0.862745}%
\pgfsetfillcolor{currentfill}%
\pgfsetlinewidth{0.803000pt}%
\definecolor{currentstroke}{rgb}{1.000000,1.000000,1.000000}%
\pgfsetstrokecolor{currentstroke}%
\pgfsetdash{}{0pt}%
\pgfpathmoveto{\pgfqpoint{1.089409in}{1.048613in}}%
\pgfpathquadraticcurveto{\pgfqpoint{1.089409in}{1.275544in}}{\pgfqpoint{1.089409in}{1.502475in}}%
\pgfpathlineto{\pgfqpoint{1.040798in}{1.502475in}}%
\pgfpathquadraticcurveto{\pgfqpoint{1.082464in}{1.585820in}}{\pgfqpoint{1.124131in}{1.669165in}}%
\pgfpathquadraticcurveto{\pgfqpoint{1.165798in}{1.585820in}}{\pgfqpoint{1.207464in}{1.502475in}}%
\pgfpathlineto{\pgfqpoint{1.158853in}{1.502475in}}%
\pgfpathquadraticcurveto{\pgfqpoint{1.158853in}{1.275544in}}{\pgfqpoint{1.158853in}{1.048613in}}%
\pgfpathlineto{\pgfqpoint{1.089409in}{1.048613in}}%
\pgfpathlineto{\pgfqpoint{1.089409in}{1.048613in}}%
\pgfpathclose%
\pgfusepath{stroke,fill}%
\end{pgfscope}%
\begin{pgfscope}%
\definecolor{textcolor}{rgb}{0.862745,0.862745,0.862745}%
\pgfsetstrokecolor{textcolor}%
\pgfsetfillcolor{textcolor}%
\pgftext[x=1.424855in,y=1.543314in,left,]{\color{textcolor}\rmfamily\fontsize{12.000000}{14.400000}\selectfont better}%
\end{pgfscope}%
\begin{pgfscope}%
\pgfsetbuttcap%
\pgfsetmiterjoin%
\definecolor{currentfill}{rgb}{1.000000,1.000000,1.000000}%
\pgfsetfillcolor{currentfill}%
\pgfsetfillopacity{0.800000}%
\pgfsetlinewidth{0.803000pt}%
\definecolor{currentstroke}{rgb}{0.800000,0.800000,0.800000}%
\pgfsetstrokecolor{currentstroke}%
\pgfsetstrokeopacity{0.800000}%
\pgfsetdash{}{0pt}%
\pgfpathmoveto{\pgfqpoint{1.591537in}{1.955264in}}%
\pgfpathlineto{\pgfqpoint{3.444365in}{1.955264in}}%
\pgfpathquadraticcurveto{\pgfqpoint{3.468810in}{1.955264in}}{\pgfqpoint{3.468810in}{1.979708in}}%
\pgfpathlineto{\pgfqpoint{3.468810in}{2.339708in}}%
\pgfpathquadraticcurveto{\pgfqpoint{3.468810in}{2.364153in}}{\pgfqpoint{3.444365in}{2.364153in}}%
\pgfpathlineto{\pgfqpoint{1.591537in}{2.364153in}}%
\pgfpathquadraticcurveto{\pgfqpoint{1.567092in}{2.364153in}}{\pgfqpoint{1.567092in}{2.339708in}}%
\pgfpathlineto{\pgfqpoint{1.567092in}{1.979708in}}%
\pgfpathquadraticcurveto{\pgfqpoint{1.567092in}{1.955264in}}{\pgfqpoint{1.591537in}{1.955264in}}%
\pgfpathlineto{\pgfqpoint{1.591537in}{1.955264in}}%
\pgfpathclose%
\pgfusepath{stroke,fill}%
\end{pgfscope}%
\begin{pgfscope}%
\pgfsetroundcap%
\pgfsetroundjoin%
\pgfsetlinewidth{1.204500pt}%
\definecolor{currentstroke}{rgb}{1.000000,0.498039,0.054902}%
\pgfsetstrokecolor{currentstroke}%
\pgfsetdash{}{0pt}%
\pgfpathmoveto{\pgfqpoint{1.615981in}{2.264292in}}%
\pgfpathlineto{\pgfqpoint{1.738203in}{2.264292in}}%
\pgfpathlineto{\pgfqpoint{1.860425in}{2.264292in}}%
\pgfusepath{stroke}%
\end{pgfscope}%
\begin{pgfscope}%
\definecolor{textcolor}{rgb}{0.150000,0.150000,0.150000}%
\pgfsetstrokecolor{textcolor}%
\pgfsetfillcolor{textcolor}%
\pgftext[x=1.958203in,y=2.221514in,left,base]{\color{textcolor}\rmfamily\fontsize{8.800000}{10.560000}\selectfont S-BDT (no Rényi filter)}%
\end{pgfscope}%
\begin{pgfscope}%
\pgfsetbuttcap%
\pgfsetroundjoin%
\pgfsetlinewidth{1.204500pt}%
\definecolor{currentstroke}{rgb}{1.000000,0.498039,0.054902}%
\pgfsetstrokecolor{currentstroke}%
\pgfsetdash{{4.440000pt}{1.920000pt}}{0.000000pt}%
\pgfpathmoveto{\pgfqpoint{1.615981in}{2.078181in}}%
\pgfpathlineto{\pgfqpoint{1.738203in}{2.078181in}}%
\pgfpathlineto{\pgfqpoint{1.860425in}{2.078181in}}%
\pgfusepath{stroke}%
\end{pgfscope}%
\begin{pgfscope}%
\definecolor{textcolor}{rgb}{0.150000,0.150000,0.150000}%
\pgfsetstrokecolor{textcolor}%
\pgfsetfillcolor{textcolor}%
\pgftext[x=1.958203in,y=2.035403in,left,base]{\color{textcolor}\rmfamily\fontsize{8.800000}{10.560000}\selectfont S-BDT (with Rényi filter)}%
\end{pgfscope}%
\end{pgfpicture}%
\makeatother%
\endgroup%

%% file: appendix_proofs.tex
\subsection{Useful RDP Properties}

\begin{lemma}[From RDP to individual RDP]
    Let $M$ be any mechanism satisfying $(\alpha, \rho(\alpha))$-Rényi differential privacy. Then $M$ satisfies $(\alpha, \rho(\alpha))$-individual Rényi differential privacy for data point $x$.
    \label{lem:rdp_implies_individual_rdp}
\end{lemma}
\begin{proof}
    Let $M$ be as defined in the lemma's statement. 
    Let $X,X'$ be two neighboring datasets, then
    \begin{align*}
        D_\alpha(M(X)||M(X')) \leq \rho(\alpha)
    \end{align*}
    Now, let $D\sim_{d_i}D'$ be two neighboring datasets differing in $d_i$. Since $\rho(\alpha)$ is a upper bound on the Rényi divergence of $M(X),M(X')$ for arbitrary inputs, it also applies for the specific inputs $D,D'$:
    \begin{align*}
        D_\alpha(M(D)||M(D')) \leq \rho(\alpha)
    \end{align*}
    Thus, $M$ is also $(\alpha,\rho(\alpha))$ individual RDP.
\end{proof}

\begin{corollary}[Adaptive sequential Composition for individual RDP]
    Let $M$ be a sequence of adaptively chosen mechanisms
    $M_i: \Pi_{j=1}^{i-1} R_j \times \mathcal{X} \mapsto \mathcal{R}_i$ ($i=1,2,...,k)$, each providing $(\alpha, \rho_i)$-individual Rényi differential privacy for data point $x$. Then for any $\alpha$ $M$ is $(\alpha, \sum_i^k \rho_i(\alpha))$-individual Rényi differentially private for data point $x$.
    \label{thm:individual_rdp_sequential_composition}
\end{corollary}

\begin{proof}
    This corollary follows directly from ~\Cref{thmrdpsequentialcomposition}. The proof of this theorem is parametric in a set of neighboring datasets. By choosing the set of neighboring datasets as all $(X,X\cup\set{x})$, the statement follows.
\end{proof}

\subsection{Main Theorem}

We state the full \Cref{thmdpinitscore} and the full proof.

\begin{reptheorem}{thm:maintheorem}[Main theorem]
    \Cref{algtrainsgbdt} (\texttt{TrainSBDT}) is $(\alpha, \rho(\alpha)$)-Rényi DP.
\end{reptheorem}
\begin{proof}
    \newcommand{\train}{\texttt{TrainSBDT}}
    \newcommand{\single}[1]{\texttt{TrainSingleTree}} 
    \newcommand{\filter}{\texttt{Filter}_\alpha^{\rho(\alpha)}}
    \newcommand{\poisson}{\texttt{Poisson}}
    \newcommand{\init}{\texttt{DPInitScore}}
    \newcommand{\grad}{\texttt{ComputeGradients}}
    \newcommand{\predict}{\texttt{Predict}}
    \newcommand{\M}[1]{M_{#1}}
    \renewcommand{\o}[1]{o_{#1}}
    \newcommand{\ou}[1]{{o^{(#1)}}}
    We prove an $(\alpha, \rho(\alpha))$-RDP bound for the learning algorithm \train. \\
    Let $M_0 := \init$, $M_i := \single{k} \circ \filter$ $\forall i>0$. $\filter$ denotes the individual Rényi filter part of our algorithm (lines \ref{algtrainsgbdtfilterstart} to \ref{algtrainsgbdtfilterend} in ~\Cref{algtrainsgbdt}): The filter computes individual RDP privacy losses for $\single{i}$ for every data point and then filters out data points that have exceeded the individual RDP bound of $\rho(\alpha)$.

    Let $T^{(0..k)}(X) := \M{k}(X, \M{k-1}(X, ...(X, \M{0}(X))))$ be the nested applications of mechanism $M_i$, comprising the learning algorithm $\train$, where $T^{(v..w)}(X) := \M{w}(X, \M{k-1}(X, ...(X, \M{v}(X))))$. The output of the training algorithm $\train$, which we call an observation, is a sequence of trees $\o{0}, ..., \o{k}$, consisting of an initial score $\o{0}$ and trees $\o{i}$ output by the application of $\M{i}$. We write $\mathbf{o} := \ou{0..k}=(\o{0}, \o{1}, ..., \o{k})$, where $\ou{v..w} = (\o{v}, ..., \o{w})$.
    \begin{align*}
        &\textstyle D_\alpha(\train(X)||\train(X')) \\
        &\textstyle= \frac{1}{\alpha-1} \log \int_\mathbf{o} \frac{\Pr[\train(X)=\mathbf{o}]^\alpha}{\Pr[\train(X')=\mathbf{o}]^{\alpha-1}} d \mathbf{o} \\
        &\textstyle= \frac{1}{\alpha-1} \log \int_\ou{0..k} \frac{\Pr[T^{(0..k)}(X)=\ou{0..k}]^\alpha}{\Pr[T^{(0..k)}(X')=\ou{0..k}]^{\alpha-1}} d \ou{0..k} \\
        \intertext{Applying the RDP sequential composition bound of  ~\Cref{thmrdpsequentialcomposition} we get}
        &\textstyle\leq \frac{1}{\alpha-1} \log \int_\ou{1..k} \frac{\Pr[T^{(1..k)}(X)=\ou{1..k}]^\alpha}{\Pr[T^{(1..k)}(X')=\ou{1..k}]^{\alpha-1}} d \ou{1..k} \\
        &\textstyle+ \frac{1}{\alpha-1} \log \int_{\o{0}} \frac{\Pr[\M{0}(X)=\o{0}]^\alpha}{\Pr[\M{0}(X')=\o{0}]^{\alpha-1}} d \o{0} \\
        &\textstyle= \frac{1}{\alpha-1} \log \int_\ou{1..k} \frac{\Pr[T^{(1..k)}(X)=\ou{1..k}]^\alpha}{\Pr[T^{(1..k)}(X')=\ou{1..k}]^{\alpha-1}} d \ou{1..k} \\
        &\textstyle+ D_\alpha(\init(X)||\init(X')) \\
        \intertext{Let $f(\varepsilon) = \log \Big\{ \frac{\alpha}{2\alpha -1} \exp\Big((\alpha-1)\cdot \varepsilon\Big) + \frac{\alpha -1}{2\alpha -1} \exp\Big(-\alpha \cdot \varepsilon\Big) \Big\}$, applying the RDP bound for \texttt{DPInitScore} of ~\Cref{thmdpinitscore} we get}
        &\textstyle\leq \frac{1}{\alpha-1} \log \int_\ou{1..k} \frac{\Pr[T^{(1..k)}(X)=\ou{1..k}]^\alpha}{\Pr[T^{(1..k)}(X')=\ou{1..k}]^{\alpha-1}} d \ou{1..k} \\
        &\textstyle+ \frac{1}{\alpha-1} \Big(f(\varepsilon_\text{init}) + f(\varepsilon_\text{ds})\Big) \\[-4ex]
    \end{align*}

    By definition of $T^{(1..k)}$, $T^{(1..k)}$ is a sequence of $M_i$ ($i=1,2,...,k$) where each $M_i$ consists of two operations: (1) The filter operation $\filter$ that applies the privacy filter and filters out data points that have exceeded the individual RDP bound $\rho(\alpha)$. (2) The mechanism $\single{i}$ that trains a tree on the filtered dataset.
    
    In ~\Cref{thm:train_single_tree_irdp_main_body} we bound the individual RDP privacy losses for $\single{i}$, so our filter is correct, i.e. it will always filter out data points that have exceeded the individual RDP bound $\rho(\alpha)$.
    By ~\Cref{thm:adaptive_composition_individual_privacy} the sequence of mechanisms $T^{(1..k)}$ then satisfies the following RDP bound
    \begin{align*}
        \textstyle\leq \rho(\alpha) 
        &\textstyle+ \frac{1}{\alpha-1} \Big(f(\varepsilon_\text{init}) + f(\varepsilon_\text{ds})\Big)\\[-4ex]
    \end{align*}
\end{proof}

\subsection{RDP Proofs of each \dpgbdt{} Component}
\subsubsection{DPInitScore}

We recall \Cref{thmdpinitscore} and state the full proof.

\begin{reptheorem}{thmdpinitscore}
    \input{theorems/initscore_rdp}
\end{reptheorem}

\begin{proof}    
    Let $X\sim X'$ be two neighboring datasets. Assume without loss of generality that $X$ and $X'$ differ in data point $(x',y')$. We first analyze the sensitivities of \texttt{DPInitialScore}'s subfunctions.

    Let $g_\text{ds}$ denote the subfunction of \texttt{DPInitialScore} without noise, that computes the dataset size. The sensitivity of $g_\text{ds}$ is $\Delta_{g_\text{ds}} = 1$ because we investigate unbounded DP, and adding or removing a single element from the dataset can change its size by at most 1.
    
    Let $g_\text{sum}$ denote the subfunction of \texttt{DPInitialScore} without noise, that computes the sum of clipped labels. All labels are clipped to have length at most $m^*$ (line \ref{algdpinitscoreclipmean} in ~\Cref{algdpinitscore}) and we denote the clipped labels for data point $(x_i,y_i) \in X$ as $\overline{y_i}$. The sensitivity of $g_\text{sum}$ is $\Delta_{g_\text{sum}} $
    \begin{align*}
        &\textstyle= \max_{X\sim X'} \big|f_\text{sum}(X) - f_\text{sum}(X')\big| \\
        &= \max_{X\sim X'} \Bigg|\sum_{(x_i,y_i) \in X} \overline{y_i} - \sum_{(x_i,y_i) \in X'} \overline{y_i}\Bigg| \\
        &\textstyle= \max_{X\sim X'} \Bigg|y_l'\Bigg| \\
        &\textstyle\leq m^*
    \end{align*}

    We analyze the Rényi divergence between two Laplace distributions shifted in $\Delta_f$ with noise scale $\lambda := \frac{\Delta_f}{\varepsilon}$ for some privacy budget $\varepsilon$, to compute the RDP bound. 
    Let $p \sim \mathcal{L}(0, \lambda), q \sim \mathcal{L}(\Delta_f, \lambda)$. $D_\alpha(p||q) = $
    \begin{align*}
        &\textstyle= \frac{1}{\alpha-1} \log \int_{-\infty}^{\infty} \frac{p(x)^\alpha}{q(x)^{\alpha-1}} dx \\
        &\textstyle= \frac{1}{\alpha-1} \log \int_{-\infty}^{\infty} \frac{(\frac{1}{2\lambda} \exp(-|x|/\lambda))^\alpha}{(\frac{1}{2\lambda}\exp(-|x-\Delta_f|/\lambda))^{\alpha-1}} dx \\
        &\textstyle= \frac{1}{\alpha-1} \log \frac{1}{2\lambda} \int_{-\infty}^{\infty} \exp(-\frac{\alpha|x|}{\lambda} + \frac{(\alpha-1)|x-\Delta_f|}{\lambda}) dx \\
        \intertext{We split the integral into three parts}
        &\textstyle\frac{1}{2\lambda}\int_{-\infty}^{\infty} \exp(-\alpha|x|/\lambda + (\alpha-1)|x-\Delta_f|/\lambda) dx \\
        &\textstyle= \frac{1}{2\lambda} \int_{-\infty}^{0} \exp(\alpha x/\lambda + (1-\alpha)(x-\Delta_f)/\lambda) dx \\
        &\textstyle+ \frac{1}{2\lambda} \int_{0}^{\Delta_f} \exp(-\alpha x/\lambda + (1-\alpha)(x-\Delta_f)/\lambda) dx \\
        &\textstyle+ \frac{1}{2\lambda} \int_{\Delta_f}^{\infty} \exp(-\alpha x/\lambda - (1-\alpha)(x-\Delta_f)/\lambda) dx \\[0em]
        &\textstyle= \frac{1}{2\lambda} \exp((\alpha-1)\Delta_f/\lambda) \cdot \lambda \\
        &\textstyle+ \frac{1}{2\lambda} \cdot \frac{\lambda}{2\alpha - 1} \Big(\exp\Big(\frac{(\alpha-1)\Delta_f}{\lambda}\Big) - \exp\Big(\frac{-\alpha \Delta_f}{\lambda}\Big)\Big) \\
        &\textstyle+ \frac{1}{2\lambda} \lambda \exp\Big(\frac{-\alpha \Delta_f}{\lambda}\Big)\\
        &\textstyle= \frac{\alpha}{2\alpha-1} \exp\Big(\frac{(\alpha-1)\Delta_f}{\lambda}\Big)
        + \frac{\alpha-1}{2\alpha-1} \exp\Big(\frac{-\alpha\Delta_f}{\lambda}\Big) \\
        &\textstyle= \frac{\alpha}{2\alpha-1} \exp\Big((\alpha-1) \cdot \varepsilon\Big)
        + \frac{\alpha-1}{2\alpha-1} \exp\Big(-\alpha \cdot \varepsilon\Big)
    \end{align*}

    We use this result to define a function:
    \begin{align*}
        f(\varepsilon) &:= \log \Big\{ \frac{\alpha}{2\alpha-1} \exp\Big((\alpha-1) \cdot \varepsilon\Big) \\
        &+ \frac{\alpha-1}{2\alpha-1} \exp\Big(-\alpha \cdot \varepsilon\Big) \Big\}
    \end{align*}

    We show, that \texttt{DPInitialScore} satisfies RDP.
    The output of \texttt{DPInitialScore}, which we call an observation, is a post-processed version $\sfrac{o_\text{sum}}{o_\text{ds}}$ of a pair of observations $(o_\text{ds}, o_\text{sum})$ consisting of the dataset size and the sum of clipped labels. As RDP is preserved by post-processing (cf. \Cref{thm:rdp_post_processing}) we investigate the pair of observations output by \texttt{DPInitialScore} prior to post-processing. Let $M_\text{ds}, M_\text{sum}$ be the two mechanisms computing the noisy dataset size and the noisy sum of clipped labels.
    \newcommand{\initialscore}{\texttt{DPInitialScore}}
    \begin{align*}
        &\textstyle D_\alpha(\initialscore(X)||\initialscore(X')) \\
        &\textstyle= \frac{1}{\alpha-1} \log \int_{(o_\text{ds}, o_\text{sum})} \frac{\Pr[\initialscore(X)=\mathbf{o}]^\alpha}{\Pr[\initialscore(X')=\mathbf{o}]^{\alpha-1}} d (o_\text{ds}, o_\text{sum}) \\
        \intertext{Applying the RDP sequential composition bound of  ~\Cref{thmrdpsequentialcomposition} we get}
        &\textstyle\leq \frac{1}{\alpha-1} \log \int_{o_\text{ds}} \frac{\Pr[M_\text{ds}(X)=o_\text{ds}]^\alpha}{\Pr[M_\text{ds}(X')=o_\text{ds}]^{\alpha-1}} d o_\text{ds} \\
        &\textstyle+ \frac{1}{\alpha-1} \log \int_{o_\text{sum}} \frac{\Pr[M_\text{sum}(X)=o_\text{sum}]^\alpha}{\Pr[M_\text{sum}(X')=o_\text{sum}]^{\alpha-1}} d o_\text{sum} \\
        &\textstyle= D_\alpha(M_\text{ds}(X)||M_\text{ds}(X')) + D_\alpha(M_\text{sum}(X)||M_\text{sum}(X')) 
        \intertext{$M_\text{ds}, M_\text{sum}$ are Laplace mechanisms with noise scales $\sfrac{\Delta_{g_\text{ds}}}{\varepsilon_\text{ds}}$, $\sfrac{\Delta_{g_\text{sum}}}{\varepsilon_\text{init}}$ and sensitivities $\Delta_{g_\text{ds}}$, $\Delta_{g_\text{sum}}$ so we can apply the previously obtained function $f(\varepsilon)$:}
        &\textstyle= \frac{1}{\alpha-1} \Big(f(\varepsilon_\text{ds}) + f(\varepsilon_\text{init})\Big)
    \end{align*}
\end{proof}

\subsubsection{DPLeaf}

We state the full \Cref{thm:rdp_nonspherical_gauss}.

\begin{repcorollary}{thm:rdp_nonspherical_gauss}[RDP of Gaussian Mechanism with leaf-balanced non-spherical noise]
    Let $f \colon \mathcal{X} \mapsto \mathbb{R}^D$ denote a function from an arbitrary input $X \in \mathcal{X}$ to a set of scalars. Let the function $f$ projected to its $d$-th output have bounded sensitivity $s_d \in \mathbb{R}$. Let $Y \sim \mathcal{N}(\mathbf{0},\Sigma)$ be a random variable of non-spherical multivariate Gaussian noise with covariance matrix $\Sigma \coloneqq D^{-1}\cdot\diag(r_1^{-1} s_{1}^{2} \sigma^2, \dots, r_D^{-1} s_{D}^{2} \sigma^2 )$ where a given variance $\sigma^2 \in \mathbb{R}^+$ is weighted in each dimension $d$ individually by $s_{d}^{2}$ and $r_d \in \mathbb{R}^+$ such that $\sum_{d=1}^D r_d = 1$. Then the Gaussian mechanism $M(X) \mapsto \set{f(X)_d + Y_d}_{d=1}^D$ is $(\alpha, \rho(\alpha))$-RDP with $\rho(\alpha) = \alpha\cdot \frac{D}{2\sigma^2}$.
\end{repcorollary}
\begin{proof}
    The corollary follows directly from ~\Cref{thm:individual_rdp_nonspherical_gauss} for a challenge data point with a worst-case sensitivity: $s_d$.
\end{proof}

We recall \Cref{cor:dpleaf_individual_rdp_main_body} and state the full proof.

\begin{repcorollary}{cor:dpleaf_individual_rdp_main_body}
    \input{theorems/dpleaf_irdp}
\end{repcorollary}

\begin{proof}
    For arbitrary input $X$, \texttt{DPLeaf} without noise computes a function $f(X) \mapsto \set{\sum_{x_i \in X} h_i, \sum_{x_i \in X} g_i}$ where $g_i, h_i$ are the gradient and Hessian of data point $x_i \in X$. The first output of $f$ is $|h_i|$-sensitivity bounded with respect to $X \sim_{x_i} X'$ (i.e., $X$ and $X'$ only  differ in $x_i$), the second output of $f$ is $|g_i|$-sensitivity bounded with respect to $X \sim_{x_i} X'$ , as $X$ and $X'$ only  differ in $x_i$:
    \begin{align*}
        &\textstyle\max_{X\sim X'} |\sum_{x_j \in X} g_j - \sum_{x_j \in X'} g_j| = |g_i|\\
        &\textstyle\max_{X\sim X'} |\sum_{x_j \in X} h_j - \sum_{x_j \in X'} h_j| = |h_i|\\
    \end{align*}
    In lines \ref{algdpleafnoisedsupport} and \ref{algdpleafnoisedsum}, \texttt{DPLeaf} applies leaf-balanced non-spherical bivariate Gaussian noise $Y \sim \mathcal{N}(\mathbf{0}, \Sigma)$ with covariance matrix $\Sigma := \frac{1}{2}\diag(r_1^{-1} (h^{*})^{2} \sigma_\text{leaf}^2, r_2^{-1} (g^{*})^{2} \sigma_\text{leaf}^2)$ to the output of $f$, denoted as $M(X) \mapsto \{f(X)_1 + Y_1, f(X)_2 + Y_2\}$.
    By ~\Cref{thm:individual_rdp_nonspherical_gauss}, $M$ satisfies $(\alpha, \rho(\alpha))$-individual RDP for $x_i$ with $\rho(\alpha) = \alpha \cdot \frac{2}{2\sigma_\text{leaf}^2} \cdot \Big(\frac{r_1 \cdot |h_i|^{2}}{(h^{*})^{2}} + \frac{r_2\cdot |g_i|^2}{(g^{*})^{2}}\Big)$.

    \texttt{DPLeaf} finally combines the outputs $f(X)_1, f(X)_2$ of $f$ in line \ref{algdpleafvalue} via post-processing (cf. ~\Cref{thm:rdp_post_processing}) without additional privacy leakage.
\end{proof}

We state the full \Cref{cor:dpleaf_rdp_main_body}.

\begin{repcorollary}{cor:dpleaf_rdp_main_body}
    \input{theorems/dpleaf_rdp}
\end{repcorollary}
\begin{proof}
    The corollary follows directly from \Cref{cor:dpleaf_individual_rdp_main_body} for a challenge data point with worst-case sensitivities $g^*, h^*$.
\end{proof}

\subsubsection{TrainSingleTree}

\begin{remark}
    Privacy amplification by subsampling (\Cref{thmsubsampling}) can be applied to individual RDP as well. The proof of ~\Cref{thmsubsampling} assumes a bound on the worst-case individual RDP of some mechanism $M$ and computes the worst-case individual RDP of the subsampled variant $M^\mathcal{P}$ of $M$. Now, if we only assume an individual RDP bound for $M$ for some data point $x_i$ we can analogously utilize ~\Cref{thmsubsampling} to obtain an individual RDP bound of subsampled $M^\mathcal{P}$ for $x_i$.
    \label{rem:subsampling}
\end{remark}

We recall \Cref{thm:train_single_tree_irdp_main_body} and state the full proof.

\begin{reptheorem}{thm:train_single_tree_irdp_main_body}
    \input{theorems/traintree_irdp}
\end{reptheorem}
\begin{proof}
    Assume that $M$ is the mechanism computing the output of ~\Cref{algtrainsingletree}, when no subsampling is applied. Then $M$ comprises a sequence of two mechanisms $M_R$ computing the random tree (line \ref{alg:trainSingleTreeRandomTree} in ~\Cref{algtrainsingletree}) and $M_L$ computing a leaf value (line \ref{alg:train_single_tree_dp_leaf} in ~\Cref{algtrainsingletree}) for each leaf. Let $M_A$ be the mechanism that outputs all $k$ of the leaves' values at once, i.e. $M_A(X) = (M_L(X), M_L(X), ..., M_L(X))$.

    Let $o:=o^{(R,1..k)}=(o_R, o_1, o_2, ..., o_k)$ be some observation of $M$, comprising the output random tree and $k$ values for $k$ leaves of the random tree. Define $o^{(1..k)}=(o_1,o_2,...,o_k)$,
    let $X,X'$ be two neighboring datasets.
    \begin{align*}
        &\textstyle\phantom{=}~D_\alpha(M(X)||M(X'))
        \textstyle= \frac{1}{\alpha-1}\log\int_{-\infty}^{\infty} \frac{\Pr[M(X)=o]^\alpha}{\Pr[M(X')=o]^{\alpha-1}} do \\
        &\textstyle= \frac{1}{\alpha-1}\log\int_{-\infty}^{\infty} \frac{\Pr[M_A(M_R(X), X)=o]^\alpha}{\Pr[M_A(M_R(X'), X')=o]^{\alpha-1}} do \\
        \intertext{We apply the individual RDP sequential composition bound of \Cref{thm:individual_rdp_sequential_composition} and get}
        &\textstyle\leq \frac{1}{\alpha-1}\log\int_{-\infty}^{\infty} \frac{\Pr[M_R(X)=o_R]^\alpha}{\Pr[M_R(X')=o_R]^{\alpha-1}} do \\
        &\textstyle+ \frac{1}{\alpha-1}\log\int_{-\infty}^{\infty} \frac{\Pr[M_A(o^{(1..k)}, X)=o^{(1..k)}]^\alpha}{\Pr[M_A(o^{(1..k)}, X')=o^{(1..k)}]^{\alpha-1}} do^{(1..k)} \\
        \intertext{Since the splits are data-independent we have}
        &\textstyle= 0 + \frac{1}{\alpha-1}\log\int_{-\infty}^{\infty} \frac{\Pr[M_A(o_R, X)=o]^\alpha}{\Pr[M_A(o_R, X')=o]^{\alpha-1}} do \\
        \intertext{As the splits are chosen uniformly at random and due to the law of total probability, it suffices to consider an arbitrary but fixed splitting function $s$ with $s(X) = (X_i)_{i=1}^k$ that partitions the data into $k$ partitions for $k$ leaves.}
        &\textstyle= \frac{1}{\alpha-1}\log
        \textstyle\int_{-\infty}^{\infty} \frac{\Pr[(M_L(o_R, X_1), ..., M_L(o_R, X_k))=o]^\alpha}{\Pr[(M_L(o_R, X'_1), ..., M_L(o_R, X'_k))=o]^{\alpha-1}} do \\
        &\textstyle= \frac{1}{\alpha-1}\log \int_{-\infty}^{\infty} \prod_{i=1}^k \frac{\Pr[M_L(o_R, X_i)]^\alpha}{\Pr[M_L(o_R, X'_i)]^{\alpha-1}} do \\
        \intertext{Each splitting function $s$ data-independently partitions the dataset into $k$ distinct subsets; hence, for neighboring datasets $D\sim_{x_i}D'$ in unbounded DP we get $s(D) = (X_j)_{j=1}^k$ and $s(D') = (X'_j)_{j=1}^k$ such that for one $i$ $X_i$ and $X'_i$ differs in the one element $x_i$ (arbitrary but fixed) and for all $j\neq i$ $X_j = X'_j$.}
        &\textstyle= \frac{1}{\alpha-1}\log
        \textstyle\int_{-\infty}^{\infty} \frac{\Pr[(M_L(o_R, X_1)]^\alpha}{\Pr[(M_L(o_R, X'_1)]^{\alpha-1}} \cdot ...\cdot \frac{\Pr[(M_L(o_R, X_i)]^\alpha}{\Pr[(M_L(o_R, X'_i)]^{\alpha-1}} \\
        &\textstyle\cdot ... \cdot \frac{\Pr[M_L(o_R, X_k))=o]^\alpha}{\Pr[M_L(o_R, X'_k))=o]^{\alpha-1}} do \\
        &\textstyle= \frac{1}{\alpha-1}\log \int_{-\infty}^{\infty} \frac{\Pr[(M_L(o_R, X_i)]^\alpha }{\Pr[(M_L(o_R, X'_i)]^{\alpha-1}} do \\
        \intertext{We apply the individual RDP bound of ~\Cref{cor:dpleaf_individual_rdp_main_body}:}
        &\textstyle\leq \alpha \cdot \frac{2}{2\sigma_\text{leaf}^2} \cdot \Big( \frac{r_1 \cdot |h_i|^{2}}{(h^{*})^{2}} + \frac{r_2\cdot |g_i|^2}{(g^{*})^{2}}\Big) 
    \end{align*}

    Finally, we show that \texttt{TrainSingleTree} with subsampling applied, yields a privacy amplification.
    As stated in ~\Cref{rem:subsampling}, we can utilize the privacy amplification by subsampling (\Cref{thmsubsampling}) for individual RDP. The individual RDP bound of \texttt{TrainSingleTree} is 
    $\alpha \cdot \frac{2}{2\sigma_\text{leaf}^2} \cdot \Big( \frac{r_1 \cdot |h_i|^{2}}{(h^{*})^{2}} + \frac{r_2\cdot |g_i|^2}{(g^{*})^{2}}\Big)$ and is linear in $\alpha$. We use our condition for tight subsampling bounds (\Cref{cor:pearson_vajda_condition_general}) to follow that the Pearson-Vajda pseudo-divergence of $D_{\mathcal{X}^{\alpha}}(M(X)||M(X')) \geq 0$ for all odd $\alpha \geq 1$. This allows us to use the tight subsampling bound of \Cref{thmsubsampling} to obtain $a_{\gamma}(\alpha, \alpha \cdot \frac{2}{2\sigma_\text{leaf}^2} \cdot \Big( \frac{r_1 \cdot |h_i|^{2}}{(h^{*})^{2}} + \frac{r_2\cdot |g_i|^2}{(g^{*})^{2}}\Big))$ as individual RDP bound.
\end{proof}

We state the full \Cref{cor:train_single_tree_rdp_main_body}.

\begin{repcorollary}{cor:train_single_tree_rdp_main_body}
    \input{theorems/traintree_rdp}
\end{repcorollary}

\begin{proof}
    This corollary follows directly from ~\Cref{thm:train_single_tree_irdp_main_body} with the worst-case sensitivities $g^*$ and $h^*$.
\end{proof}

%% file: theorems/dpleaf_rdp.tex
\Cref{algdpleaf} (\texttt{DPLeaf}) satisfies $(\alpha, \sfrac{\alpha}{\sigma_\text{leaf}^2})$-RDP with $\sigma_\text{leaf}^2$ as the unweighted variance of the leaf Gaussian.

%% file: theorems/traintree_rdp.tex
Let $\sigma_\text{leaf}^2$ be the unweighted variance of the leaf Gaussian. Let $a_\gamma: \mathbb{N} \times \mathbb{R} \mapsto \mathbb{R}$ denote the privacy amplification of \Cref{thmsubsampling} with subsampling ratio $\gamma$. Then \texttt{TrainSingleTree} (cf. \Cref{algtrainsingletree}) is  $(\alpha, a_{\gamma}(\alpha, \sfrac{\alpha}{\sigma_\text{leaf}^2}))$-RDP.

%% file: appendix_distributed_learning.tex
\label{sec:appendix_distributed_learning}

\begin{algorithm}[t!]
\small
    \newcommand{\single}[1]{\texttt{TrainSingleTree}}
    \caption{Distributed-TrainSBDT}
    \SetKwFunction{Initialization}{Initialization}
    \SetKwFunction{TrainSingleTree}{TrainSingleTree}
    \SetKwFunction{CompleteTree}{CompleteTree}
    \SetKwFunction{PrivacyAmplificationBySubsampling}{PrivacyAmplificationBySubsampling}
    \SetKwFunction{ComputeInitScore}{ComputeInitScore}
    \SetKwFunction{ComputeGradients}{ComputeGradients}
    \SetKwFunction{DPInitScore}{DPInitScore}
    \SetKwFunction{UpdateIndividualPrivacyLoss}{UpdateIndividualPrivacyLoss}
    \SetKwFunction{Append}{Append}
    \SetKwFunction{ZeroClassifier}{ZeroClassifier}
    \SetKwFunction{BulletinBoard}{BulletinBoard}
    \SetKwFunction{PoissonSubsample}{PoissonSubsample}
    \SetKwFunction{SecureAggregation}{SecureAggregation}
    \SetKwFunction{PublicUniformSampling}{PublicUniformSampling}
    \KwIn{$D_u$ : private training dataset of user $u$} 
    \SetKwInOut{KwIn}{\color{white}\phantom{Input}}
    \KwIn{$k$ : number of users} 
    \KwIn{$(r_1, r_2)$ : noise weights for leaf value}
    \KwIn{$g^*, h^*, m^*$ : clipping bounds on gradients, Hessians and labels}
    \KwIn{$\lambda, \beta$ : regularization parameters}
    \KwIn{$(T_\text{regular}, T_\text{extra})$ : number of rounds and extra rounds}
    \KwIn{$d$ : depth of trees}
    \KwIn{$\gamma$ : subsampling ratio}
    \KwIn{$m$ : number of features of $D_u$}
    \KwIn{$(v_{\min}^{(1)}, v_{\max}^{(1)}, ..., v_{\min}^{(m)}, v_{\max}^{(m)})$ : feature value ranges}

    \BlankLine
    $((r_1, r_2), g^*, \lambda, T_{\text{regular}}, T_{\text{extra}}, d, \gamma) =$ \BulletinBoard{'hyperparameters'} \label{alg:distributed_learning_hyperparameters} \\
    $T_\text{max} = T_\text{regular} + T_\text{extra}$ \\
    $\hat{\alpha}, \sigma_\text{leaf}^2, \rho(\hat{\alpha}) =$ \Initialize{$\alpha_{\max}$, $(\varepsilon_\text{trees}, \delta_\text{trees}), \varepsilon_\text{init}, \gamma$} \label{alg:distributed_learning_setup} \\ 
    $(\text{ds}_u, \text{sum}_u) = \text{\DPInitScore{$D, m^*, \varepsilon_\text{init}$}}$ \label{alg:distributed_learning_client_init} \\
    $\text{ds} =$ \SecureAggregation{$\text{ds}_1, ..., \text{ds}_u,...,\text{ds}_k$} \\
    $\text{sum} =$ \SecureAggregation{$\text{sum}_1, ..., \text{sum}_u,...,\text{sum}_k$} \\
    $\text{init}_0 = \text{sum} / \text{ds}$ \\
    $E = (\text{init}_0)$ \\
    \For{$t = 1$ to $T_\text{max}$}{
        \For{$i=1$ to $|D_u|$}{ \label{alg:distributed_learning_accounting_start}
            $\rho_t^{(i)}(\alpha) = a_\gamma(\alpha, \frac{\alpha}{\sigma_\text{leaf}^{2}} \cdot \Big( \frac{r_1\cdot |h_i|^2}{(h^{*})^{2}} + \frac{r_2\cdot |g_i|^2}{(g^{*})^{2}}\Big))$ \tcp*{by~\Cref{thm:train_single_tree_irdp_main_body} \blue{(for RDP: \Cref{cor:train_single_tree_rdp_main_body})}} \label{alg:distributed_learning_client_individual}
        }
         $D_{u,t} = (x_i : \mathcal{F}_{\hat{\alpha}, \rho(\hat{\alpha})}(\rho_{1}^{(u,i)}, \rho_{2}^{(u,i)}, \dots, \rho_{t}^{(u,i)}) = \text{CONT})$ \tcp{by ~\Cref{thmprivacyfilter}}\label{alg:distributed_learning_client_filter} \label{alg:distributed_learning_accounting_end}
         \BlankLine 
         $\text{tree}_t = $ \CompleteTree{} \label{alg:distributed_learning_client_receive_tree_start} \\
         \For{each split $(i, v)$ in $\text{tree}_t$}{
            $i = \lceil \text{\PublicUniformSampling{$[0,m)$}} \rceil$ \label{alg:distributed_learning_publicuniformsampling1} \\ 
            $v = \text{\PublicUniformSampling{$[0, v_{\max}^{(i)} - v_{\min}^{(i)})$}} - v_{\min}^{(f)}$ \label{alg:distributed_learning_client_receive_tree_end} \label{alg:distributed_learning_publicuniformsampling2}
         }
         \BlankLine
         $\text{tree}_t =$ \TrainSingleTree{$\text{tree}_t, D_t, d,$
     $\sigma_\text{leaf}^2, g^*, h^*, (r_{1}, r_{2}), \lambda, \beta, E$} \label{alg:distributed_learning_client_train_single} \\ 
         $V_u =$ list() \\
         $W_u =$ list() \\
         \For{$l=1$ to $2^{d}$}{
            let $v_l, w_l$ be gradient sum / Hessian sum of leaf $l$ in $\text{tree}_t$ \\
            \Append{$V_u, v_l$} \label{alg:distributed_learning_client_V_u} \\
            \Append{$W_u, w_l$} \label{alg:distributed_learning_client_W_u}
         }
         $A =$ \SecureAggregation{$W_1, \dots, \mathbf{W_u},$
         $ \dots, W_{k}$} \label{alg:distributed_learning_client_secure_aggregation} \\
         $B =$ \SecureAggregation{$V_1, \dots, \mathbf{V_u},$
         $ \dots, V_{k}$} \label{alg:distributed_learning_client_secure_aggregation_2} 
         \For{$l=1$ to $2^{d}$}{
            \SetLeaf{$\text{tree}_t, l, A[l] / B[l]$} \label{alg:distributed_learning_client_set_leaf}
         }
         $E = (\text{init}_0, \text{tree}_1, \text{tree}_2, \dots, \text{tree}_t)$
     }
     \BlankLine
     \Return E
    \label{alg:distributed_learning_client}
\end{algorithm}

Our distributed learning extension for \dpgbdt{} builds on prior work \cite{Maddock_2022}. We train a global ensemble $E$ that is shared amongst $k$ users with distinct private training datasets $D_1, D_2, ..., D_k$. We assume the existence of a secure bulletin board that provides each user with the same set of hyperparameters. We utilize secure aggregation \cite{secure_aggregation_bonawitz,secure_aggregation_bell}, a protocol for privately computing the sum of vectors 
$A = \sum_{u=1}^{k} W_u$. The full protocol of distributed \dpgbdt{} is described in ~\Cref{alg:distributed_learning_client}.

Every user $u$ ($u=1,2,...,k)$ receives the hyperparameters from the bulletin board (line \ref{alg:distributed_learning_hyperparameters}) and sets up the accounting (line \ref{alg:distributed_learning_setup}).

For the initial classifier (from line \ref{alg:distributed_learning_client_init}), the user computes DP releases of the sum of labels and the dataset size. This is a slight variation of  \Cref{algdpinitscore} (\texttt{DPInitialScore}) where usually the sum would already be divided by the dataset size. All users synchronize and invoke \texttt{SecureAggregation} \cite{secure_aggregation_bonawitz,secure_aggregation_bell} with fixed precision, for aggregating the dataset size first and then the label sum. The initial score is then built by dividing the aggregated label sum by the overall dataset size and added to the ensemble.

The user commences $T_{\max}$ rounds of training, and starts a single training round by updating the individual Rényi DP privacy losses for all its data points (line \ref{alg:distributed_learning_client_individual}) and then filtering out those data points that have expended all their privacy budget (line \ref{alg:distributed_learning_client_filter}). Next, the user locally initializes $\text{tree}_t$ for the current round $t$ (line \ref{alg:distributed_learning_client_receive_tree_start}) with arbitrary or even undefined splits, and then synchronizes with the other users and utilizes public uniform sampling \cite[Protocol 1]{sabater_verifiable_sampling} to collaboratively and verifiably sample uniformly random features and feature values for the splits of $\text{tree}_t$: For the random feature (line \ref{alg:distributed_learning_publicuniformsampling1}) chosen from $m$ features we draw a randomly uniform sample $i'$ from $[0, m)$ and then select the feature $\lceil i' \rceil$. For the feature value (line \ref{alg:distributed_learning_publicuniformsampling2}) of a numerical feature $i$ with feature range $[v_{\min}^{(i)}, v_{\max}^{(i)})$ we draw a randomly uniform sample $v'$ from $[0, v_{\max}^{(i)} - v_{\min}^{(i)})$ and then select the feature value $v' - v_{\min}^{(i)}$. If the feature $i$ is categorical we assume $v_{\max}^{(i)} = c$ be the number of different values of feature $i$ and $v_{\min}^{(i)}=0$. We then draw a randomly uniform sample $v'$ from $[0, c)$ and select the categorical feature value with index $\lceil v' \rceil$.

The user then locally only adjusts the leaves of $\text{tree}_t$ with its local share of sensitive data $D_u$ (line \ref{alg:distributed_learning_client_train_single}). Note that this is a slight variation of ~\Cref{algtrainsingletree} (\texttt{TrainSingleTree}): When already given a random tree, ~\Cref{algtrainsingletree} must leave the splits untouched and only train the leaves of the given tree. 

All users must now synchronize again to utilize \\\texttt{SecureAggregation} ~\cite{secure_aggregation_bonawitz}~\cite{secure_aggregation_bell} with fixed precision for collaboratively generating leaf values for the tree. User $u$ generates vectors $V_u, W_u$ to contain the gradient sum and Hessian sum of all its local leaf values (line \ref{alg:distributed_learning_client_V_u} and \ref{alg:distributed_learning_client_W_u}). The user then calls \texttt{SecureAggregation} twice (line \ref{alg:distributed_learning_client_secure_aggregation} and \ref{alg:distributed_learning_client_secure_aggregation_2}) to privately compute
\[\textstyle A = \sum_{u=1}^{k} V_u, \, B = \sum_{u=1}^{k} W_u\]
containing the collaboratively generated gradient sum and Hessian sum. User $u$ then sets the leaf values of its local tree  (line \ref{alg:distributed_learning_client_set_leaf}) and finally adds this tree to the ensemble.